\documentclass[letterpaper,11pt]{article}
\usepackage[margin=1in]{geometry}
\usepackage{amsmath,amsthm,amssymb,amsfonts}
\usepackage[dvipdfmx]{graphicx}
\usepackage{color,url,xspace,ascmac}
\usepackage[normalem]{ulem}
\usepackage{comment,tikz}
\usepackage{thmtools,thm-restate}
\usepackage{enumitem}
\usepackage{physics}
\usepackage{mathtools}
\usepackage{makecell}
\usepackage{multirow}
\newcommand{\fnote}[1]{{\color{teal} F: #1}}
\newcommand{\TRnote}[1]{{\color{red} [TR] #1}}

\usetikzlibrary{positioning}
\usetikzlibrary{calc}

\tikzset{every picture/.style={font issue=\footnotesize},font issue/.style={execute at begin picture={#1\selectfont}}}

\usepackage[unicode,colorlinks=true,citecolor=blue,linkcolor=blue,pdfusetitle,hypertexnames=false]{hyperref}

\usepackage[style=alphabetic,natbib=true,natbib=true,maxnames=99,maxalphanames=99,isbn=false,url=false,backref=true,backend=biber,giveninits=true]{biblatex}

\AtEveryBibitem{%
\ifentrytype{article}{%
  \clearlist{publisher}
  \clearfield{editor}
}{}
\ifentrytype{book}{
}{}
\ifentrytype{inproceedings}{%
  \clearfield{volume}
  \clearfield{number}
  \clearfield{editor}
  \clearlist{publisher}
  \clearlist{location}
}{}
  \clearfield{day}
  \clearfield{month}
  \clearfield{isbn}
  \clearfield{url}
}

\usepackage{doi}
\usepackage[capitalise]{cleveref}
\newtheorem{theorem}{Theorem}[section]
\newtheorem{definition}[theorem]{Definition}
\newtheorem{lemma}[theorem]{Lemma}
\newtheorem{corollary}[theorem]{Corollary}

\newtheorem{claim}[theorem]{Claim}

\newtheorem{remark}[theorem]{Remark}

\newif\ifdraft
\draftfalse

\def\defeq{\mathrel{\mathop:}=}

\DeclarePairedDelimiter{\rbra}{\lparen}{\rparen} 
\DeclarePairedDelimiter{\cbra}{\lbrace}{\rbrace} 
\DeclarePairedDelimiter{\sbra}{\lbrack}{\rbrack} 

\newcommand{\indicator}{\mathbf{1}}

\DeclareMathOperator*{\E}{\mathbb{E}}
\DeclareMathOperator*{\Var}{\mathbf{Var}}
\DeclareMathOperator*{\Cov}{\mathbf{Cov}}
\newcommand{\condition}{\;\middle\vert\;}

\newcommand{\e}{\mathrm{e}}
\newcommand{\Nat}{\mathbb{N}}
\newcommand{\Real}{\mathbb{R}}

\newcommand{\USD}{Undecided State Dynamics\xspace}
\newcommand{\pbot}{p_{\bot}}

\renewcommand{\epsilon}{\varepsilon}

\newcommand{\bounded}{D}
\newcommand{\variance}{s}
\newcommand{\opn}{\mathrm{opn}}
\newcommand{\drift}{R}
\newcommand{\calF}{\mathcal{F}}
\newcommand{\tauxplus}{\tau_X^+}
\newcommand{\tauxminus}{\tau_X^-}

\newcommand{\taubetaplus}{\tau_{\beta}^+}
\newcommand{\taubetaminus}{\tau_{\beta}^-}
\newcommand{\taupsiplus}{\tau_{\psi}^{+}}

\newcommand{\taumaxdown}{\tau_{\max}^{\downarrow}}
\newcommand{\taumaxup}{\tau_{\max}^{\uparrow}}
\newcommand{\cmaxdown}{c_{\max}^{\downarrow}}
\newcommand{\cmaxup}{c_{\max}^{\uparrow}}

\newcommand{\taudeltaup}{\tau_{\delta}^{\uparrow}}
\newcommand{\taudeltadown}{\tau_{\delta}^{\downarrow}}
\newcommand{\cdeltaup}{c_\delta^\uparrow}
\newcommand{\cdeltadown}{c_\delta^\downarrow}
\newcommand{\tauiweak}{\tau_{i}^{\mathrm{weak}}}
\newcommand{\taujweak}{\tau_{j}^{\mathrm{weak}}}
\newcommand{\cweak}{c_{\mathrm{weak}}}
\newcommand{\taugammaplus}{\tau_{\gamma}^{+}}
\newcommand{\xgamma}{x_\gamma}
\newcommand{\taudeltaplus}{\tau_{\delta}^{+}}
\newcommand{\xdelta}{x_\delta}
\newcommand{\taumaxplus}{\tau_{\max}^{+}}

\newcommand{\taujstrong}{\tau_j^{\mathrm{strong}}}
\newcommand{\cstrong}{c_{\mathrm{strong}}}
\newcommand{\taustrong}{\tau^{\mathrm{strong}}}
\newcommand{\taupsiminus}{\tau^{-}_{\psi}}

\newcommand{\tautildemaxup}{\tau^\uparrow_{\widetilde{\max}}}
\newcommand{\tautildemaxdown}{\tau^\downarrow_{\widetilde{\max}}}
\newcommand{\ctildemaxup}{c^\uparrow_{\widetilde{\max}}}
\newcommand{\ctildemaxdown}{c^\downarrow_{\widetilde{\max}}} 
\newcommand{\cdeltaplus}{c_\delta^+}
\newcommand{\ceta}{c_{\eta}}

\newcommand{\constref}[1]{C_{\mbox{\tiny\ref{#1}}}}
\newcommand{\constrefs}[2]{C_{\mbox{\tiny\ref{#1}(\ref{#2})}}}

\newcommand{\np}{\alpha}
\newcommand{\npo}{\beta}
\newcommand{\npt}{\gamma}

\newcommand{\npm}{\alpha^{\mathrm{max}}}
\newcommand{\tnp}{\widetilde{\alpha}}
\newcommand{\tnpt}{\widetilde{\gamma}}
\newcommand{\tnpm}{\widetilde{\alpha}^{\mathrm{max}}}
\newcommand{\xbeta}{x_\beta}
\newcommand{\xpsi}{x_\psi}

\newcommand{\taupsi}{\tau_{\psi}^+}

\newcommand{\xeta}{x_\eta}
\newcommand{\tauetaplus}{\tau^{+}_{\eta}}
\newcommand{\tauetaminus}{\tau^{-}_{\eta}}
\newcommand{\xetaplus}{x^{+}_{\eta}}
\newcommand{\xetaminus}{x^{-}_{\eta}}

\newcommand{\tauweak}{\tau^{\mathrm{weak}}}

\newcommand{\taucons}{\tau_{\mathrm{cons}}}

\newcommand{\tauphiplus}{\tau^{+}_{\varphi}}
\newcommand{\tauphiup}{\tau^{\uparrow}_{\varphi}}
\newcommand{\cphiup}{c_{\varphi}^{\uparrow}}

\newcommand{\Otilde}{\widetilde{O}}

\newcommand{\calE}{\mathcal{E}}

\crefname{figure}{Figure}{Figures}
\crefname{remark}{Remark}{Remarks}
\crefname{equation}{}{}

\makeatother

\title{Undecided State Dynamics with Many Opinions}
\author{
Colin Cooper\\
\small{King’s College London}\\
\small{\texttt{\href{mailto:colin.cooper@kcl.ac.uk}{colin.cooper@kcl.ac.uk}}}
\and
Frederik Mallmann-Trenn\\
\small{King’s College London}\\
\small{\texttt{\href{mailto:frederik.mallmann-trenn@kcl.ac.uk}{frederik.mallmann-trenn@kcl.ac.uk}}}
\and
Tomasz Radzik\\
\small{King’s College London}\\
\small{\texttt{\href{tomasz.radzik@kcl.ac.uk}{tomasz.radzik@kcl.ac.uk}}}
\and
Nobutaka Shimizu\\
\small{Institute of Science Tokyo}\\
\small{\texttt{\href{shimizu.n.ah@m.titech.ac.jp}{shimizu.n.ah@m.titech.ac.jp}}}
\and
Takeharu Shiraga\\
\small{Chuo University}\\
\small{\texttt{\href{shiraga.076@g.chuo-u.ac.jp}{shiraga.076@g.chuo-u.ac.jp}}}
}

\date{}
\begin{document}
\maketitle

\begin{abstract}
We study the Undecided-State Dynamics (USD), a fundamental consensus process in which each vertex holds one of $k$ decided opinions or the undecided state. We consider both the gossip model and the population protocol model.
Prior work established tight bounds 
on the consensus time of this process only for the regime $k = O(\sqrt{n}/(\log n)^2)$ (for the population protocol model) and $k = O((n/\log n)^{1/3})$ (for the gossip model), often under restrictive assumptions on the initial configuration.

In this paper, we obtain the first consensus-time guarantees for USD that hold for \emph{arbitrary} $2\le k\le n$ and for \emph{arbitrary} initial configurations in both the gossip model and the population protocol model.  In the gossip model, USD reaches consensus within $\widetilde O(\min\{k,\sqrt n\})$ synchronous rounds with probability $1-\pbot-n^{-c}$, where $\pbot$ is the gossip-specific probability of collapsing to the all-undecided state in the first round.  In the population protocol model, USD reaches consensus within $\widetilde O(\min\{kn,n^{3/2}\})$ asynchronous interactions with high probability.
We also present lower bounds that match the upper bounds up to polylogarithmic factors for a specific initial configuration and show that our upper bounds are essentially optimal. 
\end{abstract}

\newpage
\tableofcontents
\newpage
\section{Introduction} \label{sec:introduction}
The Undecided-State Dynamics (USD) \cite{undecided_Angluin,Becchetti_SODA2015}
is one of the simplest and most studied consensus dynamics on complete graphs.
Each vertex holds an opinion in a finite set $\Sigma=[k]\cup\{\bot\}$,
where $[k]:=\{1,\dots,k\}$ are decided opinions and $\bot$ is a distinguished
undecided opinion.
When a vertex $u$ interacts with a vertex $v$, 
the update follows a simple cancellation/adoption rule (with only $u$ updating its opinion):
if $u$ and $v$ hold different decided opinions in $[k]$, then $u$ becomes $\bot$; 
if $u$ is undecided and $v$ is decided, then $u$ adopts $v$’s opinion; 
otherwise $u$ keeps its opinion.
This elementary rule (formalized in \cref{sec:model and basic notation})
already gives rise to rich and nontrivial behavior.

We analyze USD with $k$ possible decided opinions on an $n$-vertex complete
graph with self-loops\footnote{Thus, choosing a random neighbor is equivalent to choosing a vertex uniformly at random.} under two classical communication models.
In the \emph{gossip model} \cite{Karp_SSV00}, time proceeds in synchronous rounds
in which each vertex samples one random neighbor and simultaneously updates its opinion accordingly.
In the \emph{population protocol model} \cite{Angluin_ADFP06}, interactions occur
asynchronously: an ordered pair $(u,v)$ is chosen uniformly at random with replacement, and
$u$ updates its opinion based on $v$'s opinion.

A central quantity of interest is the \emph{consensus time}: the number of
rounds (in  the gossip model) or pairwise interactions (in the population protocol model)
until all vertices hold the same decided opinion in
$[k]$.  The all-$\bot$ configuration, although
absorbing, does not constitute consensus and is regarded as failure. Consensus
time therefore measures the time to reach agreement on a real opinion.
In what follows, we review previous consensus time bounds on USD. For results on other relevant consensus dynamics, see \cref{sec:other related works}.

USD was first introduced by \citet{undecided_Angluin} for $k=2$ in the
population protocol model, where they showed that the consensus time is
$O(n\log n)$ with high probability.\footnote{%
A given event holds ``with high probability,'' if it holds with probability $1-O(n^{-c})$ for some constant $c>0$.}
They (and \citet{Condon_CRN_2017}) further proved that USD with $k=2$ reaches consensus with the
initial majority with high probability if the more popular opinion has
a sufficiently large initial advantage compared
to the other opinions. 
Thus USD offered a critical performance improvement over the classic pull voting, both in terms of the consensus time (expected $\Theta(n^2)$ in pull voting \cite{HP01,linear_voting}) and the chance 
for the majority opinion to win.
For larger $k$, \citet{fast_convergence_undecided} analyzed USD
for $2 \le k = O(\sqrt{n}/(\log n)^2)$ and proved that the consensus time is
$O(kn\log n)$ with high probability for any initial configuration.
Recently, \citet{El-Hayek_ES25} obtained a matching lower bound of
$\Omega(kn\log n)$ for all $k \le n^{1/2-\varepsilon}$, where $\varepsilon$ is an arbitrary 
positive constant.

In the gossip model, USD was first analyzed by \citet{Becchetti_SODA2015}, who
proved an $O(k\log n)$ upper bound for $2 \le k = O((n/\log n)^{1/3})$ 
under a strong
assumption on the initial advantage of the most popular opinion.  This
assumption was later removed for $k=2$ by \citet{undecided_MFCS}, who showed an
$O(\log n)$ bound for arbitrary initial configurations.
However, no general bound was known for any $k>2$ without any assumptions
on the initial configuration. 

It is natural to conjecture that the consensus time in the gossip model should be
$\Otilde(k)$\footnote{Throughout the paper, $\Otilde(\cdot)$ hides polylogarithmic factors in $n$.} for $2 \le k \ll \sqrt{n}$,\footnote{%
``$k \ll f(n)$'' abbreviates that $k=O(f(n)/\log^c n)$ for some constant $c >0$.
} by comparison with the
$\Otilde(kn)$ upper bound of \cite{fast_convergence_undecided} for the
population protocol model: a single synchronous round of the gossip model is
roughly equivalent to $n$ asynchronous interactions in the population protocol
model.  However, the analysis of \cite{fast_convergence_undecided} critically
exploits the asynchronous nature of population protocols and does not transfer
to the gossip model.  Indeed, they explicitly identify the gossip-model bound as
an open problem \cite[Section~9]{fast_convergence_undecided}.

Moreover, prior work provides a rather clear picture only in the regime
$k \ll \sqrt{n}$:  In the population protocol model, the best-known upper
and lower bounds are $\widetilde{\Theta}(kn)$ in this range
\cite{fast_convergence_undecided,El-Hayek_ES25}, strongly suggesting that
the consensus time scales linearly with $k$.
In the gossip model, the only existing upper bounds for USD with $k>2$
are also of order $\widetilde{O}(k)$, but they hold only under strong
assumptions on the initial configuration and only for $k\ll n^{1/3}$
\cite{Becchetti_SODA2015}.
Beyond this threshold, however, the efficiency of USD is mysterious:
the analytical techniques used before do not work for larger values of $k$
and no general consensus-time bounds are available for any $k\ge\sqrt{n}$ in 
the population protocol model or any $k\ge {n}^{1/3}$ in the gossip model.
This raises a natural and central question:  
does the linear-in-$k$ behavior of the consensus time persist for all $k$,
or does USD transition to a qualitatively different---possibly sublinear---regime
once $k$ exceeds $\sqrt{n}$?
Addressing this question requires techniques that go well beyond those used
in the $k<\sqrt{n}$ regime, and lies at the core of our contribution.

\subsection{Our Contributions}

We obtain the first general bounds on the consensus time of the
Undecided-State Dynamics (USD) for an \emph{arbitrary} 
number $k$ of decided opinions and an \emph{arbitrary} initial configuration in both the gossip model and the
population protocol model.  Our main results are summarized in the following
informal theorem. See also \cref{table:comparison}.

\begin{theorem}[Main Theorem] \label{thm:main}
    Consider USD on $n$ vertices starting from an arbitrary initial configuration\footnote{%
    A configuration is an assignment of opinions to vertices.}
    in $\{1,\dots,k,\bot\}^n$ 
    other than the all-$\bot$ configuration, where the number of 
    possible decided opinions $k$ is an arbitrary integer between $2$ and $n$ (inclusive).
    \begin{itemize}
        \item \textbf{Gossip model.}
        In the gossip model, USD reaches consensus within
        $\Otilde\rbra*{\min\{k,\sqrt{n}\}}$ synchronous rounds with probability
        $1 - \pbot - n^{-c}$ for some constant $c>0$, where $\pbot$ is the probability that all vertices hold $\bot$ after the first synchronous round.
        \item \textbf{Population protocol model.}
        In the population protocol model, USD reaches consensus within
        $\Otilde\rbra*{\min\{kn,n^{1.5}\}}$ asynchronous interactions 
        with high probability.
    \end{itemize}
\end{theorem}
\noindent
Our result for the gossip model says that, unless the process collapses to the
        all-$\bot$ configuration in the first round, it reaches consensus within
        $\Otilde\rbra*{\min\{k,\sqrt{n}\}}$ rounds with high probability.

Regarding the hidden $\mathrm{polylog}(n)$ factor in the consensus time, we have the following bounds:
In the gossip model, the consensus time is, with high probability, $O(\sqrt{n}(\log n)^3)$ for arbitrary $2\le k \le n$, and $O(k\log n)$ if $k = O\qty(\frac{\sqrt{n}}{(\log n)^2})$ (the similar bound holds for the population protocol model).
Our results hold for every $2 \le k \le n$ and for arbitrary initial
configurations, including those with undecided vertices.

\begin{table}[htbp]
\centering
\renewcommand{\arraystretch}{1.3}
\begin{tabular}{|c|c|c|c|c|}
\hline
Model & Work & Consensus Time & Range of $k$ & Initial Gap Assumption\\
\hline
\multirow{3}{*}{PP}
& \makecell[c]{\cite{undecided_Angluin}\\\cite{Condon_CRN_2017}}
& $O(n\log n)$
& $k=2$
& $A_1-A_2=\Omega(\sqrt{n\log n})$ \\
\cline{2-5}
& \cite{fast_convergence_undecided}
& $O(kn\log n)$
& $k = O(\sqrt{n}/(\log n)^2)$
& --- \\
\cline{2-5}
& \cref{thm:main}
& $O(n^{1.5}(\log n)^3)$
& $2\le k\le n$
& --- \\
\hline
\multirow{4}{*}{Gossip}
& \cite{Becchetti_SODA2015}
& $O(k\log n)$
& $k=O((n/\log n)^{1/3})$
& $A_1\ge(1+\Omega(1))\cdot A_2$ \\
\cline{2-5}
& \cite{undecided_MFCS}
& $O(\log n)$
& $k=2$
& --- \\
\cline{2-5}
& \cref{thm:main}
& $O(k\log n)$
& $k = O(\sqrt{n}/(\log n)^2)$
& --- \\
\cline{2-5}
& \cref{thm:main}
& $O(\sqrt{n}(\log n)^3)$
& $2\le k\le n$
& --- \\
\hline
\end{tabular}
\caption{
Comparison of consensus time bounds for USD in the Population Protocol (PP) and Gossip models.
Here $A_i$ denotes the number of vertices holding opinion $i$ in the initial configuration, with
$A_1\ge A_2\ge\cdots\ge A_k$.
The assumptions on the initial gap are stated only when required; entries marked by ``---''
indicate that no explicit gap assumption is imposed.
}
\label{table:comparison}
\end{table}

\begin{remark} \label{rem:all-bot}
The probability $\pbot$ of collapsing to the all-$\bot$ configuration in the
first synchronous round is a phenomenon specific to the gossip model.
Because all vertices update simultaneously, a single round can eliminate all
decided opinions.  In contrast, in the population protocol model only one
vertex updates its opinion at each interaction, and an undecided vertex cannot
be created unless two distinct decided opinions are present.
\end{remark}
 
The value of $\pbot$ strongly depends on the initial distribution of opinions.
For example, if $k = O(n/\log n)$, or if the initial configuration contains at
least one but not all undecided vertices, then $\pbot = n^{-\Omega(1)}$
(see \cref{lem:fail at first round}).
At the opposite extreme, if $k=n$ and each vertex initially holds a distinct
decided opinion, 
then it becomes undecided in the first round with probability 
$1 - 1/n$, independently of the other vertices,
and thus $\pbot = (1 - 1/n)^n \approx 1/e$.

\paragraph{Matching lower bounds.}
We complement our upper bounds for the gossip model by proving a matching lower bound up to
a polylogarithmic factor in $n$.
Such a lower bound in the population protocol model is already given by the result of \cite{El-Hayek_ES25} in the regime of $k \le n^{1/2-\varepsilon}$ for any constant $\varepsilon>0$.

\begin{theorem}[matching lower bound for the gossip model] \label{thm:main_tight}
    For any sufficiently large $n$, any constant $0<c<1/2$ and any $2\le k\le (1/2-c)n$, there exists an initial configuration over exactly $k$ decided opinions (i.e., each of the $k$ opinions is supported by at least one vertex) such that USD in the gossip model requires
        $\widetilde{\Omega}\rbra*{\min\{k,\sqrt{n}\}}$ rounds to reach
        consensus with high probability.
\end{theorem}

\subsection{Other Related Works} \label{sec:other related works}
Previous works on USD consider the plurality consensus problem and they show that USD solves it correctly with high probability if the initially most popular opinion has a significant advantage over the others.
\citet{Berenbrink_FGK16,Ghaffari_P16} presented protocols based on USD in the gossip model and showed that their protocols solve the plurality consensus problem correctly with high probability within $\Otilde(1)$ synchronous rounds with $\abs{\Sigma}=O(k)$ states under the assumption that the most popular opinion has a margin of at least $\Omega(\sqrt{n\log n})$ vertices over any other opinion in the initial configuration.
Recently, \citet{Bankhamer_SODA_2022} presented protocols having similar flavor based on USD in both the population protocol model and the gossip model.
A remarkable feature of these protocols is that the consensus time is small for any $k\ge 2$ and any initial configuration, whereas the protocols of \citet{Ghaffari_P16,Berenbrink_FGK16} require the initial margin assumption to bound the consensus time.
There is a work on another variant of USD that deals with a preferred opinion \cite{Berenbrink_BH24} and under noisy channel \cite{dACN20} for $k=2$.

A long line of works studied the consensus times of 3-Majority and related dynamics (e.g., 2-Choices) with $k>2$ opinions on an $n$-vertex complete graph \cite{simple_dynamics,stabilizing_consensus_Becchetti,nearly_tight_analysis,ignore_or_comply,Cooper_SODA25,SS25_sync,2choices}.
Recently, \citet{Cooper_SODA25} showed that 3-Majority in the asynchronous update model (i.e., at every round, a uniformly random vertex is allowed to update its opinion) reaches consensus within $\Otilde(\min\qty{k,\sqrt{n}})$ rounds.
In the synchronous update model where all vertices simultaneously update their opinions at every round, \citet{SS25_sync} showed that the consensus time of 2-Choices is $\Otilde(k)$ and that of 3-Majority is $\Otilde(\min\qty{k,\sqrt{n}})$.
Interestingly, this consensus time bound for 3-Majority is the same as the one for USD in the gossip model.

Several other relevant consensus dynamics have been considered in the literature.
A model that it somewhat similar is the model proposed in \cite{gauy2025voter}, where initialy some nodes do not have an opinion, and adopt any opnion immedieatly. Afterwards, the process behaves exactly like the voter model.
The deterministic majority model, where each agent sets its opinion to the majroity among all of its neighbors, has been studied in \cite{kaaser2015votingtimedeterministicmajority}.
See \cite{consensus_dynamics_becchetti20} for more details on consensus dynamics.
\section{Proof Overview} \label{sec:proof_overview}

We outline the proof of \cref{thm:main} for the gossip model and compare it with arguments from previous work.
For simplicity, we assume that $k\le n/\log n$ throughout this section.
This is only to ensure that with high probability USD does not fall into the all-$\bot$ configuration within one round; it can be relaxed to $k\le n$ without much difficulty.

\subsection{Process in a Nutshell} \label{sec:process in a nutshell}

Following prior work on USD, we track the evolution of the number of
vertices holding each decided opinion.
For $t\ge0$ and $i\in[k]$, define
\begin{align*}
\alpha_t(i) := \text{fraction of vertices holding opinion $i$ at the beginning of round $t$}.
\end{align*}
Then USD can be viewed as a Markov chain $(\alpha_t)_{t\ge0}$ for $\alpha_t = (\alpha_t(1),\dots,\alpha_t(k))$ on the state space 
$\{0, \frac{1}{n}, \frac{2}{n},\ldots, \frac{n-1}{n}, 1\}^k$.
The key progress parameter is
\[
    \npm_t := \max_{i\in[k]} \alpha_t(i),
\]
and the consensus is reached at the first time $t$ when $\npm_t=1$.
Unless all vertices hold $\bot$, it holds that $\npm_t \ge 1/n$.

We also introduce the following quantities:
\begin{align}
    \beta_t := \sum_{i\in[k]} \alpha_t(i), \quad
    \psi_t := \beta_t(2\beta_t-1) - \sum_{i\in[k]} \alpha_t(i)^2. \label{eq:psi_t}
\end{align}
Note that $\beta_t$ is the total fraction of vertices holding a decided opinion.
Roughly speaking, $\psi_t$ is a key parameter that controls the stability of USD.
See \cref{sec:behavior-of-beta_t-psi_t} for more details.

Call an opinion $i$ \emph{weak at round $t$} if $\alpha_t(i) \le 0.9\cdot \npm_t$ and \emph{strong at round $t$} if $\alpha_t(i) \ge 0.95\cdot \npm_t$.
Note that there is a margin between weak and strong opinions, which is crucial for our analysis.
We decompose the analysis into four parts.

\begin{lemma}[informal; see \cref{lem:taubeta,lem:taupsi,lem:growth of tnpt,lem:hitting time for tnpm and npm,lem:unique strong opinion,lem:towards consensus} for the formal statements] \label{lem:main goal}
    Consider USD in the gossip model such that $k\le n/\log n$ and $\beta_0\ge 1/n$.
    \begin{itemize}
        \item[\textup{(I)}] For some $T_\beta=O(\log n)$, with high probability, 
        for all $t\ge T_\beta$, $\beta_t\ge 1/2 - o(\log n / \sqrt{n})$. Moreover, with high probability, for all $t\ge 1$, $\psi_t = o(\log n / \sqrt{n})$.
        \item[\textup{(II)}] Suppose that $\psi_0 = o(\log n / \sqrt{n})$ and $\beta_0\ge 1/2 - o(\log n / \sqrt{n})$. For some $T_\alpha=O(\sqrt{n}(\log n)^3)$, with high probability, $\npm_t \ge \Omega\qty(\frac{(\log n)^2}{\sqrt{n}})$ holds for all $t\ge T_\alpha$.
        \item[\textup{(III)}] Suppose that $\psi_0 = o(\log n / \sqrt{n})$, $\beta_0\ge 1/2 - o(\log n / \sqrt{n})$, and $\npm_0 \ge \Omega\qty(\frac{(\log n)^2}{\sqrt{n}})$. Then, for some $T'=O(\log n / \npm_0)$, there exists exactly one strong opinion at round $T'$.
        \item[\textup{(IV)}] Suppose that $\psi_0 = o(\log n / \sqrt{n})$, $\beta_0\ge 1/2 - o(\log n / \sqrt{n})$, $\npm_0 = \alpha_0(1) \ge \Omega\qty(\frac{(\log n)^2}{\sqrt{n}})$, and $\alpha_0(1) \ge (1+c)\cdot \alpha_0(i)$ for some constant $c>0$ and all $i\ge 2$. Then, USD reaches consensus within $O(\log n / \npm_0)$ rounds with high probability.
    \end{itemize}
\end{lemma}

\begin{figure}[htbp]
    \centering
    \includegraphics[width=0.8\textwidth]{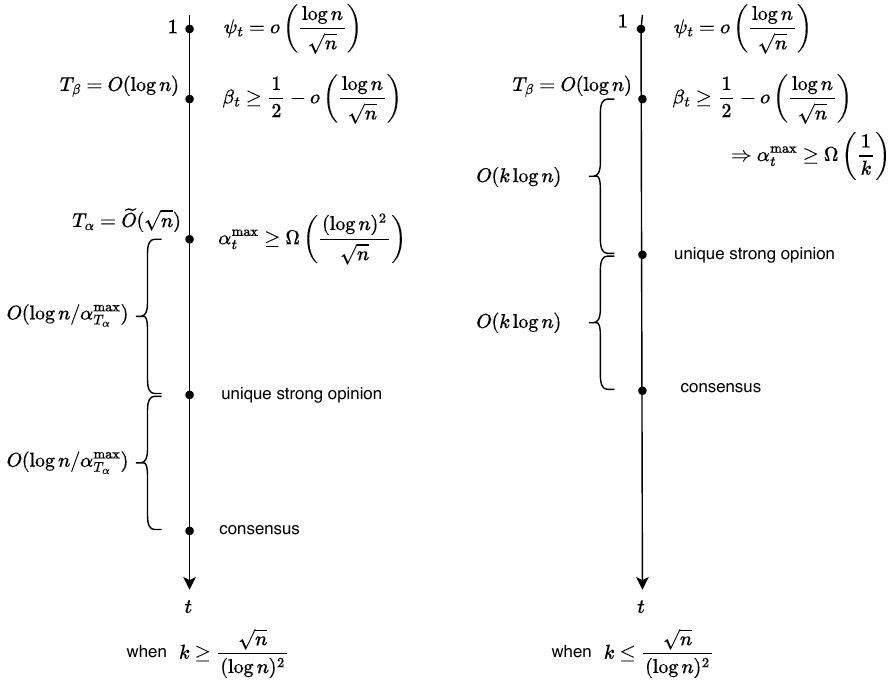}
    \caption{Combining the four parts of \cref{lem:main goal} to prove the consensus time bound.
    \label{fig:proof-flow}}
\end{figure}

The four items above immediately imply that the consensus time is $O(\sqrt{n}(\log n)^3)$.
Moreover, (I) ensures that with high probability, for some $t=O(\log n)$, $\beta_t=\Omega(1)$,
implying that $\npm_t = \Omega(1/k)$.
Thus, if $k\le \sqrt{n}/(\log n)^2$, then (III) and (IV) imply the consensus time bound of $O(\log n / \npm_0) = O(k\log n)$.
Therefore, the consensus time is $\Otilde(\min\qty{k, \sqrt{n}})$ (see \cref{fig:proof-flow}).

Our main technical contribution is to establish part (II).
An important feature is that (II) holds regardless of how large $k$ is, so 
our analysis broadens the range of $k$ from $k\ll n^{1/3}$ (or $k\ll n^{1/2}$ for the population protocol model) in the previous works to $k\le n/\log n$.

Establishing part (III) is also challenging. 
\cite[Phases 1 to 3]{fast_convergence_undecided} proved (III), for $k = O(\sqrt{n}/(\log n)^2)$ for the population protocol model, but their argument crucially relies on the local update property of the population protocol model which cannot be applied to the gossip model.
 
Part (IV) was shown in an earlier work \cite[Theorem 3.2]{Becchetti_SODA2015} for $k = O((n/\log n)^{1/3})$,
which (together with $\beta_0=\Omega(1)$) implies $\npm_0 \ge \frac{C\log^{1/3}n}{n^{1/3}}$, but we prove it under a weaker condition that $\npm_0 \ge \Omega\qty(\frac{(\log n)^2}{\sqrt{n}})$ by combining the argument of \cite{Becchetti_SODA2015} with sharp concentration inequalities from \cite{Cooper_SODA25,SS25_sync}.
In this section, we outline the proofs of (I) to (III) and omit the proof of (IV) for brevity.

\subsection{Difficulties in Extending Previous Techniques}
Before describing the ideas for the proof of \cref{lem:main goal},
we briefly discuss the techniques used in previous work and explain why they cannot be applied to USD directly.

\paragraph{(1) Breakdown of concentration-based analysis for large $k$.}
The key observation introduced by \citet{Becchetti_SODA2015} is that $\beta_t$ concentrates around $1/2$ and remains near this value for a
substantial period of time.
Using this stabilization property, they divided the USD process into several
phases according to $\npm_t$ and 
proved how the gap between the most popular opinion and the second most popular opinion grows in each phase.
This strategy is used in later work \cite{fast_convergence_undecided} to analyze USD in the population protocol model.
Our proof is based on this idea.

\citet{Becchetti_SODA2015} needed 
in their analysis a strong assumption that $k\ll n^{1/3}$  due to the limitation of the standard Chernoff bound.
In this paper, we broaden the range of $k$ to $k\ll \sqrt{n}$ by exploiting the wider martingale-based framework used in
\cite{Cooper_SODA25,SS25_sync}—particularly, we use Freedman's inequality 
together with additional technical ideas (see \cref{sec:concentration_overview} for more details).

However, for $k \gg \sqrt{n}$, the tail bounds from Freedman's inequality do not suffice to yield high-probability bounds, and the analysis based on the concentration of $\alpha_t(i)$ and $\beta_t$ as in \cite{Becchetti_SODA2015} entirely breaks down.

\begin{remark}
In the gossip model, the standard concentration argument yields that the number of remaining decided opinions becomes with high probability at most $\min(k, n\log n/k)$ after the first round if $\beta_0=1$(see \cref{sec:first round beta0=1}), i.e., if all vertices are initially decided. On the other hand, if we could extend the argument in \cite{fast_convergence_undecided} from the population protocol model to the gossip model, we would be able to derive an $\Otilde(k)$ bound on the consensus time  for $k\le O(\sqrt{n}/(\log n)^2)$.
Combining these two observations, we would be able to 
obtain a $\Otilde(\min\qty{k, \sqrt{n}})$ bound on the consensus time if either $k\le O(\sqrt{n}/(\log n)^2)$ or $k\ge\Omega(\sqrt{n}(\log n)^3)$ holds. 
Thus even if it was possible to transfer somehow the analysis in 
\cite{fast_convergence_undecided} to the gossip model, the case of 
$O(\sqrt{n}/(\log n)^2) \le k \le \Omega(\sqrt{n}(\log n)^3)$ would still remain unclear. 
Dealing with $k$ around $\sqrt{n}$ has been always challenging in the literature on USD (e.g., \cite[Section 3.3]{Bankhamer_SODA_2022}).
\end{remark}


\paragraph{(2) Limitations of the $k>\sqrt{n}$ techniques for 3-Majority.}

USD behaves similarly to the well-known consensus dynamics 3-Majority, which has been analyzed extensively for the \emph{entire} range of $2 \le k \le n$ \cite{simple_dynamics,stabilizing_consensus_Becchetti,nearly_tight_analysis,ignore_or_comply,Cooper_SODA25,SS25_sync}.
In 3-Majority, a selected vertex chooses three neighbors uniformly at random and takes the majority opinion held by the three, breaking the tie in favor of the third selected neighbor.
However, several structural and probabilistic features of USD make its analysis, especially in the regime $k \gg \sqrt{n}$, substantially more difficult.

We next explain why all known techniques that successfully handle the
large-$k$ regime in 3-Majority fail for USD.
These approaches fall into two broad categories:

\begin{itemize}

\item \textbf{Coupling + majorization (\cite{ignore_or_comply,Cooper_SODA25}).}
The successful analyses for 3-Majority by \cite{ignore_or_comply,Cooper_SODA25} rely on a coupling with Pull Voting
together with a majorization argument to show that the number of surviving opinions rapidly drops to $O(\sqrt{n})$,
after which the concentration-based analysis becomes effective.

\item \textbf{Potential-based analysis (\cite{SS25_sync}).} Another recent approach by \cite{SS25_sync} is to design a potential function that exhibits an additive drift and keeps growing during the process. 
When the potential function reaches a certain level, the concentration-based analysis becomes effective.

\end{itemize}

\noindent
Neither of the above two approaches transfers to USD.
Unlike 3-Majority, USD cannot be related to Pull Voting through coupling 
due to the existence of undecided vertices.
At the same time, the potential function used in \cite{SS25_sync} itself does not exhibit the additive drift property in USD, again due to undecided vertices.
Specifically, the potential function used in \cite{SS25_sync} is defined by
\begin{align}
    \npt_t = \sum_{i\in[k]} \alpha_t(i)^2. \label{eq:npt_t}
\end{align}
Since $\alpha_t(i)$ given the configuration at round $t-1$ can be written as the sum of $n$ independent random variables (in the gossip model), a standard moment calculation yields that, for 3-Majority,
\[
    \E_{t-1}[\npt_t] = \npt_{t-1} + \Omega(1/n),
\]
where $\E_{t-1}[\cdot]$ is the expectation conditioned on the configuration at round $t-1$.
However, in USD, again due to the emergence of undecided vertices, this inequality does not hold in general.

In summary, although several structural insights from previous work remain conceptually
useful
(e.g., the general ideas of potential-based arguments and concentration-based analysis), none of the existing analytical frameworks transfers to USD in a black-box way.
However, as shown in \cref{sec:growth-of-npm}, we can design another potential function $\tnpt_t$ that does exhibit an additive drift and keeps growing during the process.

\subsection{\texorpdfstring{Proof of \cref{lem:main goal} (I): Behavior of $\beta_t$ and $\psi_t$ (\cref{sec:behavior of beta})}{Proof of Lemma 2.1 (I): Behavior of beta and psi (Section 5.2)}
} 
\label{sec:behavior-of-beta_t-psi_t}

Now we outline the proof of the first item of \cref{lem:main goal}, which states that (i) $\beta_t$ becomes $1/2-o(1)$ with high probability within $O(\log n)$ rounds and remains at least $1/2-o(1)$ thereafter, and (ii) $\psi_t$ is small with high probability for all $t\ge 1$, where $\psi_t$ is defined in \cref{eq:psi_t}
and 
measures the 
change of $\beta_{t}$ from round $t-1$ to $t$ since 
\begin{align}
    \E_{t-1}[\beta_t] &= 2\beta_{t-1}(1-\beta_{t-1}) + \sum_{i\in[k]} \alpha_{t-1}(i)^2 \label{eq:expectation of beta_t 1} \\
    &= \beta_{t-1} - \psi_{t-1} \label{eq:expectation of beta_t 2}.
\end{align}
For example, if $\alpha_0(i)=1/k$ for all $i\in[k]$ (i.e., the initial configuration is  balanced), then $\psi_0 = 1 - 1/k$, which means that $\beta_1$ becomes $\beta_0 - \psi_0 = 1/k$ in expectation, i.e., $\beta_1$ drops significantly.

Interestingly, as shown in \cref{lem:basic inequalities for psi}, after the first round
the expectation of $\psi_t$ is always small:
\[
    \E_{t-1}[\psi_t] \le \frac{1}{n}.
\]
Combined this with a concentration inequality, we obtain the “Moreover" part of \cref{lem:main goal}(I).
Thus, once $\beta_t$ becomes $1/2-o(1)$, it remains at least $1/2-o(1)$ thereafter with high probability.

It remains to show that $\beta_t$ becomes $1/2-o(1)$ within $O(\log n)$ rounds with high probability.
Suppose that $\beta_{t-1} = o(1)$. Then, by \cref{eq:expectation of beta_t 1}, we have $\E_{t-1}[\beta_t] \ge (2-o(1))\beta_{t-1}$.
Therefore, $\beta_t$ grows by a factor of $1+\Omega(1)$ in each round, which ensures that $\beta_t$ becomes $\Omega(1)$ within $O(\log n)$ rounds with high probability.
A similar argument yields that we can show that $\beta_t=1/2-o(1)$ within $O(\log n)$ rounds with high probability.

\subsection{\texorpdfstring{Proof of \cref{lem:main goal} (II): Growth of $\npm_t$ (\cref{sec:behavior of tnpm,sec:behavior of tnpt})}{Proof of Lemma 2.1 (II): Growth of alphamax (Sections 5.3 and 5.4)}} 
\label{sec:growth-of-npm}

We now outline the proof of the second item of \cref{lem:main goal}.
The key idea is to normalize the potential $\npt_t$ of \cref{eq:npt_t} used for 3-Majority by $\beta_t^2$ to obtain the potential $\tnpt_t$ defined by
\begin{align*}
\tnpt_t :=
\begin{cases}
    \displaystyle \sum_{i\in[k]} \left(\frac{\alpha_t(i)}{\beta_t}\right)^{\!2}, & \beta_t>0,\\[1ex]
    0, & \text{otherwise}.
\end{cases}
\end{align*}
A reader familiar with the literature might notice that this quantity is somewhat similar to the \emph{monochromatic distance} introduced in \cite{Becchetti_SODA2015}.
In our notation, the monochromatic distance is defined in \cite{Becchetti_SODA2015} by
\[
    \mathrm{md}_t := \sum_{i\in[k]} \qty( \frac{\alpha_t(i)}{\npm_t} )^2.
\]
The difference between $\tnpt_t$ and $\mathrm{md}_t$ is that $\tnpt_t$ is normalized by $\beta_t^2$ while $\mathrm{md}_t$ is normalized by $(\npm_t)^2$.
This slight modification allows us to evaluate the expected growth of $\tnpt_t$, although $\tnpt_t$ is in a more complex form compared to $\npt_t$.
Using \cref{lem:main goal}(I) that ensures $\beta_{t-1}=\Omega(1)$ and a Taylor approximation, we can show that the potential $\tnpt_t$ has an additive drift
(\cref{lem:basic inequalities for tnpt}):
\[
    \E_{t-1}[\tnpt_t]\approx \frac{\E_{t-1}[\npt_t]}{\E_{t-1}[\beta_t]^2} \ge \tnpt_{t-1} + \Omega(1/n).
\]
Using a variant of the optional stopping theorem, we can show that the potential $\tnpt_t$ reaches $\Omega((\log n)^2/\sqrt{n})$ within $\Otilde(\sqrt{n})$ rounds with high probability.

Once $\tnpt_t$ reaches $\Omega((\log n)^2/\sqrt{n})$, we have
\[
    \npm_t \ge\sum_{i\in[k]} \alpha_t(i)^2 = \tnpt_t\cdot \underbrace{\beta_t^2}_{\Omega(1)\text{ from \cref{lem:main goal}(I)}} \ge \Omega\qty((\log n)^2/\sqrt{n}).
\]

We then show that $\npm_t$ does not decrease too much thereafter.
To this end, we introduce the normalized quantity
\[
    \tnpm_t :=
    \begin{cases}
        \npm_t / \beta_t, & \beta_t > 0, \\
        0, & \beta_t = 0,
    \end{cases}
\]
which is well-behaved because $\beta_t = \Omega(1)$ throughout this part of the process.  
In \cref{lem:basic inequalities for tnpm}, we will prove
\begin{align}
    \E_{t-1}[\tnpm_t] \ge \left(1 - O\qty(\frac 1 n) \right)\tnpm_{t-1}, \label{eq:Epnpm}
\end{align}
which ensures that once $\tnpm_t$ reaches $\Omega((\log n)^2/\sqrt n)$, it remains above this threshold (up to a constant factor) for $\Theta(n)$ rounds with high probability.

\subsection{\texorpdfstring{Proof of \cref{lem:main goal} (III): Gap between Two Opinions (\cref{sec:behavior of delta,sec:unique strong opinion})}{Proof of Lemma 2.1 (III): Gap between two opinions (Sections 5.5 and 5.6)}} 
\label{sec:gap-between-two-opinions}
Recall that an opinion $i$ is weak at round $t$ if $\alpha_t(i) \le 0.9\cdot \npm_t$ and strong at round $t$ if $\alpha_t(i) \ge 0.95\cdot \npm_t$.
Fix two distinct non-weak opinions $i,j\in[k]$.
The goal is to show that (i) either $i$ or $j$ becomes weak within $O(\log n / \npm_0)$ rounds with high probability, and (ii) once an opinion becomes weak, it cannot become strong during the rest of the process.

To this end, we are interested in the gap between $i$ and $j$, which is defined as
\[
    \delta_t := \alpha_t(i) - \alpha_t(j).
\]
Fix a round $t-1\ge 1$.
In \cref{lem:basic inequalities for delta_epsilon}, we will show that the expectation of $\delta_t$ given $\alpha_{t-1}$ is
\begin{align}
    \E_{t-1}[\delta_t] = \delta_{t-1}\qty(\alpha_{t-1}(i)+\alpha_{t-1}(j)+2(1-\beta_{t-1})). \label{eq:expectation of delta_t sec2}
\end{align}
Our aim is to show that $\delta_t$ grows by a factor of $1+\Omega(\npm_{t-1})$ in each round.
However, this does not hold in general.
For example, if $\beta_{t-1}=1$ and $\alpha_{t-1}(i)=\Theta(1/k)$ for all $i\in[k]$, then a standard Chernoff bound argument yields that $\alpha_t(i)\approx \alpha_{t-1}(i)^2 = \Theta(1/k^2)$, yielding $\delta_t \ll \delta_{t-1}$.
Thus we need to put some assumptions on the configuration at round $t-1$, which 
are listed below:

\begin{itemize}
    \item Suppose that both $i$ and $j$ are non-weak at round $t-1$. Then, we have  $\alpha_{t-1}(i)+\alpha_{t-1}(j) \ge 1.8\cdot \npm_{t-1}$.
    \item Suppose that $\beta_{t-1} \ge 1/2 - o(1)$ and $\psi_{t-1} \le o(\log n/\sqrt{n})$, which holds from \cref{lem:main goal}(I).
    Then, since $\psi_{t-1} = \beta_{t-1}(2\beta_{t-1}-1)-\gamma_{t-1}$ and 
    $\gamma_{t-1} = \sum_i \alpha_{t-1}(i)^2 \le \npm_{t-1}\cdot \sum_i \alpha_{t-1}(i) = \npm_{t-1}\cdot \beta_{t-1}$, we have
    \begin{align*}
        2\beta_{t-1} - 1 &\le \frac{\gamma_{t-1} + o(\log n/\sqrt{n})}{\beta_{t-1}} \\
        & = \npm_{t-1} + o(\log n/\sqrt{n}).
    \end{align*}
\end{itemize}
Applying these two items to \cref{eq:expectation of delta_t sec2}, we have
    \begin{align}
        \E_{t-1}[\delta_t] &= \delta_{t-1}\qty(1 + \underbrace{\alpha_{t-1}(i)+\alpha_{t-1}(j)}_{\ge 1.8\cdot \npm_{t-1}}-\underbrace{(2\beta_{t-1}-1)}_{\le \npm_{t-1} + o(\log n/\sqrt{n})}) \nonumber\\
        &\ge \delta_{t-1}\qty(1 + 0.8\cdot \npm_{t-1} - o(\log n/\sqrt{n})) \nonumber\\
        &\ge \delta_{t-1}\qty(1+\Omega(\npm_{t-1})). \label{eq:expectation of delta_t sec2:1}
    \end{align}
In the last inequality, note that $\npm_{t-1} \ge \Omega((\log n)^2/\sqrt{n})$ by the assumption.

Finally, from \cref{eq:Epnpm} and $\beta_t=\Omega(1)$, we have that $\npm_t \ge \Omega(\npm_0)$ for all $t=O(n)$ (in expectation).
Thus, we obtain
\[
    \E_{t-1}[\delta_t] \ge \delta_{t-1}\qty(1+\Omega(\npm_{0})).
\]
Even if we start with $\delta_0 = 0$ (i.e., $\alpha_0(i)=\alpha_0(j)$), by calculation (see \cref{lem:basic inequalities for delta_epsilon}), we have that its second moment satisfies
\[
    \Var_{t-1}[\delta_t] = \Omega\qty(\frac{(1-\beta_{t-1})^2(\alpha_{t-1}(i)+\alpha_{t-1}(j))}{n}) = \Omega\qty(\frac{\npm_0}{n}),
\]
and in particular
\[
    \E_{t-1}[\delta_t^2] \ge \delta_{t-1}^2 + \Omega\qty(\frac{\npm_0}{n}).
\]
Thus, by the optional stopping theorem combined with a widely-known argument from \cite{Doerr11}, we can show that $\abs{\delta_t}$ becomes at least $\Omega(\sqrt{\log n/n})$ within $O(\log n / \npm_0)$ rounds with high probability.
Therefore, within $O(\log n / \npm_0)$ rounds, either $i$ or $j$ becomes weak with high probability.



\subsection{Concentration Bounds via Freedman's Inequality}
\label{sec:concentration_overview}

To make the analysis of \cref{sec:gap-between-two-opinions,sec:growth-of-npm} rigorous, 
we require concentration bounds for several stochastic quantities (e.g., $\delta_t$).
Our approach follows the Freedman-based framework developed in 
\cite{Cooper_SODA25,SS25_sync}, which applies Freedman's inequality (a Bernstein-type concentration inequality for martingales) to obtain tight concentration bounds for various quantities.

In the population protocol model, these inequalities apply directly because 
each update affects only one vertex.  
In the gossip model, however, a synchronous round aggregates $n$ independent 
local changes.
Based on the observation that the one-step difference (e.g., $\alpha_t(i)-\alpha_{t-1}(i)$) can be written as the sum of $n$ independent random variables, \cite{SS25_sync} introduced the concept of \emph{Bernstein condition}, which is a sufficient condition for the application of Freedman's inequality even 
without bounded single-step jumps; see \cref{sec:Bernstein condition}.

However, the framework of \cite{SS25_sync} used for 3-Majority is insufficient for USD.
For 3-Majority, the quantities of interest 
(e.g., $\alpha_t(i)$ or $\delta_t=\alpha_t(i)-\alpha_t(j)$) are linear in 
independent random variables, which allows us to write the one-step difference as the sum of independent random variables.
In USD, however, the presence of undecided vertices destroys the clean drift 
structure of $\alpha_t(i)$, forcing us to analyze normalized 
quantities such as
\[
    \tnpt_t = \frac{\|\alpha_t\|_2^2}{\beta_t^2},
    \qquad 
    \tnpm_t = \frac{\npm_t}{\beta_t}.
\]
While $\beta_t$ concentrates reasonably well, its variance typically 
dominates that of $\alpha_t(i)$, so concentration of the numerator and the denominator 
\emph{separately} does not yield tight bounds for their ratio.  
Thus, the straightforward application of the Freedman-based framework is insufficient in its current form.

To overcome this difficulty, we apply a first-order Taylor expansion of $1/\beta_t$ around $\E_{t-1}[\beta_t]$:
\[
    \frac{1}{\beta_t}
    \ge
    \frac{1}{\E_{t-1}[\beta_t]}
    -
    \frac{\beta_t - \E_{t-1}[\beta_t]}{\E_{t-1}[\beta_t]^2}.
\]
Multiplying by $\alpha_t(i)$ gives
\[
    \frac{\alpha_t(i)}{\beta_t}
    \ge
    \frac{\alpha_t(i)}{\E_{t-1}[\beta_t]}
    -
    \frac{\alpha_t(i)}{\E_{t-1}[\beta_t]^2}
    \,(\beta_t - \E_{t-1}[\beta_t]).
\]
Both $\alpha_t(i)$ and $\beta_t$ are sums of $n$ independent random variables, 
so the right-hand side becomes a quadratic form in such sums.  
Using read-$k$ concentration bounds \cite{read-k,Duppala}, we obtain sharp 
concentration for these quadratic forms, and hence for the normalized 
quantities themselves.

This resolves the main obstacle preventing the direct use of 
Freedman-style arguments for USD and provides the concentration guarantees 
needed throughout the analysis.

\section{Definitions and Main Analysis Tools} \label{sec:preliminaries}


\subsection{Model and Basic Notation} \label{sec:model and basic notation}

For $n\in\Nat$, let $[n]=\{1,2,\dots,n\}$.
For $p>0$ and a vector $x\in\Real^n$, let $\norm{x}_p=\qty(\sum_{i=1}^n |x_i|^p)^{1/p}$ denote the $\ell^p$-norm of $x$.
For a finite set $S$, by $x\sim S$ we mean that $x$ is chosen uniformly at random from~$S$.
For $a,b\in\Real$, we use the shorthand $a\land b = \min\{a,b\}$.
We denote by $\indicator_{\calE}$ the indicator random variable of an event $\calE$, i.e., $\indicator_{\calE} = 1$ if $\calE$ occurs and $0$ otherwise.

Throughout the paper, we fix a set of $n$ vertices, denoted by $V$, with $|V|=n$.  
Each vertex holds an opinion from a finite alphabet $\Sigma = [k]\cup\{\bot\}$, 
where $k\in\mathbb{N}$ is a parameter and $\bot$ represents the undecided state.  
All configurations of the system are elements of $\Sigma^V$.  
The interaction rule (gossip or population protocol) will be specified later, 
but in all cases the opinion update is governed by the same USD update rule 
introduced below.

\begin{definition}[USD Update Rule] \label{def:USD update rule}
  The \emph{USD update rule} is a function $\mathsf{update}\colon\Sigma\times\Sigma\to\Sigma$ defined by 
  \[
    \mathsf{update}(\sigma_1,\sigma_2) = \begin{cases}
      \bot	& \text{if $|\{\sigma_1,\sigma_2,\bot\}|=3$},\\
      \sigma_2 & \text{if $\sigma_1=\bot$},\\
      \sigma_1 & \text{otherwise}.
    \end{cases}
  \]
\end{definition}
Note that, the condition $|\{\sigma_1,\sigma_2,\bot\}|=3$ means that both $\sigma_1$ and $\sigma_2$ are decided and they are distinct.

Using the update rule above,
we define USD in the gossip model.
In this model, every vertex updates simultaneously in each round by interacting with an independently chosen random neighbor.

\begin{definition}[Gossip USD] \label{def:gossip USD}
  The \emph{gossip USD} is the discrete-time Markov chain $(\opn_t)_{t\ge 0}$ on the state space $\Sigma^V$.
  Given $\opn_t\in\Sigma^V$, the next configuration $\opn_{t+1}\in\Sigma^V$ is obtained by the following procedure:
  For each vertex $u\in V$, independently select a random vertex $v\sim V$ and set 
    $\opn_{t+1}(u) = \mathsf{update}(\opn_t(u),\opn_t(v))$.
\end{definition}

We next define the Undecided-State Dynamics (USD) under the population protocol model.
Here, time evolves through a sequence of pairwise interactions:
in each step, a single ordered pair of vertices is sampled uniformly at random, and the initiator 
(the first vertex in the pair) updates its state using the USD update rule.

\begin{definition}[Population Protocol USD] \label{def:PPUSD}
  The \emph{population protocol USD} is the discrete-time Markov chain $(\opn_t)_{t\ge 0}$ on the state space $\Sigma^V$. 
  Given $\opn_t\in\Sigma^V$, the next configuration $\opn_{t+1}\in\Sigma^V$ is obtained by the following procedure:
  Select two uniformly random vertices $u,v\sim V$ with replacement and set 
    $\opn_{t+1}(u) = \mathsf{update}(\opn_t(u),\opn_t(v))$.
  All vertices other than $u$ keep their opinions unchanged, 
  i.e., $\opn_{t+1}(w)=\opn_t(w)$ for all $w\ne u$.
\end{definition}

\begin{remark}
We use the same notation $\opn_t$ for configurations in both models.
The ambient model is always clear from context, whereas introducing 
model-specific symbols such as $\opn_t^{\mathrm{gossip}}$ or 
$\opn_t^{\mathrm{pp}}$ would only create unnecessary notational branching 
without improving clarity.
\end{remark}

The main performance measure of USD is the \emph{consensus time}, defined by
  \[
  \taucons = \inf\qty{t\ge 0 \colon \exists a\in [k],\,\forall u\in V, \opn_t(u)=a}.
  \]
Note that $\taucons=\infty$ if for some $t\ge 0$, $\opn_t(u)=\bot$ for all $u\in V$.

We sometimes use $\calF_t$ to denote the natural filtration generated by the process $(\opn_s)_{s\le t}$.
Using this notation, we often use abbreviations $\Pr_{t-1}[\cdot]$, $\E_{t-1}[\cdot]$, and $\Var_{t-1}[\cdot]$ for $\Pr[\cdot\mid \calF_{t-1}]$, $\E[\cdot\mid \calF_{t-1}]$, and $\Var[\cdot\mid \calF_{t-1}]$, respectively.
For a random variable $X_t$ defined for step $t$ in USD, we often use the term ``$X_t$ conditioned on round $t-1$'' 
to denote this random variable conditioned on $\calF_{t-1}$.


In \cref{sec:key quantities}, we will define the quantities of interest mentioned in \cref{sec:proof_overview} and review their basic properties.

\subsection{Bernstein Condition}\label{sec:Bernstein condition}

The Bernstein condition \cite{SS25_sync} provides uniform control on the moment generating
function of a random variable and will be used throughout to obtain
concentration bounds for the quantities appearing in USD, such as
$\alpha_t(i)$ and their normalized forms.

\begin{definition}[Bernstein condition; \cite{SS25_sync}]\label{def:Bernstein condition}
Let $D,s\ge 0$.  
A random variable $X$ satisfies the $(D,s)$-Bernstein condition if
for all $\lambda\in\mathbb{R}$ with $|\lambda|D<3$,
\[
    \E[e^{\lambda X}]
    \le
    \exp\!\left(
        \frac{\lambda^2 s/2}{1 - |\lambda|D/3}
    \right).
\]
It satisfies the \emph{one-sided} $(D,s)$-Bernstein condition if the
same bound holds for all $\lambda\ge 0$ with $\lambda D<3$.
\end{definition}

The one-sided condition yields the following standard tail bound
(see, e.g., \cite[Proposition 2.14]{high-dim-stat}):
\begin{remark}\label{remark:Bernstein condition and concentration}
If $X$ satisfies the one-sided $(D,s)$-Bernstein condition, then
for all $h\ge 0$,
using $\Pr\left[e^{\lambda X}\ge e^{\lambda h}\right] \le e^{-\lambda h} \E[e^{\lambda X}]$ and
the bound on $\E[e^{\lambda X}]$, and taking $\lambda = h/(s+hD/3)$, we get
\[
    \Pr[X\ge h]
    \le
    \exp\qty(
        -\frac{h^2/2}{s + (hD)/3}
    ).
\]
\end{remark}

We record the basic closure properties used repeatedly later
(\cite[Lemma~3.4]{SS25_sync}).

\begin{lemma}[Basic properties of Bernstein condition; \cite{SS25_sync}]
    \label{lem:Bernstein condition}
    Let $X,Y$ be random variables. We have the following:
    \begin{enumerate}
        \renewcommand{\labelenumi}{(\roman{enumi})}
        \item \label{item:BC for bounded rv}
        If $\E[X]=0$ and $\abs{X}\leq \bounded$, then $X$ satisfies $\qty(\bounded,\Var\qty[X])$-Bernstein condition.
        \item \label{item:BC for upper bounded rv}
        If $X$ satisfies $\qty(\bounded,\variance)$-Bernstein condition, then
        $X$ satisfies $\qty(\bounded',\variance')$-Bernstein condition for any 
        $\bounded' \ge \bounded$ and $\variance' \ge \variance$ (same for the one-sided version).
        \item \label{item:BC for linear transformation}
        If $X$ satisfies $\qty(\bounded,\variance)$-Bernstein condition, then $aX$ satisfies $(\abs{a}\bounded,a^2\variance)$-Bernstein condition for any $a \in \Real$ (same for the one-sided version if $a\geq 0$).
        \item \label{item:BC for dominated rv}
        If $X$ satisfies one-sided $\qty(\bounded,\variance)$-Bernstein condition and $Y\preceq X$, 
        then $Y$ satisfies one-sided $\qty(\bounded, \variance)$-Bernstein condition, where $\preceq$ denotes the stochastic domination (see \cref{def:Stochastic domination}).
        \item \label{item:BC for independent rvs} 
        If $X_1,\ldots,X_n$ are independent and each $X_i$ satisfies $\qty(D,s_i)$-Bernstein condition, then $\sum_{i\in [n]}X_i$ satisfies $\qty(D,\sum_{i\in [n]}s_i)$-Bernstein condition. 
        \item \label{item:BC for NA rvs} 
        If $X_1,\ldots,X_n$ are negatively associated and each $X_i$ satisfies one-sided $\qty(D,s_i)$-Bernstein condition, then $\sum_{i\in [n]}X_i$ satisfies one-sided $\qty(D,\sum_{i\in [n]}s_i)$-Bernstein condition. 
    \end{enumerate}
\end{lemma}

We will need Bernstein bounds for sums of
random variables that may not be independent, namely, the sum of a read-$\ell$ family (\cref{def:read-ell family}).

\begin{lemma}
    \label{lem:Bernstein condition 2}
    We have the following:
    \begin{enumerate}
        \renewcommand{\labelenumi}{(\roman{enumi})}
        \item \label{item:BC for sum of rvs} 
        If $X_1,\ldots,X_n$ are random variables such that each $X_i$ satisfies $\qty(D,s_i)$-Bernstein condition, then $\sum_{i\in [n]}X_i$ satisfies $\qty(nD, n\sum_{i\in [n]}s_i)$-Bernstein condition. 
        \item \label{item:BC for read-k family} 
        If $Y_1, \ldots, Y_n$ is a read-$\ell$ family and each $Y_i$ satisfies $\qty(D,s_i)$-Bernstein condition, then the sum $\sum_{i\in [n]}Y_i$ satisfies $\qty(\ell D, \ell \sum_{i\in [n]}s_i)$-Bernstein condition. 
    \end{enumerate}
\end{lemma}
\begin{proof}[Proof of \cref{item:BC for sum of rvs}]
    Since $X_i$ satisfies $\qty(D,s_i)$-Bernstein condition, $nX_i$ satisfies $\qty(nD, n^2s_i)$-Bernstein condition from \cref{item:BC for linear transformation} of \cref{lem:Bernstein condition}, 
    i.e., for $\lambda\in \mathbb{R}$ such that $\abs{\lambda} n\bounded<3$, we have $\E\qty[\e^{\lambda nX_i}] \leq \exp\qty(\frac{\lambda^2 n^2s_i/2}{1-(\abs{\lambda} n\bounded)/3})$.
    Consider $\lambda\in \mathbb{R}$ such that $\abs{\lambda} n\bounded<3$.
    For such $\lambda$, we have
    \begin{align*}
        \E\qty[\prod_{i\in [n]}\e^{\lambda X_i}]
        &\leq \qty(\prod_{i\in [n]}\E\qty[\e^{\lambda X_i n}])^{1/n} &&(\text{From H\"{o}lder's inequality})\\
        &\leq \qty(\prod_{i\in [n]}\exp\qty(\frac{\lambda^2 n^2s_i  / 2}{1-(\abs{\lambda} n\bounded)/3}))^{1/n} \\
        &= \exp\qty(\frac{\lambda^2 n\sum_{i\in [n]}s_i  / 2}{1-(\abs{\lambda} n\bounded)/3}).
    \end{align*}
\end{proof}
\begin{proof}[Proof of \cref{item:BC for read-k family}]
    Since $Y_i$ satisfies $\qty(D,s_i)$-Bernstein condition, $\ell Y_i$ satisfies $\qty(\ell D, \ell^2s_i)$-Bernstein condition from \cref{item:BC for linear transformation} of \cref{lem:Bernstein condition}, 
    i.e., for $\lambda\in \mathbb{R}$ such that $\abs{\lambda} \ell\bounded<3$, we have $\E\qty[\e^{\lambda \ell Y_i}] \leq \exp\qty(\frac{\lambda^2 \ell^2s_i/2}{1-(\abs{\lambda} \ell\bounded)/3})$.
    Consider $\lambda\in \mathbb{R}$ such that $\abs{\lambda} \ell\bounded<3$.
    For such $\lambda$, we have
    \begin{align*}
        \E\qty[\prod_{i\in [n]}\e^{\lambda Y_i}]
        &\leq \qty(\prod_{i\in [n]}\E\qty[\e^{\lambda Y_i \ell}])^{1/\ell} &&(\text{From \cref{lem:moment of read-k family}})\\
        &\leq \qty(\prod_{i\in [n]}\exp\qty(\frac{\lambda^2 \ell^2s_i  / 2}{1-(\abs{\lambda} \ell\bounded)/3}))^{1/\ell} \\
        &= \exp\qty(\frac{\lambda^2 \ell\sum_{i\in [n]}s_i  / 2}{1-(\abs{\lambda} \ell\bounded)/3}).
    \end{align*}
\end{proof}

\subsection{Freedman's Inequality and Drift Analysis}\label{sec:drift analysis}

Drift analysis allows us to bound the hitting time of a stochastic process
once we control both its expected evolution and the concentration of its
increments.  We briefly recall the basic mechanism.

Let $(X_t)_{t\ge 0}$ be a process with
\[
    \E_{t-1}[X_t] \ge X_{t-1} + \drift
    \qquad\text{whenever } X_{t-1} < a,
\]
for some $\drift>0$ and threshold $a > X_0$,
and define the stopping time
\[
    \tau := \inf\{t \ge 0 : X_t \ge a\}.
\]
Consider the stopped process 
\[
    Y_t := X_{t\wedge\tau} - \drift\,(t\wedge\tau),
\]
which is a submartingale.  
If the increments $Y_t - Y_{t-1}$ satisfy a Bernstein condition,
Freedman's inequality yields, with high probability, 
\[
    Y_T \ge Y_0 - \varepsilon = X_0 - \varepsilon
    \qquad\text{for all sufficiently large } T.
\]
On the event $\{T<\tau\}$, we have $X_T<a$ and hence
\[
    X_T - \drift T = Y_T \ge X_0 - \varepsilon,
\]
which implies
\begin{align}
    T \le \frac{a - X_0 + \varepsilon}{\drift}. \label{eq:T bound by drift and concentration}
\end{align}
Thus $\tau$ is at most this value with high probability.
This is the standard way in which positive drift and concentration combine
to yield upper bounds on hitting times; see, e.g., \cite{Lengler2020-bo}.
The following lemma, introduced in \cite{SS25_sync}, summarizes the above discussion, which is based on the Freedman's inequality.

\begin{lemma}[Lemma 3.5 of \cite{SS25_sync}]
    \label{lem:Freedman stopping time additive}
    Let $(X_t)_{t\in \Nat_0}$ be a sequence of random variables and
    let $(\calF_t)_{t\in\Nat_0}$ be a filtration such that $ X_t $ is $ \calF_t $-measurable for all $t\ge 0$.
    Let $\tau$ be a stopping time with respect to $(\calF_t)_{t\in \Nat_0}$.
    Let $\bounded,\variance \ge 0$ and $\drift\in\Real$ be parameters.
    Suppose the following condition holds for any $t\geq 1$: conditioned on $\mathcal{F}_{t-1}$, 
         \begin{enumerate}[label=$(C\arabic*)$]
           \item $\indicator_{\tau>t-1}\rbra*{\E_{t-1}[X_t]-X_{t-1}-\drift}\leq 0$, \label{item:C1}
           \item $\indicator_{\tau>t-1}\rbra*{X_t-X_{t-1}-\drift}$ satisfies one-sided $\qty(\bounded,\variance)$-Bernstein condition. \label{item:C2}
         \end{enumerate}
       For a parameter $h>0$, define stopping times
       \begin{align*} 
         \tauxplus\defeq \inf\qty{t\geq 0\colon X_t\geq X_0+h} \text{ and }
         \tauxminus \defeq \inf\qty{t\geq 0\colon X_t\leq X_0-h}.
       \end{align*}
       Then, we have the following:
         \begin{enumerate}
           \renewcommand{\labelenumi}{(\roman{enumi})}
           \item \label{item:positive drift} Suppose $\drift\geq 0$. Then, for any $h,T>0$ such that $z \defeq h-\drift\cdot T > 0$, we have
           \[
           \Pr\qty[\tauxplus \le \min\{T,\tau\}] 
            \le \exp\qty(- \frac{z^2/2}{\variance T + (z \bounded)/3}).
           \]
           \item \label{item:negative drift} Suppose $\drift< 0$. Then, for any $h,T>0$ such that $ z \defeq (-\drift)\cdot T - h > 0$, we have
           \[
           \Pr\qty[\min\cbra{\tauxminus,\tau} >T] 
            \le \exp\qty(- \frac{z^2/2}{\variance T + (z \bounded)/3}).
           \]
         \end{enumerate}
     \end{lemma}
      \begin{remark}
        Since \cref{item:C1} implies $\indicator_{\tau>t-1}\rbra*{X_t-X_{t-1}-R}\leq \indicator_{\tau>t-1}\rbra*{X_t-\E_{t-1}[X_t]}$, we can use the following \cref{item:C2'} instead of \cref{item:C2}:
        \begin{enumerate}[label=$(C\arabic*')$]
          \setcounter{enumi}{1}
          \item $\indicator_{\tau>t-1}\rbra*{X_t-\E_{t-1}[X_t]}$ satisfies one-sided $\qty(\bounded,\variance)$-Bernstein condition. \label{item:C2'}
        \end{enumerate}
      \end{remark}

In this paper, we frequently use the following lemma, which is directly derived from \cref{lem:Freedman stopping time additive} and is formulated for convenience in our applications.
The proof is deferred to \cref{sec:tools for drift analysis}.
Roughly speaking, the result concerns a quantity $X_t$ of interest that satisfies the Bernstein condition and ensures that the following hold (with a sufficiently large probability):
\begin{itemize}
  \item If for all $t\ge 1$, $\E_{t-1}[X_t]\geq X_{t-1}+\drift$ for some $\drift<0$ (i.e., $\E_{t-1}[X_t]\geq X_{t-1}-r$ for some $r>0$), then $X_{t}> I^-$ holds for all $t\leq (1-\epsilon)\frac{X_0-I^-}{-\drift}$ ($=(1-\epsilon)\frac{X_0-I^-}{r}$).
  In other words, $X_t$ does not decrease too much if $\E_{t-1}[X_t]$ is not too small compared to $X_{t-1}$. For example, we apply this claim for $X_t=\tnpm_t$, which may have a tiny negative drift \cref{eq:Epnpm}.
  \item If for all $t\ge 1$, $\E_{t-1}[X_t]\geq X_{t-1}+\drift$ for some $\drift>0$, then $X_{T}\geq I^+$ holds for some $T\geq (1+\epsilon)\frac{I^+-X_0}{\drift}$. In other words, $X_t$ increases rapidly if it exhibits a positive drift. For example, we apply this claim for $X_t = \delta_t$, turning a positive multiplicative drift given in \cref{eq:expectation of delta_t sec2:1} into an additive drift
  $\E_{t-1}[\delta_t] \ge \delta_{t-1} +\Omega(\npm_{0}\delta_0)$.
\end{itemize}

\begin{lemma}
    \label{lem:Useful drift lemma}
    Let $(X_t)_{t\in \Nat_0}$ be a sequence of random variables and
    let $(\calF_t)_{t\in\Nat_0}$ be a filtration such that $ X_t $ is $ \calF_t $-measurable for all $t\ge 0$.
    Let $\tau^*$ be a stopping time with respect to $(\calF_t)_{t\in \Nat_0}$.
    For parameters $I^-$ and $I^+$, define stopping times
    \begin{align*} 
      \tau^+\defeq \inf\qty{t\geq 0\colon X_t\geq I^+} \text{ and }
      \tau^- \defeq \inf\qty{t\geq 0\colon X_t\leq I^-}.
    \end{align*}
    Then, we have the following:
    \begin{enumerate}
        \item \label{item:negative drift useful} 
        Suppose the conditions \cref{item:C1,item:C2} are valid for parameters $\bounded,\variance \ge 0$, $\drift<0$ and $\tau=\tau^*$.
        Then, for any positive constant $\epsilon>0$, and any $T\leq \frac{(1-\epsilon)(X_0-I^-)}{-\drift}$, 
        \begin{align*}
         \Pr\qty[\tau^- \le \min\{T,\tau^*\}] 
         \le \exp\qty(- \frac{\epsilon^2(X_0-I^-)^2/2}{\variance T + (\epsilon(X_0-I^-)\bounded)/3}).
        \end{align*}
        \item \label{item:positive drift useful} 
        Suppose the conditions \cref{item:C1,item:C2} are valid for parameters $\bounded,\variance \ge 0$, $\drift>0$ and $\tau=\min\{\tau^+,\tau^-,\tau^*\}$. 
        Then, for any positive constant $\epsilon>0$ and any $T\geq \frac{(1+\epsilon)(I^+-X_0)}{\drift}$, 
        \begin{align*}
        \Pr\qty[\tau^+>T \text{ and } \tau^*>T] 
          &\le \exp\qty(- \frac{(X_0-I^-)^2/2}{\variance T + ((X_0-I^-) \bounded)/3})+\exp\qty(- \frac{\epsilon^2(I^+-X_0)^2/2}{\variance T  + (\epsilon(I^+-X_0)\bounded)/3}), \\
        \Pr\qty[\tau^+>\tau^- \text{ and } \tau^*>T] 
          &\le \exp\qty(- \frac{(X_0-I^-)^2/2}{\variance T + ((X_0-I^-) \bounded)/3})+\exp\qty(- \frac{\epsilon^2(I^+-X_0)^2/2}{\variance T  + (\epsilon(I^+-X_0)\bounded)/3}).
         \end{align*}
    \end{enumerate}
\end{lemma}

\section{Key Quantities} \label{sec:key quantities}
In this section, we introduce the quantities that will be used throughout our
analysis of USD.  For each quantity, we establish its conditional expectation,
conditional variance, and a one-sided Bernstein condition.  These properties
form the backbone of the drift arguments used later.

Since the definitions are intentionally consolidated here, the presentation is
dense; readers are encouraged to refer back to this section as needed when
following the proofs in \cref{sec:proof_gossip,sec:proof_pp}.

\begin{definition}[Key Quantities] \label{def:key-quantities}
Let $(\opn_t)_{t\ge 0}$ be either the gossip USD or the population protocol USD.
For each round $t\ge 0$, define:

\begin{itemize}

\item For each $i\in[k]$, let $\alpha_t(i) := \frac{1}{n}\abs{\{u\in V : \opn_t(u)=i\}}$, and let $\beta_t := \sum_{i\in[k]} \alpha_t(i)$ denote the total fraction of vertices holding a decided opinion.

\item $\npt_t := \norm{\alpha_t}_2^2 = \sum_{i\in[k]} \alpha_t(i)^2$.

\item $\psi_t := \beta_t(2\beta_t-1) - \npt_t$.

\item 
$\tnpt_t := \npt_t/\beta_t^2$. If $\beta_t=0$, we define $\tnpt_t=0$.

\item $\npm_t := \norm{\alpha_t}_\infty = \max_{i\in[k]} \alpha_t(i)$.

\item 
$\tnpm_t := \npm_t / \beta_t$. If $\beta_t=0$, we define $\tnpm_t=0$.

\item For $i,j\in[k]$ and $\varepsilon\in[0,1]$, define $\delta_t^{(\epsilon)}(i,j)\defeq \alpha_t(i)-(1+\varepsilon)\alpha_t(j)$ and $\delta_t(i,j)\defeq \delta_t^{(0)}(i,j)$. If $i$ and $j$ are clear from the context, we write $\delta_t$ or $\delta_t^{(\epsilon)}$ instead of $\delta_t(i,j)$ or $\delta_t^{(\epsilon)}(i,j)$, respectively.

\item Let $I_t\defeq \min\{i\in [k]:\alpha_t(i)=\npm_t\}$ be the most popular opinion at round $t$.
For $i\in[k]$, let $\eta_t(i) := \delta_t^{(\ceta)}(I_t,i) = \npm_t - (1+\ceta)\alpha_t(i)$, where the constant $\ceta>0$ is fixed in \cref{def:strong}.
\end{itemize}
\end{definition}

Throughout this paper, we shall use $\E_{t-1}[\cdot]$ to denote the expectation conditioned on 
round $t-1$.
For example, $\E_{t-1}[\alpha_t]$ is the expectation of $\alpha_t$ conditioned on 
round $t-1$.

\subsection{Basic Properties in the Gossip Model} \label{sec:basic properties in the gossip model}

We now present the basic properties (expectations, variances, and Bernstein conditions) of the key quantities introduced in \cref{def:key-quantities} in the gossip model.
These results are applied repeatedly in the phase analysis.
All proofs of the lemmas in this section are deferred to \cref{sec:proof of basic properties},
as they are routine calculations or applications of lemmas in \cref{sec:Bernstein condition}.

\begin{lemma}[Basic properties for $\alpha_t(i)$ in the gossip model]
  \label{lem:basic inequalities for alpha}
  For the gossip USD, the quantity $\alpha_t(i)$ satisfies the following:
  \begin{enumerate}
    \item (Expectation) \label{item:expectation of alpha} 
    $\E_{t-1}[\alpha_t(i)]= \alpha_{t-1}(i)\qty(\alpha_{t-1}(i)+2\qty(1-\beta_{t-1}))$.
    \item (Variance) \label{item:variance of alpha} 
    $\Var_{t-1}[\alpha_t(i)] = \frac{\alpha_{t-1}(i)}{n}\qty[ (1-\beta_{t-1})(1+\beta_{t-1}-2\alpha_{t-1}(i))+\alpha_{t-1}(i)(\beta_{t-1}-\alpha_{t-1}(i))]$.
    Specifically, $\Var_{t-1}[\alpha_t(i)] \leq \frac{\alpha_{t-1}(i)}{n}$ and $\Var_{t-1}[\alpha_t(i)] \geq \frac{\qty(1-\beta_{t-1})^2\alpha_{t-1}(i)}{n}$ hold.
    \item (Bernstein condition) \label{item:bernstein condition of alpha}
    $\alpha_t(i)-\E_{t-1}[\alpha_t(i)]$ conditioned on round $t-1$ satisfies $\qty(\frac{1}{n},\frac{\alpha_{t-1}(i)}{n})$-Bernstein condition.
  \end{enumerate}
\end{lemma}

\begin{lemma}[Basic properties for $\beta_t$ in the gossip model]
  \label{lem:basic inequalities for beta}
  For the gossip USD, the quantity $\beta_t$ satisfies the following:
  \begin{enumerate}
    \item (Expectation) $\E_{t-1}[\beta_t] = 2\beta_{t-1}(1-\beta_{t-1}) + \gamma_{t-1} = \beta_{t-1} - \psi_{t-1}$. \label{item:expectation of beta} 
    \item (Variance) $\Var_{t-1}[\beta_t]=\frac{(\beta_{t-1}-\gamma_{t-1})(1-\beta_{t-1}+\gamma_{t-1})+\gamma_{t-1}^2-\norm{\alpha_{t-1}}_3^3}{n}$.
     In particular, $\Var_{t-1}[\beta_t] \leq \frac{\beta_{t-1}}{n}$. \label{item:variance of beta} 
     \item (Bernstein condition) $\beta_t-\E_{t-1}[\beta_t]$ conditioned on round $t-1$ satisfies $\qty( \frac{1}{n}, \frac{\beta_{t-1}}{n} )$-Bernstein condition. \label{item:bernstein condition of beta}
  \end{enumerate}
\end{lemma}

\begin{lemma}[Basic properties for $\delta^{(\varepsilon)}_t$ in the gossip model]
  \label{lem:basic inequalities for delta_epsilon}
  For the gossip USD, the quantity $\delta_t^{(\varepsilon)}\defeq\delta_t^{(\epsilon)}(i,j)$ satisfies the following:
  \begin{enumerate}
    \item (Expectation) \label{item:expectation of delta_epsilon} 
    $\E_{t-1}\qty[\delta_t^{(\epsilon)}]=\delta_{t-1}^{(\epsilon)}\qty(\alpha_{t-1}(i)+\alpha_{t-1}(j)+2(1-\beta_{t-1}))+\varepsilon\alpha_{t-1}(i)\alpha_{t-1}(j)$. In particular, $\E_{t-1}[\delta_t]=\delta_{t-1}+\delta_{t-1}\npm_{t-1}\qty(\frac{\alpha_{t-1}(i)+\alpha_{t-1}(j)}{\npm_{t-1}}-\frac{\gamma_{t-1}+\psi_{t-1}}{\beta_{t-1}\npm_{t-1}})$.
    \item (Variance) \label{item:variance of delta}
    $\Var_{t-1}\qty[\delta_t]\geq \frac{(1-\beta_{t-1})^2}{n}\qty(\alpha_{t-1}(i)+\alpha_{t-1}(j))$.
    \item (Bernstein condition) \label{item:bernstein condition for delta_epsilon}
    The difference $\delta_t^{(\epsilon)}-\E_{t-1}\qty[\delta_t^{(\epsilon)}]$ conditioned on round $t-1$ satisfies 
    $\qty(\frac{2(1+\epsilon)}{n},s)$-Bernstein condition for $s=\frac{2}{n}\qty(\alpha_{t-1}(i)+(1+\epsilon)^2\alpha_{t-1}(j))$
  \end{enumerate}
\end{lemma}

\begin{lemma}[Basic properties for $\npt_t$ in the gossip model]
  \label{lem:basic inequalities for gamma}
  For the gossip USD, the quantity $\gamma_t$ satisfies the following:
  \begin{enumerate}
    \item (Expectation Upper Bound) $\E_{t-1}[\gamma_t] \leq 10\gamma_{t-1}$. \label{item:expectation of gamma}
    \item (Expectation Lower Bound) $\beta_{t-1}^2\E_{t-1}[\gamma_t]\geq \E_{t-1}[\beta_t]^2\gamma_{t-1}+\frac{\beta_{t-1}^3(1-\beta_{t-1})^2}{n}$. \label{item:expectation of gammatilde}
    \item (Bernstein condition) $\npt_t-\E_{t-1}[\npt_t]$ conditioned on round $t-1$ satisfies $\qty(\frac{2}{n}, \frac{20\gamma_{t-1}}{n})$-Bernstein condition. \label{item:bernstein condition of gamma}
  \end{enumerate}
\end{lemma}

\begin{lemma}[Basic properties for $\psi_t$ in the gossip model]
  \label{lem:basic inequalities for psi}
  For the gossip USD, the quantity $\psi_t$ satisfies the following:
  \begin{enumerate}
    \item (Expectation) $\E_{t-1}[\psi_t] \leq \frac{\beta_{t-1}}{n}$. \label{item:expectation of psi}
    \item (Bernstein condition) $\psi_t - \E_{t-1}[\psi_t]$ conditioned on round $t-1$ satisfies $\qty( \frac{24}{n}, \frac{400\beta_{t-1}}{n} )$-Bernstein condition. \label{item:bernstein condition of psi}
  \end{enumerate}
\end{lemma}

\begin{lemma}[Basic properties for $\tnpm_t$ in the gossip model]
  \label{lem:basic inequalities for tnpm}
  For the gossip USD, the quantity $\tnpm_t$ satisfies the following:
  Let $C>0$ be a sufficiently large constant.
  Suppose that $\beta_{t-1}\geq 1/2-o(1)$ and $\psi_{t-1} \le o(1)$. Then,
  \begin{enumerate}
    \item (Expectation) $\E_{t-1}[\tnpm_t]\geq 
    \tnpm_{t-1}\qty(1+\frac{\npm_{t-1}-\npt_{t-1}/\npo_{t-1}}{2(1-\npo_{t-1})+\npt_{t-1}/\npo_{t-1}})-\frac{9\npm_{t-1}}{n}$. \label{item:expectation of tnpm}
    \item (Bernstein condition) $\tnpm_{t-1}\qty(1+\frac{\npm_{t-1}-\npt_{t-1}/\npo_{t-1}}{2(1-\npo_{t-1})+\npt_{t-1}/\npo_{t-1}})-\tnpm_t-\frac{9\npm_{t-1}}{n}$ conditioned on round $t-1$ satisfies one-sided $\qty(\frac{C}{n}, \frac{C\npm_{t-1}}{n})$-Bernstein condition. \label{item:bernstein condition of tnpm}
  \end{enumerate}
\end{lemma}

\begin{lemma}[Basic properties for $\tnpt_t$ in the gossip model]
  \label{lem:basic inequalities for tnpt}
  For the gossip USD, the quantity $\tnpt_t$ satisfies the following:
  Suppose $\beta_{t-1}\geq 1/2-o(1)$, $\psi_{t-1} \le o(1)$, and $\gamma_{t-1} \le o(1)$. Then, $\E_{t-1}[\tnpt_t] \ge \tnpt_{t-1} + \frac{1}{12n}$.
\end{lemma}

\subsection{Basic Properties in the Population Protocol Model} \label{sec:basic properties in the population protocol model}

We now present the basic properties (expectations, variances, and Bernstein conditions) of the key quantities introduced in \cref{def:key-quantities} in the population protocol model, which are analogous of \cref{sec:basic properties in the gossip model}.
The proofs can be found in \cref{sec:proof of basic properties in the population protocol model}.

\begin{lemma}[Basic properties for $\alpha_t(i)$ in the population protocol model]
  \label{lem:basic inequalities for alpha PP}
  For the population protocol USD, the quantity $\alpha_t(i)$ satisfies the following:
  \begin{enumerate}
    \item (Expectation) \label{item:expectation of alpha PP}
     $\E_{t-1}[\alpha_t(i)]=\alpha_{t-1}(i)\qty(1+\frac{\alpha_{t-1}(i)+1-2\beta_{t-1}}{n})$.
    \item (Variance) \label{item:variance of alpha PP} 
    $\Var_{t-1}[\alpha_t(i)]=\frac{\alpha_{t-1}(i)}{n^2}\qty(1-\alpha_{t-1}(i)-\alpha_{t-1}(i)\qty(1-2\beta_{t-1}+\alpha_{t-1}(i))^2)$.
    In particular, $\Var_{t-1}[\alpha_t(i)] \leq \frac{\alpha_{t-1}(i)}{n}$. 
    \item (Bernstein condition) \label{item:bernstein condition of alpha PP}
    $\alpha_t(i)-\E_{t-1}[\alpha_t(i)]$ conditioned on round $t-1$ satisfies $\qty(\frac{1}{n},\frac{\alpha_{t-1}(i)}{n^2})$-Bernstein condition.
  \end{enumerate}
\end{lemma}

\begin{lemma}[Basic properties for $\beta_t$ in the population protocol model]
  \label{lem:basic inequalities for beta PP}
  For the population protocol USD, the quantity $\beta_t$ satisfies the following:
  \begin{enumerate}
    \item (Expectation) $\E_{t-1}\qty[\beta_t]=\beta_{t-1}+\frac{\beta_{t-1}\qty(1-2\beta_{t-1})+\gamma_{t-1}}{n}$. \label{item:expectation of beta PP}
    \item (Variance) $\Var_{t-1}[\beta_t]\leq \frac{\beta_{t-1}-\gamma_{t-1}}{n^2}$.
    \label{item:variance of beta PP}
     \item (Bernstein condition) $\beta_t-\E_{t-1}[\beta_t]$ conditioned on round $t-1$ satisfies $\qty( \frac{1}{n}, \frac{\beta_{t-1}}{n^2} )$-Bernstein condition. \label{item:bernstein condition of beta PP}
  \end{enumerate}
\end{lemma}

\begin{lemma}[Basic properties for $\delta_t^{(\epsilon)}$ in the population protocol model]
  \label{lem:basic inequalities for delta PP} 
  For the population protocol USD, the quantity $\delta_t^{(\epsilon)}\defeq\delta_t^{(\epsilon)}(i,j)$ 
  satisfy the following:
  \begin{enumerate}
    \item (Expectation) \label{item:expectation of delta PP} 
    $\E_{t-1}\qty[\delta_t^{(\epsilon)}]
    =\delta_{t-1}^{(\epsilon)}\qty(1+\frac{\alpha_{t-1}(i)+\alpha_{t-1}(j)+1-2\beta_{t-1}}{n})+\frac{\epsilon\alpha_{t-1}(i)\alpha_{t-1}(j)}{n}$.
    \item (Variance) \label{item:variance of delta PP}
    $\Var_{t-1}\qty[\delta_t]
    \geq \frac{\alpha_{t-1}(i)\qty(1-\alpha_{t-1}(i))}{n^2}
    +\frac{\alpha_{t-1}(j)\qty(1-\alpha_{t-1}(j))}{n^2}
    -\frac{4\delta_{t-1}^2}{n^2}.$
    \item (Bernstein condition) \label{item:bernstein condition of delta PP}
    $\delta_t^{(\epsilon)}-\E_{t-1}\qty[\delta_t^{(\epsilon)}]$ conditioned on round $t-1$ satisfies 
    $\qty(\frac{2(1+\epsilon)}{n},s)$-Bernstein condition for $s=\frac{2}{n^2}\qty(\alpha_{t-1}(i)+(1+\epsilon)^2\alpha_{t-1}(j))$.
  \end{enumerate}
\end{lemma}

\begin{lemma}[Basic properties for $\npt_t$ in the population protocol model]
  \label{lem:basic inequalities for gamma PP}
  For the population protocol USD, the quantity $\gamma_t$ satisfies the following:
  \begin{enumerate}
    \item (Expectation) $\E_{t-1}[\gamma_t]=\gamma_{t-1}+\frac{2}{n}\qty((1-2\beta_{t-1})\gamma_{t-1}+\norm{\alpha_{t-1}}_3^3)+\frac{\beta_{t-1}-\gamma_{t-1}}{n^2}$. \label{item:expectation of gamma PP}
    \item (Variance) $\Var_{t-1}[\gamma_t]\leq \frac{9\norm{\alpha_{t-1}}_3^3}{n^2}$, where $\norm{\alpha_{t-1}}_3^3=\sum_{i\in[k]}\alpha_{t-1}(i)^3$. \label{item:variance of gamma PP}
    \item (Bernstein condition) $\gamma_t-\E_{t-1}[\gamma_t]$ conditioned on round $t-1$ satisfies $\qty(\frac{8\npm_{t-1}}{n},\frac{9\norm{\alpha_{t-1}}_3^3
    }{n^2})$-Bernstein condition. \label{item:bernstein condition of gamma PP}
  \end{enumerate}
\end{lemma}

\begin{lemma}[Basic properties for $\psi_t$ in the population protocol model]
  \label{lem:basic inequalities for psi PP}
  For the population protocol USD, the quantity $\psi_t$ satisfies the following:
  \begin{enumerate}
    \item (Expectation) $\E_{t-1}[\psi_t]\leq \psi_{t-1}\qty(1-\frac{1}{n})+\frac{\beta_{t-1}-\gamma_{t-1}}{n^2}$. \label{item:expectation of psi PP}
    \item (Bernstein condition) $\psi_t - \E_{t-1}[\psi_t]$ conditioned on round $t-1$ satisfies $\qty( \frac{150}{n}, \frac{150\beta_{t-1}}{n^2} )$-Bernstein condition. \label{item:bernstein condition of psi PP}
  \end{enumerate}
\end{lemma}

\begin{lemma}[Basic properties for $\tnpm_t$ in the population protocol model]
  \label{lem:basic inequalities for tnpm PP}
  For the population protocol USD, the quantity $\tnpm_t$ satisfies the following:
  Let $C>0$ be a sufficiently large constant. Suppose $\beta_{t-1}\geq 1/2-o(1)$. Then,
  \begin{enumerate}
    \item (Expectation) $\E_{t-1}\qty[\tnpm_t]\geq \tnpm_{t-1}\qty(1+\frac{\npm_{t-1}-\gamma_{t-1}/\beta_{t-1}}{2n})-\frac{18\npm_{t-1}}{n^2}$. \label{item:expectation of tnpm PP}
    \item (Bernstein condition) 
$\tnpm_{t-1}\qty(1+\frac{\npm_{t-1}-\gamma_{t-1}/\beta_{t-1}}{2n})-\tnpm_t-\frac{18\npm_{t-1}}{n^2}$ conditioned on round $t-1$ satisfies one-sided $\qty(\frac{C}{n}, \frac{C\npm_{t-1}}{n^2})$-Bernstein condition. \label{item:bernstein condition of tnpm PP}
  \end{enumerate}
\end{lemma}

\begin{lemma}[Basic properties for $\tnpt_t$ in the population protocol model]
  \label{lem:basic inequalities for tnpt PP}
  Suppose that $\beta_{t-1}\geq 1/2-o(1)$ and $\gamma_{t-1}\leq o(1)$.
  Then for the population protocol USD, 
  $\E_{t-1}\qty[\tnpt_t]\geq \tnpt_{t-1}+\frac{1}{12n^2}$.

\end{lemma}

\subsection{Stopping Times} \label{sec:stopping_times}

In the remainder of the paper we frequently refer to the first time at which
one of the key quantities introduced above crosses a given threshold.
Since these stopping times are used uniformly in both the gossip model
and the population–protocol model, we introduce them here in a unified way.

For a stochastic process $(Z_t)_{t\ge 0}$ and a threshold $\theta\in\mathbb{R}$,
we consider the upward and downward hitting times of the form
$\tau^+_Z := \inf\{t\ge 0 : Z_t \ge \theta\}$ and $\tau^-_Z := \inf\{t\ge 0 : Z_t \le \theta\}$.
We will consider such stopping times for $\beta_t$, $\npm_t$, $\delta_t$, $\tnpt_t$, etc.
These stopping times serve as canonical milestones in our drift and
concentration arguments later on.

Whenever a stopping time refers to a quantity that depends on a pair of
opinions (e.g., $\delta_t$), the stopping time is understood to be taken
with respect to that fixed pair unless otherwise stated.

\begin{definition}[Stopping Times] \label{def:stopping_times}
Consider quantities defined in \cref{def:key-quantities}.
\begin{itemize}
  \item (Stopping times for $\beta_t$) For a parameter $x>0$, define the stopping times
\[
  \taubetaplus(x) = \inf\qty{t\ge 0: \beta_t \ge \frac{1}{2}-x}
\text{ and }
\taubetaminus(x) = \inf\qty{t\ge 0: \beta_t < \frac{1}{2}-x}.
\]
By default, we set the parameter $x$ to be any positive function $\xbeta=\xbeta(n)$ satisfying $\xbeta(n) = \omega\qty( \sqrt{\log n/n} )$ and $\xbeta(n)=o(\log n/\sqrt{n})$ (for example, $\xbeta = (\log n)^{3/5}/\sqrt{n}$).
We sometimes abbreviate $\taubetaplus(\xbeta)$ and $\taubetaminus(\xbeta)$ to $\taubetaplus$ and $\taubetaminus$, respectively.

  \item (Stopping times for $\psi_t$) For a parameter $x>0$, 
  define the stopping times
\[
  \taupsiplus(x)=\inf\qty{t\geq 0:\psi_t> x}\text{ and }\taupsiminus(x)=\inf\qty{t\ge 0\colon \psi_t\le x}.
\]
By default, we set the parameter $x$ to be a function $x=\xpsi(n)$ satisfying $\xpsi=\xbeta/4$.
We sometimes abbreviate $\taupsiplus(\xpsi)$ and $\taupsiminus(\xpsi)$ to $\taupsiplus$ and $\taupsiminus$, respectively.

  \item (Stopping time for $\npt_t$)   For a parameter $\xgamma\in (0,1)$, 
  define the stopping time
  \begin{align*}
    \taugammaplus\defeq \inf\qty{t\geq 0:\gamma_t \geq \xgamma}.
  \end{align*}
  Throughout this paper, we fix $\xgamma = (\log n)^2/\sqrt{n}$.
  \item (Stopping times for $\npm_t$ and $\tnpm_t$)   For constants $\cmaxup,\cmaxdown\in (0,1)$, define the stopping times
  \begin{align*}
    \taumaxup &\defeq \inf\qty{t\geq 0:\npm_t\geq (1+\cmaxup)\npm_0}, \\
    \taumaxdown &\defeq \inf\qty{t\geq 0:\npm_t\leq (1-\cmaxdown)\npm_0}.
  \end{align*}
  For positive constants $\ctildemaxup,\ctildemaxdown\in (0,1)$, define the stopping times
  \begin{align*}
    \tautildemaxup&\defeq \inf\qty{t\geq 0:\tnpm_t\geq (1+\ctildemaxup)\tnpm_0}, \\
    \tautildemaxdown&\defeq \inf\qty{t\geq 0:\tnpm_t\leq (1-\ctildemaxdown)\tnpm_0}.  
  \end{align*}
  By default, we set $\cmaxup=\cmaxdown=\ctildemaxup=\ctildemaxdown=0.1$.

  \item (Stopping times for $\delta_t$) For constants $\cdeltaup,\cdeltadown\in(0,1)$, define the stopping times
      \begin{align*}
    \taudeltaup\defeq \inf\qty{t\geq 0:\delta_t\geq (1+\cdeltaup)\delta_0} \quad \text{and} \quad 
    \taudeltadown\defeq \inf\qty{t\geq 0:\delta_t\leq (1-\cdeltadown)\delta_0}.
    \end{align*}
    be the stopping times for $\delta_t$ growing up and down multiplicatively.
    The constants $\cdeltaup$ and $\cdeltadown$ are defined as appropriate within the proofs. 
    
    For a threshold parameter $x>0$, define the stopping time
    \begin{align*}
      \taudeltaplus(x)\defeq \inf\qty{t\geq 0\colon \abs{\delta_t} \geq x}.
    \end{align*}
    By default, we set the parameter $x$ to be $x=\cdeltaplus/\sqrt{n}$, where $\cdeltaplus$ is a positive constant defined as appropriate within the proofs.
    We sometimes abbreviate $\taudeltaplus(\xdelta)$ to $\taudeltaplus$, respectively.
    \item (Stopping times for $\eta_t$) 
    For a parameter $x>0$ and an opinion $j\in [k]$, define the stopping times
    \begin{align*}
      \tauetaplus(x)\defeq \inf\{t\geq 0: \eta_t(j)\geq  x\}\quad \text{and} \quad
      \tauetaminus(x)\defeq \inf\{t\geq 0: \eta_t(j)< x\}.
    \end{align*}
    By default, we set the parameter $x$ to be any positive function $x=\xeta(n)$ satisfying $\xeta(n) = \omega\qty( \xbeta(n) )$ and $\xeta(n)=o(\log n/\sqrt{n})$ (for example, $\xeta = (\log n)^{4/5}/\sqrt{n}$).
  We sometimes abbreviate $\tauetaplus(\xeta)$ and $\tauetaminus(\xeta)$ to $\tauetaplus$ and $\tauetaminus$, respectively.
\end{itemize}
\end{definition}

\section{Analysis for the Gossip Model} \label{sec:proof_gossip}
In this section we present the full proof of \cref{thm:main} for the gossip
model.  The structure of the argument follows the outline given in
\cref{sec:proof_overview}: each subsection implements one of the steps
described there, in the same order and with the corresponding technical tools.
Because every component of the overview requires its own drift and
concentration analysis, the section is necessarily long; however, its
organization mirrors the conceptual roadmap of \cref{sec:proof_overview}, and
the reader may consult that discussion to track how each ingredient fits into
the overall proof.

\subsection{Probability of Failure at the First Round} \label{sec:probability of failure at the first round}
In this subsection, we shall consider the configuration after the first synchronous update.
Specifically, we obtain the probability that the dynamics fails at the first round.

\begin{lemma} \label{lem:fail at first round}
  Consider USD in the gossip model. Let $\pbot=\Pr[\beta_1=0]$ be the probability that all vertices hold $\bot$ after the first round.
  Then, we have the following:
  \begin{enumerate}
    \item If $\beta_0=1$, then $\pbot = \prod_{i\in[k]}\qty(1-\alpha_0(i))^{n\alpha_0(i)} \le \exp(-n\gamma_0) \le \exp(-n/k)$.
    \item If $0<\beta_0<1$, then $\Pr[\beta_1=0] \le \frac{1}{n}$.
  \end{enumerate}
\end{lemma}
\begin{proof}
  By definition, we have
  \begin{align}
    \pbot &= (1-\beta_0)^{(1-\beta_0)n} \cdot \prod_{i\in[k]} (\beta_0 - \alpha_0(i))^{\alpha_0(i)n} \label{eq:pbot2}
  \end{align}
  Here, note that $(1-\beta_0)^{(1-\beta_0)n}$ amounts for the probability that all initially undecided vertices keep $\bot$ after the first round,
  and $\prod_{i\in[k]} (\beta_0 - \alpha_0(i))^{\alpha_0(i)n}$ amounts for the probability that all initially decided vertices becomes undecided after the first round.

  The first item follows by substituting $\beta_0=1$ in \cref{eq:pbot2}.
  Note that, since $\beta_0=1$, we have $1=\sum_{i\in[k]} \alpha_0(i) \le \sqrt{\sum_{i\in[k]} 1^2}\cdot \sqrt{\sum_{i\in[k]} \alpha_0(i)^2} = \sqrt{k\gamma_0}$ by the Cauchy-Schwarz inequality,
  so $\gamma_0 \ge 1/k$.

  We now prove the second item. From \cref{eq:pbot2}, we have
  \begin{align*}
    \pbot &\le (1-\beta_0)^{(1-\beta_0)n} \cdot \prod_{i\in[k]} \beta_0^{\alpha_0(i) n} \\
    &= (1-\beta_0)^{(1-\beta_0)n} \cdot \beta_0^{\beta_0 n} \\
    & = 2^{-n \mathrm{H}(\beta_0)},
  \end{align*}
  where $\mathrm{H}(x) = -x \log_2 x - (1-x) \log_2 (1-x)$ denotes the binary entropy function.
  Since $\frac{1}{n} \le \beta_0 \le 1-\frac{1}{n}$, we have $\mathrm{H}(\beta_0) \ge \mathrm{H}(1/n) \ge \frac{\log_2 n}{n}$. This proves the second item.
\end{proof}

\subsection{Behavior of the Fraction of Decided Vertices} \label{sec:behavior of beta}
From \cref{lem:fail at first round,lem:beta1_upper_bound}, we know that the fraction of undecided vertices is $0$ with probability $\pbot$ or otherwise lies between $1/n$ and $\gamma_0 \log n$ with high probability.
In this subsection we show that 
when the latter occurs, then it holds with high probability that $\beta_t \ge 1/2 - o(1)$ for all $\Omega(\log n) \le t \le n^{O(1)}$; \cref{lem:taubeta}.
In particular, this allows us to treat $\beta_t$ as $\Omega(1)$ for all such $t$.
See \cref{fig:beta} for an illustration.

In addition, we prove 
\cref{lem:taupsi} regarding the potential $\psi_t = \beta_t (2\beta_t - 1) - \gamma_t$, which is related to a lower bound of $\beta_t$.
Note that by definition, $\E_{t-1}[\beta_t] = \beta_{t-1} - \psi_{t-1}$.

\begin{figure}[htbp]
    \centering
    \includegraphics[width=0.8\textwidth]{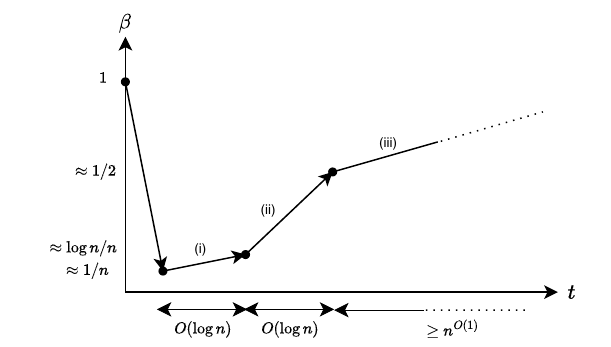}
    \caption{Behavior of $\beta_t$ over time in the gossip model.}
    \label{fig:beta}
\end{figure}


\begin{lemma}[Growth of $\beta_t$]
    \label{lem:taubeta}
    We have the following:
    \begin{enumerate}
        \item \label{item:taubeta_logn} Let $C>0$ be an arbitrary constant. 
        Suppose that $\beta_0>0$. Then, for some $T=O(\log n)$, 
        $\beta_T\geq \frac{C\log n}{n}$ with probability at least $1-O\qty(\frac{\log n}{n})$.
        \item \label{item:taubetaplus} 
        For some $T=O(\log n)$, 
      $
        \Pr\qty[ \taubetaplus > T ] \le T \qty(\exp\qty(-\Omega(n \beta_0))+\exp\qty(-\Omega(n\xbeta^2))).
      $
        \item \label{item:taubetaminus} 
        Suppose that $\psi_0 \leq \xpsi$ and $\beta_0 \ge 1/2-\xbeta$ hold.
        Then, for any $T\ge 1$, 
        \[
        \Pr\qty[ \taubetaminus \le T \text{ and } \taupsiplus > T] \le T \exp\qty(-\Omega(n \xbeta^2)).
        \]
    \end{enumerate}
\end{lemma}

\begin{lemma}[Decay of $\psi_t$]
    \label{lem:taupsi}
    We have the following:
    \begin{enumerate}
        \item \label{item:taupsiminus}
        $\Pr[\taupsiminus>1]\leq \exp(-\Omega(n\xpsi^2))$ for any initial configuration.
        \item \label{item:taupsiplus} 
        Suppose $\psi_0\leq \xpsi$. Then, for any $T\geq 1$, 
        $\Pr[\taupsiplus\leq T]\leq T\exp(-\Omega(n\xpsi^2))$.
    \end{enumerate}
\end{lemma}

\begin{proof}[Proof of \cref{item:taubeta_logn} of \cref{lem:taubeta}.]
First, for $\beta_0< C\log n/n\leq 1/4$, we have
\begin{align}
  \E_{t-1}[\beta_t] 
  &= 2\beta_{t-1}(1-\beta_{t-1}) + \gamma_{t-1} 
  \geq \frac{3}{2}\beta_{t-1}.
  \label{eq:Ebeta_1}
\end{align}
Note that we use \cref{lem:basic inequalities for beta} (\cref{item:expectation of beta}).
From \cref{lem:basic inequalities for beta} (\cref{item:bernstein condition of beta}) and \cref{remark:Bernstein condition and concentration}, 
\begin{align}
  \Pr_{t-1}\qty[ \beta_t \le \frac{5}{4}\beta_{t-1}]
  &\le \Pr_{t-1}\qty[ \beta_t \le \E_{t-1}[\beta_t]-\frac{\beta_{t-1}}{4}] 
  \le \exp\qty(-\,\frac{\frac{1}{2}\frac{\beta_{t-1}^2}{16}}{\frac{\beta_{t-1}}{n}+\frac{\beta_{t-1}}{12n}}) 
  \le \exp\qty(-\Omega(n\beta_{t-1})).
  \label{eq:one step concentration of beta}
\end{align}

\newcommand{\tauzero}{\tau_\mathrm{zero}}
Let $\tauzero=\inf\{t\ge 0: \beta_t=0\}$ be the first time when $\beta_t$ becomes $0$.
From \cref{lem:fail at first round}, we have $\Pr\qty[\tauzero\leq T]\leq T/n$.
Further, by assumption of $\beta_0>0$, we have $\tauzero > 0$.
We apply \cref{lem:nazo lemma} for $Z_t=\opn_t$, $T=1$, $\varphi(Z_t)=\sqrt{n \beta_t}$, $\tau=\tauzero$, $\cphiup=\sqrt{5/4}-1$, $x_0 = 1$, and $x^*=\sqrt{C\log n}$. 
From these settings, we have
$\tauphiplus(x_0) = \inf\{t\geq 0:\sqrt{n\beta_t}\geq 1\}=0$,
$\tauphiplus(x_*) = \inf\{t\geq 0:\sqrt{n\beta_t}\geq \sqrt{C\log n/n}\}=\inf\{t\geq 0: \beta_t\geq C\log n/n\}$, and
$\tauphiup = \inf\{t\geq 0:\sqrt{n\beta_t}\geq (1+\cphiup)\sqrt{n\beta_0}\}=\inf\{t\geq 0: \beta_t\geq (5/4)\beta_0\}$.
Hence,
\begin{align*}
\Pr\qty[\min\qty{\tauphiplus(x_0), \tau}\leq 1] = 1
\end{align*}
holds, i.e., the first condition of \cref{lem:nazo lemma} is satisfied for $C_1=1$.
Next, for any configuration of $\beta_0\leq C\log n/n$, we have
\begin{align*}
  \Pr\qty[\min\qty{\tauphiup, \tau}>1]
  \geq \Pr\qty[\beta_1\leq \frac{5}{4}\beta_0]
  \geq 1-\exp\qty(-\Omega(n\beta_0))
  =1-\exp\qty(-\Omega(\varphi(Z_0)^2)),
\end{align*}
i.e., the second condition of \cref{lem:nazo lemma} is satisfied for some positive constant $C_2>0$.
Note that we use \cref{eq:one step concentration of beta} in the second inequality.

Thus, from \cref{lem:nazo lemma} with $\varepsilon=n^{-10}$, 
for some $T'=O(\log n)$,
\begin{align*}
  \Pr\qty[\tauphiplus(x_*)>T']
  &\leq \Pr\qty[\min\qty{\tauphiplus(x_*), \tauzero}>T']+\Pr\qty[\tauzero\leq T']\\
  &\leq n^{-10}+T'/n.
\end{align*}

\end{proof}

\begin{proof}[Proof of \cref{item:taubetaplus} of \cref{lem:taubeta}.]
First, consider the case where $\beta_{t-1} \le \frac{1}{4}$.
    Recall \cref{eq:one step concentration of beta}: For any $\beta_{t-1}\leq 1/4$, 
        $\Pr_{t-1}\qty[ \beta_t \le (5/4)\beta_{t-1}]\le \exp\qty(-\Omega(n\beta_{t-1}))$.
Therefore, for some $T=O(\log n)$ rounds, 
$\beta_T \ge \frac{1}{4}$ with probability at least $1-T\exp\qty(-\Omega(n \beta_0))$.

Second, consider the case where $\frac{1}{4}\leq \beta_{t-1} \le \frac{1}{2}-\xbeta$.
Write $B_t=1-2\beta_t$ for convenience.
Recall that $\E_{t-1}[\beta_t] = 2\beta_{t-1}(1-\beta_{t-1}) + \gamma_{t-1}$.
Since $2\xbeta \leq B_{t-1}\leq 1/2$, we have
    \begin{align*}
        \E_{t-1}[B_t]
        &=1-4\beta_{t-1}(1-\beta_{t-1})-2\gamma_{t-1}
        \leq B_{t-1}^2
        \leq \frac{B_{t-1}}{2}.
    \end{align*}
    Furthermore, from $B_t-\E_{t-1}[B_t]=2(\E_{t-1}[\beta_t]-\beta_t)$, \cref{lem:Bernstein condition} (\cref{item:BC for linear transformation}), and \cref{lem:basic inequalities for beta} (\cref{item:bernstein condition of beta}), 
    $B_t-\E_{t-1}[B_t]$ satisfies $\qty(\frac{2}{n}, \frac{4\beta_{t-1}}{n})$-Bernstein condition.
    Hence,
    \begin{align*}
        \Pr_{t-1}\qty[B_t\geq \frac{3}{4}B_{t-1}]
        &\leq \Pr_{t-1}\qty[B_t\geq \E_{t-1}[B_t]+\frac{B_{t-1}}{2}]
        \leq \exp(-\Omega(nB_{t-1}^2)).
    \end{align*}
    Therefore, for some $T'=O(\log n)$ rounds, $B_{T'}\leq 2\xbeta$ holds with a probability at least $1-T_1\exp(-\Omega(n\xbeta^2))$.
    Note that $B_t\leq 2\xbeta$ implies $\beta_t\geq 1/2-\xbeta$.
    \end{proof}

\begin{proof}[Proof of \cref{item:taubetaminus} of \cref{lem:taubeta}.]
  First, consider the case where $\frac{1}{2}-\xbeta\leq \beta_{t-1}\leq \frac{1}{2}-\frac{\xbeta}{2}$.
  Note that $\frac{1}{4}\leq \frac{1}{2}-\xbeta$. 
  From \cref{lem:basic inequalities for beta} (\cref{item:expectation of beta}),
  \begin{align*}
      \E_{t-1}[\beta_t]-\beta_{t-1}
      \geq \beta_{t-1}(1-2\beta_{t-1})
      \geq \qty(\frac{1}{2}-\frac{\xbeta}{2})\xbeta
      \geq \frac{\xbeta}{4}.
  \end{align*}
  Hence,
  from \cref{lem:basic inequalities for beta} (\cref{item:bernstein condition of beta}), 
  \begin{align*}
      \Pr_{t-1}\qty[\beta_t \leq \frac{1}{2}-\xbeta]
      \leq \Pr_{t-1}\qty[\beta_{t}\leq \beta_{t-1}]
      \leq \Pr_{t-1}\qty[\beta_{t}\leq \E_{t-1}[\beta_t]-\frac{\xbeta}{4}]
      \leq \exp(-\Omega(n\xbeta^2)).
  \end{align*}

Second, consider the case where $\frac{1}{2}-\frac{\xbeta}{2}\leq \beta_{t-1}$.
Recall $\xbeta\geq 4\xpsi$.
  Then, from \cref{lem:basic inequalities for beta} (\cref{item:expectation of beta,item:bernstein condition of beta}), for $\taupsiplus>t-1$, we have
  \begin{align*}
    \Pr_{t-1}\qty[\beta_t \leq \frac{1}{2}-\xbeta]
      \leq \Pr_{t-1}\qty[\beta_{t}\leq \beta_{t-1}-\psi_{t-1}-\frac{\xbeta}{4}]
      = \Pr_{t-1}\qty[\beta_{t}\leq \E_{t-1}[\beta_t]-\frac{\xbeta}{4}]
      \leq \exp(-\Omega(n\xbeta^2)).
  \end{align*}
  Thus, we obtain
  \begin{align*}
      \Pr\qty[\taubetaminus \leq T \text{ and } \taupsiplus>T]
      &=\Pr\qty[\exists t\leq T: \beta_t\leq \frac{1}{2}-\xbeta\text{ and }\taupsiplus>T]\\
      &\leq \sum_{t=1}^T\E\qty[\indicator_{\taupsiplus>t-1}\Pr_{t-1}\qty[\beta_t \leq \frac{1}{2}-\xbeta]]\\
      &\leq T\exp(-\Omega(n\xbeta^2)).
  \end{align*}
  Note that $\Pr[\beta_t\leq \frac{1}{2}-\xbeta\text{ and }\taupsiplus>t-1]=\E\qty[\indicator_{\taupsiplus>t-1}\Pr_{t-1}\qty[\beta_t \leq \frac{1}{2}-\xbeta]]$.
\end{proof}

\begin{proof}[Proof of \cref{item:taupsiminus} of \cref{lem:taupsi}.]
  From \cref{lem:basic inequalities for psi} (\cref{item:expectation of psi}),
  $\E_{t-1}[\psi_t]+\frac{\xpsi}{2}\leq \frac{1}{n}+\frac{\xpsi}{2}\leq \xpsi$ holds.
  Further, from \cref{lem:basic inequalities for psi} (\cref{item:bernstein condition of psi}), $\psi_t-\E_{t-1}[\psi_t]$ satisfies $\qty(O(1/n),O(1/n))$-Bernstein condition.
  Hence, from \cref{remark:Bernstein condition and concentration},
  \begin{align}
    \Pr_{t-1}\qty[\psi_t> \xpsi]
    \leq \Pr_{t-1}\qty[\psi_t> \E_{t-1}[\psi_t]+\frac{\xpsi}{2}]
    \leq \exp\qty(-\Omega\qty(\frac{\xpsi^2}{\frac{1}{n}+\frac{\xpsi}{n}}))
    \leq \exp\qty(-\Omega\qty(n\xpsi^2)).
    \label{eq:psi onestep}
  \end{align}
  We obtain the claim from $\Pr\qty[\taupsiminus>1]\leq \Pr\qty[\psi_1> \xpsi]$.
\end{proof}
\begin{proof}[Proof of \cref{item:taupsiplus} of \cref{lem:taupsi}.]
      From the union bound and \cref{eq:psi onestep}, we have
      \begin{align*}
          \Pr[\taupsiplus\leq T]
          =\Pr[\bigvee_{t=1}^T \qty{\psi_t>  \xpsi}]
          \leq \sum_{t=1}^T \Pr[\psi_t>  \xpsi]
          =\sum_{t=1}^T \E\qty[\Pr_{t-1}\qty[\psi_t>\xpsi]]
          \leq  T\exp\qty(-\Omega(n\xpsi^2)).
      \end{align*}
      Note that $\taupsiplus\neq 0$ since $\psi_0\leq \xpsi$.
\end{proof}

\subsection{Behavior of Squared \texorpdfstring{$\ell^2$}{ell-2} Norm of Normalized Population} \label{sec:behavior of tnpt}
In this subsection, we track the change of the quantity $\tnpt_t=\norm{\alpha_t}_2^2/\beta_t^2$ defined in \cref{def:key-quantities}.
We show that $\tnpt_t$ becomes $\Omega((\log n)^2/\sqrt{n})$ within $\Otilde(\sqrt{n})$ rounds with high probability.
Recall the definition of $\xgamma=(\log n)^2/\sqrt{n}$ in \cref{def:stopping_times}.


\begin{lemma}[Growth of $\npt_t$]
    \label{lem:growth of tnpt}
    Suppose $\beta_0\geq 1/2-\xbeta$ and $\psi_0 \le \xpsi$.
    Then, for some $T = O(n\xgamma \log n)$, 
    \[\Pr\qty[\taugammaplus>T \text{ or } \min\{\taubetaminus,\taupsiplus\}\leq T]\leq n^{-\Omega(1)}.\]
  \end{lemma}

The intuition behind \cref{lem:growth of tnpt} is as follows.
First, we derive an upper bound on the expected value of $\tau=\min\{\taugammaplus,\taubetaminus,\taupsiplus\}$ in terms of $\E[\tnpt_\tau]$.
The key technique is to combine the additive drift of $\tnpt$ with the optional stopping theorem (\cref{thm:OST}).
Second, to obtain an upper bound on $\E[\tau]$, we provide an upper bound of $\E[\tnpt_\tau]$.

\begin{lemma}
  \label{lem:hitting time for tnpt for gossip model}
  $
  \E\qty[\min\{\taugammaplus,\taubetaminus,\taupsiplus\}]\leq 1320n\xgamma.
  $
\end{lemma}
\begin{proof}
  For $\min\{\taugammaplus,\taubetaminus,\taupsiplus\}>t-1$, we have
  \begin{align*}
    \E_{t-1}\qty[\tnpt_t]
    &\geq \tnpt_{t-1}+ \frac{1}{12n}.
  \end{align*}
  Note that we use \cref{lem:basic inequalities for tnpt}.
  Let $\tau=\min\{\taugammaplus,\taubetaminus,\taupsiplus\}$, $X_t=\tnpt_t-\frac{t}{12}$ and $Y_t=X_{t\wedge \tau}$.
  Then, we have
  \begin{align*}
    \E_{t-1}\qty[Y_t-Y_{t-1}]
    &= \indicator_{\tau>t-1}\E_{t-1}\qty[X_t-X_{t-1}]
    = \indicator_{\tau>t-1}\qty(\E_{t-1}[\tnpt_t]-\frac{t}{12}-\tnpt_{t-1}+\frac{t-1}{12})
    \geq 0,
  \end{align*}
  i.e., $(Y_t)_{t\in \Nat_0}$ is a submartingale.
  Hence, from \cref{thm:OST}, we have $\E[Y_\tau]\geq \E[Y_0] = \tnpt_0\geq 0$.
  Thus, 
  \begin{align*}
 0\leq  \E[Y_\tau]=\E[X_\tau] = \E[\tnpt_\tau] -\frac{\E[\tau]}{12n}
  \end{align*}
and we obtain $\E[\tau]\leq 12n\E[\tnpt_\tau]$.

Now, we give an upper bound on $\E[\tnpt_\tau]$.
Write $L=55\xgamma$.
    We have
    \begin{align*}
        \E\qty[\tnpt_\tau]
        &=\E\qty[\indicator_{\tnpt_\tau\leq L}\tnpt_\tau]+\E\qty[\indicator_{\tnpt_\tau> L}\tnpt_\tau]
        \leq L+\E\qty[\indicator_{\tnpt_\tau> L}\tnpt_\tau]
    \end{align*}
    and
    \begin{align*}
    \E\qty[\indicator_{\tnpt_\tau> L}\tnpt_\tau]
    &=\sum_{t=1}^\infty \E\qty[\indicator_{\tau=t}\tnpt_t\indicator_{\tnpt_t> L}]
    \leq \sum_{t=1}^\infty \E\qty[\indicator_{\tau>t-1}\tnpt_t\indicator_{\tnpt_t> L}]
    =\sum_{t=1}^\infty \E\qty[\indicator_{\tau>t-1}\E_{t-1}\qty[\tnpt_t\indicator_{\tnpt_t> L}]].
    \end{align*}
    For $t-1<\tau$, we have $\beta_{t-1}\geq 1/2-\xbeta$ and $\gamma_{t-1}\leq \xgamma$.
    Hence, we have
    \begin{align*}
        \E_{t-1}\qty[\tnpt_\tau\indicator_{\tnpt_\tau> L}]
        &\leq \Pr_{t-1}\qty[\tnpt_t> L]
        \leq \Pr_{t-1}\qty[\gamma_t> \qty(\frac{1}{2}-\xbeta)^2L]
        \leq \Pr_{t-1}\qty[\gamma_t> 11\xgamma]
        \leq \Pr_{t-1}\qty[\gamma_t> \E_{t-1}[\gamma_t]+\xgamma].
    \end{align*}
    Note that $\E_{t-1}[\gamma_t]\leq 10\gamma_{t-1}$.
    Since $\npt_t-\E_{t-1}[\npt_t]$ satisfies $\qty(\frac{2}{n}, \frac{20\gamma_{t-1}}{n})$-Bernstein condition, we obtain
    \begin{align*}
        \indicator_{\tau>t-1}\E_{t-1}\qty[\tnpt_\tau\indicator_{\tnpt_\tau> L}]
        &\leq \indicator_{\tau>t-1}\Pr_{t-1}\qty[\gamma_t> \E_{t-1}[\gamma_t]+\xgamma]\\
        &\leq \indicator_{\tau>t-1}\exp\left(-\frac{\xgamma^2/2}{\frac{20\gamma_{t-1}}{n}+\frac{2\xgamma}{3n}}\right)\\
        &\leq \indicator_{\tau>t-1}\exp\qty(-\Omega(n\xgamma))\\
        &\leq \frac{\indicator_{\tau>t-1}}{24n}.
    \end{align*}
    Consequently, 
    \begin{align*}
        \E\qty[\tnpt_\tau]
        &\leq L+\sum_{t=1}^\infty \E\qty[\frac{\indicator_{\tau>t-1}}{100n}]
        =55\xgamma+\frac{\E[\tau]}{24n}.
    \end{align*}
  Thus, we have
  \begin{align*}
      \E[\tau]\leq 12\E[\tnpt_\tau]
      \leq 12n\qty(55\xgamma+\frac{\E[\tau]}{24n})
      \leq 660n\xgamma+\frac{\E[\tau]}{2},
  \end{align*}
  i.e., $\E[\tau]\leq 1320n\xgamma$.
\end{proof}

\begin{proof}[Proof of \cref{lem:growth of tnpt}]
    Let $\tau=\min\{\taugammaplus,\taubetaminus,\taupsiplus\}$ and $T'=1320\e n^2\xgamma$.
  From the Markov inequality and \cref{lem:hitting time for tnpt for gossip model}, 
\begin{align*}
  \Pr\qty[\tau>T']
  \leq \frac{\E[\tau]}{T'}\leq \frac{1}{\e}. 
\end{align*}
Hence, for any $\ell\geq 1$, the Markov property implies that
  \begin{align*}
    \Pr\qty[\tau>\ell T' \mid \tau> (\ell-1) T']
    &=\E\qty[\Pr_{(\ell-1)T'}\qty[\tau>\ell T'] \mid \tau> (\ell-1) T']
    \leq \frac{1}{\e}.
  \end{align*}
Thus, we obtain
\begin{align*}
  \Pr\qty[\tau>\ell T']
  &=\Pr\qty[\tau>\ell T' \mid \tau> (\ell-1) T']\Pr\qty[\tau> (\ell-1) T']
  \leq \frac{1}{\e}\Pr\qty[\tau> (\ell-1) T']
  \leq \cdots \leq \frac{1}{\e^\ell}.
\end{align*}
Applying \cref{lem:taubeta} (\cref{item:taubetaminus}) and \cref{lem:taupsi} (\cref{item:taupsiplus}), we obtain 
\begin{align*}
  &\Pr\qty[\taugammaplus> \ell T' \text{ or } \min\{\taubetaminus,\taupsiplus\}\leq \ell T']\\
  &\leq \Pr\qty[\qty{\taugammaplus> \ell T' \text{ or } \min\{\taubetaminus,\taupsiplus\}\leq \ell T'} \text{ and } \min\{\taubetaminus,\taupsiplus\} > \ell T']+ \Pr\qty[\min\{\taubetaminus,\taupsiplus\}\leq \ell T']\\
  &\leq \Pr\qty[\tau>\ell T']+ \Pr\qty[\taubetaminus\leq \ell T'\text{ and } \taupsiplus > \ell T']+\Pr\qty[\taupsiplus\leq \ell T']\\
  &\leq \frac{1}{\e^\ell} + \ell T'\exp\qty(-\Omega\qty(n\xbeta^2))+\ell T'\exp\qty(-\Omega\qty(n\xpsi^2)).
\end{align*}
Taking $\ell=C\log n$ for a sufficiently large constant $C>0$, we obtain the claim.
\end{proof}

\subsection{Behavior of Maximum Population} \label{sec:behavior of tnpm}
In this subsection, we track the change of the maximum fractional population $\npm_t=\max_{i\in[k]} \alpha_t(i)$ and its normalized version $\tnpm_t = \npm_t/\beta_t$.
Note that, from \cref{lem:growth of tnpt}, we know that $\npm_t = \Omega(\tnpm_t)$ becomes $\Omega((\log n)^2/\sqrt{n})$ within $\Otilde(\sqrt{n})$ rounds with high probability.
The aim of this subsection is to show that, with high probability,
(i) $\tnpm_t \geq (1-\ctildemaxdown)\tnpm_0$ holds for all $t \leq O(n\tnpm_0/\log n)$ (see \cref{lem:hitting time for tnpm and npm}, \cref{item:tildemaxdown}), and
(ii) $\npm_t \leq (1+\cmaxup)\npm_0$ holds for all $t \leq O(1/\npm_0)$ (see \cref{lem:hitting time for tnpm and npm}, \cref{item:taualphamaxup}).

\begin{lemma}[Key properties of $\tnpm_t$ and $\npm_t$]
  \label{lem:hitting time for tnpm and npm}
  Suppose that $\psi_0\leq \xpsi$ and $\beta_0\geq 1/2-\xbeta$.
  We have the following:
  \begin{enumerate}
    \item \label{item:tildemaxdown}
    Let $\constrefs{lem:hitting time for tnpm and npm}{item:tildemaxdown}= \ctildemaxdown/36$ be a positive constant. 
    Then, for any $T\le \constrefs{lem:hitting time for tnpm and npm}{item:tildemaxdown}n$, 
    \begin{align*}
        \Pr\qty[\tautildemaxdown \le T \text{ and } \min\{\taubetaminus,\taupsiplus\}> T] \le T\exp\qty(-\Omega\qty(\frac{\tnpm_0n}{T})).
    \end{align*}
    \item \label{item:taualphamaxup}
    Suppose $\npm_0=\omega(\log n/\sqrt{n})$.
    Let $\constrefs{lem:hitting time for tnpm and npm}{item:taualphamaxup}=\frac{\cmaxup}{6(1+\cmaxup)^2}$ be a positive constant.
    Then, for any $T\leq \frac{\constrefs{lem:hitting time for tnpm and npm}{item:taualphamaxup}}{\npm_0}$,
    \begin{align*}
      \Pr\qty[\taumaxup\leq \min\{T,\taubetaminus\}]
      \leq k\exp\qty(-\Omega\qty(\frac{\npm_0n}{T})).
      \end{align*}
    \end{enumerate}
\end{lemma}

First, we introduce the following key lemma for the proof of \cref{lem:hitting time for tnpm and npm} (\cref{item:tildemaxdown}).
\begin{lemma}
    \label{lem:tildemaxdown}
    Suppose that $\psi_0\leq \xpsi$ and $\beta_0\geq 1/2-\xbeta$.
    Let $\tau^\uparrow=\max\{\taumaxup,\tautildemaxup\}$.
    Let $\constrefs{lem:hitting time for tnpm and npm}{item:tildemaxdown}= \ctildemaxdown/36$ be a positive constant defined in \cref{lem:hitting time for tnpm and npm} (\cref{item:tildemaxdown}).
    Then, for any $T\leq \constrefs{lem:hitting time for tnpm and npm}{item:tildemaxdown}n$, 
    \begin{align*}
        \Pr\qty[\tautildemaxdown \leq \min\{T,\tau^\uparrow,\taupsiplus,\taubetaminus\}]
        \leq \exp\qty(-\Omega\qty(\frac{\tnpm_0n}{T})).
      \end{align*}
\end{lemma}
\begin{proof}[Proof of \cref{lem:tildemaxdown}]
Let $\tau^\uparrow=\max\{\taumaxup,\tautildemaxup\}$
and $\tau=\min\{\tautildemaxdown,\tau^\uparrow,\taubetaminus,\taupsiplus\}$.
For $\tau>t-1$, 
\cref{item:expectation of tnpm} of \cref{lem:basic inequalities for tnpm} yields that
\begin{align*}
\E_{t-1}\qty[\tnpm_t]
\geq \tnpm_{t-1}\qty(1+\frac{\npm_{t-1}-\npt_{t-1}/\npo_{t-1}}{2(1-\npo_{t-1})+\npt_{t-1}/\npo_{t-1}})-\frac{\npm_{t-1}}{9n}
 \geq \tnpm_{t-1}-\frac{18\tnpm_{0}}{n}.
\end{align*}
Note that for $\tau^\uparrow > t-1$, we have either $\npm_{t-1} \leq (1+\cmaxup)\npm_{0} \leq 2\tnpm_{0}$ or $\npm_{t-1} \leq (1+\ctildemaxup)\tnpm_{0} \beta_{t-1} \leq 2\tnpm_{0}$, so in either case, $\npm_{t-1} \leq 2\tnpm_{0}$ holds whenever $\tau^\uparrow > t-1$.

Hence, letting $X_t=\tnpm_t$ and $\drift = -\frac{18\tnpm_{0}}{n}< 0$, we have
    \begin{align*}
      \indicator_{\tau>t-1}\qty(X_{t-1}+\drift-\E_{t-1}[X_t])
      =\indicator_{\tau>t-1}\qty(\tnpm_{t-1}-\frac{18\tnpm_{0}}{n}-\E_{t-1}[\tnpm_t])\leq 0.
    \end{align*}
    Furthermore, 
    \[\indicator_{\tau>t-1}\qty(X_{t-1}+\drift-X_t)\leq \indicator_{\tau>t-1}\qty(\tnpm_{t-1}\qty(1+\frac{\npm_{t-1}-\npt_{t-1}/\npo_{t-1}}{2(1-\npo_{t-1})+\npt_{t-1}/\npo_{t-1}})-\frac{9\npm_{t-1}}{n}-\tnpm_t)\]
     satisfies one-sided $\qty(O\qty(\frac{1}{n}),O\qty(\frac{\npm_0}{n}))$-Bernstein condition. 
  Note that we use  \cref{lem:basic inequalities for tnpm} (\cref{item:bernstein condition of tnpm}) and \cref{lem:Bernstein condition} (\cref{item:BC for upper bounded rv,item:BC for linear transformation}). 
  
    Applying \cref{lem:Useful drift lemma} (\cref{item:negative drift useful}) with $I^-=(1-\ctildemaxdown)\tnpm_{0}$, $I^+=(1+\ctildemaxup)\tnpm_{0}$, and $T\leq \constrefs{lem:hitting time for tnpm and npm}{item:tildemaxdown}n= \frac{\ctildemaxdown}{36}n= \frac{X_0-I^-}{-2\drift}$, we have
             \begin{align*}
              \Pr\qty[\tautildemaxdown \leq \min\{T,\tau^\uparrow,\taubetaminus,\taupsiplus\}] 
              \leq \exp\qty(-\Omega\qty(\frac{(\tnpm_0)^2}{\frac{\tnpm_0}{n}T+\frac{\tnpm_0}{n}}))
              \le \exp\qty(- \Omega\qty(\frac{\tnpm_0 n}{T})).
             \end{align*}
\end{proof}

  \begin{proof}[Proof of \cref{lem:hitting time for tnpm and npm} (\cref{item:tildemaxdown})]
    For $s\geq 0$, let $\tau_s^\downarrow=\inf\{t\geq s: \tnpm_t\leq (1-\ctildemaxdown)\tnpm_{s}\}$ and $\tau_s^\uparrow=\inf\{t\geq s: \tnpm_t\geq (1+\ctildemaxup)\tnpm_{s}\}$.
    Applying \cref{lem:Iterative Drift theorem} (\cref{item:bounded decrease under small negative drift}) with $\tau^*=\min\{\taubetaminus,\taupsiplus\}$ and $T\leq \constrefs{lem:hitting time for tnpm and npm}{item:tildemaxdown}n$, we have
    \begin{align*}
      \Pr\qty[\tautildemaxdown \leq T \text{ and } \tau^*> T]
      &\leq 
      \sum_{s=0}^{T-1}\E\qty[\indicator_{\tnpm_s\geq \tnpm_0 \text{ and } \tau^*> s}\Pr_s\qty[\tau_s^\downarrow \leq \min\{T,\tau_s^\uparrow,\tau^*\}]]
       \\
       &\leq \sum_{s=0}^{T-1}\E\qty[\indicator_{\tnpm_s\geq \tnpm_0 \text{ and } \tau^*> s}\exp\qty(-\Omega\qty(\frac{\tnpm_sn}{T}))]
       \\
       &\leq T\exp\qty(-\Omega\qty(\frac{\tnpm_0n}{T})).
    \end{align*}
    Note that we apply \cref{lem:tildemaxdown} in the second inequality.
  \end{proof}

\begin{proof}[Proof of \cref{lem:hitting time for tnpm and npm} (\cref{item:taualphamaxup})]
    Let $\tau=\min\{\taumaxup,\taubetaminus\}$.
  For $\tau>T$, we have
  \begin{align}
  \E_{t-1}[\alpha_t(i)]-\alpha_{t-1}(i)
  &=\alpha_{t-1}(i)\qty(\alpha_{t-1}(i)+1-2\beta_{t-1}) \nonumber\\
  &\leq \npm_{t-1}\qty(\npm_{t-1}+2\xbeta) \nonumber\\
  &\leq 3(1+\cmaxup)^2(\npm_{0})^2. \label{eq:alphamaxup drift}
  \end{align}
Note that we use assumptions on $\npm_0=\omega(\log n/\sqrt{n})$ and $\xbeta=o(\log n/\sqrt{n})$ in the second inequality.

  Hence, letting $X_t=-\alpha_t(i)$ and $\drift = -3(1+\cmaxup)^2(\npm_{0})^2<0$, we have
  \[
    \indicator_{\tau>t-1}\qty(X_{t-1}+\drift-\E_{t-1}[X_t])
    =\indicator_{\tau>t-1}\qty(\E_{t-1}[\alpha_t(i)]-\alpha_{t-1}(i)-3(1+\cmaxup)^2(\npm_{0})^2)\leq 0
  \]
  and $\indicator_{\tau>t-1}\qty(\E_{t-1}[X_t]-X_t)=\indicator_{\tau>t-1}\qty(\alpha_t(i)-\E_{t-1}[\alpha_t(i)])$ satisfies $\qty(1/n,O(\npm_0/n))$-Bernstein condition.
  Note that we use  \cref{lem:basic inequalities for alpha} (\cref{item:bernstein condition of alpha}) and \cref{lem:Bernstein condition} (\cref{item:BC for upper bounded rv,item:BC for linear transformation}). 

  Let $\tau_i=\inf\{t\geq 0: \alpha_t(i)\geq (1+\cmaxup)\npm_0\}=\inf\{t\geq 0: X_t\leq -(1+\cmaxup)\npm_0\}$.
  Applying \cref{lem:Useful drift lemma} (\cref{item:negative drift useful}) with 
 $I^-=-(1+\cmaxup)\npm_0$, for $T\leq \frac{\constrefs{lem:hitting time for tnpm and npm}{item:taualphamaxup}}{\npm_0}= \frac{\cmaxup}{6(1+\cmaxup)^2\npm_0}= \frac{X_0-I^-}{-2\drift}$, we have
 \begin{align*}
  \Pr\qty[\tau_i\leq \min\{T,\tau\}]
  \leq \exp\qty(-\Omega\qty(\frac{(\npm_0)^2}{\frac{\npm_0}{n}T+\frac{\npm_0}{n}}))
  \leq \exp\qty(-\Omega\qty(\frac{\npm_0 n}{T})).
 \end{align*}
 Thus, applying the union bound,
 \begin{align*}
  \Pr\qty[\taumaxup\leq \min\{T,\tau\}]
  &\leq \Pr\qty[\exists i\in [k]: \tau_i\leq \min\{T,\tau\}]
  \leq k\exp\qty(-\Omega\qty(\frac{\npm_0 n}{T})).
 \end{align*}
\end{proof}

Finally, as a natural consequence of the discussion so far, we introduce the following useful lemma.
\begin{lemma}
  \label{lem:tauinitial lemma}
  Suppose that $\beta_0\geq 1/2-\xbeta$ and $\psi_0\leq \xpsi$. 
  Then, for any $T\leq \constrefs{lem:hitting time for tnpm and npm}{item:tildemaxdown}n$,
  \begin{align*}
    \Pr\qty[\min\{\tautildemaxdown,\taubetaminus,\taupsiplus\} \leq  T] \leq T\exp\qty(-\Omega\qty(\frac{n\tnpm_0}{T}))+T\exp\qty(-\Omega(n\xpsi^2)).
  \end{align*}
\end{lemma}
\begin{proof}
  First, we observe that, for any stopping times $\tau_1$ and $\tau_2$, we have
  \begin{align}
    \Pr\qty[\min\{\tau_1,\tau_2\}\leq T]
    &\leq \Pr\qty[\{\tau_1 \leq T \text{ or } \tau_2\leq T\}\text{ and } \tau_2>T]+\Pr\qty[\tau_2\leq T] \nonumber\\
    &= \Pr\qty[\tau_1 \leq T \text{ and } \tau_2>T]+\Pr\qty[\tau_2\leq T]. \label{eq:min of two stopping times}
  \end{align}
  Hence, combining \cref{lem:hitting time for tnpm and npm} (\cref{item:tildemaxdown}), \cref{lem:taubeta}, and \cref{lem:taupsi}, 
  \begin{align*}
    &\Pr\qty[\min\{\tautildemaxdown,\taubetaminus,\taupsiplus\} \leq T] \\
    &\leq \Pr\qty[\tautildemaxdown \leq T \text{ and }\min\{\taubetaminus,\taupsiplus\}>T]+\Pr\qty[\taubetaminus\leq T \text{ and }\taupsiplus>T]+\Pr\qty[\taupsiplus\leq T] \nonumber\\
    &\leq T\exp\qty(-\Omega\qty(\frac{n\tnpm_0}{T}))
    +T\exp\qty(-\Omega\qty(n\xbeta^2))+T\exp\qty(-\Omega\qty(n\xpsi^2)).
  \end{align*}
\end{proof}

\subsection{Behavior of Gap between Two Opinions} \label{sec:behavior of delta}
For two opinions $i,j\in [k]$, recall $\delta_t=\delta_t(i,j)=\alpha_t(i)-\alpha_t(j)$ is the gap between the two opinions at time $t$ (\cref{def:key-quantities}).
We define the \emph{weak} opinion and its stopping time as follows:
\begin{definition}[Weak Opinion]\label{def:weak}
  For a constant $\cweak\in (0,1/2)$ and an opinion $i\in [k]$, define
  \begin{align*}
    \tauiweak \defeq \inf\{t\geq 0: \alpha_t(i)\le (1-\cweak)\npm_t\}.
  \end{align*}
  By default, we set $\cweak=0.1$.
  We call that an opinion $i\in [k]$ is \emph{weak} at time $t$ if $\alpha_t(i)\le (1-\cweak)\npm_t$.
\end{definition}

The main results of this section are the following two lemmas. 
Intuitively, they show that if both $i$ and $j$ are non-weak, then: 
(i) the gap $\delta_t(i,j)$ between them grows to at least $\Omega(\sqrt{\log n/n})$ within $O(\log n/\npm_0)$ rounds (\cref{lem:pair of non-weak opinions} (\cref{item:pair of non-weak opinions becomes weak})), and 
(ii) if the initial gap $\delta_0(i,j)$ is at least $\Omega(\sqrt{\log n/n})$, then $j$ becomes weak within $O(\log n/\npm_0)$ rounds (\cref{lem:pair of non-weak opinions} (\cref{item:delta multipricative with initial bias})).

\begin{lemma}[Either of two non-weak opinions becomes weak]
  \label{lem:pair of non-weak opinions}
  Let $i,j\in [k]$ be an arbitrary pair of two non-weak opinions.
  Suppose that $\beta_0 \geq 1/2-\xbeta$, $\psi_0 \leq \xpsi$, and $\npm_0\geq \omega(\log n/\sqrt{n})$.
  We have the following:
  \begin{enumerate}
    \item \label{item:pair of non-weak opinions becomes weak}
    Let $C$ be an arbitrary positive constant.
    Then, for some $T=O(\log n/\npm_0)$, we have
    \begin{align*}
        \Pr\qty[\min\qty{\taudeltaplus\qty(\sqrt{\frac{C\log n}{n}}),\tauiweak,\taujweak} > T \text{ or } \min\{\tautildemaxdown,\taubetaminus,\taupsiplus\} \leq T] \le n^{-10}.
    \end{align*}
    \item \label{item:delta multipricative with initial bias}
    Suppose $\delta_0(i,j)\geq \sqrt{C\log n/n}$ for a sufficiently large constant $C>0$.
  Then, for some $T=O\qty(\log n/\npm_0)$,
  \begin{align*}
      \Pr\qty[ \taujweak>T \text{ or } \min\{\tautildemaxdown,\taubetaminus,\taupsiplus\} \leq T  ] \le n^{-10}.
  \end{align*}
  \end{enumerate}
\end{lemma}

To this end, we begin with proving that $\delta_t$ employs both multiplicative and additive drifts as long as both $i$ and $j$ are non-weak.
\begin{lemma}[Multiplicative and additive drifts of $\delta_t$]
  \label{lem:delta multiplicative and additive drift}
  Let $i,j\in [k]$ be an arbitrary pair of two non-weak opinions.
  Suppose $\psi_0\leq \xpsi$, $\beta_0\geq 1/2-\xbeta$, and $\npm_0=\omega(\log n/\sqrt{n})$.
  Let $\constrefs{lem:hitting time for tnpm and npm}{item:taualphamaxup}=\frac{\cmaxup}{6(1+\cmaxup)^2}$ be a positive constant defined in \cref{lem:hitting time for tnpm and npm} (\cref{item:taualphamaxup}).
  We have the following:
  \begin{enumerate}
    \item \label{item:taudeltaup}
  Suppose $\delta_0(i,j)\geq 0$. 
  Let $\cdeltaup=\constrefs{lem:hitting time for tnpm and npm}{item:taualphamaxup}\frac{(1-\ctildemaxdown)\qty(1-2\cweak)(1-\cdeltadown)}{12}$.
  Then,
 \begin{align*}
  \Pr\qty[\min\{\tau_\delta^\uparrow,\taujweak,\tautildemaxdown,\taubetaminus,\taupsiplus\} >\frac{\constrefs{lem:hitting time for tnpm and npm}{item:taualphamaxup}}{\npm_0}]\
  \leq \exp\qty(-\Omega\qty(n\delta_0(i,j)^2))+n \exp\qty(-\Omega\qty(n(\npm_0)^2)).
\end{align*}
    \item \label{item:taudeltaplus const prob}
    Let $\xdelta=\cdeltaplus/\sqrt{n}$ for $\cdeltaplus=\frac{\cmaxup(1-2\cweak)^2(1-\ctildemaxdown)}{6\cdot 96\cdot 64(1+\cmaxup)^2(1-\cweak)}$.
    Then, there exists a positive constant $c_*\in (0,1)$ such that
    \begin{align*}
        \Pr\qty[\min\qty{\taudeltaplus,\tauiweak,\taujweak,\tautildemaxdown,\taubetaminus,\taupsiplus} > \frac{\constrefs{lem:hitting time for tnpm and npm}{item:taualphamaxup}}{\npm_0}] \leq 1-c_*.
    \end{align*}
  \end{enumerate}
\end{lemma}


\paragraph{Multiplicative Drift: Proof of \cref{lem:delta multiplicative and additive drift} (\cref{item:taudeltaup}).}
This part is devoted to proving that $\delta_t$ grows by a constant factor within $O(1/\npm_0)$ rounds.
\begin{proof}[Proof of \cref{lem:delta multiplicative and additive drift} (\cref{item:taudeltaup})]
Let $\tau^*=\min\{\taujweak,\taumaxup,\tautildemaxdown,\taubetaminus,\taupsiplus\}$ and $\tau=\min\{\taudeltaup,\taudeltadown,\tau^*\}$.
For $\tau>t-1$, we have
\begin{align}
    &\beta_{t-1}\tnpm_{t-1}\qty(\frac{\alpha_{t-1}(i)+\alpha_{t-1}(j)}{\npm_{t-1}}-\frac{\gamma_{t-1}}{\beta_{t-1}\npm_{t-1}}-\frac{\psi_{t-1}}{\beta_{t-1}^2\tnpm_{t-1}}) \nonumber \\
    &\geq \frac{1-\ctildemaxdown}{3}\tnpm_0\qty(2(1-\cweak)-1-\underbrace{\frac{\xpsi}{(1/3)^2(1-\ctildemaxdown)\tnpm_{0}}}_{o(1)}) \nonumber \\
    &\geq \frac{(1-\ctildemaxdown)(1-2\cweak)}{6}\npm_0. \label{eq:drift of delta key}
\end{align}
Note that $\gamma_{t-1}\leq \beta_{t-1}\npm_{t-1}$ and $\tnpm_0\geq \npm_0=\omega(\sqrt{\log n/n})$.
Hence, for $\tau>t-1$, we have
\begin{align*}
    \E_{t-1}[\delta_t]
    &=\delta_{t-1}+\delta_{t-1}\beta_{t-1}\tnpm_{t-1}\qty(\frac{\alpha_{t-1}(i)+\alpha_{t-1}(j)}{\npm_{t-1}}-\frac{\gamma_{t-1}}{\beta_{t-1}\npm_{t-1}}-\frac{\psi_{t-1}}{\beta_{t-1}^2\tnpm_{t-1}})\\
    &\geq \delta_{t-1}+\frac{(1-\ctildemaxdown)(1-\cdeltadown)(1-2\cweak)}{6}\delta_{0} \npm_0.
\end{align*}
Note that we use \cref{lem:basic inequalities for delta_epsilon} (\cref{item:expectation of delta_epsilon}).

Letting $X_t=\delta_t$ and $\drift = \frac{(1-\ctildemaxdown)(1-\cdeltadown)(1-2\cweak)}{6}\delta_{0} \npm_0> 0$, we have
  \[
    \indicator_{\tau>t-1}\qty(X_{t-1}+\drift-\E_{t-1}[X_t])
    =\indicator_{\tau>t-1}\qty(\delta_{t-1}+\drift-\E_{t-1}[\delta_t])\leq 0
  \]
  and $\indicator_{\tau>t-1}\qty(\E_{t-1}[X_t]-X_t)=\indicator_{\tau>t-1}\qty(\E_{t-1}[\delta_t]-\delta_t)$ satisfies $\qty(O(1/n),O(\npm_0/n))$-Bernstein condition.
  Note that we use \cref{lem:basic inequalities for delta_epsilon} (\cref{item:bernstein condition for delta_epsilon}) and \cref{lem:Bernstein condition} (\cref{item:BC for upper bounded rv,item:BC for linear transformation}). 

Recall $\cdeltaup=\constrefs{lem:hitting time for tnpm and npm}{item:taualphamaxup}\frac{(1-\ctildemaxdown)\qty(1-2\cweak)(1-\cdeltadown)}{12}$.
Let $I^+=(1+c)\delta_0$ and $I^-=(1-c)\delta_0$.
Let $T=\frac{\constrefs{lem:hitting time for tnpm and npm}{item:taualphamaxup}}{\npm_0}$.
We have $T= \frac{12\cdeltaup}{(1-\ctildemaxdown)\qty(1-2\cweak)(1-\cdeltadown)} \cdot \frac{1}{\npm_0} = \frac{2(I^+-\delta_0)}{\drift}$.
Thus, we can apply \cref{lem:Useful drift lemma} (\cref{item:positive drift useful}) and obtain
  \begin{align}
    \Pr\qty[\min\{\tau_\delta^\uparrow,\tau^*\} >T]\
    \leq \exp\qty(-\Omega\qty(n\delta_0^2)). 
    \label{eq:taudeltaup lemma}
  \end{align}

  Applying \cref{lem:hitting time for tnpm and npm} (\cref{item:taualphamaxup}), we have
  \begin{align*}
    &\Pr\qty[\min\{\taudeltaup,\taujweak,\tautildemaxdown,\taupsiplus\}>T \text{ and } \taubetaminus>T]\\
    &\leq \underbrace{\Pr\qty[\min\{\taudeltaup,\taujweak,\tautildemaxdown,\taupsiplus\}>T \text{ and } \taubetaminus>T \text{ and } \taumaxup>T]}_{\cref{eq:taudeltaup lemma}}
    +\underbrace{\Pr\qty[\taumaxup\leq T \text{ and } \taubetaminus>T]}_{\cref{lem:hitting time for tnpm and npm} (\cref{item:taualphamaxup})}\\
    &\leq \exp\qty(-\Omega\qty(n\delta_0^2)) + n \exp\qty(-\Omega(n(\npm_0)^2)).
  \end{align*}
\end{proof}

\paragraph{Additive Drift: Proof of \cref{lem:delta multiplicative and additive drift} (\cref{item:taudeltaplus const prob}).}
Next, we prove that $\abs{\delta_t}$ grows at least $\Omega(1/\sqrt{n})$ within $O(1/\npm_0)$ rounds with constant probability.
To prove \cref{lem:delta multiplicative and additive drift} (\cref{item:taudeltaplus const prob}), we first prove the following lemma.

\begin{lemma}
    \label{lem:taudeltaplus}
    Suppose that $\beta_0\geq 1/2-\xbeta$, $\psi_0\leq \xpsi$, and $\npm_0\geq \omega(\log n/\sqrt{n})$.
    Let $\tau=\min\{\taudeltaplus,\tauiweak,\taujweak,\tautildemaxdown,\taubetaminus,\taupsiplus\}$.
    Let $\constref{lem:taudeltaplus}=\frac{(1-2\cweak)^2(1-\ctildemaxdown)}{96(1-\cweak)}$.
    Then, we have
    \begin{align*}
        \E[\tau]\leq \frac{n\E[\delta_\tau^2]}{\constref{lem:taudeltaplus}\npm_0}.
    \end{align*}
  \end{lemma}
  \begin{proof}
    First, for $\min\{\tauiweak,\taujweak\}>t-1$, we observe that
    \begin{align}
        \npm_{t-1}\leq \frac{\min\{\alpha_{t-1}(i),\alpha_{t-1}(j)\}}{1-\cweak}\leq \frac{1}{2(1-\cweak)}
        \label{eq:weak npm upper bound}
      \end{align}
      holds.
      Hence, for $\min\{\taubetaminus, \taupsiplus,\tauiweak,\taujweak\}>t-1$, we have
        \begin{align*}
            1-\beta_{t-1}=\frac{1}{2}-\frac{\psi_{t-1}+\gamma_{t-1}}{2\beta_{t-1}}
            \geq \frac{1}{2}-6\xpsi-\frac{\npm_{t-1}}{2}
            \geq \frac{1}{2}-6\xpsi-\frac{1}{4(1-\cweak)}
            \geq \frac{1-2\cweak}{8(1-\cweak)}. 
        \end{align*}
        Note that we use our assumption of $\xpsi=o(1)$.
        Further, for $\min\{\taubetaminus, \taupsiplus,\tauiweak,\taujweak\}>t-1$, 
\begin{align}
    &\alpha_{t-1}(i)+\alpha_{t-1}(j)+1-2\beta_{t-1}\nonumber \\
    &=\beta_{t-1}\tnpm_{t-1}\qty(\frac{\alpha_{t-1}(i)+\alpha_{t-1}(j)}{\npm_{t-1}}-\frac{\gamma_{t-1}}{\beta_{t-1}\npm_{t-1}}-\frac{\psi_{t-1}}{\beta_{t-1}^2\tnpm_{t-1}}) \nonumber \\
    &\geq \frac{1-\ctildemaxdown}{3}\tnpm_0\qty(2(1-\cweak)-1-\underbrace{\frac{\xpsi}{(1/3)^2(1-\ctildemaxdown)\tnpm_{0}}}_{o(1)}) \geq 0. \label{eq:drift of delta plus key}
\end{align}
Note that $\gamma_{t-1}\leq \beta_{t-1}\npm_{t-1}$ and $\tnpm_0\geq \npm_0=\omega(\sqrt{\log n/n})$.
        
Let $\tau=\min\{\taudeltaplus,\tauiweak,\taujweak,\tautildemaxdown,\taubetaminus,\taupsiplus\}$.
For $t-1<\tau$, 
        \begin{align*}
            \E_{t-1}[\delta_t^2]
            &=\E_{t-1}[\delta_t]^2+\Var_{t-1}[\delta_t]\\
            &\geq \delta_{t-1}^2\qty(1+\alpha_{t-1}(i)+\alpha_{t-1}(j)+1-2\beta_{t-1})^2
            +\frac{(1-\beta_{t-1})^2}{n}(\alpha_{t-1}(i)+\alpha_{t-1}(j))\\
            &\geq \delta_{t-1}^2
            +\frac{(1-2\cweak)^2}{32(1-\cweak)n}\npm_{t-1}\\
            &\geq \delta_{t-1}^2+\frac{(1-2\cweak)^2(1-\ctildemaxdown)}{96(1-\cweak)} \cdot \frac{\npm_0}{n}.
        \end{align*}
        Note that we use \cref{lem:basic inequalities for delta_epsilon} (\cref{item:expectation of delta_epsilon,item:variance of delta}).

  Let 
  $\drift=\constref{lem:taudeltaplus} \frac{\npm_0}{n}$, 
  $X_t=\delta_t^2-\drift t$ and $Y_t=X_{t\wedge \tau}$.
  Then, we have
  \begin{align*}
  \E_{t-1}[Y_t-Y_{t-1}]
  =\indicator_{\tau>t-1}\E_{t-1}[X_t-X_{t-1}]
  =\indicator_{\tau>t-1}\qty(\E_{t-1}[\delta_t^2]-\drift t-\delta_{t-1}^2+\drift (t-1))
  \geq 0,
  \end{align*}
  i.e., $(Y_t)_{t\in \Nat_0}$ is a submartingale.
  From \cref{thm:OST}, we have
  \begin{align*}
  \E[Y_\tau]\geq \E[Y_0] = \delta_0^2\geq 0
  \end{align*}
  and
  \begin{align*}
  \E[Y_\tau]=\E[X_\tau] = \E[\delta_\tau^2]-\drift \E[\tau].
  \end{align*}
  Thus, we obtain $\E[\tau]\leq \frac{\E[\delta_\tau^2]}{\drift}$.
  \end{proof}

    \begin{lemma}
        \label{lem:taudeltaplus bounded jump}
        Let $c\in (0,1)$ be an arbitrary constant and let $\xdelta=c/\sqrt{n}$.
        Let $\constref{lem:taudeltaplus bounded jump}$ be a sufficiently large positive constant.
        Suppose $\psi_0\leq \xpsi$ and $\beta_0\geq 1/2-\xbeta$.
        Suppose $\npm_0=\omega(\log n/\sqrt{n})$ and $\npm_0\leq 16c^2/\constref{lem:taudeltaplus bounded jump}$.
        Let $\tau=\min\{\taudeltaplus,\taubetaminus,\taupsiplus,\tautildemaxdown,\tauiweak,\taujweak,\taumaxup\}$.
        Then, \[\E[\delta_\tau^2]\leq 16\xdelta^2+\frac{\constref{lem:taudeltaplus}\npm_0\E[\tau]}{2n},\] where $\constref{lem:taudeltaplus}$ is the constant defined in \cref{lem:taudeltaplus}.
    \end{lemma}
    \begin{proof}
    Write $L=16c^2/n$ and $\drift=\constref{lem:taudeltaplus}\npm_0/n$.
    We have
    \begin{align*}
    \E[\delta_\tau^2]
    &=\E[\indicator_{\delta_\tau^2\leq L}]+\E[\indicator_{\delta_\tau^2>L}]
    \leq L+\E[\indicator_{\delta_\tau^2>L}]
    \end{align*}
    and
    \begin{align*}
        \E[\indicator_{\delta_\tau^2>L}]
        &=\sum_{t=1}^\infty\E[\indicator_{\tau=t}\delta_t^2\indicator_{\delta_t^2>L}]
        \leq \sum_{t=1}^\infty\E[\indicator_{\tau>t-1}\delta_t^2\indicator_{\delta_t^2>L}]
        =\sum_{t=1}^\infty\E\qty[\indicator_{\tau>t-1}\E_{t-1}[\delta_t^2\indicator_{\delta_t^2>L}]].
    \end{align*}
    In the following, we show
    \begin{align}
        \indicator_{\tau>t-1}\E_{t-1}[\delta_t^2\indicator_{\delta_t^2>L}]\leq \frac{\drift}{2} \label{eq:key delta jump}.
    \end{align}
Note that by establishing this inequality, we have
\[
\E[\delta_\tau^2]
\leq L + \frac{\drift}{2}\sum_{t= 1}^\infty\E[\indicator_{\tau>t-1}]
= L + \frac{\drift}{2}\E[\tau],
\]
which proves the claim.

First, we observe that
\begin{align*}
\E_{t-1}[\delta_t^2\indicator_{\delta_t^2>L}]
&=\int_0^1\Pr_{t-1}[\delta_t^2\indicator_{\delta_t^2>L}>y]dy
=\int_0^1\Pr_{t-1}[\delta_t^2>(y\vee L)]dy
=\int_0^1\Pr_{t-1}\qty[\abs{\delta_t}>\sqrt{y\vee L}]dy
\end{align*}
holds.
For $\taudeltaplus>t-1$, we have
\begin{align*}
\abs{\E_{t-1}[\delta_t]}
&=\abs{\delta_{t-1}\qty(\alpha_{t-1}(i)+\alpha_{t-1}(j)+2(1-\beta_{t-1}))}
\leq 2\abs{\delta_{t-1}} \leq 2\xdelta
\leq \frac{\sqrt{L}}{2}
\leq \frac{\sqrt{y\vee L}}{2}.
\end{align*}
Hence, for $\tau>t-1$, we have
\begin{align*}
    \Pr_{t-1}\qty[\abs{\delta_t}>\sqrt{y\vee L}]
    &\leq \Pr_{t-1}\qty[\abs{\delta_t-\E_{t-1}[\delta_t]}+\frac{\sqrt{y\vee L}}{2}>\sqrt{y\vee L}]\\
    &\leq 2\exp\qty(-\frac{(y\vee L)/4}{\frac{2(\alpha_{t-1}(i)+\alpha_{t-1}(j))}{n}+\frac{\sqrt{y\vee L}}{3n}})\\
    &\leq 2\exp\qty(-\frac{(y\vee L)n/4}{4(1+\cmaxup)\npm_0+\sqrt{y\vee L}/3})\\
    &\leq 2\exp\qty(-\frac{(y\vee L)n}{32(1+\cmaxup)\npm_0})+2\exp\qty(-\frac{3\sqrt{y\vee L}n}{8}).
\end{align*}
 Note that the random variable $\delta_t-\E_{t-1}[\delta_t]$ conditioned on round $t-1$ satisfies $\qty(\frac{2}{n},\frac{2(\alpha_{t-1}(i)+\alpha_{t-1}(j))}{n})$-Bernstein condition.
Then, 
\begin{align*}
\int_0^1 \exp\qty(-\frac{3n}{8}\sqrt{y\vee L}) dy
&=L\exp\qty(-\frac{3n}{8}\sqrt{L})+\int_{L}^1 \exp\qty(-\frac{3n}{8}\sqrt{y}) dy\\
&\leq L\exp\qty(-\frac{3n}{8}\sqrt{L})+
2\frac{(3/8)n\sqrt{L}+1}{(3/8)^2n^2}\exp\qty(-\frac{3n}{8}\sqrt{L})\\
&\leq \exp\qty(-\Omega(\sqrt{n})).
\end{align*}
Note that $\int \exp(-a\sqrt{x})dx = -\frac{2\exp(-a\sqrt{x})\qty(a\sqrt{x}+1)}{a^2}$.
Furthermore, we have
\begin{align*}
&\int_0^1\exp\qty(-\frac{(y\vee L)n}{32(1+\cmaxup)\npm_0})dy\\
&=L\exp\qty(-\frac{Ln}{32(1+\cmaxup)\npm_0})+\int_L^1\exp\qty(-\frac{yn}{32(1+\cmaxup)\npm_0})dy\\
&\leq L\exp\qty(-\frac{Ln}{32(1+\cmaxup)\npm_0})
+\frac{32(1+\cmaxup)\npm_0}{n}\exp\qty(-\frac{Ln}{32(1+\cmaxup)\npm_0})\\
&\leq 2L\exp\qty(-\frac{Ln}{32(1+\cmaxup)\npm_0}).
\end{align*}
Note that $\int \exp(-ax)dx = -\frac{\exp(-ax)}{a}$.
The last inequality follows from $\frac{\npm_0}{n}\leq \frac{16c^2}{n\constref{lem:taudeltaplus bounded jump}}$ holds for sufficiently large $\constref{lem:taudeltaplus bounded jump}$.

Consequently, for $\tau>t-1$, 
\begin{align*}
\E_{t-1}[\delta_t^2\indicator_{\delta_t^2>L}]
&\leq \drift \qty( \frac{n}{\constref{lem:taudeltaplus}\npm_0} \exp\qty(-\Omega(\sqrt{n})) 
+ \frac{64c^2}{\constref{lem:taudeltaplus}\npm_0}\exp\qty(-\frac{c^2}{2(1+\cmaxup)\npm_0}) )\\
&\leq \drift/2
\end{align*}
holds since $\frac{c^2}{\npm_0}\geq \frac{\constref{lem:taudeltaplus bounded jump}}{16}$ holds for sufficiently large $\constref{lem:taudeltaplus bounded jump}$.
\end{proof}

\begin{lemma}
    \label{lem:delta CLT}
    Let $c_1, c_2\in (0,1)$ be arbitrary constants and let $\xdelta=c_1/\sqrt{n}$.
    Suppose that $\psi_0\leq \xpsi$, $\beta_0\geq 1/2-\xbeta$, $\min\{\alpha_0(i),\alpha_0(j)\}\geq (1-\cweak)\npm_0$ and $\npm_0\geq c_2$ hold.
    Then,
    \begin{align*}
        \Pr[\taudeltaplus>1]\leq 1-c_*
    \end{align*}
    for some $c_*\in (0,1)$ depending only on $c_1$, $c_2$, and $\constref{lem:taudeltaplus}=\frac{(1-2\cweak)^2(1-\ctildemaxdown)}{96(1-\cweak)}$.
\end{lemma}
\begin{proof}
Since $n\delta_1=\sum_{v\in V}(\indicator_{\opn_1(v)=i}-\indicator_{\opn_1(v)=j})$ is the sum of $n$ independent random variables, 
$\lim_{n\to \infty}\Pr\qty[\frac{n\delta_1-\E[n\delta_1]}{\sqrt{\Var[n\delta_1]}}\leq x]=\Phi(x)$ holds from the central limit theorem.
Here, $\Phi(x)=\int_{-\infty}^x \frac{1}{\sqrt{2\pi}}\exp\qty(-\frac{y^2}{2})dy$ is the cumulative distribution function of the standard normal distribution.
Noting
that
\begin{align*}
    \abs{\delta_1}
    &\geq \abs{\delta_1-\E[\delta_1]}-\abs{\E[\delta_1]}
    =\abs{\delta_1-\E[\delta_1]}-\abs{\delta_0(\alpha_1(i)+\alpha_1(j)+2(1-\beta_0))}
    \geq \abs{\delta_1-\E[\delta_1]}-2\xdelta
\end{align*}
and
\begin{align*}
\Var[\delta_1]\geq \frac{\constref{lem:taudeltaplus}\npm_0}{n}\geq \frac{\constref{lem:taudeltaplus}c_2}{n}
\end{align*}
hold, we obtain
\begin{align*}
\Pr[\taudeltaplus>1]
&=\Pr\qty[\taudeltaplus>1 \text{ and } \abs{\delta_1}<\xdelta]\\
&\leq \Pr\qty[\abs{\delta_1-\E[\delta_1]}<3\xdelta\text{ and } \taudeltaplus>1 ]  \\
&\leq \Pr\qty[\abs{\frac{\delta_1-\E[\delta_1]}{\sqrt{\Var[\delta_1]}}}<\frac{3\xdelta}{\sqrt{\Var[\delta_1]}}]\\
&\leq \Pr\qty[\abs{\frac{n\delta_1-\E[n\delta_1]}{\sqrt{\Var[n\delta_1]}}}<3\sqrt{\frac{c_1^2}{\constref{lem:taudeltaplus}c_2}}]\\
&\leq 1-2\Psi\qty(-3\sqrt{\frac{c_1^2}{\constref{lem:taudeltaplus}c_2}})-o(1)\\
&\leq 1-c_*.
\end{align*}
\end{proof}

\begin{proof}[Proof of \cref{lem:delta multiplicative and additive drift} (\cref{item:taudeltaplus const prob})]
Note that $(\cdeltaplus)^2=\constrefs{lem:hitting time for tnpm and npm}{item:taualphamaxup}\frac{\constref{lem:taudeltaplus}}{64}$ from definition.
First, consider the case where $\npm_0\geq 16(\cdeltaplus)^2/\constref{lem:taudeltaplus bounded jump}$.
For convenience, let $\tau=\min\{\taudeltaplus,\taubetaminus,\taupsiplus,\tautildemaxdown,\tauiweak,\taujweak,\taumaxup\}$.
In this case, Combining \cref{lem:taudeltaplus,lem:taudeltaplus bounded jump}, we have
\begin{align*}
    \E[\tau]\leq \frac{n}{\constref{lem:taudeltaplus}\npm_0}\qty(\frac{16(\cdeltaplus)^2}{n}+\frac{\constref{lem:taudeltaplus}\npm_0}{2n}\E[\tau])
    =\frac{16(\cdeltaplus)^2}{\constref{lem:taudeltaplus}\npm_0}+\frac{\E[\tau]}{2},
\end{align*}
i.e, $\E[\tau]\leq \frac{32(\cdeltaplus)^2}{\constref{lem:taudeltaplus}}\cdot \frac{1}{\npm_0}=\frac{\constrefs{lem:hitting time for tnpm and npm}{item:taualphamaxup}}{2\npm_0}$.
Thus, from the Markov inequality, we obtain 
\begin{align}
    \Pr[\tau>\frac{\constrefs{lem:hitting time for tnpm and npm}{item:taualphamaxup}}{2\npm_0}]\leq \frac{1}{2}.
    \label{eq:taudeltaplus const prob 1}
\end{align}
We obtain
\begin{align*}
    &\Pr\qty[\min\{\taudeltaplus,\tauiweak,\taujweak,\tautildemaxdown,\taupsiplus,\taubetaminus\}>T]\\
    &\leq \Pr\qty[\tau>T]+\Pr\qty[\taubetaminus>T \text{ and }\taumaxup\leq T]\\
    &\leq \frac{1}{2}+n\exp(-\Omega(n(\npm_0)^2)) &&(\text{by \cref{eq:taudeltaplus const prob 1,lem:hitting time for tnpm and npm} (\cref{item:taualphamaxup}}))\\
    &\leq 1-c_*.
\end{align*}

Second, consider the case where $\npm_0\leq 16(\cdeltaplus)^2/\constref{lem:taudeltaplus bounded jump}$.
In this case, we have
\begin{align*}
\Pr[\taudeltaplus>1]\leq 1-c_*
\end{align*}
from \cref{lem:delta CLT}.
Taking $\constref{lem:taudeltaplus bounded jump}\geq 16(\cdeltaplus)^2/\constrefs{lem:hitting time for tnpm and npm}{item:taualphamaxup}$, 
$\frac{\constrefs{lem:hitting time for tnpm and npm}{item:taualphamaxup}}{\npm_0}\geq \constrefs{lem:hitting time for tnpm and npm}{item:taualphamaxup}\frac{\constref{lem:taudeltaplus bounded jump}}{16(\cdeltaplus)^2}\geq 1$ and we obtain the claim.

\end{proof}

\paragraph{Either of two non-weak opinions becomes weak: Proof of \cref{lem:pair of non-weak opinions}.} 
Now, we are ready to prove \cref{lem:pair of non-weak opinions} (\cref{item:pair of non-weak opinions becomes weak,item:delta multipricative with initial bias}).
%
To this end, we invoke the known drift analysis result of \cite{Doerr11}.
Specifically, we use \cref{lem:nazo lemma}, which is a modified version from \cite{SS25_sync}.

\begin{proof}[Proof of \cref{lem:pair of non-weak opinions} (\cref{item:pair of non-weak opinions becomes weak})]
We apply \cref{lem:nazo lemma} for $Z_t=\opn_t$, $\varphi(Z_t)=\sqrt{n}\cdot \abs{\delta_t(i,j)}$, 
\[\tau=\min\qty{\tauiweak,\taujweak,\tautildemaxdown,\taubetaminus,\taupsiplus},\] $\cphiup=\cdeltaup$, $x_0 = \cdeltaplus$, and $x^*=\sqrt{C\log n}$.
Define 
\begin{align*}
  \tau^\uparrow=\begin{cases}
  \taudeltaup(i,j) & \text{if } \delta_0(i,j)\geq 0,\\
  \taudeltaup(j,i) & \text{if } \delta_0(j,i)> 0 \;(\text{i.e., } \delta_0(i,j)< 0)
  \end{cases},
\end{align*}
where $\taudeltaup(i,j)=\taudeltaup=\inf\{t\geq 0: \delta_t(i,j)\geq (1+\cdeltaup)\delta_0(i,j)\}$.
Then, from these settings, we have
\begin{align*}
    &\tauphiplus(x_0) = \inf\{t\geq 0: \sqrt{n}\abs{\delta_t}\geq \cdeltaplus\}=\taudeltaplus,\\
    &\tauphiup = \inf\{t\geq 0: \sqrt{n}\abs{\delta_t}\geq (1+\cdeltaup)\sqrt{n}\abs{\delta_0}\}\leq \tau^\uparrow.
\end{align*}
From \cref{lem:delta multiplicative and additive drift} (\cref{item:taudeltaplus const prob}), we have
\begin{align*}
  \Pr\qty[\min\qty{\tauphiplus(x_0),\tau} > \frac{\constrefs{lem:hitting time for tnpm and npm}{item:taualphamaxup}}{\npm_0}] \leq 1-c_*
\end{align*}
for some constant $c_*>0$, i.e., the first condition of \cref{lem:nazo lemma} holds for $C_1=c_*$.
Note that if $\beta_0 < \frac{1}{2}-\xbeta$ or $\psi_0 > \xpsi$ or $\npm_0\leq O(\sqrt{\log n/n})$, then $\tau=0$.
Next, from \cref{lem:delta multiplicative and additive drift} (\cref{item:taudeltaup}),
\begin{align*}
  \Pr\qty[\min\{\tauphiup,\tau\}> \frac{\constrefs{lem:hitting time for tnpm and npm}{item:taualphamaxup}}{\npm_0}]
  &\leq \Pr\qty[\min\{\tau^\uparrow,\tau\}> \frac{\constrefs{lem:hitting time for tnpm and npm}{item:taualphamaxup}}{\npm_0}]\\
  &\leq \exp\qty(-\Omega\qty(n\delta_0(i,j)^2))+n \exp\qty(-\Omega\qty(n(\npm_0)^2))\\
  &\leq \exp\qty(-\Omega\qty(\varphi(Z_0)^2))
\end{align*}
holds for $\delta_0(i,j)\leq C\sqrt{\log n/n}$, i.e., the second condition of \cref{lem:nazo lemma} holds for some positive constant $C_2$.
Note that $n\exp\qty(-\Omega\qty(n(\npm_0)^2))\leq n^{-\omega(1)}\leq \exp\qty(-\Omega\qty(n\delta_0(i,j)^2))$ from the assumption on $\npm_0=\omega(\sqrt{\log n/n})$.
Note that if $\beta_0 < \frac{1}{2}-\xbeta$ or $\psi_0 > \xpsi$ or $\npm_0\leq O(\sqrt{\log n/n})$, then $\tau=0$.

Write $\tauweak=\min\{\tauiweak,\taujweak\}$ and $\tau_{ini}=\min\{\tautildemaxdown,\taubetaminus,\taupsiplus\}$.
From \cref{lem:nazo lemma,lem:tauinitial lemma} with $\varepsilon=n^{-30}$, for some $T'=O(\log n/\npm_0)$, 
\begin{align*}
  &\Pr\qty[\min\{\taudeltaplus(\sqrt{C\log n/n}),\tauweak\}> T' \text{ or } \tau_{ini}\leq  T']\\
  &\leq \Pr\qty[\qty{\min\{\tauphiplus(x_*),\tauweak\}> T' \text{ or } \tau_{ini}\leq  T'}\text{ and } \tau_{ini}>T']+\Pr\qty[\tau_{ini}\leq T']\\
  &\leq \Pr\qty[\min\{\tauphiplus(x_*),\tauweak,\tau_{ini}\}> T']+\Pr\qty[\tau_{ini}\leq T']\\
  &\leq n^{-20}.
\end{align*}
\end{proof}

\begin{proof}[Proof of \cref{lem:pair of non-weak opinions} (\cref{item:delta multipricative with initial bias})]
Let $\tau^{\uparrow (s)}=\inf\{t\geq 0: \delta_t\geq (1+c)^s \delta_0\}$.
Let $\tau^*=\min\{\tautildemaxdown,\taubetaminus,\taupsiplus\}$.
From definition, for some $\ell=\Theta(\log n)$, $\taujweak\leq \tau^{\uparrow (\ell)}$ holds.
Let $T'=\frac{3\constrefs{lem:hitting time for tnpm and npm}{item:taualphamaxup}}{(1-\ctildemaxdown)\npm_0}$.
Applying \cref{lem:Iterative Drift theorem} (\cref{item:iterative updrift}), we have
\begin{align*}
&\Pr\qty[\min\{\taujweak,\tau^*\}>\ell T']\\
&\leq \Pr\qty[\min\{\tau^{\uparrow (\ell)},\tau^*\}>\ell T']\\
&\leq \sum_{s=1}^{\ell}\E\qty[\indicator_{\tau^*>\tau^{\uparrow (s-1)}} \Pr_{\tau^{\uparrow (s-1)}}\qty[\min\{\tau^{\uparrow (s)},\tau^*\}> \tau^{\uparrow (s-1)}+\frac{3\constrefs{lem:hitting time for tnpm and npm}{item:taualphamaxup}}{(1-\ctildemaxdown)\npm_{0}}]]\\
&\leq \sum_{s=1}^{\ell}\E\qty[\indicator_{\tau^*>\tau^{\uparrow (s-1)}} \Pr_{\tau^{\uparrow (s-1)}}\qty[\min\{\tau^{\uparrow (s)},\tau^*\}> \tau^{\uparrow (s-1)}+\frac{\constrefs{lem:hitting time for tnpm and npm}{item:taualphamaxup}}{\npm_{\tau^{\uparrow (s-1)}}}]]\\
&\leq \ell \qty(\exp(-\Omega(n\delta_0^2))+n\exp\qty(-\Omega\qty(n (\npm_0)^2))).
\end{align*}
Note that $\npm_{\tau^{\uparrow (s-1)}}\geq (1-\ctildemaxdown)\tnpm_0/3\geq (1-\ctildemaxdown)\npm_0/3\geq \omega(\sqrt{\log n/n})$ holds.
\end{proof}

\subsection{Emergence of Unique Strong Opinion} \label{sec:unique strong opinion}
In this section, we show that once an opinion becomes weak, it cannot become \emph{strong}—a key concept defined in \cref{def:strong}—for a sufficiently long time.
Recall $\eta_t(j)=\delta_t^{(\ceta)}(I_t,j)=\npm_t-(1+\ceta)\alpha_t(j)$ and stopping times for $\eta_t(j)$ as defined in \cref{def:key-quantities,def:stopping_times}.
We define the strong opinion and the stopping time for the emergence and persistence of the unique strong opinion, $\tau_{us}^+$ and $\tau_{us}^-$:
\begin{definition}[Strong Opinion]\label{def:strong}
  We call that an opinion $j\in [k]$ is \emph{strong} at time $t$ if $\alpha_t(j)\ge (1-\cstrong)\npm_t$.
  By default, we set $\cstrong=0.05<\cweak=0.1$. 

  Set $\ceta\defeq \frac{\cstrong}{1-\cstrong}$.
  For a positive parameter $x$, define
  \begin{align*}
    \tau_{us}^-(x)\defeq \inf\qty{t\geq 0: \min_{j\neq I_t}\eta_t(j)< x} \quad \text{and} \quad
    \tau_{us}^+(x)\defeq \inf\qty{t\geq 0: \min_{j\neq I_t}\eta_t(j)\geq  x}.
  \end{align*}
\end{definition}
The fundamental relationships between the notions of weak, strong, and the value of $\eta_t(j)$ are as follows:
\begin{enumerate}
  \item \label{item:weak opinion eta} For any weak opinion $j$, we have $\eta_t(j)\geq \frac{\cweak-\cstrong}{1-\cstrong}\npm_t$. 
  \item \label{item:not strong} The opinion $j$ is not strong if $\eta_t(j)> 0$.
\end{enumerate}
In summary, if an opinion $j$ is weak, then $\eta_t(j) = \Omega(\npm_t) > 0$ (see \cref{item:weak opinion eta}), and if $\eta_t(j) > 0$, then $j$ is not a strong opinion (\cref{item:not strong}). Therefore, if $\min_{j \neq I_t} \eta_t(j) \geq x$ (with $x > 0$), it means that all opinions other than $I_t$ are not strong; that is, $I_t$ is the unique strong opinion at time $t$.
Note that \cref{item:weak opinion eta} follows from $\eta_t(j) = \npm_t - \frac{1}{1-\cstrong}\alpha_t(j) \geq \npm_t - \frac{1-\cweak}{1-\cstrong}\npm_t$, while \cref{item:not strong} follows from the inequality $\frac{\alpha_t(j)}{\npm_t} < \frac{1}{1+\ceta} = 1-\cstrong$.


We show in this section that, with high probability, 
(i) there exists exactly one strong opinion within $O(\log n/\npm_0)$ rounds (\cref{item:tauusplus} of \cref{lem:unique strong opinion}), and
(ii) any non-strong opinions keep non-strong for a sufficiently long time (\cref{item:tauusminus} of \cref{lem:unique strong opinion}).

\begin{lemma}[Unique Strong Opinion Lemma]\label{lem:unique strong opinion}
  Suppose that $\psi_0\leq \xpsi$, $\beta_0\geq 1/2-\xbeta$, and $\npm_0=\omega(\log n/\sqrt{n})$.
  Then, we have the following:
  \begin{enumerate}
    \item \label{item:tauusplus} 
     For some $T=O(\log n/\npm_0)$, we have 
    \[\Pr\qty[\tau_{us}^+(x_{\eta})> T \text{ or } \min\{\tautildemaxdown,\taubetaminus,\taupsiplus\}\leq T]\leq O(n^{-10}).\]
    \item \label{item:tauusminus} 
    Suppose $\eta_0(j)\geq x_{\eta}$ for all $j\neq I_0$.
    Let $\constrefs{lem:hitting time for tnpm and npm}{item:tildemaxdown}= \ctildemaxdown/36$ be a positive constant that is defined in \cref{lem:hitting time for tnpm and npm} (\cref{item:tildemaxdown}).
    Then, for any $T\leq \constrefs{lem:hitting time for tnpm and npm}{item:tildemaxdown}n$, we have
    \[\Pr\qty[\min\{\tau_{us}^-(x_{\eta}/2),\tautildemaxdown,\taubetaminus,\taupsiplus\}\leq T ]\leq n^{-10}.\]
  \end{enumerate}
\end{lemma}

The key lemma to prove \cref{lem:unique strong opinion} is the following lemma, which shows that once $\eta_t(j)$ becomes large, it stays large for a sufficiently long time—that is, once an opinion becomes weak, it cannot become strong for a long time.
\begin{lemma}[Bounded decrease of $\eta_t(j)$]
  \label{lem:tauetaminus}
  Suppose that $\psi_0\leq \xpsi$, $\beta_0\geq 1/2-\xbeta$, $\npm_0=\omega(\log n/\sqrt{n})$, and $\eta_0(j)\geq x_{\eta}$.
  Then, for any $T\geq 0$, 
  \begin{align*}
    \Pr\qty[\tauetaminus(x_{\eta}/2) \leq T \text{ and } \min\{\tautildemaxdown,\taubetaminus,\taupsiplus\}>T]
    \leq Tn^{-\omega(1)}.
\end{align*}
\end{lemma}

\paragraph{Weak cannot be strong: Proof of \cref{lem:tauetaminus}.}
To prove \cref{lem:tauetaminus}, we first prove the following lemma:
\begin{lemma}
    \label{lem:hitting time for eta}
    Let $y=y(n)$ and $I=I(n)$ be positive parameters satisfying the following:
    $I=\omega(\xbeta)$, $y=\omega(I)$, and $y=o(\log n/\sqrt{n})$.  
    Suppose that $\beta_0\geq 1/2-\xbeta$, $\psi_0\leq \xpsi$, $\npm_0\geq \omega\qty(\log n/\sqrt{n})$, and $\eta_0(j)\in [y,y+I]$.
    Then, for any $T\geq \frac{24(1+\ceta)}{1-\ctildemaxdown}\cdot \frac{I}{(y-I)\npm_0}$, we have
    \begin{align*}
      \Pr\qty[\tauetaminus(y-I)<\tauetaplus(y+2I) \text{ and } \min\{\tautildemaxdown,\taumaxup,\taubetaminus,\taupsiplus\}>T]\leq n^{-\omega(1)}.
    \end{align*}
  \end{lemma}
  \begin{proof}
    Let $\tau^*=\min\{\tautildemaxdown,\taumaxup,\taubetaminus,\taupsiplus\}$ and
  let $\tau=\min\{\tauetaplus(y+2I),\tauetaminus(y-I),\tau^*\}$.
  For $t-1<\tau$, we have
  \begin{align}
    &\npm_{t-1}+\alpha_{t-1}(j)+1-2\beta_{t-1}\nonumber \\
    &=\npm_{t-1}\qty(1+\frac{\npm_{t-1}-\eta_{t-1}(j)}{(1+\ceta)\npm_{t-1}}-\frac{\gamma_{t-1}}{\beta_{t-1}\npm_{t-1}}-\frac{\psi_{t-1}}{\beta_{t-1}^2\tnpm_{t-1}})\nonumber \\
    &\geq \frac{1}{3}(1-\ctildemaxdown)\tnpm_{0}
    \qty(\frac{1}{1+\ceta}-\underbrace{\frac{y+2I}{(1+\ceta)(1/3)(1-\ctildemaxdown)\tnpm_{0}}}_{o(1)}-\underbrace{\frac{\xpsi}{(1/3)^2(1-\ctildemaxdown)\tnpm_{0}}}_{o(1)  }) \nonumber \\
    &\geq \frac{1-\ctildemaxdown}{6(1+\ceta)}\cdot \npm_{0}. \label{eq:drift of eta}
  \end{align}
  Hence, for $t-1<\tau$, we have
  \begin{align*}
    \E_{t-1}[\eta_t(j)]
    &\geq \E_{t-1}[\delta_{t}^{(\ceta)}(I_{t-1},j)]\\
    &\geq \delta_{t-1}^{(\ceta)}(I_{t-1},j)+\delta_{t-1}^{(\ceta)}(I_{t-1},j)\qty(\alpha_{t-1}(I_{t-1})+\alpha_{t-1}(j)+1-2\beta_{t-1})\\
    &\geq \eta_{t-1}(j)+\frac{1-\ctildemaxdown}{6(1+\ceta)}\cdot (y-I)\npm_{0}.
  \end{align*}
  Note that we use \cref{lem:basic inequalities for delta_epsilon} (\cref{item:expectation of delta_epsilon}).
  
  Letting $X_t=\eta_t(j)$ and $\drift = \frac{1-\ctildemaxdown}{6(1+\ceta)}\cdot (y-I)\npm_{0}>0$, we have
    $
      \indicator_{\tau>t-1}\qty(X_{t-1}+\drift-\E_{t-1}[X_t])
      =\indicator_{\tau>t-1}\qty(\eta_{t-1}(j)+\drift-\E_{t-1}[\eta_t(j)])\leq 0
    $
    and 
    \begin{align*}
      \indicator_{\tau>t-1}\qty(X_{t-1}+\drift-X_t)
      &=\indicator_{\tau>t-1}\qty(\eta_{t-1}(j)+\drift-\eta_t(j))\\
      &\leq \indicator_{\tau>t-1}\qty(\E_{t-1}[\delta_{t}^{(\ceta)}(I_{t-1},j)]-\delta_{t}^{(\ceta)}(I_{t-1},j)),
    \end{align*}
    i.e., $\indicator_{\tau>t-1}\qty(X_{t-1}+\drift-X_t)$ satisfies one-sided $\qty(O(1/n),O(\npm_0/n))$-Bernstein condition.
    Note that we use \cref{lem:basic inequalities for delta PP} (\cref{item:bernstein condition of delta PP}) and \cref{lem:Bernstein condition} (\cref{item:BC for upper bounded rv,item:BC for linear transformation}).
  
    Let $I^-=y-I$, $I^-_*=y$, $I^+_*=y+I$, and $I^+=y+2I$.
    Note that $I^+-I^+_*=I^+_*-I^-_*=I^-_*-I^-=I$.
    Let $T\geq \frac{24(1+\ceta)}{1-\ctildemaxdown}\frac{I}{(y-I)\npm_0}\geq \frac{2(I^+-X_0)}{\drift}$.
    Applying \cref{lem:Useful drift lemma} (\cref{item:positive drift useful}) for $\eta_0(j)\in [I_*^-,I^+_*]$, we have
    \begin{align*}
      \Pr\qty[\tauetaminus(y-I)<\tauetaplus(y+2I) \text{ and } \tau^*>T]
      \leq \exp\qty(-\Omega\qty(\frac{I^2}{\frac{\npm_0}{n}T+\frac{I}{n}}))
      \leq \exp\qty(-\Omega\qty(nI(y-I))).
    \end{align*}
  \end{proof}
  
  Next, we prove the following lemma assuring $\eta_t(j)$ does not decrease too much in one step.
  \begin{lemma}
    \label{lem:one step eta}
Let $I=I(n)$ be a positive parameter satisfying $I=\omega(\xpsi)$ and 
$\tau^{jump}_\eta \defeq \inf\{t\geq 0: \eta_t(j)\leq \eta_{t-1}(j)-I\}$.
Suppose that $\beta_0\geq 1/2-\xbeta$ and $\psi_0\leq \xpsi$.
Then, for any $T\geq 0$, 
    \begin{align*}
        \Pr\qty[\tau^{jump}_\eta \leq T \text{ and } \min\{\taubetaminus, \taupsiplus\}>T]
        \leq Tn^{-\omega(1)}.
    \end{align*}
\end{lemma}
\begin{proof}
    Let $\tau=\min\{\taubetaminus, \taupsiplus\}$.
    For $t-1<\tau$, we have $2\beta_{t-1}-1\leq \npm_{t-1}+3\xpsi$.
    Hence, for $t-1<\tau$, we have
    \begin{align*}
    \E_{t-1}[\eta_t(j)]
    &\geq \E_{t-1}\qty[\delta_t^{(\ceta)}(I_{t-1},j)]\\
    &\geq \eta_{t-1}(j)\qty(1+\npm_{t-1}+\alpha_{t-1}(j)+1-2\beta_{t-1})\\
    &\geq \eta_{t-1}(j)\qty(1-3\xpsi)\\
    &\geq \eta_{t-1}(j)-3\xpsi.
    \end{align*}
    Thus, 
    \begin{align*}
        \indicator_{\tau>t-1}\qty(\eta_{t-1}(j)-3\xpsi-\eta_t(j))
        &\leq \indicator_{\tau>t-1}\qty(\E_{t-1}\qty[\delta_t^{(\ceta)}(I_{t-1},j)]-\delta_t^{(\ceta)}(I_{t-1},j))
    \end{align*}
    holds.
    From \cref{lem:Bernstein condition} (\cref{item:BC for linear transformation,item:BC for upper bounded rv}) and \cref{lem:basic inequalities for delta_epsilon} (\cref{item:bernstein condition for delta_epsilon}), we have that the random variable $\indicator_{\tau>t-1}\qty(\E_{t-1}\qty[\delta_t^{(\ceta)}(I_{t-1},j)]- \delta_t^{(\ceta)}(I_{t-1},j))$ conditioned on round $t-1$ satisfies $\qty(\frac{4}{n},\frac{10}{n})$-Bernstein condition.
From \cref{item:BC for dominated rv} of \cref{lem:Bernstein condition}, 
$\indicator_{\tau>t-1}\qty(\eta_{t-1}(j)-3\xpsi-\eta_t(j))$ satisfies one-sided $\qty(\frac{4}{n},\frac{10}{n})$-Bernstein condition.
Thus,
\begin{align*}
    \indicator_{\tau>t-1}\Pr_{t-1}\qty[\eta_t(j)\leq \eta_{t-1}(j)-I]
    &=\E_{t-1}\qty[\indicator_{\tau>t-1}\indicator_{\eta_{t-1}(j)-3\xpsi-\eta_t(j)\geq I-3\xpsi}]\\
    &=\Pr_{t-1}\qty[\tau>t-1 \text{ and } \eta_{t-1}(j)-3\xpsi-\eta_t(j)\geq I-3\xpsi]\\
    &=\Pr_{t-1}\qty[\indicator_{\tau>t-1}\qty(\eta_{t-1}(j)-3\xpsi-\eta_t(j))\geq I-3\xpsi]\\
    &\leq \exp\qty(-\Omega\qty(nI^2))\\
    &\leq n^{-\omega(1)}.
\end{align*}
Thus, we obtain
\begin{align*}
\Pr\qty[\tau^{jump}_\eta \leq T \text{ and } \tau>T]
&=\Pr\qty[\exists t\leq T: \eta_t(j)\leq \eta_{t-1}(j)-I\text{ and }\tau>T]\\
&\leq \sum_{t=1}^T\E\qty[\indicator_{\tau>t-1}\Pr_{t-1}\qty[\eta_t(j)\leq \eta_{t-1}(j)-I]]\\
& \leq Tn^{-\omega(1)}.
\end{align*}
\end{proof}

  \begin{proof}[Proof of \cref{lem:tauetaminus}]
    Let $y=y(n)$ and $I=I(n)$ be positive parameters satisfying the following:
    $I=\omega(\xbeta)$, $y=\omega(I)$, and $y=o(\log n/\sqrt{n})$.  
    Let $I^-=y-I$, $I^-_*=y$, $I^+_*=y+I$, and $I^+=y+2I$.
    Note that $I^+-I^+_*=I^+_*-I^-_*=I^-_*-I^-=I$.
    Let $T'=\frac{24(1+\ceta)}{1-\ctildemaxdown}\frac{I}{(y-I)\npm_0}=o\qty(\frac{1}{\npm_0})$.
    Then, combining 
    \cref{lem:hitting time for eta,lem:hitting time for tnpm and npm}, \cref{lem:taubeta},
    if $\beta_0\geq 1/2-\xbeta$, $\psi_0\leq \xpsi$, $\npm_0=\omega(\log n/\sqrt{n})$, and $\eta_0(j)\in [I^-_*,I^+_*]$, 
    \begin{align*}
      &\Pr\qty[\tauetaminus(y-I)<\tauetaplus(y+2I)]\\
      &\leq \Pr\qty[\tauetaminus(y-I)<\tauetaplus(y+2I) \text{ and } \min\{\tautildemaxdown,\taumaxup,\taubetaminus,\taupsiplus\}>T']\\
      &+\Pr\qty[\tautildemaxdown\leq T'\text{ and } \min\{\taumaxup,\taubetaminus,\taupsiplus\}>T']
      +\Pr\qty[\taumaxup\leq T'\text{ and } \min\{\taubetaminus,\taupsiplus\}>T']\\
      &+\Pr\qty[\taubetaminus\leq T'\text{ and } \taupsiplus>T']
      +\Pr\qty[\taupsiplus\leq T']\\
      &\leq n^{-\omega(1)}.
    \end{align*}

    Let $\tau^*=\min\{\tau^{jump}_\eta,\tautildemaxdown,\taubetaminus,\taupsiplus\}$.
    Let $\tau_s^-=\inf\{t\geq s: \eta_t(j)\leq I^-\}$ and $\tau_s^+=\inf\{t\geq s: \eta_t(j)\geq I^+\}$.
    Suppose $\eta_0(j)\in [I^-_*,I^+_*]$, $\beta_0\geq 1/2-x$, $\psi_0\leq x$, and $\npm_0\geq \omega(\log n/\sqrt{n})$.
    Applying \cref{lem:Iterative Drift theorem} (\cref{item:no increase under negative drift}), we have
    \begin{align*}
      \Pr\qty[\tauetaminus(y-I) \leq T \text{ and } \tau^*>T]
      \leq \sum_{s=0}^{T-1}\E\qty[\indicator_{\eta_s(j)\in [I^-_*,I^{+}_*] \text{ and } \tau^*>s}\Pr_s\qty[\tau^-_s<\tau^+_s]]
      \leq Tn^{-\omega(1)}.
  \end{align*}
\end{proof}

\paragraph{Emergence of unique strong opinion: Proof of \cref{lem:unique strong opinion}.}
\begin{proof}[Proof of \cref{lem:unique strong opinion} (\cref{item:tauusplus})]
  Fix an arbitrary pair of distinct opinions $i$ and $j$.
  Observe the following holds:
  \begin{itemize}
    \item Combining \cref{item:delta multipricative with initial bias,item:pair of non-weak opinions becomes weak} of \cref{lem:pair of non-weak opinions}, 
    for some $T_1=O(\log n/\npm_0)$, $\min\{\alpha_{T_1}(i),\alpha_{T_1}(j)\}\leq (1-\cweak)\npm_{T_1}$, $\beta_{T_1}\geq 1/2-\xbeta$, $\psi_{T_1}\leq \xpsi$, and $\npm_{T_1}=\Omega(\npm_0)=\omega(\log n/\sqrt{n})$ with probability at least $1-n^{-10}$.
    Suppose that $\alpha_{T_1}(j) \le (1-\cweak)\npm_{T_1}$ without loss of generality.
    From the definition of $\eta_t(j)$, we have $\eta_{T_1}(j)\geq \frac{\cweak-\cstrong}{1-\cstrong}\npm_{T_1}=\Omega(\npm_0)\geq 2x_{\eta}$.
    \item From \cref{lem:tauetaminus}, for any $t\in [T_1,T_1+T_2]$, where $T_2\geq n \npm_{T_1} /C\log n=\Omega(n\npm_0/\log n)$, $\eta_t(j)\geq x_{\eta}$, $\beta_t\geq 1/2-\xbeta$, $\psi_t\leq \xpsi$, and $\npm_t=\Omega(\npm_0)=\omega(\log n/\sqrt{n})$ with probability at least $1-n^{-\omega(1)}$.
  \end{itemize}
  
  Since $T_1=O(\log n/\npm_0)=o(\sqrt{n})$ and $T_2=\Omega(n\npm_0/\log n)=\omega(\sqrt{n})$, we have $T_1<T_2$.
  Thus, by taking the union bound over all pairs of distinct opinions $i,j\in[k]$, with high probability, there exists $t\leq T_1$ such that $\eta_t(j)\geq x_{\eta}$ for all $j\neq I_t$.
\end{proof}
\begin{proof}[Proof of \cref{lem:unique strong opinion} (\cref{item:tauusminus})]
  Consider an arbitrary opinion $j\neq I_0$. Combining \cref{lem:tauetaminus,lem:tauinitial lemma}, 
  \begin{align}
    &\Pr\qty[\tauetaminus(x_{\eta}/2) \leq T \text{ or } \min\{\tautildemaxdown,\taubetaminus,\taupsiplus\}\leq T] \nonumber\\
    &\leq \Pr\qty[\qty{\tauetaminus(x_{\eta}/2) \leq T \text{ or } \min\{\tautildemaxdown,\taubetaminus,\taupsiplus\}\leq T} \text{ and } \min\{\tautildemaxdown,\taubetaminus,\taupsiplus\}>T] \nonumber\\
    &+ \Pr\qty[\min\{\tautildemaxdown,\taubetaminus,\taupsiplus\}\leq T] \nonumber\\
    &\leq n^{-11}. \label{eq:tauetaminus and tauinitial lemma}
  \end{align}
  Thus, from the union bound,
  \begin{align*}
    \Pr\qty[\min\{\tau_{us}^-(x_{\eta}/2),\tautildemaxdown,\taubetaminus,\taupsiplus\}\leq T ]
    \leq \sum_{j\neq I_0}\Pr\qty[\min\{\tau_\eta^-(x_{\eta}/2),\tautildemaxdown,\taubetaminus,\taupsiplus\}\leq T]
  \leq n^{-10}.
  \end{align*}
  \end{proof}

\subsection{Towards Consensus} \label{sec:towards consensus}
We now turn to establishing an upper bound on the consensus time.
Specifically, we show that once there remains exactly one strong opinion, the process will reach consensus within $O(\log n/\npm_0)$ rounds with high probability (see \cref{lem:towards consensus}).
\begin{lemma}[Unique strong opinion leads to consensus]
  \label{lem:towards consensus}
  Suppose $\psi_0\leq \xpsi$, $\beta_0\geq 1/2-\xbeta$, $\npm_0=\omega(\log n/\sqrt{n})$, and $\eta_0(j)\geq x_{\eta}$ for all $j\neq I_0$.
  Then, $\Pr\qty[\taucons>T]\leq 1/n$ holds for some $T=O\qty(\log n/\npm_0)$.
  \end{lemma}

\newcommand{\ctildemaxplus}{c^{+}_{\widetilde{\max}}}
\newcommand{\tautildemaxplus}{\tau^{+}_{\widetilde{\max}}}
\newcommand{\cmaxplus}{c^{+}_{\max}}
\newcommand{\cmaxminus}{c^{-}_{\max}}
\newcommand{\taumaxminus}{\tau^{-}_{\max}}
\newcommand{\tauall}{\tau_{\textrm{all}}}
Define the following stopping times:
\begin{align*}
&\taumaxplus(x)\defeq \inf\{t\geq 0:\tnpm_t\geq x\}, &&\taumaxminus(x)\defeq \inf\{t\geq 0:\tnpm_t\leq x\}, \\
&\tautildemaxplus(x)\defeq \inf\{t\geq 0:\tnpm_t\geq x\}, &&\tauall\defeq \inf\{t\geq 0:\npm_t=\beta_t\}.
\end{align*}
For simplicity, we sometimes assume that $\alpha_0(1)=\npm_0$. 
The following lemma describes how the unique strong opinion evolves, which in turn implies \cref{lem:towards consensus}.

\begin{lemma}
  \label{lem:unique strong opinion evolves}
  We have the following:
  \begin{enumerate}
    \item \label{item:tauplus}
    Let $c\in (0,1)$ be an arbitrary constant.
  Suppose $\psi_0\leq \xpsi$, $\beta_0\geq 1/2-\xbeta$, $\npm_0=\omega(\log n/\sqrt{n})$, and $\eta_0(j)\geq x_{\eta}$ for all $j\neq I_0$.
  Then, for some $T=O\qty(\log n/\npm_0)$, 
\begin{align*}
  \Pr\qty[\tautildemaxplus(1-c)>T \text{ or } \min\{\taupsiplus,\taubetaminus\}\leq T]\leq n^{-10}. 
 \end{align*}
    \item \label{item:taumaxmajority}
    Suppose that $\psi_0\leq \xpsi$, $\beta_0\geq 1/2-\xbeta$, and $\alpha_0(1)\geq (1-\ctildemaxdown)\beta_0$.
    Then, for some $T=O(\log n)$,
    \begin{align*}
      \Pr\qty[\taumaxplus(1-4\ctildemaxdown)>T \text{ or } \min\{\taupsi,\taubetaminus\}\leq T]\leq n^{-10}. 
     \end{align*}
    \item \label{item:tauall}
    Suppose that $\alpha_0(1)\geq 7/8$. 
    Then,
      $\Pr\qty[\tau_{all}> 8\log n \text{ or } \taumaxminus(3/4)\leq 8\log n]\leq 1/n.$
     \item \label{item:taucons}
     Suppose $\alpha_0(1)=\beta_0$ and $\beta_0\geq 1/2-\xbeta$.
     Then, $\Pr[\taucons>6\log n]\leq 1/n$.
  \end{enumerate}
\end{lemma}
\begin{proof}[Proof of \cref{lem:unique strong opinion evolves} (\cref{item:tauplus})]
  Let $\tautildemaxup = \inf\{t\geq 0:\tnpm_t\geq (1+\ctildemaxup)\tnpm_0\}$ for some positive constant $\ctildemaxup\in (0,1)$.
Let $\tau^*=\min\{\tautildemaxplus,\taubetaminus,\tau_{us}^-(x_{\eta})\}$ and $\tau=\min\{\tautildemaxup,\tautildemaxdown,\tau^*\}$.
For $t-1<\tau$, we have 
\begin{align}
  \npm_{t-1}-\frac{\gamma_{t-1}}{\beta_{t-1}}
  &\geq \npm_{t-1}-\frac{(\npm_{t-1})^2}{\beta_{t-1}}-\frac{(1-\cstrong)\npm_{t-1}\qty(\beta_{t-1}-\npm_{t-1})}{\beta_{t-1}}\nonumber\\
  &=\cstrong\npm_{t-1}\qty(1-\tnpm_{t-1})\nonumber\\
  &\geq \cstrong c\npm_{t-1}. \label{eq:tnpm additive drift}
\end{align}
Hence, for $t-1<\tau$, we have
\begin{align*}
  \E_{t-1}\qty[\tnpm_t]
  &\geq  \tnpm_{t-1}\qty(1+\frac{\npm_{t-1}-\gamma_{t-1}/\beta_{t-1}}{2})-\frac{9\npm_{t-1}}{n}\\
  &\geq \tnpm_{t-1}\qty(1+\frac{\cstrong c\npm_{t-1}}{2})-\frac{9\npm_{t-1}}{n}\\
  &= \tnpm_{t-1}+\frac{\cstrong c\beta_{t-1}(\tnpm_{t-1})^2}{2}-\frac{9\npm_{t-1}}{n}\\
  &\geq \tnpm_{t-1}+\frac{\cstrong c(1-\ctildemaxdown)^2(\npm_{0})^2}{12}.
\end{align*}
Note that we use \cref{lem:basic inequalities for tnpm} (\cref{item:expectation of tnpm})
and $\npm_{t-1}\geq \frac{(1-\ctildemaxdown)\tnpm_{0}}{3}=\omega(1/n)$.

By letting $X_t=\tnpm_t$ and $\drift = \frac{\cstrong c(1-\ctildemaxdown)^2(\npm_{0})^2}{12}> 0$, we have
  $
    \indicator_{\tau>t-1}\qty(X_{t-1}+\drift-\E_{t-1}[X_t])
    =\indicator_{\tau>t-1}\qty(\tnpm_{t-1}+\drift-\E_{t-1}[\tnpm_t])\leq 0.
  $
  Moreover, it holds that the random variable
  \begin{align*}
    \indicator_{\tau>t-1}\qty(X_{t-1}+\drift-X_t)
    &\leq \indicator_{\tau>t-1}\qty(\tnpm_{t-1}\qty(1+\frac{\npm_{t-1}-\gamma_{t-1}/\beta_{t-1}}{2})-\frac{9\npm_{t-1}}{n}-\tnpm_t)
  \end{align*}
   satisfies one-sided $\qty(O\qty(\frac{1}{n}),O\qty(\frac{\npm_0}{n}))$-Bernstein condition. 
  Note that we use \cref{lem:basic inequalities for tnpm} (\cref{item:bernstein condition of tnpm}) and \cref{lem:Bernstein condition} (\cref{item:BC for upper bounded rv,item:BC for linear transformation}).

  Applying \cref{lem:Useful drift lemma} (\cref{item:positive drift useful}) for $I^+=(1+\ctildemaxup)\tnpm_{0}$, $I^-=(1+\ctildemaxdown)\tnpm_{0}$, and for $T'=\frac{2(I^+-X_0)}{\drift} = \frac{24}{\cstrong c^+(1-\ctildemaxdown)^2} \frac{1}{\npm_0}$, we have
  \begin{align*}
            \Pr\qty[\tautildemaxup > T' \text{ and } \tau^* > T'] 
            \le \exp\qty(- \Omega\qty( n (\npm_0)^2)).
    \end{align*}

  Let $\tau^{\uparrow (s)}=\inf\{t\geq 0: \tnpm_t\geq (1+\ctildemaxup)^s \tnpm_0\}$.
  From definition, for some $\ell=\Theta(\log n)$, $\tautildemaxplus\leq \tau^{\uparrow (\ell)}$ holds.
Applying \cref{lem:Iterative Drift theorem} (\cref{item:iterative updrift}), we obtain
\begin{align*}
  \Pr\qty[\min\{\tautildemaxplus,\tau^*\}>\ell T']
  &\leq \Pr\qty[\min\{\tau^{\uparrow (\ell)},\tau^*\}>\ell T']\\
  &\leq \sum_{s=1}^{\ell}\E\qty[\indicator_{\tau^*>\tau^{\uparrow (s-1)}} \Pr_{\tau^{\uparrow (s-1)}}\qty[\min\{\tau^{\uparrow (s)},\tau^*\}> \tau^{\uparrow (s-1)}+T']]\\
  &\leq \ell \exp\qty(-\Omega\qty(n (\npm_0)^2)).
\end{align*}
Note that $\npm_{\tau^{\uparrow (s-1)}}\geq \Omega(\npm_0)$ holds.

Finally, combining the above and \cref{eq:tauetaminus and tauinitial lemma}, we have
\begin{align*}
  &\Pr\qty[\tautildemaxplus>T \text{ or }\min\{\taupsi,\taubetaminus\}\leq T]\\
  &\leq \Pr\qty[\qty{\tautildemaxplus>T \text{ or } \min\{\taupsi,\taubetaminus\}\leq T} \text{ and } \min\{\taupsi,\taubetaminus,\tau_{us}^-(x_{\eta}/2)\}>T]\\
  &+\Pr\qty[\min\{\taupsi, \taubetaminus,\tau_{us}^-(x_{\eta}/2)\}\leq T]\\
  &\leq n^{-10}.
\end{align*}
\end{proof}
\begin{proof}[Proof of \cref{lem:unique strong opinion evolves} (\cref{item:taumaxmajority})]
Let $\tau_{\max}^\downarrow=\inf\{t\geq 0:\npm_t\leq (1-c_{\max}^\downarrow)\npm_0\}$.
Let $\tau^*=\min\{\taumaxplus(1-4\ctildemaxdown),\tautildemaxdown,\taubetaminus\}$ and $\tau=\min\{\tau_{\max}^\uparrow,\tau_{\max}^\downarrow,\tau^*\}$.
For $t-1<\tau$, we have
\begin{align*}
  1-\alpha_{t-1}(1)\qty(\frac{2}{\tnpm_{t-1}}-1)
  &\geq 1-(1-4\ctildemaxdown)\qty(\frac{2}{(1-\ctildemaxdown)^2}-1)
  \geq \frac{8(\ctildemaxdown)^2}{1-2\ctildemaxdown}.
\end{align*}
Note that $\tnpm_{t-1}\geq (1-\ctildemaxdown)\tnpm_0\geq (1-\ctildemaxdown)^2\geq 1-2\ctildemaxdown$ holds.
Hence, for any $t-1<\tau$, we have
\begin{align*}
  \E_{t-1}[\alpha_t(1)]-\alpha_{t-1}(1)
  &=\alpha_{t-1}(1)\qty(1-\alpha_{t-1}(1)\qty(\frac{2}{\tnpm_{t-1}}-1))
  \geq \frac{8(\ctildemaxdown)^2(1-c_{\max}^\downarrow)}{1-2\ctildemaxdown}\alpha_{0}(1).
\end{align*}

By setting $X_t=\alpha_t(1)$ and $\drift = \frac{8(\ctildemaxdown)^2(1-c_{\max}^\downarrow)}{1-2\ctildemaxdown}\alpha_{0}(1)> 0$, we have
  $
    \indicator_{\tau>t-1}\qty(X_{t-1}+\drift-\E_{t-1}[X_t])
    =\indicator_{\tau>t-1}\qty(\alpha_{t-1}(1)+\drift-\E_{t-1}[\alpha_t(1)])\leq 0
  $.
  Moreover, it holds that the random variable
  \[
    \indicator_{\tau>t-1}\qty(\E_{t-1}[X_t]-X_t)
    = \indicator_{\tau>t-1}\qty(\E_{t-1}[\alpha_t(1)]-\alpha_t(1))\\
  \]
   satisfies one-sided $\qty(O\qty(\frac{1}{n}),O\qty(\frac{\alpha_{0}(1)}{n}))$-Bernstein condition. 
   Note that we use \cref{lem:basic inequalities for alpha} (\cref{item:bernstein condition of alpha}) and \cref{lem:Bernstein condition} (\cref{item:BC for upper bounded rv,item:BC for linear transformation}).

  Applying \cref{lem:Useful drift lemma} (\cref{item:positive drift useful}) for $I^+=(1+\cmaxup)\alpha_0(1)$, $I^-=(1-c_{\max}^\downarrow)\alpha_0(1)$, and 
  $T'= \frac{c_{\max}^\uparrow (1-2\ctildemaxdown)}{4(\ctildemaxdown)^2(1-c_{\max}^\downarrow)}=\frac{2(I^+-X_0)}{\drift}$, we have
           \begin{align*}
            \Pr\qty[\taumaxup > T' \text{ and } \tau^* > T'] 
            \le \exp\qty(- \Omega\qty( n (\npm_0)^2)).
           \end{align*}
  Let $\tau^{\uparrow (s)}=\inf\{t\geq 0: \npm_t\geq (1+c_{\max}^\uparrow)^s \npm_0\}$.
  From definition, for some $\ell=\Theta(\log n)$, $\taumaxplus(1-4\ctildemaxdown)\leq \tau^{\uparrow (\ell)}$ holds.
Applying \cref{lem:Iterative Drift theorem} (\cref{item:iterative updrift}), we have
\begin{align*}
  \Pr\qty[\min\{\taumaxplus(1-4\ctildemaxdown),\tau^*\}>\ell T']
  &\leq \Pr\qty[\min\{\tau^{\uparrow (\ell)},\tau^*\}>\ell T']\\
  &\leq \sum_{s=1}^{\ell}\E\qty[\indicator_{\tau^*>\tau^{\uparrow (s-1)}} \Pr_{\tau^{\uparrow (s-1)}}\qty[\min\{\tau^{\uparrow (s)},\tau^*\}> \tau^{\uparrow (s-1)}+T']]\\
  &\leq \ell \exp\qty(-\Omega\qty(n (\npm_0)^2)).
\end{align*}
Note that $\npm_{\tau^{\uparrow (s-1)}}\geq \Omega(\npm_0)$ holds.

Finally, combining the above and \cref{lem:tauinitial lemma}, we have
\begin{align*}
  &\Pr\qty[\taumaxplus(1-4\ctildemaxdown)>T \text{ or } \min\{\taupsi,\taubetaminus\}\leq T]\\
  &\leq \Pr\qty[\qty{\taumaxplus(1-4\ctildemaxdown)>T \text{ or } \min\{\taupsi,\taubetaminus\}\leq T} \text{ and } \min\{\taupsi,\taubetaminus,\tautildemaxdown\}>T]\\
  &+\Pr\qty[\min\{\taupsi, \taubetaminus,\tautildemaxdown\}\leq T]\\
  &\leq n^{-10}.
\end{align*}
\end{proof}
\begin{proof}[Proof of \cref{lem:unique strong opinion evolves} (\cref{item:tauall})]
  To begin with, we prove the following claim.
  \begin{claim}\label{claim:taumaxminus}
    Let $c\in (0,1/2)$ be an arbitrary constant.
    Suppose that $\psi_0\leq \xpsi$, $\beta_0\geq 1/2-\xbeta$, and $\alpha_0(1)\geq 1-c$.
    Then, for some $T=\Omega(1/\xpsi)$, 
    \begin{align*}
      \Pr\qty[\taumaxminus(1-2c) \leq \min\{T,\taubetaminus,\taupsiplus\}]
        \leq \exp\qty(-\Omega\qty(n)).
    \end{align*}
  \end{claim}
  \begin{proof}
    Let $\tau=\min\{\taubetaminus,\taupsiplus\}$.
    Then, for $t-1<\tau$, we have
    \begin{align*}
        \E_{t-1}[\alpha_t(1)]
        &= \alpha_{t-1}(1)\qty(1+\alpha_{t-1}(1)-\frac{\gamma_{t-1}+\xpsi}{\beta_{t-1}})
        \geq \alpha_{t-1}(1)-3\xpsi.
      \end{align*}
      Note that we use \cref{lem:basic inequalities for alpha} (\cref{item:expectation of alpha}) in the first equality.
  
      Hence, letting $X_t=\alpha_t(1)$ and $\drift = -3\xpsi<0$, we have
      $
        \indicator_{\tau>t-1}\qty(X_{t-1}+\drift-\E_{t-1}[X_t])
        =\indicator_{\tau>t-1}\qty(\alpha_{t-1}(1)-3\xpsi-\E_{t-1}[\alpha_t(1)])\leq 0
      $
      and $\indicator_{\tau>t-1}\qty(\E_{t-1}[X_t]-X_t)=\indicator_{\tau>t-1}\qty(\E_{t-1}[\alpha_t(1)]-\alpha_t(1))$ satisfies $\qty(O(1/n),O(1/n))$-Bernstein condition. 
      Note that we use \cref{lem:basic inequalities for alpha} (\cref{item:bernstein condition of alpha}) and \cref{lem:Bernstein condition} (\cref{item:BC for upper bounded rv,item:BC for linear transformation}).
    
      Hence, applying \cref{lem:Useful drift lemma} (\cref{item:negative drift useful}) with $I^-=1-2c$, for $T\leq \frac{c}{6\xpsi}\leq \frac{X_0-I^-}{-2\drift}$, we have
      \begin{align}
        \Pr\qty[\taumaxminus(1-2c) \leq \min\{T,\taubetaminus,\taupsiplus\}]
        \leq \exp\qty(-\Omega\qty(n)). 
      \end{align}
  \end{proof}
  Write $g_t=\beta_t-\alpha_t(1)=\sum_{j\geq 2}\alpha_t(j)$ for convenience.
  Let $\tau=\min\{\tau_{all},\taumaxminus(3/4)\}$.
  Then, for $t-1<\tau$, we have
  \begin{align*}
    \E_{t-1}[g_t]
    &=\sum_{j\geq 2}\alpha_{t-1}(j)\qty(1+\alpha_{t-1}(j)+1-2\beta_{t-1})
    \leq g_{t-1}\qty(1-\frac{1}{4}).
  \end{align*}
  Note that $\alpha_{t-1}(j)\leq 1-\alpha_{t-1}(1)\leq 1/4$ and $\beta_{t-1}\geq \alpha_{t-1}(1)\geq 3/4$.
  
  Let $r=1-\frac{1}{4}=\frac{3}{4}$, $X_t=r^{-t}g_t$, and $Y_t=X_{t\wedge \tau}$.
Then, 
\begin{align*}
\E_{t-1}[Y_t]-Y_{t-1}
=\indicator_{\tau>t-1}\qty(\E_{t-1}[X_t]-X_{t-1})
\leq \indicator_{\tau>t-1}\qty(r^{-t}\E_{t-1}[g_t]-r^{-(t-1)}g_{t-1})
\leq 0,
\end{align*}
i.e., $Y_t$ is a submartingale.
Hence, we have $\E[Y_T]\leq \E[Y_0]=g_0\leq 1$ and 
\begin{align*}
\E[Y_T]
\geq \E[X_{T}\mid \tau>T]\Pr[\tau>T]
=r^{-T}\E[g_T\mid \tau>T]\Pr[\tau>T]
\geq r^{-T}n^{-1}\Pr[\tau>T].
\end{align*}
Consequently, we have
\begin{align*}
    \Pr[\tau>T]\leq nr^T\leq n\exp\qty(-\frac{T}{4})\leq 1/n^2.
\end{align*}
Thus, combining the above, \cref{claim:taumaxminus,lem:taubeta,lem:taupsi} gives
\begin{align*}
  &\Pr\qty[\tau_{all}> T \text{ or } \taumaxminus(3/4)\leq T]\\
  &\leq \Pr\qty[\qty{\tau_{all}> T \text{ or } \taumaxminus(3/4)\leq T} \text{ and } \min\{\taumaxminus(3/4),\taubetaminus,\taupsiplus\}>T]\\
  &+\Pr\qty[\min\{\taumaxminus(3/4),\taubetaminus,\taupsiplus\}\leq T]\\
  &\leq \Pr\qty[\tau>T]+\Pr\qty[\taumaxminus(3/4)\leq T \text{ and } \min\{\taubetaminus,\taupsiplus\}>T] + \Pr\qty[\min\{\taubetaminus,\taupsiplus\}\leq T]\\
  &\leq 1/n.
\end{align*}
\end{proof}
\begin{proof}[Proof of \cref{lem:unique strong opinion evolves} (\cref{item:taucons})]
  Let $\tau=\min\{\taucons,\taubetaminus\}$.
  Let $u_t=1-\alpha_t(1)$. Then, 
  \begin{align*}
    \E_{t-1}[u_t]=1-\alpha_{t-1}(1)\qty(1+1+\alpha_{t-1}(1)-2\beta_{t-1})
    =u_{t-1}\qty(1-\alpha_{t-1}(1))
    \leq u_{t-1}\qty(1-\frac{1}{3}).
  \end{align*}
  Let $r=1-\frac{1}{3}=\frac{2}{3}$, $X_t=r^{-t}u_t$, and $Y_t=X_{t\wedge \tau}$.
Then, 
\begin{align*}
\E_{t-1}[Y_t]-Y_{t-1}
=\indicator_{\tau>t-1}\qty(\E_{t-1}[X_t]-X_{t-1})
\leq \indicator_{\tau>t-1}\qty(r^{-t}\E_{t-1}[u_t]-r^{-(t-1)}u_{t-1})
\leq 0,
\end{align*}
i.e., $Y_t$ is a submartingale.
Hence, we have $\E[Y_T]\leq \E[Y_0]=u_0\leq 1$ and 
\begin{align*}
\E[Y_T]
\geq \E[X_{T}\mid \tau>T]\Pr[\tau>T]
=r^{-T}\E[u_T\mid \tau>T]\Pr[\tau>T]
\geq r^{-T}n^{-1}\Pr[\tau>T].
\end{align*}
Consequently, we have
\begin{align*}
    \Pr[\tau>T]\leq nr^T\leq n\exp\qty(-\frac{T}{3})\leq 1/n^2.
\end{align*}
Thus, from \cref{lem:taubeta,lem:taupsi}, we have
\begin{align*}
  \Pr[\taucons>T]
  &\leq \Pr[\tau>T]+\Pr[\taubetaminus\leq T]\\
  &\leq \Pr[\tau>T]+\Pr[\taubetaminus\leq T \text{ and } \taupsiplus>T]+\Pr[\taupsiplus\leq T]\\
  &\leq 1/n.
\end{align*}
\end{proof}

\begin{proof}[Proof of \cref{lem:towards consensus}]
  We have the following:
\begin{itemize}
  \item From \cref{lem:unique strong opinion evolves} (\cref{item:tauplus}), for some $T_1=O\qty(\log n/\npm_0)$, we have $\alpha_{T_1}(1)\geq (1-\ctildemaxdown)\beta_{T_1}$, $\beta_{T_1}\geq 1/2-\xbeta$, and $\psi_{T_1}\leq \xpsi$, with probability at least $1-n^{-10}$.
  \item From \cref{lem:unique strong opinion evolves} (\cref{item:taumaxmajority}), for some $T_2=O(\log n)$, we have $\alpha_{T_1+T_2}(1)\geq 1-4\ctildemaxdown$, $\beta_{T_1+T_2}\geq 1/2-\xbeta$, and $\psi_{T_1+T_2}\leq \xpsi$, with probability at least $1-n^{-10}$.
  \item Assume that $\ctildemaxdown = 1/32$. From \cref{lem:unique strong opinion evolves} (\cref{item:tauall}), for some $T_3= O(\log n)$, we have $\alpha_{T_1+T_2+T_3}(1)=\beta_{T_1+T_2+T_3}$ and $\alpha_{T_1+T_2+T_3}(1)\geq 3/4$ with probability at least $1-1/n^2$.
  \item From \cref{lem:unique strong opinion evolves} (\cref{item:taucons}), for some $T_4=O(\log n)$, we have $\alpha_{T_1+T_2+T_3+T_4}(1)=1$ with probability at least $1-1/n^2$.
\end{itemize}
Thus, we obtain the claim.
\end{proof}

\subsection{Putting All Together}

\begin{lemma}
    \label{lem:good configuration}
    Let $C>0$ be any constant.
    For some 
    \begin{align}
        T=
    \begin{cases}
        O\qty(\log n) & \qty(\text{if } k\le \frac{C\sqrt{n}}{(\log n)^2}), \\
        O\qty(\sqrt{n}(\log n)^3) & (\text{otherwise}),
    \end{cases}
    \label{eq:good configuration}
    \end{align}
    we have
    \begin{align*}
    \Pr\qty[\psi_T\leq \xpsi \text{ and } \beta_T\geq \frac{1}{2}-\xbeta \text{ and } \npm_T\geq \frac{(\log n)^{1.5}}{\sqrt{n}}] \geq 1 -\pbot -O\qty(\frac{\log n}{n}).
    \end{align*}
\end{lemma}
\begin{proof}
  Combining \cref{lem:taubeta,lem:taupsi,lem:growth of tnpt}, we have the following:
  \begin{enumerate}
  \item \label{item:good configuration 1} From \cref{lem:taupsi} (\cref{item:taupsiminus}) and the definition of $\pbot$, $\psi_1\leq \xpsi$ and $\beta_1\geq 1/n$ hold with probability at least $1-\pbot - n^{-\Omega(1)}$.
  \item \label{item:good configuration 2} From \cref{lem:taubeta} (\cref{item:taubeta_logn,item:taubetaplus}), 
  for some $T=O(\log n)$, $\beta_{1+T}\geq 1/2-\xbeta$ holds with probability at least $1-O(\log n/n)$.
  \item \label{item:good configuration 3} From \cref{lem:taupsi} (\cref{item:taupsiplus}), $\max_{t\in [n^2]}\psi_t\leq \xpsi$ holds with probability at least $1-n^{-10}$.
  \item \label{item:good configuration 4} From \cref{lem:growth of tnpt}, for some $T'=O\qty(\sqrt{n}(\log n)^3)$, $\npt_{1+T+T'}\geq (\log n)^2/\sqrt{n}$, $\psi_{1+T+T'}\leq \xpsi$, and $\beta_{1+T+T'}\geq 1/2-\xbeta$ with probability at least $1-n^{-\Omega(1)}$.
  \end{enumerate}
  Since $\npm_t\geq \npt_t/\beta_t\geq \npt_t$ holds for any $t$, combining \cref{item:good configuration 1,item:good configuration 2,item:good configuration 3,item:good configuration 4}, we obtain the claim for general $k$.
 
  For the case where $k=O\qty(\sqrt{n}/(\log n)^2)$, 
  $\npm_t\geq \beta_t/k=\Omega(\beta_t(\log n)^2/\sqrt{n})$ holds for any $t$.
  Hence, from \cref{item:good configuration 1,item:good configuration 2,item:good configuration 3}, we obtain the claim.

\end{proof}

\begin{lemma}
    \label{lem:consensus time starting from a goood intial configuration}
    Suppose that $\psi_0\leq \xpsi$, $\beta_0\geq 1/2-\xbeta$, and $\npm_0=\omega(\log n/\sqrt{n})$.
    Then, $\taucons= O\qty(\log n/\npm_0)$ with high probability. 
\end{lemma}
\begin{proof}
  We have the following:
  \begin{itemize}
    \item From \cref{lem:unique strong opinion} (\cref{item:tauusplus}), for some $T_1=O\qty(\log n/\npm_0)$, we have $\min_{j\neq I_{T_1}}\eta_{T_1}(j)\geq x_{\eta}$, $\beta_{T_1}\geq 1/2-\xbeta$, $\psi_{T_1}\leq \xpsi$, and $\npm_{T_1}=\Omega(\npm_0)=\omega(\log n/\sqrt{n})$ with probability at least $1-n^{-10}$.
    \item From \cref{lem:towards consensus}, for some $T_2=O\qty(\log n/\npm_0)$, we have $\taucons\leq T_1+T_2$ with probability at least $1-1/n$.
  \end{itemize}
  Thus, we obtain the claim.
\end{proof}
\begin{proof}[Proof of \cref{thm:main}]
    From \cref{lem:good configuration}, we have that $\psi_T\leq \xpsi$ and $\beta_T\geq 1/2-\xbeta$ and $\npm_T\geq (\log n)^{1.5}/\sqrt{n}$ hold with high probability for some $T$ as defined in \cref{eq:good configuration} (\cref{lem:good configuration}).
    Then, from \cref{lem:consensus time starting from a goood intial configuration}, we reach a consensus within additional $O\qty(\frac{\log n}{\npm_T}) = O\qty( \min\qty{ k\log n, \sqrt{n/\log n} } )$ rounds with high probability.
    Here, we use $\npm_T \ge \frac{\beta_T}{k} = \Omega(k)$ if $k$ is small and $\npm_T\ge (\log n)^{1.5}/\sqrt{n}$ if $k$ is large.
    Therefore, the consensus time is bounded by
    \begin{align*}
        \taucons \le T + O\qty( \min\qty{ k\log n, \sqrt{n/\log n} } ) = \Otilde(\min\qty{k, \sqrt{n}})
    \end{align*}
    and obtain the claim.
\end{proof}
\subsection{Lower Bound}

In this subsection, we prove \cref{thm:main_tight}.
\begin{lemma}
    \label{lem:consensus time lower bound exercise}
    If $\beta_0\ge 1/2 - \xbeta$, $\psi_0\le \xpsi$, and $\npm_0 \ge C\sqrt{\frac{\log n}{n}}$ for a sufficiently large constant $C>0$, then $\taucons=\Omega(1/\npm_0)$ with high probability.
\end{lemma}

\begin{proof}
    Let $T=\frac{\constrefs{lem:hitting time for tnpm and npm}{item:taualphamaxup}}{\npm_0}$.
    From \cref{lem:hitting time for tnpm and npm} (\cref{item:taualphamaxup}) and \cref{lem:taubeta} (\cref{item:taubetaminus}), we have
    \begin{align*}
    \Pr\qty[\taucons\leq T]
    &\leq \Pr\qty[\taumaxup\leq T]\\
    &\leq \Pr\qty[\taumaxup\leq T \text{ and } \taubetaminus> T]+\Pr\qty[\taubetaminus\leq T]\\
    &\leq \Pr\qty[\taumaxup\leq \min\{T,\taubetaminus\}]+\Pr\qty[\taubetaminus\leq T]\\
    &\leq k\underbrace{\exp\qty(-\Omega\qty(\frac{\npm_0n}{T}))}_{=\exp(-\Omega((\npm_0)^2 n))}+T\exp\qty(-\Omega\qty(n\xbeta^2))\\
    &\leq 1/n.
    \end{align*}
\end{proof}

\begin{proof}[Proof of \cref{thm:main_tight}]
    For a constant $c>0$ such that $k\le (1/2-c)\cdot n$, define
    \[
    k^* = \frac{2\sqrt{n}}{c\log n}.
    \]

    Suppose that $k \le k^*$.
    Consider the initial configuration defined by
    \[
        \alpha_0(i) = \frac{1}{2k} \text{ for all } i\in [k].
    \]
    Then, $\beta_0 = 1/2$, $\psi_0 \le 0$, and $\npm_0 = \frac{1}{2k} \ge \Omega\qty(\frac{\log n}{\sqrt{n}})$.
    Therefore, from \cref{lem:consensus time lower bound exercise}, with high probability, $\taucons=\Omega(1/\npm_0)=\Omega(k)$.

    Suppose that $k > k^*$.
    Consider the initial configuration defined by
    \begin{align*}
        \alpha_0(i) = \begin{cases}
            \frac{1}{2k^*} - \frac{k-k^*}{nk^*} & \text{if } i\in [k^*], \\
            \frac{1}{n} & \text{if } k^* < i \le k.
        \end{cases}
    \end{align*}
    
    Then, $\beta_0 = 1/2$, $\psi_0 \le 0$, and
    \begin{align*}
        \npm_0 &= \frac{1}{2k^*} \qty(1 - \frac{2k-2k^*}{n}) \\
        &\ge \frac{1}{2k^*} \qty(1 - \frac{2k}{n}) \\
        &\ge \frac{c}{2k^*} & & \because k\le (1/2-c)\cdot n \\
        &= \Omega\qty(\frac{\log n}{\sqrt{n}}).
    \end{align*}
    Therefore, from \cref{lem:consensus time lower bound exercise}, with high probability, $\taucons=\Omega(1/\npm_0)=\Omega(k^*) = \widetilde{\Omega}(\sqrt{n})$.
\end{proof}

\section{Analysis for the Population Protocol Model}\label{sec:proof_pp}
Throughout this section, we recall quantities and hitting times defined in \cref{def:key-quantities,def:stopping_times}.
We frequently use results from \cref{sec:basic properties in the population protocol model}.

The key distinction from the gossip model is that, in the population protocol model, we introduce a \emph{margin} to the parameters in $\taubetaminus$ and $\taupsiplus$.
In the gossip model, we can show that once $\beta_t$ exceeds the threshold $1/2-\xbeta$, it never drops below this threshold due to concentration in every round.
In contrast, in the population protocol model, even after $\beta_t$ exceeds $1/2-\xbeta$, it is still possible for it to fall below the threshold.
To address this, we instead prove that $\beta_t$ does not fall below a lower threshold, $1/2-2\xbeta$, for a sufficiently long period of time (\cref{item:hitting times for beta for PP:keeping phase} of \cref{lem:taubeta for PP}).
We also introduce a similar margin for the parameter $\xpsi$ in the definition of $\taupsiplus(x)$ (\cref{item:taupsiplus for PP} of \cref{lem:taupsi for PP}).

For clarity and readability, we present complete proofs for the population protocol model separately from those for the gossip model, including all necessary lemmas and arguments. While the two models share a common high-level structure, interleaving their proofs would obscure the exposition and hinder readability.


\subsection{Behavior of the Fraction of Decided Vertices}
Consider the stopping times defined in \cref{def:stopping_times}.
In this section, we present the following two lemmas:
(i) $\beta_t$ reaches at least $1/2 - x$ within $O(n \log n)$ steps and then stays at least $1/2 - O(x)$ for a sufficiently long period (\cref{lem:taubeta for PP}); and
(ii) $\psi_t$ drops to at most $x$ within $O(n \log n)$ steps and then remains at most $O(x)$ for a sufficiently long period (\cref{lem:taupsi for PP}).
Specifically, we intend to apply the following results for $x$ such that $x = \omega(\sqrt{\log n / n})$ and $x = o(\log n / \sqrt{n})$.

\begin{lemma}[Growth of $\beta_t$]
  \label{lem:taubeta for PP}
  Let $x=x(n)$ be an arbitrary positive function such that $x=\omega(\sqrt{\log n / n})$ and $x=o(\log n/\sqrt{n})$.
  We have the following:
  \begin{enumerate}
    \item \label{item:taubeta_logn for PP}
    Let $C>0$ be an arbitrary constant.
    Suppose that $\beta_0>0$.
    Then, for some $\ell=O(\log n)$,
    $\beta_{\ell n}\geq \frac{C\log n}{n}$ with probability at least $1-O(n^{-10})$.
    \item \label{item:taubetaplus for PP}
    For some $\ell=O(\log n)$, 
    $\Pr[\tau_\beta^+(x)>\ell n]\leq 
    \ell\qty(\exp\qty(-\Omega\qty(n\beta_0))+\exp\qty(-\Omega\qty(nx^2))).$
    \item \label{item:hitting times for beta for PP:keeping phase}
    Suppose $\beta_0\geq 1/2-x$.
    Then, for any $T>0$, $\Pr\qty[\tau_\beta^-(2x)\leq T]\leq T\exp\qty(-\Omega\qty(nx^2))$.
  \end{enumerate}
\end{lemma}
\begin{lemma}[Decay of $\psi_t$]
  \label{lem:taupsi for PP}
  Let $x=x(n)$ be an arbitrary positive function such that $x=\omega(\sqrt{\log n / n})$ and $x=o(\log n/\sqrt{n})$.
  We have the following:
  \begin{enumerate}
    \item 
    \label{item:taupsiminus for PP}
    Suppose $\psi_0\geq x$. 
    Then, for some $\ell=O(\log n)$,
   $\Pr[\tau_\psi^-(x)>\ell n]\leq \ell\exp\qty(-\Omega\qty(nx^2))$.
    \item \label{item:taupsiplus for PP} 
    Suppose $\psi_0\leq x$.
    Then, for any $T\geq 1$, $\Pr[\taupsiplus(2x)\leq T]\leq T\exp\qty(-\Omega\qty(nx^2))$.
  \end{enumerate}
\end{lemma}
\begin{proof}[Proof of \cref{item:taubeta_logn for PP} of \cref{lem:taubeta for PP}]
  For positive constants $c_\beta^\uparrow, c_\beta^\downarrow\in (0,1)$, let $\tau_\beta^\uparrow \defeq \inf\{t\geq 0: \beta_t\geq (1+c_\beta^\uparrow)\beta_0\}$ and $\tau_\beta^\downarrow \defeq \inf\{t\geq 0: \beta_t\leq (1-c_\beta^\downarrow)\beta_0\}$.
  For $y\leq 1/2-c$, let $\tau=\min\{\tau_\beta^\uparrow,\tau_\beta^\downarrow,\taubetaplus(y)\}$.
  Then, for $\tau>t-1$, we have
  \begin{align}
    \E_{t-1}[\beta_t]=\beta_{t-1}+\frac{\beta_{t-1}\qty(1-2\beta_{t-1})+\gamma_{t-1}}{n}
    \geq \beta_{t-1}+\frac{2c(1-c_\beta^\downarrow)\beta_0}{n}.
    \label{eq:beta hitting time for population protocol model:1}
  \end{align}
  Note that we use \cref{lem:basic inequalities for beta PP} (\cref{item:expectation of beta PP}) in the first equality.

  Hence, by setting $X_t=\beta_t$ and $\drift = \frac{2c(1-c_\beta^\downarrow)\beta_0}{n}> 0$, we obtain
  $
    \indicator_{\tau>t-1}\qty(X_{t-1}+\drift-\E_{t-1}[X_t])
    =\indicator_{\tau>t-1}\qty(\beta_{t-1}+\drift-\E_{t-1}[\beta_t])\leq 0.
  $
  Moreover, $\indicator_{\tau>t-1}\qty(\E_{t-1}[X_t]-X_t)=\indicator_{\tau>t-1}\qty(\E_{t-1}[\beta_t]-\beta_t)$ satisfies $\qty(\frac{1}{n},\frac{(1+c_\beta^\uparrow)\beta_0}{n^2})$-Bernstein condition. 
  Note that we use \cref{lem:basic inequalities for beta PP} (\cref{item:bernstein condition of beta PP}), $\indicator_{\tau>t-1}\beta_{t-1}\leq (1+c_\beta^\downarrow)\beta_0$, and \cref{lem:Bernstein condition} (\cref{item:BC for upper bounded rv,item:BC for linear transformation}).

  Applying \cref{lem:Useful drift lemma} (\cref{item:positive drift useful}) with $I^+=(1+c_\beta^\uparrow)\beta_0$, $I^-=(1-c_\beta^\downarrow)\beta_0$, and 
  $T'= \frac{(1+\epsilon)(I^+-X_0)}{\drift}=\frac{(1+\epsilon)c_\beta^\uparrow}{2c(1-c_\beta^\downarrow)}n=\Theta(n)$ for an arbitrary constant $\epsilon>0$, we have
  \begin{align}
    \Pr\qty[\min\{\tau_\beta^\uparrow,\taubetaplus(y)\} >T']
    \leq \exp\qty(-\Omega\qty(\frac{\beta_0^2}{\frac{\beta_0}{n^2}\cdot n+\frac{\beta_0}{n}}))
    \leq \exp\qty(-\Omega\qty(n\beta_0)).
    \label{eq:beta hitting time for population protocol model:1/4}
  \end{align}
  
  \newcommand{\tauzero}{\tau_\mathrm{zero}}
  We apply \cref{lem:nazo lemma} for $Z_t=\opn_t$, $T=T'$, $\varphi(Z_t)=\sqrt{n \beta_t}$, $\tau=\taubetaplus(y)$, $\cphiup=c_\beta^\uparrow$, $x_0 = 1$, and $x^*=\sqrt{C\log n}$. 
  From these settings, we have
  $\tauphiplus(x_0) = \inf\{t\geq 0:\sqrt{n\beta_t}\geq 1\}=0$,
  $\tauphiplus(x_*) = \inf\{t\geq 0:\beta_t\geq C\log n/n\}$, and
  $\tauphiup = \inf\{t\geq 0:\sqrt{n\beta_t}\geq (1+\cphiup)\sqrt{n\beta_0}\}=\tau_\beta^\uparrow$.
  Hence,
  \begin{align*}
  \Pr\qty[\min\qty{\tauphiplus(x_0), \tau}\leq T'] = 1
  \end{align*}
  holds, i.e., the first condition of \cref{lem:nazo lemma} is satisfied for $C_1=1$.
  Next, for any configuration of $\beta_0\leq C\log n/n$, we have
  \begin{align*}
    \Pr\qty[\min\qty{\tauphiup, \tau}>T']
    \geq 1-\exp\qty(-\Omega(n\beta_0))
    =1-\exp\qty(-\Omega(\varphi(Z_0)^2)).
  \end{align*}
  i.e., the second condition of \cref{lem:nazo lemma} is satisfied for some positive constant $C_2>0$.
  Note that we use \cref{eq:beta hitting time for population protocol model:1/4} in the second inequality.
  
  Thus, from \cref{lem:nazo lemma} with $\varepsilon=n^{-10}$, 
  for some $\ell=O(\log n)$,
  $
    \Pr\qty[\tauphiplus(x_*)>\ell T']
    \leq  n^{-10}.
$
\end{proof}
\begin{proof}[Proof of \cref{item:taubetaplus for PP} of \cref{lem:taubeta for PP}]
  Let $c$ be an arbitrary constant such that $0<c<1/2$.
  We split the proof into two phases: $0<\beta_0< \frac{1}{2}-c$ and $\frac{1}{2}-c\leq \beta_0< \frac{1}{2}-x$.
  \paragraph{Phase 1. $0<\beta_0< \frac{1}{2}-c$.}

  Thus, applying \cref{eq:beta hitting time for population protocol model:1/4} for some $\ell=O(\log n)$ times, $\beta_t$ gets greater than or equal to $1/2-c$ within $\ell n=O(n\log n)$ steps with probability at least $1-\ell\exp\qty(-\Omega\qty(n\beta_0))$.

  \paragraph{Case 2. $\frac{1}{2}-c\leq \beta_0< \frac{1}{2}-x$.}
  In this case, consider the parameter $B_t=1-2\beta_t$ (note that $\beta_t=(1-B_t)/2$ and $2x \leq B_0\leq 2c$).
  For positive constants $c_B^\uparrow, c_B^\downarrow\in (0,1)$, let $\tau_B^\uparrow \defeq \inf\{t\geq 0: B_t\geq (1+c_B^\uparrow)B_0\}$, $\tau_B^\downarrow \defeq \inf\{t\geq 0: B_t\leq (1-c_B^\downarrow)B_0\}$, 
  $\tau^{2x} \defeq \inf\{t\geq 0: B_t\leq 2x\}=\tau_\beta^+(x)$, 
  and $\tau=\min\{\tau_B^\uparrow,\tau_B^\downarrow,\tau^{2x}\}$.
  For $\tau>t-1$, we have
  \begin{align}
    \E_{t-1}[B_t]=1-2\E_{t-1}[\beta_t]
    = B_{t-1}-\frac{B_{t-1}(1-B_{t-1})+\gamma_{t-1}}{n}
    \leq B_{t-1}-\frac{(1-c_B^\downarrow)(1-(1+c_B^\downarrow)B_0)B_0}{n}. \label{eq:B hitting time for population protocol model:1}
  \end{align}
  We use \cref{lem:basic inequalities for beta PP} (\cref{item:expectation of beta PP}) in the first equality.

  Hence, letting $X_t=-B_t$ and $\drift = \frac{(1-c_B^\downarrow)(1-(1+c_B^\downarrow)B_0)B_0}{n}\geq 0$, we have
  $
    \indicator_{\tau>t-1}\qty(X_{t-1}+\drift-\E_{t-1}[X_t])
    =\indicator_{\tau>t-1}\qty(\E_{t-1}[B_t]-B_{t-1}+\drift)\leq 0
  $
  and $\indicator_{\tau>t-1}\qty(\E_{t-1}[X_t]-X_t)=2\indicator_{\tau>t-1}\qty(\E_{t-1}[\beta_t]-\beta_t)$ satisfies $\qty(\frac{2}{n},\frac{4}{n^2})$-Bernstein condition. 
  Note that we use \cref{lem:basic inequalities for beta PP} (\cref{item:bernstein condition of beta PP}), $\beta_{t-1}\leq 1$, and \cref{lem:Bernstein condition} (\cref{item:BC for upper bounded rv,item:BC for linear transformation}).

Applying \cref{lem:Useful drift lemma} (\cref{item:positive drift useful}) with $I^+=X_0+c^\downarrow_BB_0$, $I^-=X_0-c^\uparrow_BB_0$, and $T' = \frac{(1+\epsilon)c_B^\downarrow}{(1-c_B^\downarrow)(1-(1+c_B^\downarrow)B_0)}n=\Theta(n)$ for an arbitrary constant $\epsilon>0$, we have
  \begin{align}
    \Pr\qty[\min\{\tau_B^\downarrow,\tau_\beta^+(x) \} >T']
    =
    \Pr\qty[\min\{\tau_B^\downarrow,\tau^{2x} \} >T']
    \leq \exp\qty(-\Omega\qty(nB_0^2))
    \leq \exp\qty(-\Omega\qty(nx^2)).
    \label{eq:beta hitting time for population protocol model:B}
  \end{align}
  Thus, applying \cref{eq:beta hitting time for population protocol model:B} for some $\ell = O(\log n)$ times, $B_t$ gets less than or equal to $2x$, i.e., $\beta_t\geq 1/2-x$ within $T=O(n\log n)$ steps with probability at least $1-\ell \exp\qty(-\Omega\qty(nx^2))$.
\end{proof}
\begin{proof}[Proof of \cref{item:hitting times for beta for PP:keeping phase} of \cref{lem:taubeta for PP}]
  Let $I^-=1/2-2x$, $I_*^-=1/2-x$, $I_*^+=1/2-x/2$ and $I^+=1/2-x/4$.
  Note that $I^+-I^+_*=I^+_*-I^-_*=I^-_*-I^-=\Theta(x)$.
  Let $\tau^+=\inf\{t\geq 0: \beta_t\geq I_*^+\}$ and $\tau=\min\{\tau^+,\tau_\beta^-(2x)\}$.
  Then, for $\tau>t-1$, we have
  \begin{align}
    \E_{t-1}[\beta_t]
    &=\beta_{t-1}+\frac{\beta_{t-1}\qty(1-2\beta_{t-1})+\gamma_{t-1}}{n}
    \geq \beta_{t-1}+\frac{x}{6n}.
  \end{align}
  We use \cref{lem:basic inequalities for beta PP} in the first equality.
  Hence, letting $X_t=\beta_t$ and $\drift = \frac{x}{6n}>0$, we have
  $
    \indicator_{\tau>t-1}\qty(X_{t-1}+\drift-\E_{t-1}[X_t])
    =\indicator_{\tau>t-1}\qty(\beta_{t-1}+\frac{x}{6n}-\E_{t-1}[\beta_t])\leq 0
  $
  and $\indicator_{\tau>t-1}\qty(\E_{t-1}[X_t]-X_t)=\indicator_{\tau>t-1}\qty(\E_{t-1}[\beta_t]-\beta_t)$ satisfies $\qty(\frac{1}{n},\frac{1}{n^2})$-Bernstein condition. 
  Note that we use \cref{lem:basic inequalities for beta PP} (\cref{item:bernstein condition of beta PP}), $\beta_{t-1}\leq 1$, and \cref{lem:Bernstein condition} (\cref{item:BC for upper bounded rv,item:BC for linear transformation}).

  Applying \cref{lem:Useful drift lemma} (\cref{item:positive drift useful}) for $\beta_0\in [I_*^+,I^+]$, we have
  \begin{align}
    \Pr\qty[\tau_\beta^-(2x)<\tau^+]\leq \exp\qty(-\Omega\qty(nx^2)).
  \end{align}
  Consequently, from \cref{lem:Iterative Drift theorem} (\cref{item:no increase under negative drift}) with $\tau_s^-=\inf\{t\geq s: \beta_t\leq I^-\}$ and $\tau_s^+=\inf\{t\geq s: \beta_t\geq I^+\}$, we obtain
    \begin{align*}
      \Pr\qty[\tau^-_\beta(2x)\leq T]
      \leq \sum_{s=0}^{T-1}\E\qty[\indicator_{\beta_s\in [I^-_*,I^{+}_*]}\Pr_s\qty[\tau^-_s<\tau^+_s]]\leq T\exp\qty(-\Omega\qty(n x^2)).
  \end{align*}
\end{proof}

\begin{proof}[Proof of \cref{item:taupsiminus for PP} of \cref{lem:taupsi for PP}]
  Let $c_\psi^\uparrow,c_\psi^\downarrow\in (0,1)$ be constants, 
  $\tau_\psi^\uparrow=\inf\{t\geq 0: \psi_t\geq (1+c_\psi^\uparrow)\psi_0\}$, and
  $\tau_\psi^\downarrow=\inf\{t\geq 0: \psi_t\leq (1-c_\psi^\downarrow)\psi_0\}$.
  Write $\tau=\min\{\tau_\psi^\uparrow,\tau_\psi^\downarrow,\tau_\psi^-(x)\}$.

  For any $\tau>t-1$, we have
  \begin{align}
    \E_{t-1}[\psi_t]
    &=\psi_{t-1}-\frac{\psi_{t-1}}{n}+\frac{\beta_{t-1}-\gamma_{t-1}}{n^2}
    \leq \psi_{t-1}-\frac{\psi_{t-1}}{n}+\frac{\psi_{t-1}}{2n}
    \leq \psi_{t-1}-\frac{(1-c_\psi^\downarrow)\psi_{0}}{2n}.
  \end{align}
  We use \cref{lem:basic inequalities for psi PP} (\cref{item:expectation of psi PP}) in the first equality.
  Note that $\psi_{t-1}\geq x\geq \frac{2}{n}$.

  Hence, letting $X_t=-\psi_t$ and $\drift = \frac{(1-c_\psi^\downarrow)\psi_{0}}{2n}\geq 0$, we have
  $
    \indicator_{\tau>t-1}\qty(X_{t-1}+\drift-\E_{t-1}[X_t])
    =\indicator_{\tau>t-1}\qty(\E_{t-1}[\psi_t]-\psi_{t-1}+\frac{(1-c_\psi^\downarrow)\psi_{0}}{2n})\leq 0
  $
  and $\indicator_{\tau>t-1}\qty(\E_{t-1}[X_t]-X_t)=\indicator_{\tau>t-1}\qty(\psi_t-\E_{t-1}[\psi_t])$ satisfies $\qty(O(1/n),O(1/n^2))$-Bernstein condition.
  Note that we use \cref{lem:basic inequalities for psi PP} (\cref{item:bernstein condition of psi PP}), $\beta_{t-1}\leq 1$, and \cref{lem:Bernstein condition} (\cref{item:BC for upper bounded rv,item:BC for linear transformation}). 

  Applying \cref{lem:Useful drift lemma} (\cref{item:positive drift useful}) with $I^+=X_0+c^\downarrow_\psi \psi_0$, $I^-=X_0-c^\uparrow_\psi \psi_0$, and $T'=\frac{2(1+\epsilon)c_\psi^\downarrow}{1-c_\psi^\downarrow}n=\Theta(n)$ for an arbitrary constant $\epsilon>0$, we have
  \begin{align}
    \Pr\qty[\min\{\tau_\psi^\downarrow,\tau_\psi^- \} >T']
    \leq \exp\qty(-\Omega\qty(n\psi_0^2))
    \leq \exp\qty(-\Omega\qty(nx^2)).
    \label{eq:psi hitting time for population protocol model:down}
  \end{align}
  Thus, applying \cref{eq:psi hitting time for population protocol model:down} for some $\ell=O(\log n)$ times repeatedly, $\psi_t$ gets less than or equal to $x$ within $\ell n=O(n\log n)$ steps with probability at least $1-\ell\exp\qty(-\Omega\qty(nx^2))$.
\end{proof}
\begin{proof}[Proof of \cref{item:taupsiplus for PP} of \cref{lem:taupsi for PP}]
  Let $\tau^{x/4}=\inf\{t\geq 0: \psi_t\leq x/4\}$.
  Let $\tau=\min\{\tau_\psi^+(2x),\tau^{x/4}\}$.
  Then, for $\tau>t-1$, we have
  \begin{align}
    \E_{t-1}[\psi_t]
    &=\psi_{t-1}-\frac{\psi_{t-1}}{n}+\frac{\beta_{t-1}-\gamma_{t-1}}{n^2}
    \leq \psi_{t-1}-\frac{\psi_{t-1}}{n}+\frac{\psi_{t-1}}{2n}
    \leq \psi_{t-1}-\frac{x}{8n}.
  \end{align}
  We use \cref{lem:basic inequalities for psi PP} in the first equality.
  Note that $\psi_{t-1}\geq x/4\geq \frac{2}{n}$.
  Hence, letting $X_t=-\psi_t$ and $\drift = \frac{x}{8n}>0$, we have
  $
    \indicator_{\tau>t-1}\qty(X_{t-1}+\drift-\E_{t-1}[X_t])
    =\indicator_{\tau>t-1}\qty(\E_{t-1}[\psi_t]-\psi_{t-1}+\frac{x}{8n})\leq 0
  $
  and $\indicator_{\tau>t-1}\qty(\E_{t-1}[X_t]-X_t)=\indicator_{\tau>t-1}\qty(\psi_t-\E_{t-1}[\psi_t])$ satisfies $\qty(O(1/n),O(1/n^2))$-Bernstein condition. 
  Note that we use \cref{lem:basic inequalities for psi PP} (\cref{item:bernstein condition of psi PP}), $\beta_{t-1}\leq 1$, and \cref{lem:Bernstein condition} (\cref{item:BC for upper bounded rv,item:BC for linear transformation}). 

  Let $I^-=-2x$, $I^-_*=-x$, $I^+_*=-x/2$, and $I^+=-x/4$.
  Let $\tau^+=\inf\{t\geq 0: X_t\geq I^+\}=\tau^{x/4}$ and $\tau^-=\inf\{t\geq 0: X_t\leq I^-\}=\tau_\psi^+(2x)$.
  Note that $I^+-I^+_*=I^+_*-I^-_*=I^-_*-I^-=\Theta(x)$.
  Applying \cref{lem:Useful drift lemma} (\cref{item:positive drift useful}) for $X_0\in [I_*^-,I^+_*]$, i.e., $\psi_0\in [x/2,x]$, we have
  \begin{align}
    \Pr[\tau^-<\tau^+]=\Pr\qty[\taupsiplus(2x)<\tau^-]\leq \exp\qty(-\Omega\qty(nx^2)).
  \end{align}
  Consequently, applying \cref{lem:Iterative Drift theorem} (\cref{item:no increase under negative drift}) with $\tau_s^-=\inf\{t\geq s: X_t\leq I^-\}=\inf\{t\geq s:\psi_t\geq 2x\}$ and $\tau_s^+=\inf\{t\geq s: X_t\geq I^+\}=\inf\{t\geq s:\psi_t\leq x/4\}$, we obtain
    \begin{align*}
      \Pr\qty[\taupsiplus (2x)\leq T]
      \leq \sum_{s=0}^{T-1}\E\qty[\indicator_{\psi_s\in [I^-_*,I^{+}_*]}\Pr_s\qty[\tau^-_s<\tau^+_s]]\leq T\exp\qty(-\Omega\qty(nx^2)).
  \end{align*}
\end{proof}

\subsection{Behavior of Squared \texorpdfstring{$\ell^2$}{ell-2} Norm of Normalized Population}
In this section, we show that $\npt_t$ reaches $\xgamma$ within $O(n^2\xgamma \log n)$ steps with high probability (\cref{lem:growth of tnpt for PP}).
Recall $\taugammaplus=\inf\{t\geq 0: \gamma_t\geq \xgamma\}$ for some $\xgamma=(\log n)^2/\sqrt{n}$ (see \cref{def:stopping_times}). 
The main difference from the gossip model is the magnitude of the drift of $\tnpt_t$:
In the gossip model, we have $\E_{t-1}[\tnpt_t]\geq \tnpt_{t-1}+\Omega(1/n)$, whereas in the population protocol model, $\E_{t-1}[\tnpt_t]\geq \tnpt_{t-1}+\Omega(1/n^2)$ holds.
\begin{lemma}[Growth of $\npt_t$]
  \label{lem:growth of tnpt for PP}
  Let $x=x(n)$ be an arbitrary positive function such that $x=\omega(\sqrt{\log n / n})$ and $x=o(\log n/\sqrt{n})$.
  Suppose $\beta_0\geq 1/2-x$.
  Then, for some $T = O(n^2\xgamma \log n)$, 
  \begin{align*}
    \Pr\qty[\taugammaplus>T \text{ or } \taubetaminus(2x)\leq T]\leq n^{-10}.
  \end{align*}
\end{lemma}
The key idea of the proof of \cref{lem:growth of tnpt for PP} is to combine the additive drift of $\tnpt_t$ (\cref{lem:basic inequalities for tnpt PP}) and the optional stopping theorem.
\begin{lemma}
  \label{lem:hitting time for gamma for population protocol model}
Let $x=x(n)$ be an arbitrary positive function such that $x=\omega(\sqrt{\log n / n})$ and $x=o(\log n/\sqrt{n})$.
  Suppose $\beta_0\geq 1/2-x$.
  Then, $\E\qty[\min\{\taugammaplus,\taubetaminus(x)\}]\leq 72 n^2\xgamma$.
\end{lemma}
\begin{proof}
  From \cref{lem:basic inequalities for tnpt PP},
  for $\min\{\taugammaplus,\taubetaminus(x)\}>t-1$, we have
  \begin{align*}
    \E_{t-1}\qty[\tnpt_t]
    &\geq \tnpt_{t-1}+\frac{1}{12n^2}.
  \end{align*}
  Let $\tau=\min\{\taugammaplus,\taubetaminus(x)\}$, $X_t=\tnpt_t-\frac{t}{12n^2}$ and $Y_t=X_{t\wedge \tau}$.
  Then, we have
  \begin{align*}
    \E_{t-1}\qty[Y_t-Y_{t-1}]
    &= \indicator_{\tau>t-1}\E_{t-1}\qty[X_t-X_{t-1}]
    = \indicator_{\tau>t-1}\qty(\E_{t-1}[\tnpt_t]-\frac{t}{12n^2}-\tnpt_{t-1}+\frac{t-1}{12n^2})
    \geq 0,
  \end{align*}
  i.e., $(Y_t)_{t\in \Nat_0}$ is a submartingale.
  Hence, from \cref{thm:OST}, we have $\E[Y_\tau]\geq \E[Y_0] = \tnpt_0\geq 0$.
  Thus, 
$  
 0\leq  \E[Y_\tau]=\E[X_\tau] = \E[\tnpt_\tau] -\frac{\E[\tau]}{12n^2}
 $
and we obtain 
\begin{align*}
  \E[\tau]\leq 12n^2\E[\tnpt_\tau]\leq 36n^2(\xgamma+1)\leq 72n^2\xgamma.
\end{align*}
\end{proof}

\begin{proof}[Proof of \cref{lem:growth of tnpt for PP}.]
  Let $\tau=\min\{\taugammaplus,\taubetaminus(2x)\}$ and $T'=72\e n^2\xgamma$.
  From the Markov inequality and \cref{lem:hitting time for gamma for population protocol model}, if $\beta_0\geq 1/2-2x$ for some $x=\omega(\sqrt{\log n/n})$ and $x=o(1)$ (i.e., $2x=\omega(\sqrt{\log n/n})$ and $2x=o(1)$), 
$
  \Pr\qty[\tau>T']
  \leq \frac{\E[\tau]}{T'}\leq \frac{1}{\e}. 
$
Hence, for any $\ell\geq 1$, the Markov property implies that
  \begin{align*}
    \Pr\qty[\tau>\ell T' \mid \tau> (\ell-1) T']
    &=\E\qty[\Pr_{(\ell-1)T'}\qty[\tau>\ell T'] \mid \tau> (\ell-1) T']
    \leq \frac{1}{\e}.
  \end{align*}
Thus, we obtain
\begin{align*}
  \Pr\qty[\tau>\ell T']
  &=\Pr\qty[\tau>\ell T' \mid \tau> (\ell-1) T']\Pr\qty[\tau> (\ell-1) T']
  \leq \frac{1}{\e}\Pr\qty[\tau> (\ell-1) T']
  \leq \cdots \leq \frac{1}{\e^\ell}.
\end{align*}
Applying \cref{lem:taubeta for PP} (\cref{item:hitting times for beta for PP:keeping phase}), 
we obtain 
\begin{align*}
  &\Pr\qty[\taugammaplus> \ell T' \text{ or } \taubetaminus(2x) \leq \ell T']\\
  &\leq \Pr\qty[\qty{\taugammaplus> \ell T' \text{ or } \taubetaminus(2x) \leq \ell T'} \text{ and } \taubetaminus(2x) > \ell T']+ \Pr\qty[\taubetaminus(2x)\leq \ell T']\\
  &\leq \frac{1}{\e^\ell} + \ell T'\exp\qty(-\Omega\qty(nx^2)). 
\end{align*}
Taking $\ell=C\log n$ for a sufficiently large constant $C>0$, we obtain the claim.
\end{proof}

\subsection{Behavior of Maximum Population}
\label{sec:bounded decrease of tnpm}
In this section, we show that 
(i) $\tnpm_t\geq (1-\ctildemaxdown)\tnpm_0$ holds for all $t\leq O(n^2\npm_0/\log n)$ with high probability (\cref{lem:hitting time for tnpm and npm for PP} (\cref{item:tautildemaxdown is large for PP})) and
(ii) $\npm_t\leq (1+\cmaxup)\tnpm_0$ holds for all $t\leq O(n/\npm_0)$ with high probability (\cref{lem:hitting time for tnpm and npm for PP} (\cref{item:taualphamaxup for PP})).
Intuitively, in the population protocol model, both the drift and the variance $\variance$ in the Bernstein condition for $\tnpm_t$ (or $\npm_t$) are roughly a factor of $1/n$ smaller compared to those in the gossip model.

\begin{lemma}[Key properties of $\tnpm_t$ and $\npm_t$]
  \label{lem:hitting time for tnpm and npm for PP}
  Let $x=x(n)$ be an arbitrary positive function such that $x=\omega(\sqrt{\log n / n})$ and $x=o(\log n/\sqrt{n})$.
  Suppose $\beta_0\geq 1/2-x$. 
  We have the following:
  \begin{enumerate}
    \item \label{item:tautildemaxdown is large for PP}
  Let $\constrefs{lem:hitting time for tnpm and npm for PP}{item:tautildemaxdown is large for PP}=\ctildemaxdown/72$ be a positive constant.
  Then, for any $T\leq \constrefs{lem:hitting time for tnpm and npm for PP}{item:tautildemaxdown is large for PP}n^2$, 
  \[\Pr\qty[\tautildemaxdown \leq T \text{ and } \taubetaminus(x)> T]\leq T\exp\qty(-\Omega\qty(\frac{\tnpm_0n^2}{T+n})).\]
  \item \label{item:taualphamaxup for PP}
  Suppose $\npm_0=\omega(\log n/\sqrt{n})$.
  Let $\constrefs{lem:hitting time for tnpm and npm for PP}{item:taualphamaxup for PP}=\frac{\cmaxup}{6(1+\cmaxup)^2}$.
  Then, for any $T\leq \frac{\constrefs{lem:hitting time for tnpm and npm for PP}{item:taualphamaxup for PP}n}{\npm_0}$, we have
  \begin{align*}
    \Pr\qty[\taumaxup\leq \min\qty{T,\taubetaminus(x)}]\leq k\exp\qty(-\Omega\qty(\frac{\npm_0 n^2}{T+n})).
  \end{align*}
  \end{enumerate}
\end{lemma}

The key tool for the proof of \cref{lem:hitting time for tnpm and npm for PP} (\cref{item:tautildemaxdown is large for PP}) is the following lemma.
\begin{lemma}
  \label{lem:tildemaxdown for population protocol model}
  Let $x=x(n)$ be an arbitrary positive function such that $x=\omega(\sqrt{\log n / n})$ and $x=o(\log n/\sqrt{n})$.
  Suppose $\beta_0\geq 1/2-x$.
  Let $\tau^\uparrow=\max\{\taumaxup,\tautildemaxup\}$.
  Let $\constrefs{lem:hitting time for tnpm and npm for PP}{item:tautildemaxdown is large for PP}=\ctildemaxdown/72$ be a positive constant defined in \cref{lem:hitting time for tnpm and npm for PP} (\cref{item:tautildemaxdown is large for PP}).
  Then, for any $T\leq \constrefs{lem:hitting time for tnpm and npm for PP}{item:tautildemaxdown is large for PP}n^2$, 
  \begin{align*}
    \Pr\qty[\tautildemaxdown \leq \min\{T,\tau^\uparrow,\taubetaminus(x)\}] \leq \exp\qty(- \Omega\qty(\frac{\tnpm_0 n^2}{T+n})).
  \end{align*}
\end{lemma}
\begin{proof}
  Let $\tau^\uparrow=\max\{\taumaxup,\tautildemaxup\}$ and $\tau=\min\{\tautildemaxdown,\tau^\uparrow,\taubetaminus(x)\}$.
  From \cref{lem:basic inequalities for tnpm PP} (\cref{item:expectation of tnpm PP}),
for $\tau>t-1$, we have
\begin{align*}
  \E_{t-1}\qty[\tnpm_t]
  &\geq \tnpm_{t-1}\qty(1+\frac{\npm_{t-1}-\gamma_{t-1}/\beta_{t-1}}{2n})-\frac{18\npm_{t-1}}{n^2}
  \geq \tnpm_{t-1}-\frac{36\tnpm_{0}}{n^2}.
\end{align*}
Note that
$\npm_{t-1}\leq (1+\cmaxup)\npm_{0}\leq 2\tnpm_{0}$ 
or $\npm_{t-1}\leq (1+\ctildemaxup)\tnpm_{0}\beta_{t-1}\leq 2\tnpm_{0}$, i.e., $\npm_{t-1}\leq 2\tnpm_{0}$ holds for $\tau^\uparrow>t-1$.

Hence, letting $X_t=\tnpm_t$ and $\drift = -\frac{36\tnpm_{0}}{n^2}< 0$, we have
  \begin{align*}
    \indicator_{\tau>t-1}\qty(X_{t-1}+\drift-\E_{t-1}[X_t])
    =\indicator_{\tau>t-1}\qty(\tnpm_{t-1}-\frac{36\tnpm_{0}}{n^2}-\E_{t-1}[\tnpm_t])\leq 0.
  \end{align*}
  Furthermore, 
  \[\indicator_{\tau>t-1}\qty(X_{t-1}+\drift-X_t)\leq \indicator_{\tau>t-1}\qty(\tnpm_{t-1}\qty(1+\frac{\npm_{t-1}-\gamma_{t-1}/\beta_{t-1}}{2n})-\frac{18\npm_{t-1}}{n^2}-\tnpm_t)\]
   satisfies one-sided $\qty(O\qty(\frac{1}{n}),O\qty(\frac{\npm_0}{n^2}))$-Bernstein condition. 
Note that we use  \cref{lem:basic inequalities for tnpm PP} (\cref{item:bernstein condition of tnpm PP}) and \cref{lem:Bernstein condition} (\cref{item:BC for upper bounded rv,item:BC for linear transformation}).

  Applying \cref{lem:Useful drift lemma} (\cref{item:negative drift useful}) with $I^-=(1+\ctildemaxdown)\tnpm_{0}$, $I^+=(1+\ctildemaxup)\tnpm_{0}$, and $T\leq \frac{X_0-I^-}{-2\drift}=\frac{\ctildemaxdown}{72}n^2=\constrefs{lem:hitting time for tnpm and npm for PP}{item:tautildemaxdown is large for PP}n^2$, 
           \begin{align*}
            \Pr\qty[\tautildemaxdown \leq \min\{T,\tau^\uparrow,\taubetaminus(x)\}] 
            \le \exp\qty(- \Omega\qty(\frac{\tnpm_0 n^2}{T+n})).
           \end{align*}
\end{proof}

\begin{proof}[Proof of \cref{lem:hitting time for tnpm and npm for PP} (\cref{item:tautildemaxdown is large for PP})]
  For $s\geq 0$, let $\tau_s^\downarrow=\inf\{t\geq s: \tnpm_t\leq (1+\ctildemaxdown)\tnpm_{s}\}$ and $\tau_s^\uparrow=\inf\{t\geq s: \tnpm_t\geq (1+\ctildemaxup)\tnpm_{s}\}$.
  Applying \cref{lem:Iterative Drift theorem} (\cref{item:bounded decrease under small negative drift}) with $T\leq \constrefs{lem:hitting time for tnpm and npm for PP}{item:tautildemaxdown is large for PP}n^2$, we have
  \begin{align*}
    \Pr\qty[\tautildemaxdown \leq T \text{ and } \taubetaminus(x)> T]
    &\leq 
    \sum_{s=0}^{T-1}\E\qty[\indicator_{\tnpm_s\geq \tnpm_0 \text{ and } \taubetaminus(x)> s}\Pr_s\qty[\tau_s^\downarrow \leq \min\{T,\tau_s^\uparrow,\taubetaminus(x)\}]]
     \\
     &\leq \sum_{s=0}^{T-1}\E\qty[\indicator_{\tnpm_s\geq \tnpm_0 \text{ and } \taubetaminus(x)> s}\exp\qty(-\Omega\qty(\frac{\tnpm_sn^2}{T+n}))]
     \\
     &\leq T\exp\qty(-\Omega\qty(\frac{\tnpm_0n^2}{T+n})).
  \end{align*}
  Note that we apply \cref{lem:tildemaxdown for population protocol model} in the second inequality.
\end{proof}


\begin{proof}[Proof of \cref{lem:hitting time for tnpm and npm for PP} (\cref{item:taualphamaxup for PP})]
  Let $\tau=\min\{\taumaxup,\taubetaminus(x)\}$.
  For $\tau > t-1$, similarly to \cref{eq:alphamaxup drift}, we have
  \begin{align*}
    \E_{t-1}[\alpha_t(i)]-\alpha_{t-1}(i)
    &=\alpha_{t-1}(i) \cdot \frac{\alpha_{t-1}(i)+1-2\beta_{t-1}}{n}
    \leq \alpha_{t-1}(i)+\frac{3(1+\cmaxup)^2(\npm_{0})^2}{n}.
  \end{align*}
  Note that we use \cref{lem:basic inequalities for alpha PP} (\cref{item:expectation of alpha PP}).
  Hence, letting $X_t=-\alpha_t(i)$ and $\drift = -\frac{3(1+\cmaxup)^2(\npm_{0})^2}{n}<0$, we have
  \[
    \indicator_{\tau>t-1}\qty(X_{t-1}+\drift-\E_{t-1}[X_t])
    =\indicator_{\tau>t-1}\qty(\E_{t-1}[\alpha_t(i)]-\alpha_{t-1}(i)-\frac{3(1+\cmaxup)^2(\npm_{0})^2}{n})\leq 0
  \]
  and $\indicator_{\tau>t-1}\qty(\E_{t-1}[X_t]-X_t)=\indicator_{\tau>t-1}\qty(\alpha_t(i)-\E_{t-1}[\alpha_t(i)])$ satisfies $\qty(1/n,O(\npm_0/n^2))$-Bernstein condition.
  Note that we use  \cref{lem:basic inequalities for alpha PP} (\cref{item:bernstein condition of alpha PP}) and \cref{lem:Bernstein condition} (\cref{item:BC for upper bounded rv,item:BC for linear transformation}). 

  Let $\tau_i=\inf\{t\geq 0: \alpha_t(i)\geq (1+\cmaxup)\npm_0\}=\inf\{t\geq 0: X_t\leq -(1+\cmaxup)\npm_0\}$.
  Recall $\constrefs{lem:hitting time for tnpm and npm for PP}{item:taualphamaxup for PP}=\frac{\cmaxup}{6(1+\cmaxup)^2}$.
  Applying \cref{lem:Useful drift lemma} (\cref{item:negative drift useful}) with 
 $I^-=-(1+\cmaxup)\npm_0$, for $T\leq \frac{\constrefs{lem:hitting time for tnpm and npm for PP}{item:taualphamaxup for PP}n}{\npm_0}= \frac{X_0-I^-}{-2\drift}$, we have
 \begin{align*}
  \Pr\qty[\tau_i\leq \min\{T,\tau\}]\leq \exp\qty(-\Omega\qty(\frac{\npm_0 n^2}{T+n})).
 \end{align*}
 Thus, applying the union bound,
 \begin{align*}
  \Pr\qty[\taumaxup\leq \min\{T,\tau\}]
  &\leq \Pr\qty[\exists i\in [k]: \tau_i\leq \min\{T,\tau\}]
  \leq k\exp\qty(-\Omega\qty(\frac{\npm_0 n^2}{T+n})).
 \end{align*}
\end{proof}

The following lemma follows naturally from the preceding discussion.
Note that the bounds concerning to $\beta_t$ and $\psi_t$ are different from those for the gossip model.
\begin{lemma}
  \label{lem:tauinitial lemma for PP}
  Let $x=x(n)$ be an arbitrary positive function such that $x=\omega(\sqrt{\log n / n})$ and $x=o(\log n/\sqrt{n})$.
  Suppose that $\beta_0\geq 1/2-x$ and $\psi_0\leq x$. 
  Then, for any $T\leq \constrefs{lem:hitting time for tnpm and npm for PP}{item:tautildemaxdown is large for PP}n^2$,
  \begin{align*}
    \Pr\qty[\min\{\tautildemaxdown,\taubetaminus(2x),\taupsiplus(2x)\} \leq  T] \leq T\exp\qty(-\Omega\qty(\frac{n^2\tnpm_0}{T+n}))+T\exp\qty(-\Omega(nx^2)).
  \end{align*}
\end{lemma}
\begin{proof}
  Combining \cref{eq:min of two stopping times,lem:hitting time for tnpm and npm for PP} (\cref{item:tautildemaxdown is large for PP}), \cref{lem:taubeta for PP}, \cref{lem:taupsi for PP}, 
  \begin{align*}
    &\Pr\qty[\min\{\tautildemaxdown,\taubetaminus(2x),\taupsiplus(2x)\} \leq T] \\
    &\leq \Pr\qty[\tautildemaxdown \leq T \text{ and }\min\{\taubetaminus(2x),\taupsiplus(2x)\}>T]+\Pr\qty[\taubetaminus(2x)\leq T]+\Pr\qty[\taupsiplus(2x)\leq T] \nonumber\\
    &\leq T\exp\qty(-\Omega\qty(\frac{n^2\tnpm_0}{T+n}))
    +T\exp\qty(-\Omega\qty(nx^2)). 
  \end{align*}
\end{proof}

\subsection{Behavior of Gap between Two Opinions}
Recall the definition of the weak opinion and its stopping time in \cref{def:weak}.
The main results of this section is the following lemma.

\begin{lemma}[Either of two non-weak opinions becomes weak]
  \label{lem:pair of non-weak opinions for PP}
  Let $i,j\in [k]$ be an arbitrary pair of two non-weak opinions.
  Let $x=x(n)$ be an arbitrary positive function such that $x=\omega(\sqrt{\log n / n})$ and $x=o(\log n/\sqrt{n})$.
  Suppose $\psi_0\leq x$, $\beta_0\geq 1/2-x$, and $\npm_0=\omega(\log n/\sqrt{n})$.
  We have the following:
  \begin{enumerate}
    \item \label{item:pair of non-weak opinions becomes weak for PP}
    Let $C$ be an arbitrary positive constant.
    Then, for some $T=O(n\log n/\npm_0)$, we have
    \begin{align*}
        \Pr\qty[\min\qty{\taudeltaplus\qty(\sqrt{\frac{C\log n}{n}}),\tauiweak,\taujweak} > T \text{ or } \min\{\tautildemaxdown,\taubetaminus(2x),\taupsiplus(2x)\} \leq T] \le n^{-10}.
    \end{align*}
    \item \label{item:delta multipricative with initial bias for PP}
    Suppose $\delta_0(i,j)\geq C\sqrt{\log n/n}$ for a sufficiently large constant $C>0$.
  Then, for some $T=O\qty(n\log n/\npm_0)$,
  \begin{align*}
      \Pr\qty[ \taujweak>T \text{ or } \min\{\tautildemaxdown,\taubetaminus(2x),\taupsiplus(2x)\} \leq T  ] \le n^{-10}.
  \end{align*}
  \end{enumerate}
\end{lemma}

In the case of the population protocol model, both the drift and the variance $\variance$ in the Bernstein condition for $\delta_t$ become smaller by about a factor of $1/n$ compared to the gossip model.
We show the following lemma.
\begin{lemma}[Multiplicative and additive drifts of $\delta_t$]
  \label{lem:delta multiplicative and additive drift for PP}
  Let $i,j\in [k]$ be an arbitrary pair of two non-weak opinions.
  Let $x=x(n)$ be an arbitrary positive function such that $x=\omega(\sqrt{\log n / n})$ and $x=o(\log n/\sqrt{n})$.
  Suppose that $\beta_0\geq 1/2-x$, $\psi_0\leq x$, and $\npm_0=\omega(\log n/\sqrt{n})$.
  We have the following:
  \begin{enumerate}
    \item \label{item:taudeltaup for population protocol model}
    Suppose $\delta_0(i,j)\geq 0$. 
     Let $\cdeltaup= \constrefs{lem:hitting time for tnpm and npm for PP}{item:taualphamaxup for PP}\frac{(1-\ctildemaxdown)(1-\cdeltadown)\qty(1-2\cweak)}{12}$, where $\constrefs{lem:hitting time for tnpm and npm for PP}{item:taualphamaxup for PP}$ is a positive constant defined in \cref{lem:hitting time for tnpm and npm for PP} (\cref{item:taualphamaxup for PP}).
    Then, 
    \begin{align*}
     &\Pr\qty[\min\{\tau_\delta^\uparrow,\taujweak,\tautildemaxdown,\taubetaminus(2x),\taupsiplus(2x)\} >\frac{\constrefs{lem:hitting time for tnpm and npm for PP}{item:taualphamaxup for PP} n}{\npm_0}]\\
     &\leq \exp\qty(-\Omega\qty(n\delta_0(i,j)^2)) + n \exp\qty(-\Omega(n(\npm_0)^2)).
   \end{align*}
    \item \label{item:taudeltaplus const prob for population protocol model}
    Let $\xdelta=c/\sqrt{n}$, where $c\in (0,1)$ is an arbitrary constant.
    Let $\constrefs{lem:delta multiplicative and additive drift for PP}{item:taudeltaplus const prob for population protocol model}=\frac{24c}{(1-2\cweak)(1-\ctildemaxdown)}$.
    Then, 
    \begin{align*}
      \Pr\qty[\min\qty{\taudeltaplus,\tauiweak,\taujweak,\tautildemaxdown,\taubetaminus(x),\taupsiplus(x)} >\frac{\constrefs{lem:delta multiplicative and additive drift for PP}{item:taudeltaplus const prob for population protocol model} n}{\npm_0}]\leq \frac{1}{2}.
    \end{align*}
    
  \end{enumerate}
\end{lemma}

\begin{proof}[Proof of \cref{lem:delta multiplicative and additive drift for PP} (\cref{item:taudeltaup for population protocol model})]
Let $\tau^*=\min\{\taujweak,\taumaxup,\tautildemaxdown,\taubetaminus(2x),\taupsiplus(2x)\}$ and $\tau=\min\{\taudeltaup,\taudeltadown,\tau^*\}$.
For $\tau>t-1$, we have
\begin{align*}
  \E_{t-1}[\delta_t]
  &=\delta_{t-1}+\frac{\delta_{t-1}\beta_{t-1}\tnpm_{t-1}}{n}\qty(\frac{\alpha_{t-1}(i)+\alpha_{t-1}(j)}{\npm_{t-1}}-\frac{\gamma_{t-1}}{\beta_{t-1}\npm_{t-1}}-\frac{\psi_{t-1}}{\beta_{t-1}^2\tnpm_{t-1}})\\
  &\geq \delta_{t-1}+\frac{(1-\ctildemaxdown)(1-\cdeltadown)\qty(1-2\cweak)}{6}\cdot \frac{\delta_{0}\npm_{0}}{n}.
\end{align*}
Note that we use \cref{lem:basic inequalities for delta PP} (\cref{item:expectation of delta PP}) and \cref{eq:drift of delta key}.

  Letting $X_t=\delta_t$ and $\drift = \frac{(1-\ctildemaxdown)(1-\cdeltadown)\qty(1-2\cweak)}{6}\cdot \frac{\delta_{0}\npm_{0}}{n}> 0$, we have
  \[
    \indicator_{\tau>t-1}\qty(X_{t-1}+\drift-\E_{t-1}[X_t])
    =\indicator_{\tau>t-1}\qty(\delta_{t-1}+\drift-\E_{t-1}[\delta_t])\leq 0
  \]
  and $\indicator_{\tau>t-1}\qty(\E_{t-1}[X_t]-X_t)=\indicator_{\tau>t-1}\qty(\E_{t-1}[\delta_t]-\delta_t)$ satisfies $\qty(O(1/n),O(\npm_0/n^2))$-Bernstein condition.
  Note that we use \cref{lem:basic inequalities for delta PP} (\cref{item:bernstein condition of delta PP}) and \cref{lem:Bernstein condition} (\cref{item:BC for upper bounded rv,item:BC for linear transformation}). 

Let $I^+=(1+c)\delta_0$ and $I^-=(1-c)\delta_0$.
Recall that $\cdeltaup= \constrefs{lem:hitting time for tnpm and npm for PP}{item:taualphamaxup for PP}\frac{(1-\ctildemaxdown)(1-\cdeltadown)\qty(1-2\cweak)}{12}$. 
Let $T=\frac{\constrefs{lem:hitting time for tnpm and npm for PP}{item:taualphamaxup for PP}n}{\npm_0}$.
Then, we have $T= \frac{2(I^+-\delta_0)}{\drift}$.
Thus, we can apply \cref{lem:Useful drift lemma} (\cref{item:positive drift useful}) and obtain
  \begin{align*}
    \Pr\qty[\min\{\tau_\delta^\uparrow,\tau^*\} >T]\
    \leq \exp\qty(-\Omega\qty(n\delta_0^2)). 
  \end{align*}
  Applying \cref{lem:hitting time for tnpm and npm for PP} (\cref{item:taualphamaxup for PP}), we have
  \begin{align*}
    &\Pr\qty[\min\{\taudeltaup,\taujweak,\tautildemaxdown,\taupsiplus(2x)\}>T \text{ and } \taubetaminus(2x)>T]\\
    &\leq \Pr\qty[\min\{\taudeltaup,\taujweak,\tautildemaxdown,\taupsiplus(2x)\}>T \text{ and } \taubetaminus(2x)>T \text{ and } \taumaxup>T]\\
    &+\Pr\qty[\taumaxup\leq T \text{ and } \taubetaminus(2x)>T]\\
    &\leq \exp\qty(-\Omega\qty(n\delta_0^2)) + n \exp\qty(-\Omega(n(\npm_0)^2)).
  \end{align*}
  
\end{proof}
%
\begin{proof}[Proof of \cref{lem:delta multiplicative and additive drift for PP} (\cref{item:taudeltaplus const prob for population protocol model})]
Let $\tau=\min\{\taudeltaplus,\tauiweak,\taujweak,\tautildemaxdown,\taubetaminus(x),\taupsiplus(x)\}$.
For $t-1<\tau$, 
\begin{align*}
  \Var_{t-1}[\delta_t]
  &\geq \frac{\alpha_{t-1}(i)\qty(1-\alpha_{t-1}(i))}{n^2}
  +\frac{\alpha_{t-1}(j)\qty(1-\alpha_{t-1}(j))}{n^2}
  -\frac{4\delta_{t-1}(i,j)^2}{n^2} & &(\text{by \cref{lem:basic inequalities for delta PP}})\\
  &\geq \frac{(1-2\cweak)\npm_{t-1}-4c^2/n}{n^2} & &(\text{by \cref{eq:weak npm upper bound}})\\
  &\geq \frac{(1-2\cweak)(1-\ctildemaxdown)\npm_{0}-12c^2/n}{3n^2}\\
  &\geq \frac{(1-2\cweak)(1-\ctildemaxdown)\npm_{0}}{6n^2}.
\end{align*}
Note that $\npm_0=\omega(1/n)$.
Hence, for $t-1<\tau$, we have 
\begin{align*}
  \E_{t-1}[\delta_t^2]
  =\E_{t-1}[\delta_t]^2+\Var_{t-1}[\delta_t]
  \geq \delta_{t-1}^2+\frac{(1-2\cweak)(1-\ctildemaxdown)\npm_{0}}{6n^2}.
\end{align*}
Note that $\E_{t-1}[\delta_t]^2\geq \delta_{t-1}^2$ holds for $\tau>t-1$ (\cref{eq:drift of delta plus key}).
Let 
$\drift=\frac{(1-2\cweak)(1-\ctildemaxdown)\npm_{0}}{6n^2}$ and 
$X_t=\delta_t^2-\drift t$ and $Y_t=X_{t\wedge \tau}$.
Then, we have
\begin{align*}
\E_{t-1}[Y_t-Y_{t-1}]
=\indicator_{\tau>t-1}\E_{t-1}[X_t-X_{t-1}]
=\indicator_{\tau>t-1}\qty(\E_{t-1}[\delta_t^2]-\drift t-\delta_{t-1}^2+\drift (t-1))
\geq 0,
\end{align*}
i.e., $(Y_t)_{t\in \Nat_0}$ is a submartingale.
From \cref{thm:OST}, we have
$
\E[Y_\tau]\geq \E[Y_0] = \delta_0^2\geq 0
$
and
$
\E[Y_\tau]=\E[X_\tau] = \E[\delta_\tau^2]-\drift \E[\tau]
$
hold.
Thus, 
\begin{align*}
  \E[\tau]\leq \frac{\E[\delta_\tau^2]}{\drift}\leq \frac{\qty(\frac{c}{\sqrt{n}}+\frac{1}{n})^2}{\drift}
  \leq \frac{24c}{(1-2\cweak)(1-\ctildemaxdown)}\frac{n}{\npm_0}= \frac{\constrefs{lem:delta multiplicative and additive drift for PP}{item:taudeltaplus const prob for population protocol model}n}{2\npm_0}.
\end{align*}
We obtain the claim from the Markov inequality.
\end{proof}


\begin{proof}[Proof of \cref{lem:pair of non-weak opinions for PP} (\cref{item:pair of non-weak opinions becomes weak for PP})] 
  We apply \cref{lem:nazo lemma} for $Z_t=\opn_t$, $\varphi(Z_t)=\sqrt{n}\cdot \abs{\delta_t(i,j)}$, 
  \[\tau=\min\qty{\tauiweak,\taujweak,\tautildemaxdown,\taubetaminus(2x),\taupsiplus(2x)},\] $\cphiup=\cdeltaup$, $x_0 = \cdeltaplus$, and $x^*=\sqrt{C\log n}$.
  Define 
  \begin{align*}
    \tau^\uparrow=\begin{cases}
    \taudeltaup(i,j) & \text{if } \delta_0(i,j)\geq 0,\\
    \taudeltaup(j,i) & \text{if } \delta_0(j,i)> 0 \;(\text{i.e., } \delta_0(i,j)< 0)
    \end{cases},
  \end{align*}
  where $\taudeltaup(i,j)=\taudeltaup=\inf\{t\geq 0: \delta_t(i,j)\geq (1+\cdeltaup)\delta_0(i,j)\}$.
  Then, from these settings, we have
  \begin{align*}
      &\tauphiplus(x_0) = \inf\{t\geq 0: \sqrt{n}\abs{\delta_t}\geq \cdeltaplus\}=\taudeltaplus,\\
      &\tauphiup = \inf\{t\geq 0: \sqrt{n}\abs{\delta_t}\geq (1+\cdeltaup)\sqrt{n}\abs{\delta_0}\}\leq \tau^\uparrow.
  \end{align*}
  From \cref{lem:delta multiplicative and additive drift for PP} (\cref{item:taudeltaplus const prob for population protocol model}), we have
  \begin{align*}
    \Pr\qty[\min\qty{\tauphiplus(x_0),\tau} > \frac{\constrefs{lem:delta multiplicative and additive drift for PP}{item:taudeltaplus const prob for population protocol model} n}{\npm_0}] \leq 1/2,
  \end{align*}
  i.e., the first condition of \cref{lem:nazo lemma} holds for $C_1=1/2$.
  Note that if $\beta_0 < 1/2-2x$ or $\psi_0 > 2x$ or $\npm_0\leq O(\sqrt{\log n/n})$, then $\tau=0$.
  Next, from \cref{lem:delta multiplicative and additive drift for PP} (\cref{item:taudeltaup for population protocol model}),
  \begin{align*}
    \Pr\qty[\min\{\tauphiup,\tau\}> \frac{\constrefs{lem:hitting time for tnpm and npm for PP}{item:taualphamaxup for PP} n}{\npm_0}]
    &\leq \Pr\qty[\min\{\tau^\uparrow,\tau\}> \frac{\constrefs{lem:hitting time for tnpm and npm for PP}{item:taualphamaxup for PP} n}{\npm_0}]\\
    &\leq \exp\qty(-\Omega\qty(n\delta_0(i,j)^2))+n \exp\qty(-\Omega\qty(n(\npm_0)^2))\\
    &\leq \exp\qty(-\Omega\qty(\varphi(Z_0)^2))
  \end{align*}
  holds for $\delta_0(i,j)\leq C\sqrt{\log n/n}$, i.e., the second condition of \cref{lem:nazo lemma} holds for some positive constant $C_2$.
  Note that $n\exp\qty(-\Omega\qty(n(\npm_0)^2))\leq n^{-\omega(1)}\leq \exp\qty(-\Omega\qty(n\delta_0(i,j)^2))$ from the assumption on $\npm_0=\omega(\sqrt{\log n/n})$.
  Note that if $\beta_0 < 1/2-2x$ or $\psi_0 > 2x$ or $\npm_0\leq O(\sqrt{\log n/n})$, then $\tau=0$.
  
  Write $\tauweak=\min\{\tauiweak,\taujweak\}$ and $\tau_{ini}=\min\{\tautildemaxdown,\taubetaminus(2x),\taupsiplus(2x)\}$.
  From \cref{lem:nazo lemma,lem:tauinitial lemma for PP} with $\varepsilon=n^{-30}$, for some $T'=O(n\log n/\npm_0)$, 
  \begin{align*}
    &\Pr\qty[\min\{\taudeltaplus(\sqrt{C\log n/n}),\tauweak\}> T' \text{ or } \tau_{ini}\leq  T']\\
    &\leq \Pr\qty[\qty{\min\{\tauphiplus(x_*),\tauweak\}> T' \text{ or } \tau_{ini}\leq  T'}\text{ and } \tau_{ini}>T']+\Pr\qty[\tau_{ini}\leq T']\\
    &\leq \Pr\qty[\min\{\tauphiplus(x_*),\tauweak,\tau_{ini}\}> T']+\Pr\qty[\tau_{ini}\leq T']\\
    &\leq n^{-20}.
  \end{align*}
  \end{proof}

  \begin{proof}[Proof of \cref{lem:pair of non-weak opinions for PP} (\cref{item:delta multipricative with initial bias for PP})]
  Let $\tau^{\uparrow (s)}=\inf\{t\geq 0: \delta_t\geq (1+c)^s \delta_0\}$.
  Let $\tau^*=\min\{\tautildemaxdown,\taubetaminus(2x),\taupsiplus(2x)\}$.
  From definition, for some $\ell=\Theta(\log n)$, $\taujweak\leq \tau^{\uparrow (\ell)}$ holds.
  Let $T'=\frac{3\constrefs{lem:hitting time for tnpm and npm for PP}{item:taualphamaxup for PP}n}{(1-\ctildemaxdown)\npm_0}$.
  Applying \cref{lem:Iterative Drift theorem} (\cref{item:iterative updrift}), we have $\Pr\qty[\min\{\taujweak,\tau^*\}>\ell T']\geq \Pr\qty[\min\{\tau^{\uparrow (\ell)},\tau^*\}>\ell T']$ and 
  \begin{align*}
    \Pr\qty[\min\{\tau^{\uparrow (\ell)},\tau^*\}>\ell T']
  &\leq \sum_{s=1}^{\ell}\E\qty[\indicator_{\tau^*>\tau^{\uparrow (s-1)}} \Pr_{\tau^{\uparrow (s-1)}}\qty[\min\{\tau^{\uparrow (s)},\tau^*\}> \tau^{\uparrow (s-1)}+\frac{3\constrefs{lem:hitting time for tnpm and npm for PP}{item:taualphamaxup for PP}}{(1-\ctildemaxdown)\npm_{0}}]]\\
  &\leq \sum_{s=1}^{\ell}\E\qty[\indicator_{\tau^*>\tau^{\uparrow (s-1)}} \Pr_{\tau^{\uparrow (s-1)}}\qty[\min\{\tau^{\uparrow (s)},\tau^*\}> \tau^{\uparrow (s-1)}+\frac{\constrefs{lem:hitting time for tnpm and npm for PP}{item:taualphamaxup for PP}}{\npm_{\tau^{\uparrow (s-1)}}}]]\\
  &\leq \ell \qty(\exp(-\Omega(n\delta_0^2))+n\exp\qty(-\Omega\qty(n (\npm_0)^2))).
  \end{align*}
  Note that $\npm_{\tau^{\uparrow (s-1)}}\geq (1-\ctildemaxdown)\tnpm_0/3\geq (1-\ctildemaxdown)\npm_0/3\geq \omega(\sqrt{\log n/n})$ holds.
  \end{proof}

\subsection{Emergence of Unique Strong Opinion} 
Recall the definition of stopping times in \cref{def:strong}.
We can show the following lemma that assures the emergence and persistence of the unique strong opinion, in a similar manner as in the gossip model.
\begin{lemma}[Unique strong opinion lemma]\label{lem:unique strong opinion for PP}
  Let $x=x(n)$ and $y=y(n)$ be arbitrary positive functions such that $x=\omega(\sqrt{\log n / n})$, $y=\omega(x)$, and $y=o(\log n/\sqrt{n})$.
  Suppose $\psi_0\leq x$, $\beta_0\geq 1/2-x$, and $\npm_0=\omega(\log n/\sqrt{n})$.
  We have the following:
  \begin{enumerate}
    \item \label{item:tauusplus for PP} 
    For some $T=O(n\log n/\npm_0)$, we have 
    \[\Pr\qty[\tau_{us}^+(y)> T \text{ or } \min\{\tautildemaxdown,\taubetaminus(8x),\taupsiplus(8x)\}\leq T]\leq O(n^{-10}).\]
    \item \label{item:tauusminus for PP} 
    Suppose $\eta_0(j)\geq y$ for all $j\neq I_0$.
    Then, for any $T\leq \constrefs{lem:hitting time for tnpm and npm for PP}{item:tautildemaxdown is large for PP}n^2$, we have
    \[\Pr\qty[\min\{\tau_{us}^-(y/2),\tautildemaxdown,\taubetaminus(2x),\taupsiplus(2x)\}\leq T ]\leq n^{-10}.\]
  \end{enumerate}
\end{lemma}

The proof of \cref{lem:unique strong opinion for PP} relies on the following lemma.
Essentially, this lemma follows from the drift of $\eta_t(j)$ in the population protocol model, which is approximately $1/n$ times as large as the drift in the gossip model.
\begin{lemma}[Bounded decrease of $\eta_t(j)$]
  \label{lem:tauetaminus for PP}
  Let $x=x(n)$ and $y=y(n)$ be arbitrary positive functions such that $x=\omega(\sqrt{\log n / n})$, $y=\omega(x)$, and $y=o(\log n/\sqrt{n})$.
  Suppose $\psi_0\leq x$, $\beta_0\geq 1/2-x$, $\npm_0=\omega(\log n/\sqrt{n})$, and $\eta_0(j)\geq y$ for all $j\neq I_0$.
  Then, for any $T\geq 0$, 
  \begin{align*}
    \Pr\qty[\tauetaminus(y/2) \leq T \text{ and } \min\{\tautildemaxdown,\taubetaminus(x),\taupsiplus(x)\}>T]
    \leq Tn^{-\omega(1)}.
\end{align*}
\end{lemma}

\paragraph{Weak cannot be strong: Proof of \cref{lem:tauetaminus for PP}.}
To prove \cref{lem:tauetaminus for PP}, we use the following lemma:
\begin{lemma}
  \label{lem:hitting time for eta for population protocol model}
  Let $x=x(n)$, $y=y(n)$, and $I=I(n)$ be arbitrary positive functions such that $x=\omega(\sqrt{\log n / n})$, $I=\omega(x)$, $y=\omega(I)$, and $y=o(\log n/\sqrt{n})$.
  Suppose that $\beta_0\geq 1/2-x$, $\psi_0\leq x$, $\npm_0\geq \omega\qty(\log n/\sqrt{n})$, and $\eta_0(j)\in [y,y+I]$.
  Then, for any $T\geq \frac{24(1+\ceta)}{1-\ctildemaxdown}\frac{In}{(y-I)\npm_0}$, we have
  \begin{align*}
    \Pr\qty[\tauetaminus(y-I)<\tauetaplus(y+2I) \text{ and } \min\{\tautildemaxdown,\taumaxup,\taubetaminus(x),\taupsiplus(x)\}>T]\leq n^{-\omega(1)}.
  \end{align*}
\end{lemma}
\begin{proof}
  Let $\tau^*=\min\{\tautildemaxdown,\taubetaminus(x),\taupsiplus(x)\}$ and
let $\tau=\min\{\tauetaplus(y+2I),\tauetaminus(y-I),\tau^*\}$.
For $t-1<\tau$, we have $\E_{t-1}[\eta_t(j)]\geq \E_{t-1}[\delta_{t}^{(\ceta)}(I_{t-1},j)]$ and 
\begin{align*}
  \E_{t-1}[\delta_{t}^{(\ceta)}(I_{t-1},j)]
  &\geq \delta_{t-1}^{(\ceta)}(I_{t-1},j)+\frac{\delta_{t-1}^{(\ceta)}(I_{t-1},j)}{n}\qty(\alpha_{t-1}(I_{t-1})+\alpha_{t-1}(j)+1-2\beta_{t-1})\\
    &\geq \eta_{t-1}(j)+\frac{1-\ctildemaxdown}{6(1+\ceta)}\cdot \frac{(y-I)\npm_{0}}{n}.
\end{align*}
Note that we use \cref{lem:basic inequalities for delta PP} (\cref{item:expectation of delta PP}) and \cref{eq:drift of eta}.

Letting $X_t=\eta_t(j)$ and $\drift = \frac{1-\ctildemaxdown}{6(1+\ceta)}\cdot \frac{(y-I)\npm_{0}}{n}>0$, we have
  $
    \indicator_{\tau>t-1}\qty(X_{t-1}+\drift-\E_{t-1}[X_t])
    =\indicator_{\tau>t-1}\qty(\eta_{t-1}(j)+\drift-\E_{t-1}[\eta_t(j)])\leq 0
  $. Further, we have $X_{t-1}+\drift-X_t
    =\eta_{t-1}(j)+\drift-\eta_t(j)$
  and 
  \begin{align*}
    \indicator_{\tau>t-1}\qty(\eta_{t-1}(j)+\drift-\eta_t(j))
    &\leq \indicator_{\tau>t-1}\qty(\E_{t-1}[\delta_{t}^{(\ceta)}(I_{t-1},j)]-\delta_{t}^{(\ceta)}(I_{t-1},j)),
  \end{align*}
  i.e., $\indicator_{\tau>t-1}\qty(X_{t-1}+\drift-X_t)$ satisfies one-sided $\qty(O(1/n),O(\npm_0/n^2))$-Bernstein condition.
  Note that we use \cref{lem:basic inequalities for delta PP} (\cref{item:bernstein condition of delta PP}) and \cref{lem:Bernstein condition} (\cref{item:BC for upper bounded rv,item:BC for linear transformation}).

  Let $I^-=y-I$, $I^-_*=y$, $I^+_*=y+I$, and $I^+=y+2I$.
  Note that $I^+-I^+_*=I^+_*-I^-_*=I^-_*-I^-=I$.
  Let $T\geq \frac{24(1+\ceta)}{1-\ctildemaxdown}\frac{In}{(y-I)\npm_0}\geq \frac{2(I^+-X_0)}{\drift}$.
  Applying \cref{lem:Useful drift lemma} (\cref{item:positive drift useful}) for $\eta_0(j)\in [I_*^-,I^+_*]$, we have
  \begin{align*}
    \Pr\qty[\tauetaminus(y-I)<\tauetaplus(y+2I) \text{ and } \tau^*>T]
      \leq \exp\qty(-\Omega\qty(\frac{I^2}{\frac{\npm_0}{n^2}T+\frac{I}{n}}))
      \leq \exp\qty(-\Omega\qty(nI(y-I))).
  \end{align*}
\end{proof}

\begin{proof}[Proof of \cref{lem:tauetaminus for PP}]
  Let $I=I(n)$ be a positive parameter such that $I=\omega(x)$ and $I=o(y)$.
  Let $I^-=y-I$, $I^-_*=y$, $I^+_*=y+I$, and $I^+=y+2I$.
  Note that $I^+-I^+_*=I^+_*-I^-_*=I^-_*-I^-=I$.
  Let $T'=\frac{24(1+\ceta)}{1-\ctildemaxdown}\frac{In}{(y-I)\npm_0}=o\qty(\frac{n}{\npm_0})$.
  Then, combining 
  \cref{lem:taubeta for PP,lem:taupsi for PP,lem:tildemaxdown for population protocol model,lem:hitting time for tnpm and npm for PP} (\cref{item:taualphamaxup for PP}), \cref{lem:hitting time for eta for population protocol model}
  if $\beta_0\geq 1/2-x\geq 1/2-2x$, $\psi_0\leq x\leq 2x$, $\npm_0\geq \omega\qty(\log n/\sqrt{n})$, and $\eta_0(j)\in [I^-_*,I^+_*]$, 
  \begin{align*}
    &\Pr\qty[\tauetaminus(y-I)<\tauetaplus(y+2I)]\\
    &\leq \Pr\qty[\tauetaminus(y-I)<\tauetaplus(y+2I) \text{ and } \min\{\tautildemaxdown,\taumaxup,\taubetaminus(2x),\taupsiplus(2x)\}>T']\\
    &+\Pr\qty[\min\{\tautildemaxdown,\taumaxup,\taubetaminus(2x)\}\leq T']+\Pr\qty[\taupsiplus(2x)\leq T']\\
    &\leq n^{-\omega(1)}.
  \end{align*}
  Let $\tau^*=\min\{\tautildemaxdown,\taubetaminus(x),\taupsiplus(x)\}$.
  Let $\tau_s^-=\inf\{t\geq s: \eta_t(j)\leq I^-\}$ and $\tau_s^+=\inf\{t\geq s: \eta_t(j)\geq I^+\}$.
  Suppose $\eta_0(j)\in [I^-_*,I^+_*]$, $\beta_0\geq 1/2-x$, $\psi_0\leq x$, and $\npm_0\geq \omega(\log n/\sqrt{n})$.
  Applying \cref{lem:Iterative Drift theorem} (\cref{item:no increase under negative drift}), we have
  \begin{align*}
    \Pr\qty[\tauetaminus(y-I) \leq T \text{ and } \tau^*>T]
    \leq \sum_{s=0}^{T-1}\E\qty[\indicator_{\eta_s(j)\in [I^-_*,I^{+}_*] \text{ and } \tau^*>s}\Pr_s\qty[\tau^-_s<\tau^+_s]]
    \leq Tn^{-\omega(1)}.
\end{align*}
\end{proof}

\paragraph{Emergence of unique strong opinion: Proof of \cref{lem:unique strong opinion for PP}.}
\begin{proof}[Proof of \cref{lem:unique strong opinion for PP} (\cref{item:tauusplus for PP})]
  Fix an arbitrary pair of distinct opinions $i$ and $j$.
  Observe the following holds:
  \begin{itemize}
    \item Combining \cref{item:pair of non-weak opinions becomes weak for PP,item:delta multipricative with initial bias for PP} of \cref{lem:pair of non-weak opinions for PP}, 
    for some $T_1=O(n\log n/\npm_0)$, $\min\{\alpha_{T_1}(i),\alpha_{T_1}(j)\}\leq (1-\cweak)\npm_{T_1}$, $\beta_{T_1}\geq 1/2-4x$, $\psi_{T_1}\leq 4x$, and $\npm_{T_1}=\Omega(\npm_0)=\omega(\log n/\sqrt{n})$ with probability at least $1-n^{-10}$.
    Suppose that $\alpha_{T_1}(j) \le (1-\cweak)\npm_{T_1}$ without loss of generality.
    From the definition of $\eta_t(j)$, we have $\eta_{T_1}(j)\geq \frac{\cweak-\cstrong}{1-\cstrong}\npm_{T_1}=\Omega(\npm_0)\geq 2y$.
    \item From \cref{lem:tauetaminus for PP}, for any $t\in [T_1,T_1+T_2]$, where $T_2\geq n^2 \npm_{T_1} /C\log n=\Omega(n^2\npm_0/\log n)$, $\eta_t(j)\geq y$, $\beta_t\geq 1/2-8x$, $\psi_t\leq 8x$, and $\npm_t=\Omega(\npm_0)=\omega(\log n/\sqrt{n})$ with probability at least $1-n^{-\omega(1)}$.
  \end{itemize}
  
  Since $T_1=O(n\log n/\npm_0)=o(n\sqrt{n})$ and $T_2=\Omega(n^2\npm_0/\log n)=\omega(n\sqrt{n})$, we have $T_1<T_2$.
  Thus, by taking the union bound over all pairs of distinct opinions $i,j\in[k]$, with high probability, there exists $t\leq T_1$ such that $\eta_t(j)\geq y$ for all $j\neq I_t$.
\end{proof}
\begin{proof}[Proof of \cref{lem:unique strong opinion for PP} (\cref{item:tauusminus for PP})]
  Consider an arbitrary opinion $j\neq I_0$. Combining \cref{lem:tauetaminus for PP,lem:tauinitial lemma for PP}, 
  \begin{align}
    &\Pr\qty[\tauetaminus(y/2) \leq T \text{ or } \min\{\tautildemaxdown,\taubetaminus(2x),\taupsiplus(2x)\}\leq T] \nonumber\\
    &\leq \Pr\qty[\qty{\tauetaminus(y/2) \leq T \text{ or } \min\{\tautildemaxdown,\taubetaminus(2x),\taupsiplus(2x)\}\leq T} \text{ and } \min\{\tautildemaxdown,\taubetaminus(2x),\taupsiplus(2x)\}>T] \nonumber\\
    &+ \Pr\qty[\min\{\tautildemaxdown,\taubetaminus(2x),\taupsiplus(2x)\}\leq T] \nonumber\\
    &\leq n^{-11}. \label{eq:tauetaminus and tauinitial lemma for PP}
  \end{align}
  Thus, from the union bound, we obtain the claim.
  \end{proof}

  \subsection{Towards Consensus} 
  Once we have shown that there exists exactly one strong opinion within $O(n\log n/\npm_0)$ rounds, we can show that with high probability, the consensus is reached within $O(n\log n/\npm_0)$ rounds.
  The main lemma in this section is the following:
  \begin{lemma}[Unique strong opinion leads to consensus]
    \label{lem:towards consensus for PP}
    Let $x=x(n)$ and $y=y(n)$ be arbitrary positive functions such that $x=\omega(\sqrt{\log n / n})$, $y=\omega(x)$, and $y=o(\log n/\sqrt{n})$.
    Suppose $\psi_0\leq x$, $\beta_0\geq 1/2-x$, $\npm_0=\omega(\log n/\sqrt{n})$, and $\eta_0(j)\geq y$ for all $j\neq I_0$.
    Then, $\Pr\qty[\taucons>T]\leq 1/n$ holds for some $T=O\qty(n\log n/\npm_0)$.
    \end{lemma}

    Recall the stopping times defined in \cref{sec:towards consensus}.
  We introduce the following lemmas to show \cref{lem:towards consensus for PP}.
  These are analogous to the lemmas in the gossip model, but with different drifts and variances in the Bernstein condition.

\begin{lemma}
\label{lem:unique strong opinion evolves for PP}
Let $x=x(n)$ and $y=y(n)$ be arbitrary positive functions such that $x=\omega(\sqrt{\log n / n})$, $y=\omega(x)$, and $y=o(\log n/\sqrt{n})$.
We have the following:
\begin{enumerate}
  \item \label{item:tauplus for PP}
  Let $c\in (0,1)$ be an arbitrary constant.
  Suppose $\psi_0\leq x$, $\beta_0\geq 1/2-x$, $\npm_0=\omega(\log n/\sqrt{n})$, and $\eta_0(j)\geq y$ for all $j\neq I_0$.
  Then, for some $T=O\qty(n\log n/\npm_0)$, 
\begin{align*}
  \Pr\qty[\tautildemaxplus(1-c)>T \text{ or } \min\{\taupsiplus(2x),\taubetaminus(2x)\}\leq T]\leq n^{-10}. 
 \end{align*}
 \item \label{item:taumaxmajority for PP}
 Suppose that $\psi_0\leq x$, $\beta_0\geq 1/2-x$, and $\alpha_0(1)\geq (1-\ctildemaxdown)\beta_0$.
 Then, for some $T=O(n\log n)$,
 \begin{align*}
   \Pr\qty[\taumaxplus(1-4\ctildemaxdown)>T \text{ or } \min\{\taupsi(2x),\taubetaminus(2x)\}\leq T]\leq n^{-10}. 
  \end{align*}
  \item \label{item:tauall for PP}
  Suppose that $\alpha_0(1)\geq 7/8$. 
  Then, 
  $
    \Pr\qty[\tau_{all}> 8n\log n \text{ or } \taumaxminus(3/4)\leq 8n\log n]\leq 1/n.
  $
  \item \label{item:taucons for PP}
  Suppose $\alpha_0(1)=\beta_0$ and $\beta_0\geq 1/2-x$.
  Then, $\Pr[\taucons>6n\log n]\leq 1/n$.
\end{enumerate}
\end{lemma}
  \begin{proof}[Proof of \cref{lem:unique strong opinion evolves for PP} (\cref{item:tauplus for PP})]
    Let $\tautildemaxup = \inf\{t\geq 0:\tnpm_t\geq (1+\ctildemaxup)\tnpm_0\}$ for some positive constant $\ctildemaxup\in (0,1)$.
  Let $\tau^*=\min\{\tautildemaxplus,\taubetaminus(2x),\tau_{us}^-(y/2)\}$ and $\tau=\min\{\tautildemaxup,\tautildemaxdown,\tau^*\}$.
  For $t-1<\tau$, in a similar calculation as the proof of \cref{lem:unique strong opinion evolves} (\cref{item:tauplus}), we have 
  \begin{align*}
    \E_{t-1}\qty[\tnpm_t]
    &\geq  \tnpm_{t-1}\qty(1+\frac{\npm_{t-1}-\gamma_{t-1}/\beta_{t-1}}{2n})-\frac{18\npm_{t-1}}{n^2}
    \geq \tnpm_{t-1}+\frac{\cstrong c(1-\ctildemaxdown)^2(\npm_{0})^2}{12n^2}.
  \end{align*}
  Note that we use \cref{lem:basic inequalities for tnpm PP}, \cref{eq:tnpm additive drift}
  and $\npm_{t-1}\geq \frac{(1-\ctildemaxdown)\tnpm_{0}}{3}=\omega(1/n)$.
  
  Hence, letting $X_t=\tnpm_t$ and $\drift = \frac{\cstrong c(1-\ctildemaxdown)^2(\npm_{0})^2}{12n}> 0$, we have
    $
      \indicator_{\tau>t-1}\qty(X_{t-1}+\drift-\E_{t-1}[X_t])
      =\indicator_{\tau>t-1}\qty(\tnpm_{t-1}+\drift-\E_{t-1}[\tnpm_t])\leq 0
    $
    and 
    \begin{align*}
      \indicator_{\tau>t-1}\qty(X_{t-1}+\drift-X_t)
      &\leq \indicator_{\tau>t-1}\qty(\tnpm_{t-1}\qty(1+\frac{\npm_{t-1}-\gamma_{t-1}/\beta_{t-1}}{2n})-\frac{18\npm_{t-1}}{n^2}-\tnpm_t)
    \end{align*}
     satisfies one-sided $\qty(O\qty(\frac{1}{n}),O\qty(\frac{\npm_0}{n^2}))$-Bernstein condition. 
    Note that we use \cref{lem:basic inequalities for tnpm PP} and \cref{lem:Bernstein condition} (\cref{item:BC for upper bounded rv,item:BC for linear transformation}).
  
    Applying \cref{lem:Useful drift lemma} (\cref{item:positive drift useful}) for $I^+=(1+\ctildemaxup)\tnpm_{0}$, $I^-=(1+\ctildemaxdown)\tnpm_{0}$, and for $T'=\frac{2(I^+-X_0)}{\drift} = \frac{24}{\cstrong c^+(1-\ctildemaxdown)^2} \frac{n}{\npm_0}$, we have
    \begin{align*}
              \Pr\qty[\tautildemaxup > T' \text{ and } \tau^* > T'] 
              \le \exp\qty(- \Omega\qty( n (\npm_0)^2)).
      \end{align*}
  
    Let $\tau^{\uparrow (s)}=\inf\{t\geq 0: \tnpm_t\geq (1+\ctildemaxup)^s \tnpm_0\}$.
    From definition, for some $\ell=\Theta(\log n)$, $\tautildemaxplus\leq \tau^{\uparrow (\ell)}$ holds.
  Applying \cref{lem:Iterative Drift theorem} (\cref{item:iterative updrift}), we obtain $\Pr\qty[\min\{\tautildemaxplus,\tau^*\}>\ell T'] \leq \Pr\qty[\min\{\tau^{\uparrow (\ell)},\tau^*\}>\ell T']$ and 
  \begin{align*}
    \Pr\qty[\min\{\tau^{\uparrow (\ell)},\tau^*\}>\ell T']
    &\leq \sum_{s=1}^{\ell}\E\qty[\indicator_{\tau^*>\tau^{\uparrow (s-1)}} \Pr_{\tau^{\uparrow (s-1)}}\qty[\min\{\tau^{\uparrow (s)},\tau^*\}> \tau^{\uparrow (s-1)}+T']]\\
    &\leq \ell \exp\qty(-\Omega\qty(n (\npm_0)^2)).
  \end{align*}
  Note that $\npm_{\tau^{\uparrow (s-1)}}\geq \Omega(\npm_0)$ holds.
  Finally, combining the above and \cref{eq:tauetaminus and tauinitial lemma for PP}, we have
  \begin{align*}
    &\Pr\qty[\tautildemaxplus>T \text{ or }\min\{\taupsi(2x),\taubetaminus(2x)\}\leq T]\\
    &\leq \Pr\qty[\qty{\tautildemaxplus>T \text{ or } \min\{\taupsi(2x),\taubetaminus(2x)\}\leq T} \text{ and } \min\{\taupsi(2x),\taubetaminus(2x),\tau_{us}^-(y/2)\}>T]\\
    &+\Pr\qty[\min\{\taupsi(2x), \taubetaminus(2x),\tau_{us}^-(y/2)\}\leq T]\\
    &\leq n^{-10}.
  \end{align*}
  \end{proof}
  \begin{proof}[Proof of \cref{lem:unique strong opinion evolves for PP} (\cref{item:taumaxmajority for PP})]
  Let $\tau_{\max}^\downarrow=\inf\{t\geq 0:\npm_t\leq (1-c_{\max}^\downarrow)\npm_0\}$.
  Let $\tau^*=\min\{\taumaxplus(1-4\ctildemaxdown),\tautildemaxdown,\taubetaminus(2x)\}$ and $\tau=\min\{\tau_{\max}^\uparrow,\tau_{\max}^\downarrow,\tau^*\}$.
  For $t-1<\tau$, in a similar calculation as the proof of \cref{lem:unique strong opinion evolves} (\cref{item:taumaxmajority}), 
  \begin{align*}
    \E_{t-1}[\alpha_t(1)]-\alpha_{t-1}(1)
    &=\frac{\alpha_{t-1}(1)}{n}\qty(1-\alpha_{t-1}(1)\qty(\frac{2}{\tnpm_{t-1}}-1))
    \geq \frac{8(\ctildemaxdown)^2(1-c_{\max}^\downarrow)}{1-2\ctildemaxdown}\frac{\alpha_{0}(1)}{n}.
  \end{align*}
  
  Hence, letting $X_t=\alpha_t(1)$ and $\drift = \frac{8(\ctildemaxdown)^2(1-c_{\max}^\downarrow)}{1-2\ctildemaxdown}\frac{\alpha_{0}(1)}{n}> 0$, we have
    $
      \indicator_{\tau>t-1}\qty(X_{t-1}+\drift-\E_{t-1}[X_t])
      =\indicator_{\tau>t-1}\qty(\alpha_{t-1}(1)+\drift-\E_{t-1}[\alpha_t(1)])\leq 0
    $
    and 
    $
      \indicator_{\tau>t-1}\qty(\E_{t-1}[X_t]-X_t)
      = \indicator_{\tau>t-1}\qty(\E_{t-1}[\alpha_t(1)]-\alpha_t(1))\\
    $
     satisfies one-sided $\qty(O\qty(\frac{1}{n}),O\qty(\frac{\alpha_{0}(1)}{n^2}))$-Bernstein condition. 
     Note that we use \cref{lem:basic inequalities for alpha} (\cref{item:bernstein condition of alpha}) and \cref{lem:Bernstein condition} (\cref{item:BC for upper bounded rv,item:BC for linear transformation}).
  
    Applying \cref{lem:Useful drift lemma} (\cref{item:positive drift useful}) for $I^+=(1+\cmaxup)\alpha_0(1)$, $I^-=(1-c_{\max}^\downarrow)\alpha_0(1)$, and 
    $T'= \frac{c_{\max}^\uparrow (1-2\ctildemaxdown)}{4(\ctildemaxdown)^2(1-c_{\max}^\downarrow)}n=\frac{2(I^+-X_0)}{\drift}$, we have
             \begin{align*}
              \Pr\qty[\taumaxup > T' \text{ and } \tau^* > T'] 
              \le \exp\qty(- \Omega\qty( n (\npm_0)^2)).
             \end{align*}
    Let $\tau^{\uparrow (s)}=\inf\{t\geq 0: \npm_t\geq (1+c_{\max}^\uparrow)^s \npm_0\}$.
    From definition, for some $\ell=\Theta(\log n)$, $\taumaxplus(1-4\ctildemaxdown)\leq \tau^{\uparrow (\ell)}$ holds.
  Applying \cref{lem:Iterative Drift theorem} (\cref{item:iterative updrift}), we have
  \begin{align*}
    \Pr\qty[\min\{\taumaxplus(1-4\ctildemaxdown),\tau^*\}>\ell T']
    &\leq \Pr\qty[\min\{\tau^{\uparrow (\ell)},\tau^*\}>\ell T']\\
    &\leq \sum_{s=1}^{\ell}\E\qty[\indicator_{\tau^*>\tau^{\uparrow (s-1)}} \Pr_{\tau^{\uparrow (s-1)}}\qty[\min\{\tau^{\uparrow (s)},\tau^*\}> \tau^{\uparrow (s-1)}+T']]\\
    &\leq \ell \exp\qty(-\Omega\qty(n (\npm_0)^2)).
  \end{align*}
  Note that $\npm_{\tau^{\uparrow (s-1)}}\geq \Omega(\npm_0)$ holds.
  Finally, combining the above and \cref{lem:tauinitial lemma for PP}, we have
  \begin{align*}
    &\Pr\qty[\taumaxplus(1-4\ctildemaxdown)>T \text{ or } \min\{\taupsi(2x),\taubetaminus(2x)\}\leq T]\\
    &\leq \Pr\qty[\qty{\taumaxplus(1-4\ctildemaxdown)>T \text{ or } \min\{\taupsi(2x),\taubetaminus(2x)\}\leq T} \text{ and } \min\{\taupsi(2x), \taubetaminus(2x),\tautildemaxdown\}>T]\\
    &+\Pr\qty[\min\{\taupsi(2x), \taubetaminus(2x),\tautildemaxdown\}\leq T]\\
    &\leq n^{-10}.
  \end{align*}
  \end{proof}
  \begin{proof}[Proof of \cref{lem:unique strong opinion evolves for PP} (\cref{item:tauall for PP})]
    To begin with, we prove the following claim.
    \begin{claim}\label{claim:taumaxminus for PP}
      Let $c\in (0,1/2)$ be an arbitrary constant.
      Let $x=x(n)$ be an arbitrary positive function such that $x=\omega(\sqrt{\log n / n})$ and $x=o(\log n/\sqrt{n})$.
      Suppose that $\psi_0\leq x$, $\beta_0\geq 1/2-x$, and $\alpha_0(1)\geq 1-c$.
      Then, for some $T=\Omega(n/x)$, 
      \begin{align*}
        \Pr\qty[\taumaxminus(1-2c) \leq \min\{T,\taubetaminus(x),\taupsiplus(x)\}]
          \leq \exp\qty(-\Omega\qty(n)).
      \end{align*}
    \end{claim}
    \begin{proof}
      Let $\tau=\min\{\taubetaminus(x),\taupsiplus(x)\}$.
      Then, for $t-1<\tau$, we have
      \begin{align*}
          \E_{t-1}[\alpha_t(1)]
          &= \alpha_{t-1}(1)\qty(1+\frac{\alpha_{t-1}(1)-\frac{\gamma_{t-1}+x}{\beta_{t-1}}}{n})
          \geq \alpha_{t-1}(1)-\frac{3x}{n}.
        \end{align*}
        Note that we use \cref{lem:basic inequalities for alpha} (\cref{item:expectation of alpha}) in the first equality.
    
        Hence, letting $X_t=\alpha_t(1)$ and $\drift = -\frac{3x}{n}<0$, we have
        $
          \indicator_{\tau>t-1}\qty(X_{t-1}+\drift-\E_{t-1}[X_t])
          =\indicator_{\tau>t-1}\qty(\alpha_{t-1}(1)-3x-\E_{t-1}[\alpha_t(1)])\leq 0
        $
        and $\indicator_{\tau>t-1}\qty(\E_{t-1}[X_t]-X_t)=\indicator_{\tau>t-1}\qty(\E_{t-1}[\alpha_t(1)]-\alpha_t(1))$ satisfies $\qty(O(1/n),O(1/n))$-Bernstein condition. 
        Note that we use \cref{lem:basic inequalities for alpha PP} (\cref{item:bernstein condition of alpha PP}) and \cref{lem:Bernstein condition} (\cref{item:BC for upper bounded rv,item:BC for linear transformation}).
        Hence, applying \cref{lem:Useful drift lemma} (\cref{item:negative drift useful}) with $I^-=1-2c$, and $T\leq \frac{c}{6x}n\leq \frac{X_0-I^-}{-2\drift}$, we obtain the claim.
    \end{proof}
    Write $g_t=\beta_t-\alpha_t(1)=\sum_{j\geq 2}\alpha_t(j)$ for convenience.
    Let $\tau=\min\{\tau_{all},\taumaxminus(3/4)\}$.
    Then, for $t-1<\tau$, we have
    \begin{align*}
      \E_{t-1}[g_t]
      &=\sum_{j\geq 2}\alpha_{t-1}(j)\qty(1+\frac{\alpha_{t-1}(j)+1-2\beta_{t-1}}{n})
      \leq g_{t-1}\qty(1-\frac{1}{4n}).
    \end{align*}
    Note that $\alpha_{t-1}(j)\leq 1-\alpha_{t-1}(1)\leq 1/4$ and $\beta_{t-1}\geq \alpha_{t-1}(1)\geq 3/4$.
    
    Let $r=1-\frac{1}{4n}$, $X_t=r^{-t}g_t$, and $Y_t=X_{t\wedge \tau}$.
  Then, 
  \begin{align*}
  \E_{t-1}[Y_t]-Y_{t-1}
  =\indicator_{\tau>t-1}\qty(\E_{t-1}[X_t]-X_{t-1})
  \leq \indicator_{\tau>t-1}\qty(r^{-t}\E_{t-1}[g_t]-r^{-(t-1)}g_{t-1})
  \leq 0,
  \end{align*}
  i.e., $Y_t$ is a submartingale.
  Hence, we have $\E[Y_T]\leq \E[Y_0]=g_0\leq 1$ and 
  \begin{align*}
  \E[Y_T]
  \geq \E[X_{T}\mid \tau>T]\Pr[\tau>T]
  =r^{-T}\E[g_T\mid \tau>T]\Pr[\tau>T]
  \geq r^{-T}n^{-1}\Pr[\tau>T].
  \end{align*}
  Consequently, we have
     $ \Pr[\tau>T]\leq nr^T\leq n\exp\qty(-\frac{T}{4n})\leq 1/n^2.$
  Thus, combining the above, \cref{claim:taumaxminus for PP,lem:taubeta for PP,lem:taupsi for PP} gives
  \begin{align*}
    &\Pr\qty[\tau_{all}> T \text{ or } \taumaxminus(3/4)\leq T]\\
    &\leq \Pr\qty[\qty{\tau_{all}> T \text{ or } \taumaxminus(3/4)\leq T} \text{ and } \min\{\taumaxminus(3/4),\taubetaminus(2x),\taupsiplus(2x)\}>T]\\
    &+\Pr\qty[\min\{\taumaxminus(3/4),\taubetaminus(2x),\taupsiplus(2x)\}\leq T]\\
    &\leq \Pr\qty[\tau>T]+\Pr\qty[\taumaxminus(3/4)\leq T \text{ and } \min\{\taubetaminus(2x),\taupsiplus(2x)\}>T] + \Pr\qty[\min\{\taubetaminus(2x),\taupsiplus(2x)\}\leq T]\\
    &\leq 1/n.
  \end{align*}
  \end{proof}
  \begin{proof}[Proof of \cref{lem:unique strong opinion evolves for PP} (\cref{item:taucons for PP})]
    Let $\tau=\min\{\taucons,\taubetaminus\}$.
    Let $u_t=1-\alpha_t(1)$. Then, 
    \begin{align*}
      \E_{t-1}[u_t]=1-\alpha_{t-1}(1)\qty(1+\frac{1+\alpha_{t-1}(1)-2\beta_{t-1}}{n})
      =u_{t-1}\qty(1-\frac{\alpha_{t-1}(1)}{n})
      \leq u_{t-1}\qty(1-\frac{1}{3n}).
    \end{align*}
    Let $r=1-\frac{1}{3n}$, $X_t=r^{-t}u_t$, and $Y_t=X_{t\wedge \tau}$.
  Then, 
  \begin{align*}
  \E_{t-1}[Y_t]-Y_{t-1}
  =\indicator_{\tau>t-1}\qty(\E_{t-1}[X_t]-X_{t-1})
  \leq \indicator_{\tau>t-1}\qty(r^{-t}\E_{t-1}[u_t]-r^{-(t-1)}u_{t-1})
  \leq 0,
  \end{align*}
  i.e., $Y_t$ is a submartingale.
  Hence, we have $\E[Y_T]\leq \E[Y_0]=u_0\leq 1$ and 
  \begin{align*}
  \E[Y_T]
  \geq \E[X_{T}\mid \tau>T]\Pr[\tau>T]
  =r^{-T}\E[u_T\mid \tau>T]\Pr[\tau>T]
  \geq r^{-T}n^{-1}\Pr[\tau>T].
  \end{align*}
  Consequently, we have
     $ \Pr[\tau>T]\leq nr^T\leq n\exp\qty(-\frac{T}{3n})\leq 1/n^2.
$
  Thus, from \cref{lem:taubeta for PP}, we have
  \begin{align*}
    \Pr[\taucons>T]
    &\leq \Pr[\min\{\taucons,\taubetaminus(2x)\}>T]+\Pr[\taubetaminus(2x)\leq T]
    \leq 1/n.
  \end{align*}
  \end{proof}

  \begin{proof}[Proof of \cref{lem:towards consensus for PP}]
    We have the following:
  \begin{itemize}
    \item From \cref{lem:unique strong opinion evolves for PP} (\cref{item:tauplus for PP}), for some $T_1=O\qty(n\log n/\npm_0)$, we have $\alpha_{T_1}(1)\geq (1-\ctildemaxdown)\beta_{T_1}$, $\beta_{T_1}\geq 1/2-2x$, and $\psi_{T_1}\leq 2x$, with probability at least $1-n^{-10}$.
    \item From \cref{lem:unique strong opinion evolves for PP} (\cref{item:taumaxmajority for PP}), for some $T_2=O(n\log n)$, we have $\alpha_{T_1+T_2}(1)\geq 1-4\ctildemaxdown$, $\beta_{T_1+T_2}\geq 1/2-4x$, and $\psi_{T_1+T_2}\leq 4x$, with probability at least $1-n^{-10}$.
    \item Assume that $\ctildemaxdown = 1/32$. From \cref{lem:unique strong opinion evolves for PP} (\cref{item:tauall for PP}), for some $T_3= O(n\log n)$, we have $\alpha_{T_1+T_2+T_3}(1)=\beta_{T_1+T_2+T_3}$ and $\alpha_{T_1+T_2+T_3}(1)\geq 3/4$ with probability at least $1-1/n^2$.
    \item From \cref{lem:unique strong opinion evolves for PP} (\cref{item:taucons for PP}), for some $T_4=O(n\log n)$, we have $\alpha_{T_1+T_2+T_3+T_4}(1)=1$ with probability at least $1-1/n^2$.
  \end{itemize}
  Thus, we obtain the claim.
  \end{proof}

\subsection{Putting All Together}

\begin{lemma}
    \label{lem:good configuration for PP}
    Let $C>0$ be any constant.
    Let $x=x(n)$ be an arbitrary positive function such that $x=\omega(\sqrt{\log n / n})$ and $x=o(\log n/\sqrt{n})$.
    For some 
    \begin{align}
        T=
    \begin{cases}
        O\qty(n\log n) & \qty(\text{if } k\le \frac{C\sqrt{n}}{(\log n)^2}), \\
        O\qty(n^{1.5}(\log n)^3) & (\text{otherwise})
    \end{cases},
    \label{eq:good configuration for PP}
    \end{align}
    we have
    \begin{align*}
    \Pr\qty[\psi_T\leq x \text{ and } \beta_T\geq \frac{1}{2}-x \text{ and } \npm_T\geq \frac{(\log n)^{1.5}}{\sqrt{n}}] \geq 1 -O\qty(n^{-10}).
    \end{align*}
\end{lemma}
\begin{proof}
 First, from \cref{lem:taubeta for PP,lem:taupsi for PP}, we have the following:
  \begin{enumerate}
    \item \label{item:good configuration 1 for PP} Combining \cref{item:taubeta_logn for PP,item:taubetaplus for PP,item:hitting times for beta for PP:keeping phase} of \cref{lem:taubeta for PP}, for some $T_1=O(n\log n)$ and $T_2=\Omega(n^3)$, $\min_{t\in [T_2]}\beta_{T_1+t}\geq 1/2-x/2$ holds with probability at least $1-O(n^{-10})$. 
  \item \label{item:good configuration 2 for PP} Combining \cref{item:taupsiminus for PP,item:taupsiplus for PP} of \cref{lem:taupsi for PP}, for some $T_1'=O(n\log n)$ and $T_2'=\Omega(n^3)$, $\min_{t\in [T_2']}\psi_{T_1'+t}\leq x/2$ holds with probability at least $1-n^{-\omega(1)}$. 
  \end{enumerate}
\Cref{item:good configuration 1 for PP,item:good configuration 2 for PP} implies that, for some $T=O(n\log n)$, $\beta_T\geq 1/2-x/2$ and $\psi_T\leq x/2$ hold with probability at least $1-O(n^{-10})$.
For the case where $k=O\qty(\sqrt{n}/(\log n)^2)$, 
since $\npm_t\geq \beta_t/k=\Omega(\beta_t(\log n)^2/\sqrt{n})$ holds for any $t$, we obtain the claim.
For the general case, from \cref{lem:growth of tnpt for PP}, for some $T'=O\qty(n^{1.5}(\log n)^3)$, we have that $\npt_{T+T'}\geq (\log n)^2/\sqrt{n}$, $\psi_{T+T'}\leq x$, and $\beta_{T+T'}\geq 1/2-x$ with probability at least $1-n^{-\omega(1)}$.

\end{proof}

\begin{lemma}
    \label{lem:consensus time starting from a goood intial configuration for PP}
    Let $x=x(n)$ be an arbitrary positive function such that $x=\omega(\sqrt{\log n / n})$ and $x=o(\log n/\sqrt{n})$.
    Suppose that $\psi_0\leq x$, $\beta_0\geq 1/2-x$, and $\npm_0=\omega(\log n/\sqrt{n})$.
    Then, $\taucons= O\qty(n\log n/\npm_0)$ with high probability. 
\end{lemma}
\begin{proof}
  We have the following:
  \begin{itemize}
    \item From \cref{lem:unique strong opinion for PP} (\cref{item:tauusplus for PP}), for some $T_1=O\qty(n\log n/\npm_0)$, we have $\min_{j\neq I_{T_1}}\eta_{T_1}(j)\geq y$, $\beta_{T_1}\geq 1/2-8x$, $\psi_{T_1}\leq 8x$, and $\npm_{T_1}=\Omega(\npm_0)=\omega(\log n/\sqrt{n})$ with probability at least $1-n^{-10}$.
    \item From \cref{lem:towards consensus for PP}, for some $T_2=O\qty(n\log n/\npm_0)$, we have $\taucons\leq T_1+T_2$ with probability at least $1-1/n$.
  \end{itemize}
  Thus, we obtain the claim.
\end{proof}
\begin{proof}[Proof of \cref{thm:main}]
    From \cref{lem:good configuration for PP}, we have that $\psi_T\leq x$ and $\beta_T\geq 1/2-x$ and $\npm_T\geq (\log n)^2/\sqrt{n}$ hold with high probability for some $T$ as defined in \cref{eq:good configuration for PP} (\cref{lem:good configuration for PP}).
    Then, from \cref{lem:consensus time starting from a goood intial configuration}, we reach a consensus within additional $O\qty(\frac{n\log n}{\npm_T}) = O\qty( \min\qty{ kn\log n, n^{1.5}/\sqrt{\log n} } )$ steps with high probability.
    Here, we use $\npm_T \ge \frac{\beta_T}{k} = \Omega(k)$ if $k$ is small and $\npm_T\ge (\log n)^{1.5}/\sqrt{n}$ if $k$ is large.
    Therefore, the consensus time is bounded by
    \begin{align*}
        \taucons \le T + O\qty( \min\qty{ kn\log n, n^{1.5}/\sqrt{\log n}} ) = \Otilde(\min\qty{kn, n^{1.5}})
    \end{align*}
    and obtain the claim.
\end{proof}

\printbibliography

\appendix
\section{Technical Inequalities and Auxiliary Lemmas} \label{sec:tools}

This appendix collects several technical tools that are used repeatedly in the analysis but are not specific to USD.
We include variants of Chernoff-type bounds, concentration results for read-$\ell$ families, the optional stopping theorem, and basic notions such as stochastic domination.
None of these results are new; they are provided here for completeness and to keep the main proofs self-contained.

\begin{lemma}[\mbox{\cite[Corollary 1.10.4]{Doe18}}] \label{lem:chernoff-variant}
  Let $X_1,\dots,X_n$ be independent random variables that take values in $[0,1]$.
  Let $X=\sum_{i=1}^n X_i$.
  Then, for any $h\ge 2\mathrm{e}\E[X]$, we have
  \begin{align*}
    \Pr\qty[X \ge h] \le 2^{-h}.
  \end{align*}
\end{lemma}

While the Chernoff bound gives a tail bound for the sum of independent random variables, we will need a more general concentration result for the sum of random variables that are not necessarily independent.
In the following, we introduce the notion of \emph{read-$\ell$ family}, which is a family of random variables that are not necessarily independent but have a certain structure \cite{read-k,Duppala}.

\begin{definition}[Read-$\ell$ family; \cite{read-k}] \label{def:read-ell family}
  A family of real-valued random variables $(Y_1,\dots,Y_n)$ is called a \emph{read-$\ell$ family} if there exist $m\in\Nat$, independent random variables $X_1,\dots,X_m$, and subsets $S_1,\dots,S_n \subseteq [m]$ that satisfy the following:
  \begin{enumerate}
    \item Each $Y_i$ can be written as a function of $(x_s)_{s\in S_i}$ for $i\in [n]$.
    \item For each $j\in [m]$, the number of subsets $S_i$ that contain $j$ is at most $\ell$.
  \end{enumerate}
\end{definition}

\citet{read-k} shows for the first time a tail bound for the sum of a read-$\ell$ family using an information-theoretic argument, which does not yield a bound for the moment generating function.
Very recently, \citet{Duppala} shows a concentration result for the sum of a read-$\ell$ family based on the moment generating function using the following general inequality.

\begin{lemma}[\mbox{\cite[Lemma 16]{Duppala}}]
    \label{lem:moment of read-k family}
    Let $F_1,\ldots,F_n$ be a read-$\ell$ family.
    Then, $\E\qty[\prod_{i\in [n]}F_i] \leq  \qty(\prod_{i\in [n]}\E[F_i^\ell])^{1/\ell}$.
\end{lemma}

Note that, if $(Y_1,\dots,Y_n)$ forms a read-$1$ family, then so does the family $(F_1,\dots,F_n)$ for $F_i = \exp\qty(\lambda Y_i)$ for any $\lambda\in \Real$.
Applying \cref{lem:moment of read-k family} to this family $(F_1,\dots,F_n)$ yields a bound for the moment generating function.

\begin{theorem}[Optional stopping theorem. See, e.g.,~Theorem 4.8.5 of \cite{Dur19}]\label{thm:OST}
    Let $(X_t)_{t\in \mathbb{N}_0}$ be a submartingale (resp.\ supermartingale) such that $\E_{t-1}\sbra*{|X_t-X_{t-1}|}<\infty$ a.s.\
    and let $\tau$ be a stopping time such that $\E[\tau]<\infty$.
    Then, $\E[X_\tau]\geq \E[X_0]$ (resp.\ $\E[X_\tau]\leq \E[X_0]$).
\end{theorem}

\begin{definition}[Stochastic domination]
    \label{def:Stochastic domination}
    For two random variables $X$ and $Y$, we say that $Y$ stochastically dominates $X$, written as $X\preceq Y$, if for all $\lambda \in \mathbb{R}$ we have $\Pr\qty[X\leq \lambda]\geq \Pr\qty[Y\leq \lambda]$.
\end{definition}

\begin{lemma}[\mbox{\cite[Lemma 5.1]{SS25_sync}}] \label{lem:nazo lemma}
Let $ (Z_t)_{t\ge 0} $ be a Markov chain over a state space $ \Omega $ associated with natural filtration $\calF=(\calF_t)_{t\ge 0}$ and let $\tau$ be any stopping time with respect to $\calF$.
Let $ \varphi \colon \Omega \to \Real_{\ge 0} $ be a function.
For a parameter $ x\in\Real_{\ge 0} $, let $ \tauphiplus(x) = \inf\qty{t\ge 0 \colon \varphi(Z_t) \ge x} $.
Let $ T,x_0,\cphiup>0,x^*>x_0 $ be parameters and $\Omega^*\subseteq \Omega$ be the set of states $z\in\Omega$ such that $\varphi(z)\le x^*$.
Suppose that the following holds:
\begin{enumerate}\renewcommand{\labelenumi}{(\roman{enumi})}
    \item There exists $ C_1>0 $ such that for any $ z \in \Omega^* $, \[ \Pr\qty[\min\qty{\tauphiplus(x_0), \tau}\le T\condition Z_0=z] \ge C_1. \] 
    \item Define $ \tauphiup = \inf\qty{ t\ge 0 \colon \varphi(Z_t) \ge (1+\cphiup) \cdot \varphi(Z_0) } $. Then, there exists $ C_2>0$ such that for any $ z \in \Omega^* $, \[ \Pr\qty[ \min\qty{\tauphiup, \tau} \le T \condition Z_0 = z ] \ge 1-\exp\qty( - C_2 \varphi(z)^2 ).\] 
\end{enumerate}
Then, there exists a constant $ \constref{lem:nazo lemma} = \constref{lem:nazo lemma}(C_1,C_2,\cphiup,x_0) >0 $ (independent of $x^*$) such that,
for any $ z \in \Omega^* $ and any $ \varepsilon>0 $, we have
\begin{align*} 
  \Pr\qty[ \min\qty{\tauphiplus(x^*), \tau} \le \constref{lem:nazo lemma}\cdot T\cdot \qty(\log(1/\varepsilon) +  \log(x^*/x_0)) \condition Z_0 = z] \ge 1-\varepsilon.
\end{align*}
\end{lemma}

\subsection{Drift Analysis Results} \label{sec:tools for drift analysis}
In this section, we present some drift analysis tools that are useful in our analysis for USD.

\begin{proof}[Proof of \cref{lem:Useful drift lemma} (\cref{item:negative drift useful})]
    We apply \cref{lem:Freedman stopping time additive} (\cref{item:positive drift}) to $Y_t=-X_t$ with $\drift_Y=-\drift>0$.
We have
$
    \indicator_{\tau>t-1}\qty(\E_{t-1}[Y_t]-Y_{t-1}-\drift_Y)
    =\indicator_{\tau>t-1}\qty(X_{t-1}+\drift-\E_{t-1}[X_t])\leq 0
$
and $\indicator_{\tau>t-1}\qty(Y_t-Y_{t-1}-\drift_Y)=\indicator_{\tau>t-1}\qty(X_{t-1}+\drift-X_t)$ satisfies one-sided $\qty(\bounded,\variance)$-Bernstein condition.
Hence, letting $h=X_0-I^-$, $T=\frac{(1-\epsilon)h}{-\drift}=\frac{(1-\epsilon)h}{\drift_Y}$, and $z=h-\drift_Y T=\epsilon h =\epsilon(X_0-I^-)$, we obtain
\begin{align*}
    \Pr\qty[\tau^- \le \min\{T,\tau\}] 
    \leq \exp\qty(- \frac{\epsilon^2(X_0-I^-)^2/2}{\variance T + (\epsilon(X_0-I^-) \bounded)/3}).
\end{align*}
Note that $\tau^-=\inf\{t\geq 0: X_t\leq I^-\}=\inf\{t\geq 0: Y_t\geq Y_0+h\}$.
\end{proof}
\begin{proof}[Proof of \cref{lem:Useful drift lemma} (\cref{item:positive drift useful})]
    We have
    \begin{align*}
        &\Pr\qty[\tau^+>T \text{ and } \tau^*>T] \\
        &= \Pr\qty[\tau^+>T \text{ and } \tau^*>T \text{ and } \tau^->T] 
        + \Pr\qty[\tau^+>T \text{ and } \tau^*>T \text{ and } \tau^-\leq T] \\
        &\leq \Pr\qty[\min\{\tau^+,\tau^-,\tau^*\}>T] +\Pr\qty[\tau^-\leq \min\{T,\tau^+,\tau^*\}]
    \end{align*}
    and
    \begin{align*}
        &\Pr\qty[\tau^+>\tau^- \text{ and } \tau^*>T] \\
        &= \Pr\qty[\tau^+>\tau^- \text{ and } \tau^*>T \text{ and } \tau^->T] 
        + \Pr\qty[\tau^+>\tau^- \text{ and } \tau^*>T \text{ and } \tau^-\leq T] \\
        &\leq \Pr\qty[\min\{\tau^+,\tau^-,\tau^*\}>T] +\Pr\qty[\tau^-\leq \min\{T,\tau^+,\tau^*\}].
    \end{align*}
    For $\Pr\qty[\min\{\tau^+,\tau^-,\tau^*\}>T]$, we apply \cref{lem:Freedman stopping time additive} (\cref{item:negative drift}) to $Y_t=-X_t$ with $\drift_Y=-\drift<0$.
    We have
    $
        \indicator_{\tau>t-1}\qty(\E_{t-1}[Y_t]-Y_{t-1}-\drift_Y)
        =\indicator_{\tau>t-1}\qty(X_{t-1}+\drift-\E_{t-1}[X_t])\leq 0
    $
    and $\indicator_{\tau>t-1}\qty(Y_t-Y_{t-1}-\drift_Y)=\indicator_{\tau>t-1}\qty(X_{t-1}+\drift-X_t)$ satisfies one-sided $\qty(\bounded,\variance)$-Bernstein condition.
    Hence, letting $h=I^+-X_0$, 
    $T=\frac{(1+\epsilon)h}{\drift}=\frac{(1+\epsilon)h}{-\drift_Y}$, and $z=(-\drift_Y)T-h=\epsilon h =\epsilon(I^+-X_0)$, we obtain
    \begin{align*}
        \Pr\qty[\min\{\tau^+,\tau\} >T] 
        \le \exp\qty(- \frac{\epsilon^2(I^+-X_0)^2/2}{\variance T + (\epsilon(I^+-X_0) \bounded)/3}).
    \end{align*}
    Note that $\tau^+=\inf\{t\geq 0: X_t\geq I^+\}=\inf\{t\geq 0: Y_t\leq Y_0-h\}$.

    For $\Pr\qty[\tau^-\leq \min\{T,\tau^+,\tau^*\}]$, we apply \cref{lem:Freedman stopping time additive} (\cref{item:positive drift}) to $Y_t=-X_t$ and $\drift_Y=0$.
    Then, since $\drift>0$, we have
    $
        \indicator_{\tau>t-1}\qty(\E_{t-1}[Y_t]-Y_{t-1}-\drift_Y)
        \leq \indicator_{\tau>t-1}\qty(X_{t-1}+\drift-\E_{t-1}[X_t])\leq 0
    $
    and $\indicator_{\tau>t-1}\qty(Y_t-Y_{t-1}-\drift_Y)\leq \indicator_{\tau>t-1}\qty(X_{t-1}+\drift-X_t)$ satisfies one-sided $\qty(\bounded,\variance)$-Bernstein condition.
    Hence, letting $z=h=X_0-I^-$, we obtain 
    \begin{align*}
        \Pr\qty[\tau^- \le \min\{T,\tau\}] 
        \leq \exp\qty(- \frac{(X_0-I^-)^2/2}{\variance T + ((X_0-I^-) \bounded)/3}).
    \end{align*}
    Note that $\tau^-=\inf\{t\geq 0: X_t\leq I^-\}=\inf\{t\geq 0: Y_t\geq Y_0+h\}$.
    
\end{proof}

\begin{lemma}
    \label{lem:Iterative Drift theorem}
    Let $(X_t)_{t\in \Nat_0}$ be a sequence of random variables and
    let $(\calF_t)_{t\in\Nat_0}$ be a filtration such that $ X_t $ is $ \calF_t $-measurable for all $t\ge 0$.
    Let $\tau^*$ be a stopping time with respect to $(\calF_t)_{t\in \Nat_0}$.
    We have the following:
    \begin{enumerate}
        \item \label{item:bounded decrease under small negative drift}
        Let $c^\uparrow, c^\downarrow > 0$ be positive constants.
        For $s\geq 0$, let
        \begin{align*}
            \tau^\uparrow_s \defeq \inf\{t\geq s: X_t\geq (1+c^\uparrow)X_s\}, \;
            \tau^\downarrow_s \defeq \inf\{t\geq s: X_t\leq (1-c^\downarrow)X_s\}.
        \end{align*}
        Suppose $X_0>0$.
        Then, for $\tau^\downarrow=\tau_0^\downarrow$ and any $T> 0$, we have
        \begin{align*}
            \Pr\qty[\tau^\downarrow \leq T \text{ and } \tau^*> T]\leq 
            \sum_{s=0}^{T-1}\E\qty[\indicator_{X_s\geq X_0 \text{ and } \tau^*> s}\Pr_s\qty[\tau_s^\downarrow \leq \min\{T,\tau_s^\uparrow,\tau^*\}]].
        \end{align*}
        \item \label{item:no increase under negative drift}
        Let $I^-< I^-_*<I^{+}_*<I^{+}$ be parameters.
        For $s\geq 0$, let
        \begin{align*}
            \tau^+_s\defeq \inf\qty{t\geq s\colon X_t\geq I^+},\;
            \tau^-_s \defeq \inf\qty{t\geq s\colon X_t\leq I^-}.
        \end{align*}
        Let $\tau^{jump}\defeq \inf\qty{t> 0: X_t\leq X_{t-1} - (I^{+}_*-I^-_*)}$.
        Then, for $\tau^-=\tau_0^-$, $X_0\geq I^-_*$ and $T> 0$, we have
        \begin{align*}
            \Pr\qty[\tau^-\leq T \text{ and } \min\{\tau^{jump},\tau^*\}> T]
            \leq \sum_{s=0}^{T-1}\E\qty[\indicator_{X_s\in [I^-_*,I^{+}_*] \text{ and } \tau^*>s}\Pr_s\qty[\tau^-_s<\tau^+_s]].
        \end{align*}
        \item \label{item:iterative updrift}
        Let $c^\uparrow > 0$ be a positive constant.
    For $s \geq 0$, let
    \begin{align*}
        \tau^{\uparrow (s)} \defeq \inf\{t\geq 0: X_t\geq (1+c^\uparrow)^s X_0\}.
    \end{align*}
    Suppose $X_0>0$.
    Then, for any $T>0$ and $\ell>0$, 
    \begin{align*}
        \Pr\qty[\min\{\tau^{\uparrow (\ell)},\tau^*\}>\ell T]
        \leq \sum_{s=1}^{\ell}\E\qty[\indicator_{\tau^*>\tau^{\uparrow (s-1)}} \Pr_{\tau^{\uparrow (s-1)}}\qty[\min\{\tau^{\uparrow (s)},\tau^*\}> \tau^{\uparrow (s-1)}+T]].
    \end{align*}
    \end{enumerate}
\end{lemma}
    \begin{proof}[Proof of \cref{lem:Iterative Drift theorem} (\cref{item:bounded decrease under small negative drift})]
        First, we claim the following:
        Let $\omega$ be any sample path such that $\tau^\downarrow(\omega) \leq T$.
        Then, there exists $s\in \{0,1,\ldots,T-1\}$ such that $X_s(\omega)\geq X_0(\omega)$ and $\tau_s^\downarrow(\omega) \leq \min\{T,\tau_s^\uparrow(\omega)\}$ hold.
    
        To see this, define $s$ to be the largest integer such that $X_s(\omega) = \max_{0\leq t< \tau^\downarrow(\omega)} X_t(\omega)$ holds.
        For such $s$, we have $X_s(\omega)\geq X_0(\omega)$ and $\tau_s^\downarrow(\omega) < \tau_s^\uparrow(\omega)$.
        Furthermore, $\tau_s^\downarrow(\omega) \leq \tau^\downarrow(\omega) \leq T$ since $(1-c^\downarrow)X_s(\omega)\geq (1-c^\downarrow)X_0(\omega)$.
        Thus, $\tau_s^\downarrow(\omega) \leq \min\{T,\tau_s^\uparrow(\omega)\}$ holds and this completes the proof of the claim.
   
        From the above claim, we obtain
        \begin{align*}
            \Pr\qty[\tau^\downarrow \leq T \text{ and } \tau^*> T]
            &\leq \Pr\qty[\exists s\in \{0,\ldots,T-1\}: X_s\geq X_0 \text{ and } \tau_s^\downarrow \leq \min\{T,\tau_s^\uparrow\} \text{ and } \tau^*> T]\\
            &\leq \sum_{s=0}^{T-1} \Pr\qty[X_s\geq X_0 \text{ and } \tau_s^\downarrow \leq \min\{T,\tau_s^\uparrow\} \text{ and } \tau^*> T \text{ and } \tau^*> s]\\
            &\leq \sum_{s=0}^{T-1}\E\qty[\indicator_{X_s\geq X_0 \text{ and } \tau^*> s}\Pr_s\qty[\tau_s^\downarrow \leq \min\{T,\tau_s^\uparrow,\tau^*\}]].
        \end{align*}
    \end{proof}
    \begin{proof}[Proof of \cref{lem:Iterative Drift theorem} (\cref{item:no increase under negative drift})]
    First, we claim the following:
    Let $\omega$ be any sample path such that $\tau^{-}(\omega) \leq T$ and $\tau^{jump}(\omega)>T$.
    Then, there exists $s\in \{0,1,\ldots,T-1\}$ such that $I^-_*\leq X_s(\omega)\leq I^{+}_*$ and $\tau^{-}_s(\omega)<\tau^{+}_s(\omega)$.
    
    To see this, define $s\in\{0,1,\ldots,T-1\}$ to be the largest integer such that
    $X_s(\omega)\in [I^-_*,I^{+}_*]$.
    Such an $s$ exists since $X_0(\omega)\geq I^-_*$ and $\tau^{jump}(\omega)>T$.
    By the definition of $s$, we have $X_t(\omega)\notin [I^-_*,I^{+}_*]$ for any $t>s$.
    Furthermore, since $\tau^{jump}(\omega)>T$, that is, the process never jumps
    across the interval $[I^-_*,I^{+}_*]$ in one step, it follows that $X_t(\omega)\le I^{+}_*$ for
    all $t\ge s$, and hence $\tau^{-}_s(\omega)<\tau^{+}_s(\omega)$.
    Indeed, if there existed some $t\ge s$ such that $X_t(\omega)>I^{+}_*$, then the occurrence
    of $\tau^-_s(\omega)< \tau^+_s(\omega)$ would imply the existence of a time $t'>t$ for which
    $X_{t'}(\omega)\in [I^-_*,I^{+}_*]$, contradicting the definition of $s$.
    
    Thus, 
\begin{align*}
    &\Pr\qty[\tau^-\leq T \text{ and } \tau^{jump}>T \text{ and } \tau^*> T]\\
    &\leq \Pr\qty[\exists s \in \{0,1,\ldots, T-1\}: \tau^{-}_s<\tau^{+}_s \text{ and } X_s\in [I^-_*,I^{+}_*] \text{ and } \tau^*> T]\\
    &\leq \sum_{s=0}^{T-1}\Pr\qty[X_s\in [I^-_*,I^{+}_*] \text{ and } \tau^{-}_s<\tau^{+}_s \text{ and }  \tau^*> s]\\
    &\leq \sum_{s=0}^{T-1}\E\qty[\indicator_{X_s\in [I^-_*,I^{+}_*] \text{ and } \tau^*> s}\Pr_s\qty[\tau^{-}_s<\tau^{+}_s]].
\end{align*}
    \end{proof}
\begin{proof}[Proof of \cref{lem:Iterative Drift theorem} (\cref{item:iterative updrift})]
    Let $\sigma_0=0$ and for $s\geq 1$, let
    \begin{align*}
        \sigma_s \defeq \min\{\tau^{\uparrow (s)}, \tau^*, \sigma_{s-1}+T\}.
    \end{align*}
    First, we claim the following:
    If the event $\min\{\tau^{\uparrow (\ell)},\tau^*\}>\ell T$ occurs, then there exists $s\in \{1,\ldots,\ell\}$ such that the events $\sigma_{s-1}=\tau^{\uparrow (s-1)}$, $\tau^*>\sigma_{s-1}$, and $\min\{\tau^{\uparrow (s)},\tau^*\}> \sigma_{s-1}+T$ occur.
    To see this, define $s$ to be the smallest integer such that $\min\{\tau^{\uparrow (s)},\tau^*\}> \sigma_{s-1}+T$.
    By the minimality of $s$, for any $i\leq s-1$, we have $\min\{\tau^{\uparrow (i)},\tau^*\}\leq \sigma_{i-1}+T$.
    Moreover, $\min\{\tau^{\uparrow (s)},\tau^*\}> \sigma_{s-1}+T$ implies $\tau^*>\sigma_{s-1}+T$, and hence $\tau^*>\sigma_{s-1}$.
    Thus, for any $i\leq s-1$, it follows that $\tau^*>\sigma_{i-1}$, and therefore $\sigma_{i}=\tau^{\uparrow (i)}$.
    This proves the claim.
    From the claim, we obtain
    \begin{align*}
        &\Pr\qty[\min\{\tau^{\uparrow (\ell)},\tau^*\}>\ell T]\\
        &\leq \Pr\qty[\exists s \in \{1,\ldots,\ell\}: \sigma_{s-1}=\tau^{\uparrow (s-1)} \text{ and } \tau^*>\sigma_{s-1} \text{ and } \min\{\tau^{\uparrow (s)},\tau^*\}> \sigma_{s-1}+T]\\
        &\leq \sum_{s=1}^{\ell}\Pr\qty[\tau^*>\tau^{\uparrow (s-1)} \text{ and } \min\{\tau^{\uparrow (s)},\tau^*\}> \tau^{\uparrow (s-1)}+T]\\
        & = \sum_{s=1}^{\ell}\E\qty[\indicator_{\tau^*>\tau^{\uparrow (s-1)}} \Pr_{\tau^{\uparrow (s-1)}}\qty[\min\{\tau^{\uparrow (s)},\tau^*\}> \tau^{\uparrow (s-1)}+T]].
    \end{align*}
\end{proof}

\section{Proof of Basic Properties} \label{sec:proof of basic properties}
In this section, we give a proof for basic inequalities for the key quantities in the gossip and population protocol models.
We first introduce the following lemma that will be used in the proof of $\tnpt_t$ and $\tnpm_t$.
\begin{lemma}
  \label{lem:basic inequalities for ratio distribution}
  Define a function $f(x,y)$ as
  \begin{align*}
    f(x,y)=
    \begin{cases}
      \frac{x}{y} & \text{if } y>0,\\
      0 & \text{if } y=0.
    \end{cases}
  \end{align*} 
  Then, for non-negative random variables $X$ and $Y$ such that $0\leq X\leq Y$ a.s.~and $\E[Y]>0$, we have the following:
  \begin{enumerate}
    \item \label{item:expectation of ratio distribution}
    $\E\qty[f(X,Y)]\geq \frac{\E[X]}{\E[Y]}-\frac{\Cov[X,Y]}{\E[Y]^2}$.
    \item \label{item:Bernstein condition of ratio distribution}
    $\frac{\E[X]}{\E[Y]}-f(X,Y)-\frac{\Cov[X,Y]}{\E[Y]^2}\leq \frac{2\qty(\E[X]-X)}{\E[Y]}+\frac{XY-\E[XY]}{\E[Y]^2}$ a.s.
  \end{enumerate}
\end{lemma}
\begin{proof}
  Since \cref{item:expectation of ratio distribution} follows immediately by taking expectations on both sides of \cref{item:Bernstein condition of ratio distribution}, we give a proof of \cref{item:Bernstein condition of ratio distribution}.
  
  First, we observe that the following inequality holds: For any $a\geq 0$, $b>0$, $x\geq 0$, and $y\geq x\geq 0$, 
  \begin{align*}
      \frac{a}{b}-f(x,y)\leq \frac{a-x}{b}+\frac{x}{b^2}(y-b).
  \end{align*}
  To see this, consider the following two cases: $y=0$ and $y>0$.
  If $y=0$, then $f(x,y)=0$ and $x=0$ from the definition.
  Thus, $\frac{a}{b}-f(x,y)=\frac{a}{b}=\frac{a-x}{b}+\frac{x}{b^2}(y-b)$ and the inequality holds.
  If $y>0$, then the inequality holds since $\frac{a}{b}-f(x,y)=\frac{a}{b}-\frac{x}{y}=\frac{a-x}{b}+\frac{x}{b^2}(y-b)-\frac{x\qty(y-b)^2}{b^2y}$.
  
  Hence, we obtain
  \begin{align*}
    \frac{\E[X]}{\E[Y]}-f(X,Y)-\frac{\Cov[X,Y]}{\E[Y]^2}
    &\leq \frac{\E[X]-X}{\E[Y]}+\frac{X}{\E[Y]^2}(Y-\E[Y])-\frac{\Cov[X,Y]}{\E[Y]^2}\\
    &= \frac{2\qty(\E[X]-X)}{\E[Y]}+\frac{XY-\E[XY]}{\E[Y]^2}.
  \end{align*}
\end{proof}

\subsection{Gossip Model} \label{sec:proof of basic properties in the gossip model}
By definition of the gossip USD (Definition~\ref{def:gossip USD}), we have for any $v\in V$ and $t\ge 1$:
\begin{align}
\Pr_{t-1}[\opn_t(v)=i]
=\begin{cases}
\alpha_{t-1}(i)+1-\beta_{t-1} & (\text{if } \opn_{t-1}(v)=i),\\
\alpha_{t-1}(i) & (\text{if } \opn_{t-1}(v)=\bot),\\
0 & (\text{otherwise}).
\end{cases}
\label{eq:updating probability for alpha}
\end{align}

\begin{proof}[Proof of \cref{lem:basic inequalities for alpha}]
\medskip
\noindent\textbf{(i) Expectation.}
From \cref{eq:updating probability for alpha}, we have
\begin{align*}
\E_{t-1}[\alpha_t(i)]
&=\frac{1}{n}\sum_{v\in V:\opn_{t-1}(v)=i}\Pr_{t-1}[\opn_t(v)=i]
   +\frac{1}{n}\sum_{v\in V:\opn_{t-1}(v)=\bot}\Pr_{t-1}[\opn_t(v)=i]\\
&=\alpha_{t-1}(i)\qty(\alpha_{t-1}(i)+1-\beta_{t-1})
  +(1-\beta_{t-1})\alpha_{t-1}(i)\\
&=\alpha_{t-1}(i)\qty(\alpha_{t-1}(i)+2(1-\beta_{t-1})).
\end{align*}
\medskip
\noindent\textbf{(ii) Variance.}
Since $\alpha_t(i) = (\sum_{v\in V}\indicator_{\opn_t(v)=i})/n$, then from \cref{eq:updating probability for alpha},
\begin{align*}
\Var_{t-1}[\alpha_t(i)]
&=\frac{1}{n^2}\sum_{v\in V}\Var_{t-1}\qty[\indicator_{\opn_t(v)=i}] \\
&=\frac{1}{n^2}\sum_{v\in V}
   \Pr_{t-1}[\opn_t(v)=i]\Pr_{t-1}[\opn_t(v)\neq i] \\
&=\frac{\alpha_{t-1}(i)}{n}
   \qty(\alpha_{t-1}(i)+1-\beta_{t-1})\qty(\beta_{t-1}-\alpha_{t-1}(i))
  +\frac{1-\beta_{t-1}}{n}\alpha_{t-1}(i)\qty(1-\alpha_{t-1}(i))\\
&=\frac{\alpha_{t-1}(i)}{n}\qty[
   (1-\beta_{t-1})(1+\beta_{t-1}-2\alpha_{t-1}(i))
   +\alpha_{t-1}(i)(\beta_{t-1}-\alpha_{t-1}(i)) ].
\end{align*}
Thus,
\[
\Var_{t-1}[\alpha_t(i)]
\le \frac{\alpha_{t-1}(i)}{n}
\]
and
\[
\Var_{t-1}[\alpha_t(i)]
\ge \frac{(1-\beta_{t-1})^2\alpha_{t-1}(i)}{n}.
\]
\medskip
\noindent\textbf{(iii) Bernstein condition.}
For $v\in V$, define
\[
    X_t(v)=\frac{1}{n}\qty(\indicator_{\opn_t(v)=i}
    -\E_{t-1}[\indicator_{\opn_t(v)=i}]).
\]
Since $\lvert X_t(v)\rvert \le 1/n$, Lemma~\ref{lem:Bernstein condition}(i)
implies that $X_t(v)$ satisfies a
\[
\qty(\tfrac{1}{n},\,\Var_{t-1}[X_t(v)])
\text{-Bernstein condition}.
\]
Because
\[
\alpha_t(i)-\E_{t-1}[\alpha_t(i)]
=\sum_{v\in V} X_t(v)
\]
and the variables $X_t(v)$, $v\in V$ are independent, Lemma~\ref{lem:Bernstein condition}(v)
implies that this sum satisfies a
\[
\qty(\tfrac{1}{n},\,\sum_{v\in V}\Var_{t-1}[X_t(v)])
\text{-Bernstein condition}.
\]
Using \cref{item:variance of alpha} of Lemma~\ref{lem:basic inequalities for alpha},
\[
\sum_{v\in V}\Var_{t-1}[X_t(v)]
=\Var_{t-1}[\alpha_t(i)]
\le \frac{\alpha_{t-1}(i)}{n}.
\]
Thus the result follows.
\end{proof}

\begin{proof}[Proof of \cref{lem:basic inequalities for beta}]
\medskip
\noindent\textbf{(i) Expectation.}
  From \cref{item:expectation of alpha} of \cref{lem:basic inequalities for alpha}, we have
  \begin{align*}
      \E_{t-1}\qty[\npo_t]
      &=\sum_{i\in [k]}\np_{t-1}(i)\qty(\np_{t-1}(i)+2\qty(1-\npo_{t-1}))
      =\npt_{t-1}+2\npo_{t-1}(1-\npo_{t-1}).
  \end{align*}
\medskip
\noindent\textbf{(ii) Variance.}
  We begin by calculating $\mathbf{Cov}_{t-1}\qty[\alpha_t(i),\alpha_t(j)]$ and $\mathbf{Cov}_{t-1}\qty[\alpha_t(i),\beta_t]$, which will be used later.
  \begin{claim}
  \label{claim:covariance of alpha}
  For any distinct $i,j\in [k]$ and $t\geq 1$, 
      $\Cov_{t-1}[\alpha_t(i),\alpha_t(j)] = -\frac{\alpha_{t-1}(i)\alpha_{t-1}(j)}{n}\qty(1-\beta_{t-1})$.  
  \end{claim}
  \begin{proof}
    From \cref{eq:updating probability for alpha},
\begin{align*}
\Cov_{t-1}\qty[\alpha_t(i), \alpha_t(j)]
&=\frac{1}{n^2}\sum_{v\in V}\sum_{u\in V}
  \Cov_{t-1}\qty[\indicator_{\opn_{t}(v)=i},
                  \indicator_{\opn_{t}(u)=j}]\\
&=\frac{1}{n^2}\sum_{v\in V}
  \Cov_{t-1}\qty[\indicator_{\opn_t(v)=i},
                  \indicator_{\opn_t(v)=j}]\\
&=-\frac{1}{n^2}\sum_{v\in V}
  \E_{t-1}[\indicator_{\opn_t(v)=i}]
  \E_{t-1}[\indicator_{\opn_t(v)=j}]\\
&=-\frac{1-\beta_{t-1}}{n}\alpha_{t-1}(i)\alpha_{t-1}(j).
\end{align*}
  \end{proof}
  \begin{claim}
  \label{claim:covariance of alpha and beta}
  For any $i\in [k]$ and $t\geq 1$, it holds that
  \[\mathbf{Cov}_{t-1}\qty[\alpha_t(i),\beta_t]=\frac{\alpha_{t-1}(i)}{n}\qty[(1-\beta_{t-1})(1-\alpha_{t-1}(i))+\alpha_{t-1}(i)(\beta_{t-1}-\alpha_{t-1}(i))].\]
  Specifically, $\mathbf{Cov}_{t-1}\qty[\alpha_t(i),\beta_t]\leq \frac{\alpha_{t-1}(i)}{n}$.
  \end{claim}
  \begin{proof}
    From \cref{claim:covariance of alpha} and \cref{lem:basic inequalities for alpha} (\cref{item:variance of alpha}), we have
    \begin{align*}
      \mathbf{Cov}_{t-1}\qty[\alpha_t(i),\beta_t]
      &=\Var_{t-1}\qty[\alpha_t(i)]+\sum_{j\in [k]\setminus \{i\}}\mathbf{Cov}_{t-1}\qty[\alpha_t(i),\alpha_t(j)] \nonumber\\
      &=\Var_{t-1}\qty[\alpha_t(i)]-\frac{(1-\beta_{t-1})\alpha_{t-1}(i)(\beta_{t-1}-\alpha_{t-1}(i))}{n}\nonumber\\
      &=\frac{\alpha_{t-1}(i)}{n}\qty[(1-\beta_{t-1})(1-\alpha_{t-1}(i))+\alpha_{t-1}(i)(\beta_{t-1}-\alpha_{t-1}(i))]\\
  &\leq \frac{\alpha_{t-1}(i)}{n}. 
    \end{align*}
  \end{proof}
  Since $\Var_{t-1}\qty[\beta_{t-1}]
  =\sum_{i\in [k]}\sum_{j\in [k]}\Cov_{t-1}\qty[\alpha_t(i), \alpha_t(j)]
  =\sum_{i\in [k]}\Cov_{t-1}\qty[\alpha_t(i), \beta_t]$, we obtain
  \begin{align*}
    \Var_{t-1}\qty[\beta_{t-1}]
    &=\frac{(1-\beta_{t-1})(\beta_{t-1}-\gamma_{t-1})+\beta_{t-1}\gamma_{t-1}-\norm{\alpha_{t-1}}_3^3}{n}\\
    &= \frac{(\beta_{t-1}-\gamma_{t-1})(1-\beta_{t-1}+\gamma_{t-1})+\gamma_{t-1}^2-\norm{\alpha_{t-1}}_3^3}{n}\\
    &\leq \frac{\beta_{t-1}}{n}.
\end{align*}
Note that $\gamma_{t-1}^2\leq \beta_{t-1}\norm{\alpha_{t-1}}_3^3$ holds from the Cauchy-Schwarz inequality.
Hence, the result follows.
\medskip
\noindent\textbf{(iii) Bernstein condition.}
  Write $Y_t(v)=\frac{1}{n}(\indicator_{\opn_t(v)\neq \bot}-\E_{t-1}[\indicator_{\opn_t(v)\neq \bot}])$ for convenience.
  From \cref{item:BC for bounded rv} of \cref{lem:Bernstein condition}, $Y_t(v)$ satisfies $\qty(\frac{1}{n}, \Var_{t-1}[Y_t(v)])$-Bernstein condition.
  From definition, we have
  $\npo_t-\E_{t-1}[\npo_t]=\sum_{v\in V}Y_t(v)$, i.e., $\npo_t-\E_{t-1}[\npo_t]$ conditioned on $t-1$ round is the sum of $n$ independent random variables $(Y_t(v))_{v\in V}$.
  Hence, by \cref{item:BC for independent rvs} of \cref{lem:Bernstein condition}, $\npo_t-\E_{t-1}[\npo_t]$ satisfies $\qty(\frac{1}{n}, \sum_{v\in V}\Var_{t-1}[Y_t(v)])$-Bernstein condition.
  Since
  \begin{align*}
      \sum_{v\in V}\Var_{t-1}[Y_t(v)]
      &=\Var[\npo_t]
      \leq \frac{\npo_{t-1}(i)}{n}
  \end{align*}
  holds from \cref{lem:basic inequalities for beta} (\cref{item:variance of beta}), we obtain the claim.
\end{proof}

\begin{proof}[Proof of \cref{lem:basic inequalities for delta_epsilon}]
  \medskip
  \noindent\textbf{(i) Expectation.}
      From \cref{item:expectation of alpha} of \cref{lem:basic inequalities for alpha}, we have
      \begin{align*}
          \E_{t-1}\qty[\delta_t^{(\epsilon)}]
          &=\alpha_{t-1}(i)\qty(\alpha_{t-1}(i)+2(1-\beta_{t-1}))
          -(1+\varepsilon)\alpha_{t-1}(j)\qty(\alpha_{t-1}(j)+2(1-\beta_{t-1}))\\
          &=\qty(\alpha_{t-1}(i)-(1+\varepsilon)\alpha_{t-1}(j)) \qty(\alpha_{t-1}(i)+\alpha_{t-1}(j)+2(1-\beta_{t-1}))+\varepsilon\alpha_{t-1}(i)\alpha_{t-1}(j).
      \end{align*}
  \medskip
  \noindent\textbf{(ii) Variance.}
      From \cref{item:variance of alpha,claim:covariance of alpha} of \cref{lem:basic inequalities for alpha}, we have
      \begin{align*}
          \Var_{t-1}\qty[\delta_t]
          =\Var_{t-1}\qty[\alpha_t(i)]+\Var_{t-1}\qty[\alpha_t(j)]-2\Cov_{t-1}\qty[\alpha_t(i),\alpha_t(j)]
          \geq \frac{(1-\beta_{t-1})^2}{n}\qty(\alpha_{t-1}(i)+\alpha_{t-1}(j)).
      \end{align*}
  \medskip
  \noindent\textbf{(iii) Bernstein condition.}
      From \cref{lem:basic inequalities for alpha} (\cref{item:bernstein condition of alpha}) and \cref{lem:Bernstein condition} (\cref{item:BC for linear transformation}), for any $l\in [k]$ and $x\geq 0$,
      $x\cdot \qty(\alpha_t(l)-\E_{t-1}[\alpha_t(l)])$ satisfies $\qty(\frac{x}{n},\frac{x^2\alpha_{t-1}(i)}{n})$-Bernstein condition.
      Since 
      \[\delta_t^{(\epsilon)}-\E_{t-1}\qty[\delta_t^{(\epsilon)}]
      =\qty(\alpha_t(i)-\E_{t-1}[\alpha_t(i)])-(1+\epsilon)\qty(\alpha_t(j)-\E_{t-1}[\alpha_t(j)]),
      \]
      We obtain the claim by \cref{lem:Bernstein condition} (\cref{item:BC for upper bounded rv}) and \cref{lem:Bernstein condition 2} (\cref{item:BC for sum of rvs}).
  \end{proof}

\begin{proof}[Proof of \cref{lem:basic inequalities for gamma}]
\medskip
\noindent\textbf{(i) Expectation Upper Bound.}
  First, we observe that
  \begin{align}
    \sum_{i\in [k]} \E_{t-1}[\alpha_t(i)]^2
    =4(1-\beta_{t-1})^2\gamma_{t-1}+4(1-\beta_{t-1})\norm{\alpha_t}_3^3+\norm{\alpha_t}_4^4
    \label{eq:sum of squared expectation of alpha}
  \end{align}
  holds from \cref{item:expectation of alpha} of \cref{lem:basic inequalities for alpha}.
  Hence, 
  \begin{align*}
    \E_{t-1}[\gamma_t]
    &=4(1-\beta_{t-1})^2\gamma_{t-1}+4(1-\beta_{t-1})\norm{\alpha_t}_3^3+\norm{\alpha_t}_4^4+\sum_{i\in [k]} \Var_{t-1}[\alpha_t(i)] 
    \leq 10\gamma_{t-1}.
  \end{align*}
  holds. Note that $\sum_{i\in [k]} \Var_{t-1}[\alpha_t(i)] \leq \frac{\beta_{t-1}}{n}$ from \cref{item:variance of alpha} of \cref{lem:basic inequalities for alpha} and $\beta_{t-1}\leq n\gamma_{t-1}$ holds.
\medskip
\noindent\textbf{(ii) Expectation Lower Bound.}
  From the Cauchy-Schwarz inequality, we have
  \begin{align*}
    \gamma_t^2 = \qty(\sum_{i\in [k]}\alpha_t(i)^{0.5}\alpha_t(i)^{1.5})^2\leq \beta_t\norm{\alpha_t}_3^3 \quad\text{and}\quad
    \norm{\alpha_t}_3^6 = \qty(\sum_{i\in [k]}\alpha_t(i)\alpha_t(i)^{2})^2\leq \gamma_t\norm{\alpha_t}_4^4.
  \end{align*}
  Hence, from \cref{eq:sum of squared expectation of alpha} and \cref{lem:basic inequalities for beta} (\cref{item:expectation of beta}), we have
  \begin{align*}
    \npo_{t-1}^2\npt_{t-1}\sum_{i\in [k]} \E_{t-1}[\alpha_t(i)]^2
    &\geq 4(1-\npo_{t-1})^2 \npo_{t-1}^2\npt_{t-1}^2 + 4(1-\npo_{t-1})\npo_{t-1}^2\npt_{t-1}\norm{\np_t}_3^3+\npo_{t-1}^2\npt_{t-1}\norm{\np_{t-1}}_4^4 \\
    &\geq 4(1-\npo_{t-1})^2 \npo_{t-1}^2\npt_{t-1}^2 + 4(1-\npo_{t-1})\npo_{t-1}\npt_{t-1}^3+\npt_{t-1}^4 \\
    &=\npt_{t-1}^2\E_{t-1}[\npo_t]^2. 
\end{align*}
  Furthermore, $\Var_{t-1}[\alpha_t(i)] \geq \frac{\qty(1-\beta_{t-1})^2\alpha_{t-1}(i)}{n}$ from \cref{lem:basic inequalities for alpha} (\cref{item:variance of alpha}), we obtain
  \begin{align*}
  \beta_{t-1}^2\E_{t-1}[\gamma_t]
  =\beta_{t-1}^2\qty(\sum_{i\in [k]} \E_{t-1}[\alpha_t(i)]^2+\sum_{i\in [k]} \Var_{t-1}[\alpha_t(i)]) 
  \geq \gamma_{t-1}\E_{t-1}[\beta_t]^2+\frac{\beta_{t-1}^3(1-\beta_{t-1})^2}{n}.
  \end{align*}
\medskip
\noindent\textbf{(iii) Bernstein condition.}
  Write 
  \[Z_t(v,u)=\frac{1}{n^2}\qty(\indicator_{\opn_t(v)=\opn_t(u)}\indicator_{\opn_t(v)\neq \bot}-\E_{t-1}[\indicator_{\opn_t(v)=\opn_t(u)}\indicator_{\opn_t(v)\neq \bot}])\]
  for convenience.
  From \cref{item:BC for bounded rv} of \cref{lem:Bernstein condition}, $Z_t(v,u)$ satisfies $\qty(\frac{1}{n^2}, \Var_{t-1}[Z_t(v,u)])$-Bernstein condition.
  Observe that $\npt_t-\E_{t-1}[\npt_t]=\sum_{v\in V}\sum_{u\in V}Z_t(v,u)$.
  In other words, $\npt_t-\E_{t-1}[\npt_t]$ conditioned on $t-1$ round is the sum of $n^2$ random variables $(Z_t(v,u))_{v,u\in V}$.
  Furthermore, $(Z_t(v,u))_{v,u\in V}$ is a read-$2n-1$ family of independent random variables $(\opn_t(v))_{v\in V}$.
  Thus, from \cref{item:BC for read-k family} of \cref{lem:Bernstein condition}, $\npt_t-\E_{t-1}[\npt_t]$ satisfies $\qty(\frac{2n-1}{n^2}, (2n-1)\sum_{v,u\in V}\Var_{t-1}[Z_t(v,u)])$-Bernstein condition.
  Now, we bound $\sum_{v,u\in V}\Var_{t-1}[Z_t(v,u)]$. 
  From \cref{lem:basic inequalities for gamma} (\cref{item:expectation of gamma}), we have
  \begin{align*}
      \sum_{v,u\in V}\Var_{t-1}[Z_t(v,u)]
      \leq \frac{1}{n^4}\sum_{v,u\in V}\E_{t-1}[\indicator_{\opn_t(v)=\opn_t(u)}\indicator_{\opn_t(v)\neq \bot}]
      = \frac{1}{n^2}\E_{t-1}[\npt_t]
      \leq \frac{10\gamma_{t-1}}{n^2},
  \end{align*}
  and we obtain the result.
\end{proof}

\begin{proof}[Proof of \cref{lem:basic inequalities for psi}]
\medskip
\noindent\textbf{(i) Expectation.}
  From definition of $\psi_t$ and \cref{lem:basic inequalities for beta} (\cref{item:expectation of beta}), we have
  \begin{align*}
    &\E_{t-1}[\psi_t]
    =2\E_{t-1}\qty[\beta_t]^2-\E_{t-1}[\beta_t]-\E_{t-1}[\gamma_t]+\Var_{t-1}[\beta_t], \\
    &\E_{t-1}[\beta_t]=\beta_{t-1}-\psi_{t-1},\\
    &2\beta_{t-1}^2-\gamma_{t-1}=\beta_{t-1}+\psi_{t-1}.
  \end{align*}
    Applying \cref{lem:basic inequalities for tnpt} (\cref{item:expectation of gammatilde}), we have
    \begin{align*}
      \beta_{t-1}^2\qty(2\E_{t-1}\qty[\beta_t]^2-\E_{t-1}[\beta_t]-\E_{t-1}[\gamma_t])
      &\leq \beta_{t-1}^2\qty(2\E_{t-1}\qty[\beta_t]^2-\E_{t-1}[\beta_t]) - \E_{t-1}[\beta_t]^2\gamma_{t-1}\\
      &=-\psi_{t-1}^2\E_{t-1}[\npo_t]\\
      &\leq 0.
  \end{align*}
  Thus, from \cref{lem:basic inequalities for beta} (\cref{item:variance of beta}), we obtain
  \begin{align*}
    \E_{t-1}[\psi_t]
    \leq \Var_{t-1}[\beta_t]\leq \frac{\beta_{t-1}}{n}.
  \end{align*}
  \medskip
\noindent\textbf{(ii) Bernstein condition.}
   First, we show the following claim: 
  \begin{claim}
  $\npo_t^2-\E_{t-1}[\npo_t^2]$ satisfies $\qty(\frac{2}{n},\frac{20\beta_{t-1}^2}{n})$-Bernstein condition.
  \end{claim}
  \begin{proof}
    Write 
    \[
        Y_t(v,u)=\frac{1}{n^2}\qty(\indicator_{\opn_t(v)\neq \bot}\indicator_{\opn_t(u)\neq \bot}-\E_{t-1}[\indicator_{\opn_t(v)\neq \bot}\indicator_{\opn_t(u)\neq \bot}])
    \]
   for convenience.
   From \cref{item:BC for bounded rv} of \cref{lem:Bernstein condition}, $Y_t(v,u)$ satisfies $\qty(\frac{1}{n^2}, \Var_{t-1}[Y_t(v,u)])$-Bernstein condition.
   Observe that $\npo_t^2-\E_{t-1}[\npo_t^2]=\sum_{v\in V}\sum_{u\in V}Y_t(v,u)$.
   In other words, $\npo_t^2-\E_{t-1}[\npo_t^2]$ conditioned on $t-1$ round is the sum of $n^2$ random variables $(Y_t(v,u))_{v,u\in V}$.
   Furthermore, $(Y_t(v,u))_{v,u\in V}$ is a read-$2n-1$ family of independent random variables $(\opn_t(v))_{v\in V}$.
   Thus, from \cref{item:BC for read-k family} of \cref{lem:Bernstein condition}, $\npo_t^2-\E_{t-1}[\npo_t^2]$ satisfies $\qty(\frac{2n-1}{n^2}, (2n-1)\sum_{v,u\in V}\Var_{t-1}[Y_t(v,u)])$-Bernstein condition.
   Now, we bound $\sum_{v,u\in V}\Var_{t-1}[Y_t(v,u)]$.
   From \cref{lem:basic inequalities for beta} (\cref{item:expectation of beta,item:variance of beta}), we have
   \begin{align*}
        \sum_{v,u\in V}\Var_{t-1}[Y_t(v,u)]
        &\leq \frac{1}{n^4}\sum_{v,u\in V}\E_{t-1}[\indicator_{\opn_t(v)\neq \bot}\indicator_{\opn_t(u)\neq \bot}]\\
        &= \frac{1}{n^2}\E_{t-1}[\npo_t^2]\\
        &= \frac{1}{n^2}\qty(\E_{t-1}[\npo_t]^2+\Var_{t-1}[\npo_t])\\
        &\leq \frac{1}{n^2}\qty(9\npo_{t-1}^2+\frac{\npo_{t-1}}{n})\\
        &\leq \frac{10\beta_{t-1}^2}{n^2}.
   \end{align*}
   holds, and we obtain the claim. 
  \end{proof}

  From definition, $\psi_t-\E_{t-1}[\psi_t]=2\qty(\npo_t^2-\E_{t-1}[\npo_t^2])+\qty(\E_{t-1}[\npo_t]-\npo_t)+\qty(\E_{t-1}[\npt_t]-\npt_t)$.
  Combining \cref{item:BC for sum of rvs,item:BC for linear transformation} of \cref{lem:Bernstein condition}, 
  $2\qty(\npo_t^2-\E_{t-1}[\npo_t^2])+\qty(\E_{t-1}[\npo_t]-\npo_t)+\qty(\E_{t-1}[\npt_t]-\npt_t)$
  satisfies $\qty(3\cdot \frac{8}{n}, 3\cdot \qty(\frac{80\npo_{t-1}}{n}+\frac{\npo_{t-1}}{n}+\frac{20\npo_{t-1}}{n}))$-Bernstein condition.
\end{proof}

\begin{proof}[Proof of \cref{lem:basic inequalities for tnpm}]
  \medskip
  \noindent\textbf{(i) Expectation.}
    From \cref{lem:basic inequalities for alpha,lem:basic inequalities for beta},
    \begin{align}
        \frac{\E_{t-1}\qty[\np_t(I_{t-1})]}{\E_{t-1}[\npo_t]}
        &=\frac{\np_{t-1}(I_{t-1})}{\npo_{t-1}}\qty(1+\frac{\frac{\E_{t-1}\qty[\np_t(I_{t-1})]}{\np_{t-1}(I_{t-1})}-\frac{\E_{t-1}[\npo_t]}{\npo_{t-1}}}{\frac{\E_{t-1}[\npo_t]}{\npo_{t-1}}}) \nonumber \\
        &=\frac{\np_{t-1}(I_{t-1})}{\npo_{t-1}}\cdot \frac{\np_{t-1}(I_{t-1})+2(1-\npo_{t-1})}{2(1-\npo_{t-1})+\npt_{t-1}/\npo_{t-1}}\nonumber \\
        &=\tnpm_{t-1}\qty(1+\frac{\npm_{t-1}-\npt_{t-1}/\npo_{t-1}}{2(1-\npo_{t-1})+\npt_{t-1}/\npo_{t-1}}) 
        \label{eq:expectation of alphatilde lower bound}
    \end{align}
    holds. 
    Note that $I_{t-1}=\min\qty{i\in [k]:\alpha_{t-1}(i)=\npm_{t-1}}=\min\qty{i\in [k]:\tnp_{t-1}(i)=\tnpm_{t-1}}$.
    Combining the above and \cref{lem:basic inequalities for ratio distribution} (\cref{item:expectation of ratio distribution}), we obtain 
    \begin{align}
      \E_{t-1}\qty[\tnp_t(I_{t-1})] 
      &\geq \frac{\E_{t-1}\qty[\np_t(I_{t-1})]}{\E_{t-1}[\npo_t]}-\frac{\mathbf{Cov}_{t-1}\qty[\np_t(I_{t-1}),\npo_t]}{\E_{t-1}[\npo_t]^2} \nonumber\\
      &\geq \tnpm_{t-1}\qty(1+\frac{\npm_{t-1}-\npt_{t-1}/\npo_{t-1}}{2(1-\npo_{t-1})+\npt_{t-1}/\npo_{t-1}})-\frac{9\npm_{t-1}}{n}.
      \label{eq:expectation of tnpm lower bound key}
    \end{align}
    The last inequality holds from \cref{claim:covariance of alpha and beta} ($\mathbf{Cov}_{t-1}\qty[\np_t(i),\npo_t]\leq \alpha_{t-1}(i)/n$) and the assumptions of $\beta_{t-1}\geq 1/2-o(1)$, $\psi_{t-1}\leq o(1)$, 
    and $\E_{t-1}[\beta_t]=\beta_{t-1}-\psi_{t-1}\geq 1/2-o(1)\geq 1/3$.
  
    \medskip
    \noindent\textbf{(ii) Bernstein condition.}
    From \cref{claim:covariance of alpha and beta}, we have
    $\Cov_{t-1}\qty[\np_t(I_{t-1}),\npo_t]\leq \frac{\np_{t-1}(I_{t-1})}{n}=\frac{\npm_{t-1}}{n}$.
  From \cref{eq:expectation of alphatilde lower bound} and \cref{lem:basic inequalities for ratio distribution} (\cref{item:Bernstein condition of ratio distribution}), we have
  \begin{align*}
    &\tnpm_{t-1}\qty(1+\frac{\npm_{t-1}-\npt_{t-1}/\npo_{t-1}}{2(1-\npo_{t-1})+\npt_{t-1}/\npo_{t-1}})-\tnpm_t-\frac{9\npm_{t-1}}{n}\\
    &\leq \frac{\E_{t-1}\qty[\np_t(I_{t-1})]}{\E_{t-1}[\npo_t]}-\tnp_t(I_{t-1})-\frac{\Cov_{t-1}\qty[\np_t(I_{t-1}),\npo_t]}{\E_{t-1}[\npo_t]^2} && (\text{from \cref{eq:expectation of tnpm lower bound key}})\\
    &\leq  \frac{2\qty(\E_{t-1}\qty[\np_t(I_{t-1})]-\np_t(I_{t-1}))}{\E_{t-1}[\npo_t]}+\frac{\np_t(I_{t-1})\npo_t-\E_{t-1}\qty[\np_t(I_{t-1})\npo_t]}{\E_{t-1}[\npo_t]^2}. && (\text{from \cref{lem:basic inequalities for ratio distribution} (\cref{item:Bernstein condition of ratio distribution})})
  \end{align*}
  Note that we use $\tnp_t(I_{t-1})\leq \tnpm_t$ in the second inequality.
  From \cref{item:bernstein condition of alpha} of \cref{lem:basic inequalities for alpha} and \cref{item:BC for linear transformation} of \cref{lem:Bernstein condition}, the random variable $\frac{2\qty(\E_{t-1}\qty[\np_t(I_{t-1})]-\np_t(I_{t-1}))}{\E_{t-1}[\npo_t]}$ conditioned on round $t-1$ satisfies $\qty(\frac{2}{n\E_{t-1}[\npo_t]}, \frac{4\npm_{t-1}}{n\E_{t-1}[\npo_t]^2})$-Bernstein condition.
  We now check the Bernstein condition of $\np_t(i)\npo_t-\E_{t-1}[\np_t(i)\npo_t]$.
  \begin{claim}
    \label{claim:bernstein condition of alpha beta}
    For any $t\geq 1$ and $i\in [k]$,
    $\np_t(i)\npo_t-\E_{t-1}[\np_t(i)\npo_t]$ conditioned on round $t-1$ satisfies $\qty(\frac{2}{n}, \frac{14\np_{t-1}(i)\npo_{t-1}}{n})$-Bernstein condition.
  \end{claim}
  \begin{proof}[Proof of \cref{claim:bernstein condition of alpha beta}]
    Write 
    \[
        X_t(v,u)=\frac{1}{n^2}\qty(\indicator_{\opn_t(v)=i}\indicator_{\opn_t(u)\neq \bot}-\E_{t-1}[\indicator_{\opn_t(v)=i}\indicator_{\opn_t(u)\neq \bot}])
    \]
    for convenience.
    From \cref{item:BC for bounded rv} of \cref{lem:Bernstein condition}, the random variable $X_t(v,u)$ conditioned on round $t-1$ satisfies $\qty(\frac{1}{n^2}, \Var_{t-1}[X_t(v,u)])$-Bernstein condition.
    Observe that $\np_t(i)\npo_t-\E_{t-1}[\np_t(i)\npo_t]=\sum_{v\in V}\sum_{u\in V}X_t(v,u)$.
    In other words, $\np_t(i)\npo_t-\E_{t-1}[\np_t(i)\npo_t]$ conditioned on $t-1$ round is the sum of $n^2$ random variables $(X_t(v,u))_{v,u\in V}$.
    Furthermore, $(X_t(v,u))_{v,u\in V}$ is a read-$2n-1$ family of independent random variables $(\opn_t(v))_{v\in V}$.
    Thus, from \cref{item:BC for read-k family} of \cref{lem:Bernstein condition}, $\np_t(i)\npo_t-\E_{t-1}[\np_t(i)\npo_t]$ satisfies $\qty(\frac{2n-1}{n^2}, (2n-1)\sum_{v,u\in V}\Var_{t-1}[X_t(v,u)])$-Bernstein condition.
    Now, we bound the summation $\sum_{v,u\in V}\Var_{t-1}[X_t(v,u)]$ as follows:
    Since
    \begin{align*}
        \sum_{v,u\in V}\Var_{t-1}[X_t(v,u)]
        &\leq \frac{1}{n^4}\sum_{v,u\in V}\E_{t-1}[\indicator_{\opn_t(v)=i}\indicator_{\opn_t(u)\neq \bot}]\\
        &= \frac{1}{n^2}\E_{t-1}[\np_t(i)\npo_t]\\
        &= \frac{1}{n^2}\qty(\E_{t-1}[\np_t(i)]\E_{t-1}[\npo_t]+\mathbf{Cov}_{t-1}\qty[\np_t(i),\npo_t])\\
        &\leq \frac{1}{n^2}\qty(6\np_{t-1}(i)\npo_{t-1}+\frac{\np_{t-1}(i)}{n})\\
        &\leq \frac{7\np_{t-1}(i)\npo_{t-1}}{n^2}
    \end{align*}
    holds from \cref{lem:basic inequalities for alpha,lem:basic inequalities for beta,claim:covariance of alpha and beta}, we obtain the claim.
  \end{proof}
  
  From \cref{claim:bernstein condition of alpha beta} and \cref{item:BC for linear transformation} of \cref{lem:Bernstein condition}, the random variable $\frac{\np_t(I_{t-1})\npo_t-\E_{t-1}\qty[\np_t(I_{t-1})\npo_t]}{\E_{t-1}[\npo_t]^2}$ conditioned on round $t-1$ satisfies $\qty(\frac{2}{n\E_{t-1}[\npo_t]^2}, \frac{14\npm_{t-1}\npo_{t-1}}{n\E_{t-1}[\npo_t]^4})$-Bernstein condition.
  Therefore, from \cref{item:BC for sum of rvs} of \cref{lem:Bernstein condition} and \cref{item:BC for upper bounded rv} of \cref{lem:Bernstein condition}, we obtain the claim.
  Note that, from the assumptions of $\beta_{t-1}\geq 1/2-o(1)$ and $\psi_{t-1}\leq o(1)$, we have $\E_{t-1}[\beta_t]=\beta_{t-1}-\psi_{t-1}\geq 1/2-o(1)\geq 1/3$.
  \end{proof}

\begin{proof}[Proof of \cref{lem:basic inequalities for tnpt}]
\medskip
\noindent\textbf{(i) Expectation.}
  From \cref{lem:basic inequalities for tnpt} (\cref{item:expectation of gammatilde}), we have
  \begin{align*}
      \frac{\E_{t-1}\qty[\npt_t]}{\E_{t-1}[\npo_t]^2}
      &\geq \tnpt_{t-1}+\frac{\beta_{t-1}(1-\beta_{t-1})^2}{n\E_{t-1}[\beta_t]^2}.
  \end{align*}

  Hence, from \cref{lem:basic inequalities for ratio distribution} (\cref{item:expectation of ratio distribution}), \cref{lem:basic inequalities for beta} (\cref{item:variance of beta}), \cref{lem:basic inequalities for gamma} (\cref{item:expectation of gamma}), and \cref{claim:covariance of gamma and beta2} (we will prove this later), we have
  \begin{align*}
      \E_{t-1}\qty[\tnpt_t]
      &\geq \frac{\E_{t-1}\qty[\npt_t]}{\E_{t-1}[\npo_t^2]} - \frac{\mathbf{Cov}_{t-1}\qty[\npt_t,\npo_t^2]}{\E_{t-1}[\npo_t^2]^2}\\
      &=\frac{\E_{t-1}\qty[\npt_t]}{\E_{t-1}[\npo_t]^2} - \frac{\E_{t-1}\qty[\npt_t]\Var_{t-1}[\npo_t]}{\E_{t-1}[\npo_{t}^2]\E_{t-1}[\npo_t]^2}- \frac{\mathbf{Cov}_{t-1}\qty[\npt_t,\npo_t^2]}{\E_{t-1}[\npo_t^2]^2}\\
      &\geq \tnpt_{t-1}+\frac{\beta_{t-1}(1-\beta_{t-1})^2}{n\E_{t-1}[\beta_t]^2}
      - \frac{\E_{t-1}[\gamma_t]\npo_{t-1}}{n\E_{t-1}[\npo_t]^4}
      - \frac{4\E_{t-1}[\gamma_t]}{n\E_{t-1}[\npo_t]^4} && (\text{\cref{claim:covariance of gamma and beta2}})\\
      &\geq \tnpt_{t-1}+ \frac{1}{n\E_{t-1}[\npo_t]^2}\qty(\npo_{t-1}(1-\npo_{t-1})^2-\frac{50\npt_{t-1}}{\E_{t-1}[\npo_t]^2}) && (\text{\cref{lem:basic inequalities for gamma} (\cref{item:expectation of gamma}}))\\
      &=\tnpt_{t-1} + \frac{1}{n\E_{t-1}[\npo_t]^2}\qty(\frac{\npo_{t-1}}{4}\qty(1-\frac{\psi_{t-1}+\gamma_{t-1}}{\beta_{t-1}})^2-\frac{50\npt_{t-1}}{(\beta_{t-1}-\psi_{t-1})^2}).
  \end{align*}
Note that $1-\beta_{t-1}=\frac{1}{2}-\frac{\psi_{t-1}+\gamma_{t-1}}{2\beta_{t-1}}$holds from the definition of $\psi_t=\beta_t(2\beta_t-1)-\npt_t$.
From the assumption that $\psi_{t-1},\gamma_{t-1}\leq o(1)$, and $\beta_{t-1}\geq 1/2-o(1)$, we obtain the claim.
  
  \begin{claim}
  \label{claim:covariance of gamma and beta2}
  For any $t\geq 1$, we have
  $\mathbf{Cov}_{t-1}\qty[\gamma_t,\beta_t^2]\leq \frac{4\E_{t-1}[\gamma_t]}{n}$.
  \end{claim}
  \begin{proof}
   We have
  \begin{align}
  &\gamma_t=\sum_{i\in [k]}\qty(\frac{1}{n}\sum_{v\in V}\indicator_{\opn_t(v)=i})^2
  =\frac{1}{n^2}\sum_{v\in V}\sum_{u\in V}\indicator_{\opn_t(v)=\opn_t(u)}\indicator_{\opn_t(v)\neq \bot}, \\
  &\beta_t^2=\qty(\frac{1}{n}\sum_{v\in V}\indicator_{\opn_t(v)=\bot})^2=\frac{1}{n^2}\sum_{v\in V}\sum_{u\in V}\indicator_{\opn_t(v)\neq \bot}\indicator_{\opn_t(u)\neq \bot}.
  \end{align}
  Hence, 
  \begin{align*}
  \Cov_{t-1}\qty[\gamma_t,\beta_t^2]
  &=\frac{1}{n^4}\sum_{v,u,a,b\in V}\Cov_{t-1}\qty[\indicator_{\opn_t(v)=\opn_t(u)}\indicator_{\opn_t(v)\neq \bot},\indicator_{\opn_t(a)\neq \bot}\indicator_{\opn_t(b)\neq \bot}]\\
  &=\frac{1}{n^4}\sum_{v,u,a,b\in V:\{v,u\}\cap\{a,b\}\neq \emptyset}\Cov_{t-1}\qty[\indicator_{\opn_t(v)=\opn_t(u)}\indicator_{\opn_t(v)\neq \bot},\indicator_{\opn_t(a)\neq \bot}\indicator_{\opn_t(b)\neq \bot}]\\
  &\leq \frac{1}{n^4}\sum_{v,u,a,b\in V:\{v,u\}\cap\{a,b\}\neq \emptyset}\E_{t-1}\qty[\indicator_{\opn_t(v)=\opn_t(u)}\indicator_{\opn_t(v)\neq \bot}\indicator_{\opn_t(a)\neq \bot}\indicator_{\opn_t(b)\neq \bot}]\\
  &\leq \frac{1}{n^4}\sum_{v,u\in V}\sum_{a,b\in V:\{v,u\}\cap\{a,b\}\neq \emptyset}\E_{t-1}\qty[\indicator_{\opn_t(v)=\opn_t(u)}\indicator_{\opn_t(v)\neq \bot}]\\
  &\leq \frac{4n}{n^2}\E_{t-1}\qty[\gamma_t].
  \end{align*}
  \end{proof}
\end{proof}

\subsection{Population Protocol Model} \label{sec:proof of basic properties in the population protocol model}

\begin{proof}[Proof of \cref{lem:basic inequalities for alpha PP}]
\medskip
\noindent\textbf{(i) Expectation.}
First, observe that, for any $i\in [k]$ and $t\geq 1$, 
\begin{align}
  \alpha_t(i)-\alpha_{t-1}(i)
  &=\begin{cases}
    \frac{1}{n} & \text{with probability} \; (1-\beta_{t-1})\alpha_{t-1}(i),\\
    -\frac{1}{n} & \text{with probability} \; \alpha_{t-1}(i)(\beta_{t-1}-\alpha_{t-1}(i)), \\
    0 & \text{otherwise}.
  \end{cases}
  \label{eq:one step difference for alpha PP}
\end{align}
Hence, we have
\begin{align*}
\E_{t-1}[\alpha_t(i)] &= \alpha_{t-1}(i)
+ \frac{1}{n}(1-\beta_{t-1})\alpha_{t-1}(i) - \frac{1}{n}\alpha_{t-1}(i)(\beta_{t-1}-\alpha_{t-1}(i)) \\
&=\alpha_{t-1}(i)\qty(1+\frac{\alpha_{t-1}(i)+1-2\beta_{t-1}}{n}).
\end{align*}

\medskip
\noindent\textbf{(ii) Variance.}
  From $\Var_{t-1}[\alpha_t(i)]=\Var_{t-1}\qty[\alpha_{t}(i)-\alpha_{t-1}(i)]$ and \cref{eq:one step difference for alpha PP}, we have
  \begin{align*}
    \Var_{t-1}[\alpha_t(i)]
    &=\E_{t-1}\qty[\qty(\alpha_{t}(i)-\alpha_{t-1}(i))^2]-\E_{t-1}\qty[\alpha_{t}(i)-\alpha_{t-1}(i)]^2\\
    &=\frac{\alpha_{t}(i)\qty(1-\alpha_{t-1}(i))}{n^2}-\frac{\alpha_{t-1}(i)^2}{n^2}\qty(1-2\beta_{t-1}+\alpha_{t-1}(i))^2.
  \end{align*}

\medskip
\noindent\textbf{(iii) Bernstein condition.}
  First, $\alpha_t(i)-\E_{t-1}[\alpha_t(i)]$ satisfies $\qty(\frac{1}{n},\Var_{t-1}[\alpha_t(i)])$-Bernstein condition by \cref{lem:Bernstein condition} (\cref{item:BC for bounded rv}).
  Further, we have $\Var_{t-1}[\alpha_t(i)]\leq \frac{\alpha_{t-1}(i)}{n}$ from \cref{item:variance of alpha} of \cref{lem:basic inequalities for alpha}.
  Hence, from \cref{lem:Bernstein condition} (\cref{item:BC for upper bounded rv}), we obtain the claim.
\end{proof}

\begin{proof}[Proof of \cref{lem:basic inequalities for beta PP}]
\medskip
\noindent\textbf{(i) Expectation.} By definition,
\begin{align}
  \beta_t - \beta_{t-1} = \begin{cases}
      \frac{1}{n}	& \text{with probability } \; (1-\beta_{t-1})\beta_{t-1},\\
      -\frac{1}{n}	& \text{with probability } \; \sum_{i\in [k]}\alpha_{t-1}(i)\qty(\beta_{t-1}-\alpha_{t-1}(i))=\beta_{t-1}^2-\gamma_{t-1},\\
      0    & \text{otherwise}.
  \end{cases} \label{eq:one step difference for beta}
\end{align}
Thus, we obtain
\begin{align*}
\E_{t-1}\qty[\beta_t]
&=\beta_{t-1}+\frac{(1-\beta_{t-1})\beta_{t-1}}{n}-\frac{\beta_{t-1}^2-\gamma_{t-1}}{n}\\
&=\beta_{t-1}+\frac{\beta_{t-1}\qty(1-2\beta_{t-1})+\gamma_{t-1}}{n}.
\end{align*}

\medskip
\noindent\textbf{(ii) Variance.}
From $\Var_{t-1}[\beta_t]=\Var_{t-1}\qty[\beta_{t}-\beta_{t-1}]$ and \cref{eq:one step difference for beta}, we have
\begin{align}
  \Var_{t-1}[\beta_t]
  &\leq \E_{t-1}\qty[\qty(\beta_{t}-\beta_{t-1})^2]
  =\frac{(1-\beta_{t-1})\beta_{t-1}+\beta_{t-1}^2-\gamma_{t-1}}{n^2}
  =\frac{\beta_{t-1}-\gamma_{t-1}}{n^2}.
  \label{eq:squared difference of beta}
\end{align}

\medskip
\noindent\textbf{(iii) Bernstein condition.}
  From \cref{lem:Bernstein condition} (\cref{item:BC for bounded rv}), 
  $\beta_t-\E_{t-1}[\beta_t]$ satisfies $\qty(\frac{1}{n},\Var_{t-1}[\beta_t])$-Bernstein condition.
  Hence, from \cref{lem:Bernstein condition} (\cref{item:BC for upper bounded rv}) and \cref{lem:basic inequalities for beta PP} (\cref{item:variance of beta PP}), we obtain the claim.

\end{proof}

\begin{proof}[Proof of \cref{lem:basic inequalities for delta PP}]
  \noindent\textbf{(i) Expectation.}
  From \cref{lem:basic inequalities for alpha PP} (\cref{item:expectation of alpha PP}), we have
  \begin{align*}
    \E_{t-1}\qty[\delta_t^{(\epsilon)}]
    &=\delta_{t-1}^{(\epsilon)}\qty(1+\frac{\alpha_{t-1}(i)+\alpha_{t-1}(j)+1-2\beta_{t-1}}{n})+\frac{\epsilon\alpha_{t-1}(i)\alpha_{t-1}(j)}{n}.
  \end{align*}
  \noindent\textbf{(ii) Variance.}
  Form \cref{lem:basic inequalities for delta PP} (\cref{item:expectation of delta PP}), we have
  \begin{align*}
    \Var_{t-1}\qty[\delta_t]
    &=\E_{t-1}\qty[\qty(\delta_t-\delta_{t-1})^2]-\E_{t-1}\qty[\delta_t-\delta_{t-1}]^2\\
    &=\E_{t-1}\qty[\qty(\alpha_t(i)-\alpha_{t-1}(i))^2]+\E_{t-1}\qty[\qty(\alpha_t(j)-\alpha_{t-1}(j))^2]-\E_{t-1}\qty[\delta_t-\delta_{t-1}]^2\\
    &\frac{\alpha_{t-1}(i)\qty(1-\alpha_{t-1}(i))}{n^2}
    +\frac{\alpha_{t-1}(j)\qty(1-\alpha_{t-1}(j))}{n^2}
    -\frac{\delta_{t-1}^2}{n^2}\qty(1-2\beta_{t-1}+\alpha_{t-1}(i)+\alpha_{t-1}(j))^2\\
    &\geq \frac{\alpha_{t-1}(i)\qty(1-\alpha_{t-1}(i))}{n^2}
    +\frac{\alpha_{t-1}(j)\qty(1-\alpha_{t-1}(j))}{n^2}
    -\frac{4\delta_{t-1}^2}{n^2}.
  \end{align*}
  \noindent\textbf{(iii) Bernstein condition.}
  Since
  \begin{align*}
    \delta_t^{(\epsilon)}-\E_{t-1}\qty[\delta_t^{(\epsilon)}]
    =\qty(\alpha_t(i)-\E_{t-1}\qty[\alpha_t(i)])-(1+\varepsilon)\qty(\alpha_t(j)-\E_{t-1}\qty[\alpha_t(j)]),
  \end{align*}
  applying \cref{lem:Bernstein condition 2} (\cref{item:BC for sum of rvs}) and \cref{lem:basic inequalities for alpha PP} (\cref{item:bernstein condition of alpha PP}), 
  $\delta_t^{(\epsilon)}-\E_{t-1}\qty[\delta_t^{(\epsilon)}]$ satisfies $\qty(\frac{2(1+\epsilon)}{n},s)$-Bernstein condition for $s=\frac{2}{n^2}\qty(\alpha_{t-1}(i)+(1+\epsilon)^2\alpha_{t-1}(j))$.
\end{proof}

\begin{proof}[Proof of \cref{lem:basic inequalities for gamma PP}]
Suppose that $v\in V$ chooses $u\in V$ at time $t-1$.
Then, 
\begin{align}
\gamma_t-\gamma_{t-1}
&=\begin{cases}
  -\frac{2\alpha_{t-1}(i)}{n}+\frac{1}{n^2} & \text{if $\opn_{t-1}(v)=i$ and $\opn_{t-1}(u)=j$ for $i\neq j$},\\
  \frac{2\alpha_{t-1}(i)}{n}+\frac{1}{n^2} & \text{if $\opn_{t-1}(v)=\bot$ and $\opn_{t-1}(u)=i$},\\
  0 & \text{otherwise}.
\end{cases}
\label{eq:one step gamma difference for PP}
\end{align}
Indeed, for the first case, we have
\begin{align*}
\gamma_t-\gamma_{t-1}
&=\alpha_t(i)^2-\alpha_{t-1}(i)^2
=\qty(\alpha_{t-1}(i)-\frac{1}{n})^2-\alpha_{t-1}(i)^2
=-\frac{2\alpha_{t-1}(i)}{n}+\frac{1}{n^2},
\end{align*}
and for the second case, we have
\begin{align*}
\gamma_t-\gamma_{t-1}
&=\alpha_t(i)^2-\alpha_{t-1}(i)^2
=\qty(\alpha_{t-1}(i)+\frac{1}{n})^2-\alpha_{t-1}(i)^2
=\frac{2\alpha_{t-1}(i)}{n}+\frac{1}{n^2}.
\end{align*}

Note that \cref{eq:one step gamma difference for PP} implies the following fact: For any $t\geq 1$, we have
\begin{align}
\abs{\gamma_t-\gamma_{t-1}}
\leq \frac{1}{n^2}+\frac{2\npm_{t-1}}{n} \leq \frac{3\npm_{t-1}}{n}.
\label{eq:one step gamma difference bound for PP}
\end{align}
Here, we use the fact that $\npm_{t-1}\leq \frac{1}{n}$ holds (if this does not hold, then $\gamma_t-\gamma_{t-1}=0$ and \cref{eq:one step gamma difference bound for PP} holds trivially).

\noindent\textbf{(i) Expectation.}
From \cref{eq:one step gamma difference for PP}, 
\begin{align*}
&\E_{t-1}\qty[\gamma_t-\gamma_{t-1}]\\
&=\sum_{i\in [k]}\sum_{j\in [k]\setminus \{i\}}\qty(\frac{1}{n^2}-\frac{2\alpha_{t-1}(i)}{n})\alpha_{t-1}(i)\alpha_{t-1}(j)
+\sum_{i\in [k]}\qty(\frac{1}{n^2}+\frac{2\alpha_{t-1}(i)}{n})(1-\beta_{t-1})\alpha_{t-1}(i)\\
&=\frac{2}{n}\sum_{i\in [k]}\alpha_{t-1}(i)^2\qty(1-2\beta_{t-1}+\alpha_{t-1}(i))
+\frac{1}{n^2}\sum_{i\in [k]}\alpha_{t-1}(i)\qty(1-\alpha_{t-1}(i))\\
&=\frac{2}{n}\qty((1-2\beta_{t-1})\gamma_{t-1}+\norm{\alpha_{t-1}}_3^3)+\frac{\beta_{t-1}-\gamma_{t-1}}{n^2}
\end{align*}
holds and we obtain the claim.

\medskip
\noindent\textbf{(ii) Variance.}
Note that $\Var_{t-1}[\gamma_t]=\Var_{t-1}[\gamma_t-\gamma_{t-1}]\leq \E_{t-1}[(\gamma_t-\gamma_{t-1})^2]$.
From \cref{eq:one step gamma difference for PP}, we have
\begin{align*}
&\E_{t-1}\qty[\qty(\gamma_t-\gamma_{t-1})^2]\\
&=\sum_{i\in [k]}\sum_{j\in [k]\setminus \{i\}}\qty(\frac{1}{n^2}-\frac{2\alpha_{t-1}(i)}{n})^2\alpha_{t-1}(i)\alpha_{t-1}(j)
+\sum_{i\in [k]}\qty(\frac{1}{n^2}+\frac{2\alpha_{t-1}(i)}{n})^2(1-\beta_{t-1})\alpha_{t-1}(i)\\
&=\frac{1}{n^4}\sum_{i\in [k]}\alpha_{t-1}(i)\qty(1-\alpha_{t-1}(i))
+\frac{4}{n^2}\sum_{i\in [k]}\alpha_{t-1}(i)^3\qty(1-\alpha_{t-1}(i))\\
&\hspace{1em}+\frac{4}{n^3}\sum_{i\in [k]}\alpha_{t-1}(i)^2\qty(1-2\beta_{t-1}+\alpha_{t-1}(i))\\
&=\frac{\beta_{t-1}-\gamma_{t-1}}{n^4}
+\frac{4}{n^2}\qty(\norm{\alpha_{t-1}}_3^3-\norm{\alpha_{t-1}}_4^4)
+\frac{4}{n^3}\qty((1-2\beta_{t-1})\gamma_{t-1}+\norm{\alpha_{t-1}}_3^3).
\end{align*}
Since
\begin{align*}
  \sum_{i\in [k]}\alpha_{t-1}(i)^{h+\ell}
  =\sum_{i\in [k]:\alpha_{t-1}(i)>0}\alpha_{t-1}(i)^{h+\ell}
  \geq \frac{1}{n^\ell}\sum_{i\in [k]:\alpha_{t-1}(i)>0}\alpha_{t-1}(i)^h
  =\frac{1}{n^\ell}\sum_{i\in [k]}\alpha_{t-1}(i)^h
\end{align*}
holds for any $h,\ell\geq 1$, 
we have 
$\norm{\alpha_{t-1}}_3^3\leq n\norm{\alpha_{t-1}}_4^4$,
$\gamma_{t-1}\leq n\norm{\alpha_{t-1}}_3^3$, and
$\beta_{t-1}\leq n^2\norm{\alpha_{t-1}}_3^3$.
Hence, we obtain the claim.

\medskip
\noindent\textbf{(iii) Bernstein condition.}
From \cref{eq:one step gamma difference bound for PP}, we have
\begin{align}
  \abs{\gamma_t-\E_{t-1}[\gamma_t]}
  =\abs{\gamma_t-\gamma_{t-1}-\E_{t-1}[\gamma_t-\gamma_{t-1}]}
  \leq 2\abs{\gamma_t-\gamma_{t-1}}
  \leq \frac{6\npm_{t-1}}{n}.
  \label{eq:one step gamma difference bound for expectation of gamma PP}
\end{align}
Hence, from \cref{lem:Bernstein condition} (\cref{item:BC for bounded rv}), 
$\gamma_t-\E_{t-1}[\gamma_t]$ satisfies $\qty(\frac{6\npm_{t-1}}{n},\Var_{t-1}[\gamma_t])$-Bernstein condition.
Combining this with \cref{lem:Bernstein condition} (\cref{item:BC for upper bounded rv}) and \cref{lem:basic inequalities for gamma PP} (\cref{item:variance of gamma PP}), we obtain the claim.
\end{proof}

\begin{proof}[Proof of \cref{lem:basic inequalities for psi PP}]
  \medskip
  \noindent\textbf{(i) Expectation.}
  First, from \cref{lem:basic inequalities for gamma PP} (\cref{item:expectation of gamma PP}), we have
  \begin{align*}
    \beta_{t-1}\E_{t-1}[\gamma_t-\gamma_{t-1}]
    &=\frac{2}{n}\qty(\beta_{t-1}(1-2\beta_{t-1})\gamma_{t-1}+\beta_{t-1}\norm{\alpha_{t-1}}_3^3)+\beta_{t-1}\frac{\beta_{t-1}-\gamma_{t-1}}{n^2}\\
    &\geq \frac{2}{n}\qty(\beta_{t-1}(1-2\beta_{t-1})\gamma_{t-1}+\gamma_{t-1}^2)+\beta_{t-1}\frac{\beta_{t-1}-\gamma_{t-1}}{n^2}\\
    &=-\frac{2\gamma_{t-1}}{n}\psi_{t-1}+\beta_{t-1}\frac{\beta_{t-1}-\gamma_{t-1}}{n^2}.
  \end{align*}
  Note that $\gamma_{t-1}^2\leq \beta_{t-1}\norm{\alpha_{t-1}}_3^3$ holds from the Cauchy-Schwarz inequality.
  Hence, from \cref{eq:one step difference for beta,eq:squared difference of beta},
  \begin{align*}
  &\beta_{t-1}\E_{t-1}\qty[\psi_t-\psi_{t-1}]\\
  &=2\beta_{t-1}\E_{t-1}\qty[\beta_t^2-\beta_{t-1}^2]-\beta_{t-1}\E_{t-1}[\beta_t-\beta_{t-1}]-\beta_{t-1}\E_{t-1}[\gamma_t-\gamma_{t-1}]\\
  &=2\beta_{t-1}\E_{t-1}\qty[\qty(\beta_t-\beta_{t-1})^2]+\beta_{t-1}(4\beta_{t-1}-1)\E_{t-1}[\beta_t-\beta_{t-1}]-\beta_{t-1}\E_{t-1}[\gamma_t-\gamma_{t-1}]\\
  &\leq 2\beta_{t-1}\frac{\beta_{t-1}-\gamma_{t-1}}{n^2}-\beta_{t-1}(4\beta_{t-1}-1)\frac{\psi_{t-1}}{n}
  +\frac{2\gamma_{t-1}}{n}\psi_{t-1}-\beta_{t-1}\frac{\beta_{t-1}-\gamma_{t-1}}{n^2}\\
  &=\frac{\psi_{t-1}}{n}\qty(-2\psi_{t-1}-\beta_{t-1})+\beta_{t-1}\frac{\beta_{t-1}-\gamma_{t-1}}{n^2}\\
  &\leq -\frac{\psi_{t-1}}{n}\beta_{t-1}+\beta_{t-1}\frac{\beta_{t-1}-\gamma_{t-1}}{n^2}
  \end{align*}
  holds and we obtain the claim.

  \medskip
  \noindent\textbf{(ii) Bernstein condition.}
  From definition, 
  $\psi_t-\E_{t-1}[\psi_t]=2\qty(\beta_t^2-\E_{t-1}[\beta_t^2])-\qty(\beta_t-\E_{t-1}[\beta_t])-\qty(\gamma_t-\E_{t-1}[\gamma_t]).$
  In what follows, we show that $\beta_t^2-\E_{t-1}[\beta_t^2]$ satisfies $\qty(\frac{16}{n},\frac{36\beta_{t-1}^4}{n^2})$-Bernstein condition.
  If this holds, then by applying \cref{lem:Bernstein condition} (\cref{item:BC for upper bounded rv}), together with \cref{lem:basic inequalities for beta PP} (\cref{item:bernstein condition of beta PP}) and \cref{lem:basic inequalities for gamma PP} (\cref{item:bernstein condition of gamma PP}), the claim follows.

First, we have
  \begin{align*}
    \abs{\beta_t^2-\E_{t-1}[\beta_t^2]}
    =\abs{\beta_t^2-\beta_{t-1}^2-\E_{t-1}[\beta_t^2-\beta_{t-1}^2]}
    \leq 2\abs{\beta_t^2-\beta_{t-1}^2}
    \leq 4\abs{\beta_t-\beta_{t-1}}
    \leq\frac{4}{n}.
  \end{align*}

  Now, we show that $\Var_{t-1}[\beta_t^2]\leq \frac{9\beta_{t-1}^4}{n^2}$.
  If $\beta_{t-1}=0$, then $\beta_t^2-\beta_{t-1}^2=0$ and the inequality holds.
  Suppose $v\in V$ chooses $u\in V$ at time $t-1$.
  If $\opn_{t-1}(v)=i$ and $\opn_{t-1}(u)=j$ for $i\neq j$, we have
  \begin{align}
    \beta_t^2-\beta_{t-1}^2
    =-\frac{2\beta_{t-1}}{n}+\frac{1}{n^2}.
    \label{eq:one step difference of squared beta 1}
  \end{align}
  If $\opn_{t-1}(v)=\bot$ and $\opn_{t-1}(u)=i$, we have
  \begin{align}
    \beta_t^2-\beta_{t-1}^2
    =\frac{2\beta_{t-1}}{n}+\frac{1}{n^2}.
    \label{eq:one step difference of squared beta 2}
  \end{align}
  Otherwise, we have $\beta_t^2-\beta_{t-1}^2=0$.
  Thus, we obtain
  \begin{align*}
    &\E_{t-1}\qty[\qty(\beta_t^2-\beta_{t-1}^2)^2]\\
    &=\sum_{i\in [k]}\sum_{j\in [k]\setminus \{i\}}\qty(\frac{1}{n^2}-\frac{2\beta_{t-1}}{n})^2\alpha_{t-1}(i)\alpha_{t-1}(j)
    +\sum_{i\in [k]}\qty(\frac{1}{n^2}+\frac{2\beta_{t-1}}{n})^2(1-\beta_{t-1})\alpha_{t-1}(i)\\
    &=\qty(\frac{1}{n^2}-\frac{2\beta_{t-1}}{n})^2\qty(\beta_{t-1}^2-\gamma_{t-1})
    +\qty(\frac{1}{n^2}+\frac{2\beta_{t-1}}{n})^2\beta_{t-1}\qty(1-\beta_{t-1})\\
    &=\qty(\frac{1}{n^4}+\frac{4\beta_{t-1}^2}{n^2})\qty(\beta_{t-1}^2-\gamma_{t-1})
    +\frac{4\beta_{t-1}^2}{n^3}\qty(\beta_{t-1}(1-2\beta_{t-1})+\gamma_{t-1})\\
    &\leq \frac{9\beta_{t-1}^4}{n^2}.
  \end{align*}
  Note that $\Var_{t-1}[\beta_t^2]=\Var_{t-1}[\beta_t^2-\beta_{t-1}^2]\leq \E_{t-1}[\qty(\beta_t^2-\beta_{t-1}^2)^2]$.
  Hence, from \cref{lem:Bernstein condition} (\cref{item:BC for bounded rv}), we obtain the claim.
\end{proof}

\begin{proof}[Proof of \cref{lem:basic inequalities for tnpm PP}]
  \medskip
  \noindent\textbf{(i) Expectation.}
  From \cref{lem:basic inequalities for ratio distribution} (\cref{item:expectation of ratio distribution}), we have
    \begin{align*}
      \E_{t-1}\qty[\tnpm_t]
      \geq \E_{t-1}\qty[\tnp_t(I_{t-1})]
      \geq \frac{\E_{t-1}\qty[\alpha_t(I_{t-1})]}{\E_{t-1}[\beta_t]}-\frac{\mathbf{Cov}_{t-1}\qty[\alpha_t(I_{t-1}),\beta_t]}{\E_{t-1}[\beta_t]^2}.
    \end{align*}
    From \cref{lem:basic inequalities for alpha PP} (\cref{item:expectation of alpha PP}) and \cref{lem:basic inequalities for beta PP} (\cref{item:expectation of beta PP}), we have
    \begin{align}
      \frac{\E_{t-1}\qty[\alpha_t(I_{t-1})]}{\E_{t-1}[\beta_t]}
      &=\frac{\alpha_{t-1}(I_{t-1})}{\beta_{t-1}}\qty(1+\frac{\frac{\E_{t-1}[\alpha_t(I_{t-1})]}{\alpha_{t-1}(I_{t-1})}-\frac{\E_{t-1}[\beta_t]}{\beta_{t-1}}}{\frac{\E_{t-1}[\beta_t]}{\beta_{t-1}}})\nonumber\\
      &=\tnpm_{t-1}\qty(1+\frac{\npm_{t-1}-\gamma_{t-1}/\beta_{t-1}}{n\E_{t-1}[\beta_t]/\beta_{t-1}})\nonumber\\
      &\geq \tnpm_{t-1}\qty(1+\frac{\npm_{t-1}-\gamma_{t-1}/\beta_{t-1}}{2n}). \label{eq:expectation of tnpm ratio PP}
    \end{align}
    Note that $\gamma_{t-1}\leq \npm_{t-1}\beta_{t-1}$ holds.
    
    Now, we complete the proof by showing that $\mathbf{Cov}_{t-1}\qty[\alpha_t(I_{t-1}),\beta_t]\leq \frac{2\npm_{t-1}}{n^2}$.
    To see this, from \cref{lem:basic inequalities for alpha PP} (\cref{item:variance of alpha PP}) and \cref{lem:basic inequalities for beta PP} (\cref{item:expectation of beta PP}), we have
    \begin{align*}
    \Cov_{t-1}\qty[\alpha_{t}(i),\beta_{t}]
    &=\Cov_{t-1}\qty[\alpha_t(i)-\alpha_{t-1}(i),\beta_{t-1}-\beta_{t-1}] \nonumber\\
    &=\E_{t-1}\qty[\qty(\alpha_t(i)-\alpha_{t-1}(i))\qty(\beta_t-\beta_{t-1})]-\E_{t-1}\qty[\alpha_t(i)-\alpha_{t-1}(i)]\E_{t-1}\qty[\beta_t-\beta_{t-1}] \nonumber\\
    &=\E_{t-1}\qty[\qty(\alpha_t(i)-\alpha_{t-1}(i))^2]-\E_{t-1}\qty[\alpha_t(i)-\alpha_{t-1}(i)]\E_{t-1}\qty[\beta_t-\beta_{t-1}] \nonumber\\
    &=\frac{\alpha_{t-1}(i)}{n^2}\qty(1-\alpha_{t-1}(i)-\qty(1-2\beta_{t-1}+\alpha_{t-1}(i))\qty(\beta_{t-1}(1-2\beta_{t-1})+\gamma_{t-1}))) \nonumber\\
    &\leq \frac{2\alpha_{t-1}(i)}{n^2}. 
    \end{align*}
    Hence, 
    \begin{align}
      \frac{\mathbf{Cov}_{t-1}\qty[\alpha_t(I_{t-1}),\beta_t]}{\E_{t-1}[\beta_t]^2}
      \leq \frac{18\npm_{t-1}}{n^2}
      \label{eq:covariance of alpha beta PP}
    \end{align}
    and we obtain the claim.
    Note that we use assumptions of $\beta_{t-1}\geq 1/2-o(1)$, 
    and $\E_{t-1}[\beta_t]\geq \beta_{t-1}-1/n\geq 1/3$.
  
  \medskip
  \noindent\textbf{(ii) Bernstein condition.}
    From \cref{lem:basic inequalities for ratio distribution} (\cref{item:Bernstein condition of ratio distribution}), \cref{lem:basic inequalities for beta PP,eq:covariance of alpha beta PP,eq:expectation of tnpm ratio PP}, we have
    \begin{align*}
      &\tnpm_{t-1}\qty(1+\frac{\npm_{t-1}-\gamma_{t-1}/\beta_{t-1}}{2n})-\tnpm_t-\frac{18\npm_{t-1}}{n^2}\\
      &\leq \frac{\E_{t-1}\qty[\alpha_t(I_{t-1})]}{\E_{t-1}[\beta_t]}-\alpha_t(I_{t-1})-\frac{\Cov_{t-1}\qty[\alpha_t(I_{t-1}),\beta_t]}{\E_{t-1}[\beta_t]^2}\\
      &\leq \frac{2\qty(\E_{t-1}\qty[\alpha_t(I_{t-1})]-\alpha_{t-1}(I_{t-1}))}{\E_{t-1}[\beta_t]}-\frac{\alpha_t(I_{t-1})\beta_t-\E_{t-1}\qty[\alpha_t(I_{t-1})\beta_t]}{\E_{t-1}[\beta_t]^2}.
    \end{align*}
    From \cref{lem:basic inequalities for alpha PP} (\cref{item:bernstein condition of alpha PP}), the former term $\frac{2\qty(\E_{t-1}\qty[\alpha_t(I_{t-1})]-\alpha_{t-1}(I_{t-1}))}{\E_{t-1}[\beta_t]}$ satisfies $\qty(\frac{2}{n\E_{t-1}[\beta_t]},\frac{4\npm_{t-1}}{n^2\E_{t-1}[\beta_t]^2})$-Bernstein condition.
    
    Now, we show that $\frac{\alpha_t(I_{t-1})\beta_t-\E_{t-1}\qty[\alpha_t(I_{t-1})\beta_t]}{\E_{t-1}[\beta_t]^2}$ satisfies the $\qty(\frac{6\beta_{t-1}}{n\E_{t-1}[\beta_t]^2},\frac{20\npm_{t-1}\beta_{t-1}}{n^2\E_{t-1}[\beta_t]^4})$-Bernstein condition.
    Once this is established, applying \cref{lem:Bernstein condition 2} (\cref{item:BC for sum of rvs}) and \cref{lem:Bernstein condition} (\cref{item:BC for bounded rv}) with $\E_{t-1}[\beta_t]\geq \beta_{t-1}-1/n \geq 1/3$ completes the proof.

    First, 
    we have 
    \begin{align*}
      \abs{\alpha_t(i)\beta_t-\E_{t-1}\qty[\alpha_{t}(i)\beta_{t}]}
      &\leq 2\abs{\alpha_t(i)\beta_t-\alpha_{t-1}(i)\beta_{t-1}}\\
      &\leq 2\alpha_t(i)\abs{\beta_t(i)-\beta_{t-1}(i)}+2\beta_{t-1}\abs{\alpha_t(i)-\alpha_{t-1}(i)}\\
      &\leq \frac{6\beta_{t-1}}{n}.
    \end{align*}
Furthermore, we have 
\begin{align*}
  \Var_{t-1}\qty[\alpha_t(i)\beta_t]=\Var_{t-1}\qty[\alpha_t(i)\beta_t-\alpha_{t-1}(i)\beta_{t-1}]\leq \E_{t-1}\qty[\qty(\alpha_t(i)\beta_t - \alpha_{t-1}(i)\beta_{t-1})^2]
\end{align*}
and
\begin{align*}
\E_{t-1}\qty[\qty(\alpha_t(i)\beta_t - \alpha_{t-1}(i)\beta_{t-1})^2]
&=\qty(\frac{\alpha_{t-1}(i)}{n}+\frac{\beta_{t-1}}{n}+\frac{1}{n^2})^2 \qty(1-\beta_{t-1})\alpha_{t-1}(i)\\
&+\qty(\frac{\alpha_{t-1}(i)}{n})^2 (1-\beta_{t-1})\qty(\beta_{t-1}-\alpha_{t-1}(i))\\
&+\qty(-\frac{\alpha_{t-1}(i)}{n}-\frac{\beta_{t-1}}{n}+\frac{1}{n^2})^2 \alpha_{t-1}(i)\qty(\beta_{t-1}-\alpha_{t-1}(i))\\
&+\qty(-\frac{\alpha_{t-1}(i)}{n})^2 \sum_{j\in [k]\setminus \{i\}}\alpha_{t-1}(j)\qty(\beta_{t-1}-\alpha_{t-1}(j))\\
&\leq \frac{20\alpha_{t-1}(i)\beta_{t-1}}{n^2},
\end{align*}
and we obtain the claim from \cref{lem:Bernstein condition} (\cref{item:BC for bounded rv}).
  \end{proof}

\begin{proof}[Proof of \cref{lem:basic inequalities for tnpt PP}]
  Write $x=\frac{1}{n}\qty(1-2\beta_{t-1}+\frac{\gamma_{t-1}}{\beta_{t-1}})$ for convenience.
      Then, from \cref{lem:basic inequalities for beta PP} (\cref{item:expectation of beta PP}), we have
      \begin{align*}
      \E_{t-1}\qty[\beta_t]=\beta_{t-1}+\frac{\beta_{t-1}\qty(1-2\beta_{t-1})+\gamma_{t-1}}{n}=\beta_{t-1}(1+x),
      \end{align*}
      and from \cref{lem:basic inequalities for gamma PP} (\cref{item:expectation of gamma PP}), we have
      \begin{align*}
        \E_{t-1}\qty[\gamma_t]
        &\geq \gamma_{t-1}+\frac{2}{n}\qty((1-2\beta_{t-1})\gamma_{t-1}+\frac{\gamma_{t-1}^2}{\beta_{t-1}})+\frac{\beta_{t-1}-\gamma_{t-1}}{n^2}
        =\gamma_{t-1}\qty(1+2x)+\frac{\beta_{t-1}-\gamma_{t-1}}{n^2}.
      \end{align*}
      Note that $\norm{\alpha_{t-1}}_3^3\geq \beta_{t-1}\left(\sum_{i\in [k]}\frac{\alpha_{t-1}(i)}{\beta_{t-1}}\right)^2=\gamma_{t-1}^2/\beta_{t-1}$ holds from Jensen's inequality.
      Hence,
      \begin{align*}
        \frac{\E_{t-1}\qty[\gamma_t]}{\E_{t-1}\qty[\beta_t]^2}
        &\geq \frac{\gamma_{t-1}(1+2x)+\frac{\beta_{t-1}-\gamma_{t-1}}{n^2}}{\beta_{t-1}^2(1+x)^2}\\
        &=\tnpt_{t-1}+\frac{\frac{\beta_{t-1}-\gamma_{t-1}}{n^2}-\gamma_{t-1}x^2}{\beta_{t-1}^2(1+x)^2}\\
        &=\tnpt_{t-1}+\frac{\beta_{t-1}-\gamma_{t-1}-\gamma_{t-1}\qty(1-2\beta_{t-1}+\frac{\gamma_{t-1}}{\beta_{t-1}})^2}{n^2\E_{t-1}\qty[\beta_t]^2}.
      \end{align*}

      Now, we give an upper bound of $\Cov_{t-1}\qty[\gamma_t,\beta_t^2]$.
      From \cref{eq:one step difference of squared beta 1,eq:one step difference of squared beta 2,eq:one step gamma difference for PP}, we have
      \begin{align*}
      &\E_{t-1}\qty[\qty(\gamma_t-\gamma_{t-1})(\beta_t^2-\beta_{t-1}^2)]\\
      &=\sum_{i\in [k]}\sum_{j\in [k]\setminus \{i\}}\qty(\frac{1}{n^2}-\frac{2\alpha_{t-1}(i)}{n})\qty(\frac{1}{n^2}-\frac{2\beta_{t-1}}{n})\alpha_{t-1}(i)\alpha_{t-1}(j)\\
      &+\sum_{i\in [k]}\qty(\frac{1}{n^2}+\frac{2\alpha_{t-1}(i)}{n})\qty(\frac{1}{n^2}+\frac{2\beta_{t-1}}{n})(1-\beta_{t-1})\alpha_{t-1}(i)\\
      &=\frac{1}{n^3}\qty(\frac{1}{n}-2\beta_{t-1})\qty(\beta_{t-1}^2-\gamma_{t-1})
      -\frac{2}{n^2}\qty(\frac{1}{n}-2\beta_{t-1})\qty(\gamma_{t-1}\beta_{t-1}-\norm{\alpha_{t-1}}_3^3)\\
      &+\frac{1}{n^3}\qty(\frac{1}{n}+2\beta_{t-1})\beta_{t-1}\qty(1-\beta_{t-1})
      +\frac{2}{n^2}\qty(\frac{1}{n}+2\beta_{t-1})\gamma_{t-1}\qty(1-\beta_{t-1})\\
      &\leq \frac{24\gamma_{t-1}}{n^2}.
      \end{align*}
      Note that $\beta_{t-1}\geq 1/n$, $\gamma_{t-1}\geq \frac{\beta_{t-1}^2}{k}$, and $\norm{\alpha_{t-1}}_3^3\leq \gamma_{t-1}\npm_{t-1}$ hold.
    Furthermore, we have
      \begin{align*}
        &\E_{t-1}\qty[\beta_t^2-\beta_{t-1}^2]\\
        &=\sum_{i\in [k]}\sum_{j\in [k]\setminus \{i\}}\qty(\frac{1}{n^2}-\frac{2\beta_{t-1}}{n})\alpha_{t-1}(i)\alpha_{t-1}(j)
        +\sum_{i\in [k]}\qty(\frac{1}{n^2}+\frac{2\beta_{t-1}}{n})(1-\beta_{t-1})\alpha_{t-1}(i)\\
        &=\frac{1}{n}\qty(\frac{1}{n}-2\beta_{t-1})\qty(\beta_{t-1}^2-\gamma_{t-1})
        +\frac{1}{n}\qty(\frac{1}{n}+2\beta_{t-1})\beta_{t-1}\qty(1-\beta_{t-1})\\
        &\leq \frac{6\beta_{t-1}^2}{n}
        \end{align*}
        and $\E_{t-1}\qty[\gamma_t-\gamma_{t-1}]\leq \frac{5\gamma_{t-1}}{n\beta_{t-1}}$ holds from \cref{lem:basic inequalities for gamma PP} (\cref{item:expectation of gamma PP}).
        Hence, 
        \begin{align*}
          \Cov_{t-1}\qty[\gamma_t,\beta_t^2]
          &=\E_{t-1}\qty[\gamma_t-\gamma_{t-1}]\E_{t-1}\qty[\beta_t^2-\beta_{t-1}^2]-\E_{t-1}\qty[\gamma_t-\gamma_{t-1}]\E_{t-1}\qty[\beta_t^2-\beta_{t-1}^2]\\
          &\leq \frac{24\gamma_{t-1}}{n^2}+\frac{5\gamma_{t-1}}{n\beta_{t-1}}\cdot \frac{6\beta_{t-1}^2}{n}\\
          &\leq \frac{54\gamma_{t-1}}{n^2}.
        \end{align*}

      Combining the above inequalities, 
      we obtain
      \begin{align*}
        \E_{t-1}\qty[\tnpt_t]
      &\geq \frac{\E_{t-1}\qty[\npt_t]}{\E_{t-1}[\npo_t^2]} - \frac{\mathbf{Cov}_{t-1}\qty[\npt_t,\npo_t^2]}{\E_{t-1}[\npo_t^2]^2}\\
      &=\frac{\E_{t-1}\qty[\npt_t]}{\E_{t-1}[\npo_t]^2} - \frac{\E_{t-1}\qty[\npt_t]\Var_{t-1}[\npo_t]}{\E_{t-1}[\npo_{t}^2]\E_{t-1}[\npo_t]^2}- \frac{\mathbf{Cov}_{t-1}\qty[\npt_t,\npo_t^2]}{\E_{t-1}[\npo_t^2]^2}\\
        &\geq \tnpt_{t-1}+\frac{\beta_{t-1}}{n^2\E_{t-1}\qty[\beta_t]^2}-\frac{5\gamma_{t-1}}{n^2\E_{t-1}\qty[\beta_t]^2}-\frac{6\gamma_{t-1}\beta_{t-1}}{n^2\E_{t-1}\qty[\beta_t]^4}-\frac{54\gamma_{t-1}}{n^2\E_{t-1}\qty[\beta_t]^4}\\
        &\geq \tnpt_{t-1}+\frac{1}{6n^2}-\frac{80000\gamma_{t-1}}{n^2}\\
        &\geq \tnpt_{t-1}+\frac{1}{12n^2}.
      \end{align*}
      Here, we use assumptions $\beta_{t-1}\geq 1/3$ and $\gamma_{t-1}=o(1)$.
      Note that $\beta_t\geq \beta_{t-1}-\frac{1}{n}\geq \frac{\beta_{t-1}}{2}\geq \frac{1}{6}$.

\end{proof}

\section{First Round \texorpdfstring{$\beta_0=1$}{beta0=1}} \label{sec:first round beta0=1}
In this section, we prove that in the gossip USD, the number of decided vertices after the first round becomes $\Otilde(n\gamma_0)$ with high probability.

\begin{lemma} \label{lem:beta1_upper_bound}
  Consider \USD in the gossip model with $\beta_0=1$.
  Then, it holds with high probability that $ \gamma_0  - \sqrt{\gamma_0\log n/n} \le \beta_1 \le \gamma_0 \log n$.
\end{lemma}
\begin{proof}
  Since $\beta_0=1$, a vertex $u$ is still decided after the first synchronous update if and only if it communicates with a vertex holding the same opinion as $u$, which occurs with probability $\alpha_0(i)$ if $u$ has opinion $i$.
  Therefore, we have
  $$
    \E[\beta_1] = \sum_{i\in[k]} \alpha_0(i)\cdot \alpha_0(i) = \gamma_0.
  $$
  Moreover, the quantity $n\beta_1$ is the sum of $n$ independent random variables, each of which is bounded by $1$.
  Therefore, by a variant of the Chernoff bound (\cref{lem:chernoff-variant}) for $h=n\gamma_0\log n \ge 6\E[\beta_1]$ (for all sufficiently large $n$), we have
  $$
    \Pr\qty[ n\beta_1 \ge n\gamma_0\log n ] \le 2^{-n\gamma_0\log n} \le n^{-\Omega(1)}.
  $$
\end{proof}

From this lemma, we know that the number of remaining opinions becomes at most $\gamma_0 \log n$ with high probability after the first round.

\begin{corollary} \label{cor:k_upper_bound}
  Consider \USD in the gossip model with $\beta_0=1$.
  Let $k'$ be the number of opinions that remain after the first round.
  Then, with high probability, $k' \le \min\qty{ k, n\gamma_0\log n}$.
\end{corollary}
\begin{proof}
  From \cref{lem:beta1_upper_bound}, with high probability, we have $k'\le \min\{k,n\beta_1\} \le \min\{k,n\gamma_0\log n\}$.
\end{proof}


\end{document}